%% file: main.tex
\documentclass[twoside, openany, 12pt]{book}

\usepackage{verbatim}
\usepackage{amsthm}
\theoremstyle{definition}
\usepackage{appendix}
\usepackage{epsfig}
\usepackage{graphicx}
\usepackage{amsmath}
\usepackage{amssymb}
\usepackage{float}
\usepackage{url}
\newtheorem{definition}{Definition}[section]
\newtheorem{lemma}{Lemma}[section]
\newtheorem{proposition}{Proposition}[section]
\newtheorem{corollary}{Corollary}[section]
\newtheorem{theorem}{Theorem}[section]

\newtheorem{remark}{Remark}[section]
\newtheorem{example}{Example}[section]
\newtheorem{fact}{Fact}[section]

\numberwithin{equation}{section}

\newcommand{\dsum}{\displaystyle\sum}
\newcommand{\dint}{\displaystyle\int}
\newcommand{\dfr}{\displaystyle\frac}

\newcommand{\naturals}{\ensuremath{\mathbb{N}}}
\newcommand{\reals}{\ensuremath{\mathbb{R}}}

\newcommand{\pr}{\ensuremath{\mathbb{P}}}
\newcommand{\expectation}{\ensuremath{\mathbb{E}}}
\newcommand{\LL}{\ensuremath{\mathbb{L}}}

\newcommand{\integers}{\ensuremath{\mathbb{Z}}}

\def\cA{\ensuremath{\mathcal{A}}}
\def\cB{\ensuremath{\mathcal{B}}}
\def\cC{\ensuremath{\mathcal{C}}}
\def\cD{\ensuremath{\mathcal{D}}}
\def\cF{\ensuremath{\mathcal{F}}}
\def\cL{\ensuremath{\mathcal{L}}}
\def\cN{\ensuremath{\mathcal{N}}}
\def\cP{\ensuremath{\mathcal{P}}}
\def\cR{\ensuremath{\mathcal{R}}}
\def\cS{\ensuremath{\mathcal{S}}}
\def\cT{\ensuremath{\mathcal{T}}}
\def\cU{\ensuremath{\mathcal{U}}}
\def\cV{\ensuremath{\mathcal{V}}}
\def\cX{\ensuremath{\mathcal{X}}}
\def\cY{\ensuremath{\mathcal{Y}}}
\def\cZ{\ensuremath{\mathcal{Z}}}

\def\sF{\ensuremath{\mathsf{F}}}
\def\sP{\ensuremath{\mathsf{P}}}

\def\d{\ensuremath{\mathrm{d}}}

\def\eps{\ensuremath{\varepsilon}}

\def\deq{\triangleq}
\def\eqdist{\ensuremath{\stackrel{{\rm d}}{=}}}

\def\LSI{\ensuremath{\mathrm{LSI}}}
\def\Ent{\ensuremath{\mathrm{Ent}}}
\def\mse{\ensuremath{\mathsf{mse}}}
\def\lmmse{\ensuremath{\mathsf{lmmse}}}
\def\mmse{\ensuremath{\mathsf{mmse}}}
\def\var{\ensuremath{\mathsf{var}}}
\def\OU{\ensuremath{\mathrm{OU}}}
\def\Bernoulli{\ensuremath{\mathrm{Bernoulli}}}
\def\Poisson{\ensuremath{\mathrm{Poisson}}}
\def\CP{\ensuremath{\mathsf{CP}}}
\def\sphere{\ensuremath{\mathbb{S}}}
\def\emp{\ensuremath{\mathsf{P}}}

\usepackage[margin=0.8in]{geometry}

\begin{document}

\begin{center}
\vspace*{1cm} {\Huge{\bf{Concentration of Measure Inequalities\\[0.1cm]
in Information Theory, Communications\\[0.2cm]
and Coding}}}

\vspace*{2.5cm} MONOGRAPH\\[0.2cm]
Last updated: September 29, 2014.

\vspace*{2cm} {\LARGE{Foundations and Trends\\
in Communications and Information Theory,\\[0.2cm]
{\em Second Edition}, 2014.}}

\vspace*{2.0cm} {\bf{Maxim Raginsky}} \\[0.2cm]
Department of Electrical and Computer Engineering,\\
Coordinated Science Laboratory,\\
University of Illinois at Urbana-Champaign,\\
Urbana, IL 61801, USA.\\
E-mail: maxim@illinois.edu\\[0.5cm]
and \\[0.5cm]
{\bf{Igal Sason}}\\[0.2cm]
Department of Electrical Engineering,\\
Technion -- Israel Institute of Technology,\\
Haifa 32000, Israel.\\
E-mail: sason@ee.technion.ac.il

\end{center}

\newpage
\begin{center}
{\huge{\bf{Abstract}}}
\end{center}
\vspace*{1cm}

During the last two decades, concentration inequalities have been the subject of exciting
developments in various areas, including convex geometry,
functional analysis, statistical physics, high-dimensional statistics, pure and applied
probability theory (e.g., concentration of measure phenomena in random graphs, random
matrices, and percolation), information theory, theoretical computer science, and learning
theory. This monograph focuses on some of the key modern mathematical
tools that are used for the derivation of concentration inequalities,
on their links to information theory, and on their various applications
to communications and coding. In addition to being a survey, this monograph
also includes various new recent results derived by the authors.

The first part of the monograph introduces classical concentration inequalities
for martingales, as well as some recent refinements and extensions. The power
and versatility of the martingale approach is exemplified in the context of codes defined
on graphs and iterative decoding algorithms, as well as codes for wireless communication.

The second part of the monograph introduces the entropy method, an information-theoretic technique for
deriving concentration inequalities. The basic ingredients of the entropy method are discussed first
in the context of logarithmic Sobolev inequalities, which underlie the so-called functional approach to
concentration of measure, and then from a complementary information-theoretic viewpoint based on
transportation-cost inequalities and probability in metric spaces. Some
representative results on concentration for dependent random variables are
briefly summarized, with emphasis on their connections to the entropy method.
Finally, we discuss several applications of the entropy method to problems in
communications and coding, including strong converses, empirical distributions
of good channel codes, and an information-theoretic converse for concentration
of measure.

\newpage
\begin{center}
{\huge{\bf{Acknowledgments}}}
\end{center}
\vspace*{1cm}

It is a pleasure to thank several individuals, who have carefully read parts of the manuscript in various stages and provided constructive comments, suggestions, and corrections. These include Ronen Eshel, Peter Harremo\"es, Eran Hof, Nicholas Kalouptsidis, Leor Kehaty, Aryeh Kontorovich, Ioannis Kontoyannis, Mokshay Madiman, Daniel Paulin, Yury Polyanskiy, Boaz Shuval, Emre Telatar, Tim van Erven, Sergio Verd\'u, Yihong Wu and Kostis Xenoulis. Among these people, Leor Kehaty is gratefully acknowledged for a very detailed report on the initial draft of this manuscript, and Boaz Shuval is acknowledged for some helpful comments on the first edition. The authors are thankful to the three anonymous reviewers and the Editor in Chief, Sergio Verd\'u, for very constructive and detailed suggestions, which contributed a lot to the presentation of the first edition of this manuscript. The authors accept full responsibility for any remaining omissions or errors.

The work of M.~Raginsky was supported in part by the U.S.\ National Science Foundation (NSF) under CAREER award no.\ CCF--1254041. The work of I.~Sason was supported by the Israeli Science Foundation (ISF), grant number 12/12. The hospitality of the Bernoulli inter-faculty center at EPFL, the Swiss Federal Institute of Technology in Lausanne, during the summer of 2011 is acknowledged by I.~Sason. We would like to thank the organizers of the {\em Information Theory and Applications Workshop} in San-Diego, California; our collaboration in this project was initiated during this successful workshop in Feb.~2012. Finally, we are grateful to the publishers of the {\em Foundations and Trends (FnT) in Communications and Information Theory}: Mike Casey, James Finlay and Alet Heezemans for their assistance in both the first and second editions of this monograph
(dated: Oct.~2013 and Sept.~2014, respectively).

\tableofcontents
\setcounter{section}{1}
\include{chapter1}

\setcounter{section}{2}
\include{chapter2}
\setcounter{section}{3}
\include{chapter3}

\bibliography{references}
\end{document}

%% file: chapter1.tex
\chapter{Introduction}

\section{An overview and a brief history}
Concentration-of-measure inequalities provide bounds on the
probability that a random variable $X$ deviates from its mean,
median or other typical value $\overline{x}$ by a given amount. These
inequalities have been studied for several decades, with some
fundamental and substantial contributions during
the last two decades. Very roughly speaking, the concentration of
measure phenomenon can be stated in the following simple way:
``A random variable that depends in a smooth way on many independent
random variables (but not too much on any of them) is essentially
constant'' \cite{Talagrand96}. The exact meaning of such a statement
clearly needs to be clarified rigorously, but it often means that
such a random variable $X$ concentrates around $\overline{x}$ in a
way that the probability of the event $\{|X-\overline{x}| \geq t\}$,
for a given $t > 0$, decays exponentially in $t$. Detailed treatments
of the concentration of measure phenomenon, including historical
accounts, can be found, e.g., in \cite{Boucheron_Lugosi_Massart_book,
Ledoux, Lugosi_Lecture_Notes, Massart_book, McDiarmid_tutorial, Talagrand95}.

In recent years, concentration inequalities have been intensively studied
and used as a powerful tool in various areas. These include convex geometry,
functional analysis, statistical physics, dynamical systems, probability
(random matrices, Markov processes, random graphs,
percolation etc.), statistics, information theory, coding theory, learning theory,
and theoretical computer science. Several techniques have been developed so
far to prove concentration of measure inequalities. These include:
\begin{itemize}
\item The martingale approach (see, e.g.,
\cite{McDiarmid_tutorial,Azuma,Hoeffding}, \cite[Chapter~7]{AlonS_tpm3},
\cite{Chung_LU2006,survey2006}),
and its information-theoretic applications
(see, e.g., \cite{RiU_book} and references therein, \cite{SeldinLCTA_IT2012}).
This methodology will be covered in Chapter~\ref{chapter: martingales},
which is focused on concentration inequalities for discrete-time martingales
with bounded differences, as well as on some of their potential applications in
information theory, coding and communications.
A recent interesting avenue that follows from the martingale-based
concentration inequalities which are introduced in Chapter~\ref{chapter: martingales}
refers to their generalization to random matrices (see, e.g.,
\cite{Tropp_FoCM_2011, Tropp_ECP_2011}).
\item The entropy method and logarithmic Sobolev inequalities
(see, e.g., \cite[Chapter~5]{Ledoux}, \cite{Lugosi_Lecture_Notes}
and references therein).
This methodology and its many remarkable links to information theory
will be considered in Chapter~\ref{chapter: entropy method}.
\item Transportation-cost inequalities that originated from
information theory (see, e.g., \cite[Chapter~6]{Ledoux},
\cite{Gozlan_Leonard}, and references therein).
This methodology, which is closely related to the entropy
method and log-Sobolev inequalities, will be
considered in Chapter~\ref{chapter: entropy method}.
\item Talagrand's inequalities for product
measures (see, e.g., \cite{Talagrand96},
\cite[Chapter~4]{McDiarmid_tutorial}, \cite{Talagrand95}
and \cite[Chapter~6]{Steele_book}) and their links to
information theory \cite{Dembo}. These inequalities
proved to be very useful in combinatorial applications (such as the
study of common and/or increasing subsequences), in statistical
physics, and in functional analysis. We do
not discuss Talagrand's inequalities in detail.
\item Stein's method (or the method of exchangeable pairs) was recently
used to prove concentration inequalities (see, e.g., \cite{Chatterjee_phd,
Chatterjee1, Chatterjee2, Ross2011, Ghosh_Goldstein11a, Ghosh_Goldstein11b,
GoldsteinI14, Mackey_AOP14, Paulin_EJP14}).
\item Concentration inequalities that follow from rigorous methods in
statistical physics (see, e.g., \cite{Abbe_Montanari_KSAT,
Korada_Macris_ISIT07, Korada_Kudekar_Macris_ISIT08, Kudekar_thesis09,
Kudekar_Macris_IT2009, Korada_Macris_IT10, Montanari05, Talagrand_2010}).
\item The so-called reverse Lyapunov inequalities were recently
used to derive concentration inequalities for multi-dimensional
log-concave distributions \cite{Bobkov_Madiman_paper2} (see also a
related work in \cite{Bobkov_Madiman_paper1}).
The concentration inequalities in \cite{Bobkov_Madiman_paper2}
imply an extension of the Shannon--McMillan--Breiman strong ergodic theorem
to the class of discrete-time processes with log-concave marginals.
\end{itemize}
The last three items are not addressed in this monograph.

We now give a synopsis of some of the main ideas underlying the
martingale approach (Chapter~\ref{chapter: martingales}) and the
entropy method (Chapter~\ref{chapter: entropy method}).

The Azuma--Hoeffding inequality, as is introduced in Chapter~\ref{chapter: martingales},
is by now a well-known tool to establish concentration results for
{\em discrete-time bounded-difference martingales}.
It is due to Hoeffding \cite{Hoeffding}, who proved this inequality
for a sum of independent and bounded random variables, and to Azuma
\cite{Azuma}, who later extended it to bounded-difference martingales.
This inequality was introduced into the computer science literature by
Shamir and Spencer \cite{ShamirS87}, who used it to prove concentration
of the chromatic number for random graphs around its expected value
(the chromatic number of a graph is defined as the minimal number of
colors required to color all the vertices of this graph such that no
two adjacent vertices have the same color).
Shamir and Spencer \cite{ShamirS87} established concentration of the chromatic
number for the so-called {\em Erd\"os--R\'enyi} ensemble of random graphs,
where an arbitrary pair of vertices is connected by an edge with probability
$p \in (0,1)$, independently of all other edges. Note that the
concentration result in \cite{ShamirS87} was established without
knowing the expected value of the chromatic number over this ensemble.
This approach has been imported into coding theory in \cite{LubyMSS_IT01},
\cite{RichardsonU2001} and \cite{SipserS96}, especially for
exploring concentration of measure phenomena pertaining to codes defined
on graphs and iterative message-passing decoding algorithms. The last
decade has seen an ever-expanding use of the Azuma--Hoeffding inequality
for proving concentration inequalities in coding theory
(see, e.g., \cite{RiU_book} and references therein). All
these concentration inequalities serve in general to justify
theoretically the ensemble approach to codes defined on graphs;
nevertheless, much stronger concentration of measure phenomena are
observed in practice.

Let $f \colon \reals^n \rightarrow \reals$ be a function that has
{\em bounded differences}, i.e., the value of $f$ changes by a bounded
amount whenever any of its $n$ input variables is changed arbitrarily
while others are held fixed. A common method for proving
concentration of such a function of $n$ independent random variables
around its expected value $\expectation[f]$ revolves around the so-called
McDiarmid's inequality or the ``independent bounded-differences
inequality'' \cite{McDiarmid_tutorial}. This inequality, as is introduced
in Chapter~\ref{chapter: martingales}, was originally
proved via the martingale approach \cite{McDiarmid_tutorial}.
Although the proof of McDiarmid's inequality has some similarity to the
proof of the Azuma--Hoeffding inequality, the bounded-difference
assumption on $f$ that is used for the derivation of the former inequality
yields an improvement in the exponent by a factor of~$4$.
Nice applications of martingale-based concentration inequalities
in discrete mathematics and random graphs, based on the Azuma--Hoeffding
and McDiarmid inequalities, are exemplified in \cite[Section~3]{McDiarmid_tutorial},
\cite[Chapter~7]{AlonS_tpm3}, \cite{RiU_book} and \cite[Chapters~1,~2]{Steele_book}.

In spite of the large variety of problems where concentration of measure phenomena can be
asserted via the martingale approach, as pointed out by Talagrand \cite{Talagrand96},
``for all its qualities, the martingale method has a great drawback: it does not seem
to yield results of optimal order in several key situations. In particular,
it seems unable to obtain even a weak version of concentration of measure
phenomenon in Gaussian space.'' In Chapter~\ref{chapter: entropy method}
of this monograph, we focus on another set of techniques, fundamentally rooted
in information theory, that provide very strong concentration inequalities.
These powerful techniques, commonly referred to as the {\em entropy method}, have
originated in the work of Michel Ledoux \cite{Ledoux_paper}, who found an
alternative route to a class of concentration inequalities for product measures
originally derived by Talagrand \cite{Talagrand95} using an ingenious inductive
technique. Specifically, Ledoux noticed that the well-known Chernoff bounding
technique, which bounds the deviation probability of the form $\pr(|X-\bar{x}| > t)$,
for an arbitrary $t > 0$, in terms of the moment-generating function (MGF)
$\expectation[\exp(\lambda X)]$, can be combined with the so-called
{\em logarithmic Sobolev inequalities}, which can be used to control the MGF
in terms of the relative entropy.

Perhaps the best-known log-Sobolev inequality, first explicitly referred to
as such by Leonard Gross \cite{Gross}, pertains to the standard Gaussian
distribution in Euclidean space $\reals^n$, and bounds the relative entropy
$D(P \| G_n)$ between an arbitrary probability distribution $P$ on $\reals^n$
and the standard Gaussian measure $G_n$ by an ``energy-like'' quantity related
to the squared norm of the gradient of the density of $P$ w.r.t.\ $G_n$. By a
clever analytic argument which he attributed to an unpublished note by Ira Herbst,
Gross has used his log-Sobolev inequality to show that the logarithmic MGF
$\Lambda(\lambda) = \ln \expectation[\exp(\lambda U)]$ of $U = f(X^n)$, where
$X^n \sim G_n$ and $f \colon \reals^n \to \reals$ is an arbitrary sufficiently
smooth function with $\| \nabla f \| \le 1$, can be bounded as
$\Lambda(\lambda) \le \lambda^2/2$. This bound then yields the optimal Gaussian
concentration inequality $\pr\left(\left|f(X^n)-\expectation[f(X^n)]\right| > t\right)
\le 2\exp\left(-t^2/2\right)$ for $X^n \sim G_n$ and $t>0$.
(It should be pointed out that the Gaussian log-Sobolev inequality has a curious history,
and it seems to have been discovered independently in various equivalent forms by several
people, e.g., by Stam \cite{Stam} in the context of information theory, and by Federbush
\cite{Federbush} in the context of mathematical quantum field theory. Through the work of
Stam \cite{Stam}, the Gaussian log-Sobolev inequality has been linked to several other
information-theoretic notions, such as the concavity of entropy power
\cite{Costa_IT85, Dembo_Cover_Thomas, Villani_EP_concavity,Toscani}.)

In a nutshell, the entropy method takes this idea and applies it beyond the Gaussian case.
In abstract terms, log-Sobolev inequalities are functional inequalities that relate the
relative entropy between an arbitrary distribution $Q$ w.r.t.\ the distribution $P$ of
interest to some ``energy functional'' of the density $f = \d Q / \d P$. If one is interested
in studying concentration properties of some function $U = f(Z)$ with $Z \sim P$, the core
of the entropy method consists in applying an appropriate log-Sobolev inequality to the
{\em tilted distributions} $P^{(\lambda f)}$ with
$\d P^{(\lambda f)}/\d P \propto \exp(\lambda f)$. Provided the function $f$ is well-behaved
in the sense of having bounded ``energy,'' one can use the Herbst argument to pass from the
log-Sobolev inequality to the bound $\ln \expectation[\exp(\lambda U)] \le c\lambda^2/(2C)$,
where $c > 0$ depends only on the distribution $P$, while $C > 0$ is determined by the energy
content of $f$. While there is no general technique for deriving log-Sobolev inequalities,
there are nevertheless some underlying principles that can be exploited for that purpose.
We discuss some of these principles in Chapter~\ref{chapter: entropy method}. More information
on log-Sobolev inequalities can be found in several excellent monographs and lecture notes
\cite{Ledoux,Massart_book,Guionnet_Zegarlinski,Ledoux_lecture_notes,Royer}, as well as in recent
papers \cite{Bobkov_Gotze_expint,Bobkov_Ledoux,Bobkov_Tetali,Chafai,KitsosT_IT09} and references
therein.

Around the same time that Michel Ledoux first introduced the entropy method \cite{Ledoux_paper},
Katalin Marton showed in a breakthrough paper \cite{Marton_dbar}  that one can bypass functional
inequalities and work directly on the level of probability measures (see also the survey paper
\cite{Marton_ISIT2013}, presented at the 2013 Shannon Award Lecture). More specifically,
Marton has shown that Gaussian concentration bounds can be deduced from the so-called
{\em transportation-cost inequalities}. These inequalities, discussed in detail in
Section~\ref{sec:transportation}, relate information-theoretic quantities, such as the relative
entropy, to a certain class of distances between probability measures on the metric space where
the random variables of interest are defined. These so-called {\em Wasserstein distances} have
been the subject of intense research activity that touches upon probability theory, functional
analysis, dynamical systems, partial differential equations, statistical physics, and differential
geometry. A great deal of information on this field of {\em optimal transportation} can be found
in two books by C\'edric Villani --- \cite{Villani_TOT} offers a concise and fairly elementary
introduction, while a more recent monograph \cite{Villani_newbook} is a lot more detailed and
encyclopedic. Multiple connections between optimal transportation, concentration of measure, and
information theory are also explored in
\cite{Gozlan_Leonard,Dembo,Cattiaux_Guillin,Dembo_Zeitouni_TC,Djellout_Guillin_Wu,Gozlan,E_Milman}.
Note that Wasserstein distances have been also used in information theory in the context of
lossy source coding \cite{Gray_Neuhoff_Shields_dbar,Gray_Neuhoff_Omura,SV_1996}.

The first explicit invocation of concentration inequalities in an information-theoretic context
appears in the work of Ahlswede {\em et al.}~\cite{Ahlswede_Gacs_Korner,Ahlswede_Dueck}. These authors
have shown that a certain delicate probabilistic inequality, which was referred to as the ``blowing
up lemma'', and which we now (thanks to the contributions by Marton \cite{Marton_dbar,Marton_blowup})
recognize as a Gaussian concentration bound in the Hamming space, can be used to derive strong converses
for a wide variety of information-theoretic problems, including multi-terminal scenarios.
The importance of sharp concentration inequalities for characterizing fundamental limits of coding schemes
in information theory is evident from the recent flurry of activity on {\em finite-blocklength} analysis
of source and channel codes (see, e.g.,
\cite{Altug_Wagner_asymmetric,AmraouiMRU,NozakiKS,Kontoyiannis_Verdu,Kostina_Verdu_IT2012,Matthews_2012,
Polyanskiy_Poor_Verdu_IT2010,Wiechman_Sason_ISP}).
Thus, it is timely to revisit the use of concentration-of-measure ideas in information theory from
a modern perspective. We hope that our treatment, which, above all, aims to distill the core
information-theoretic ideas underlying the study of concentration of measure, will be helpful
to researchers in information theory and related fields.

\vspace*{-0.1cm}
\section{A reader's guide}

This monograph is mainly focused on the interplay between concentration of measure and
information theory, as well as applications to problems related
to information theory, communications and coding. For this reason, it is
primarily aimed at researchers and graduate students working in these fields.
The necessary mathematical background is real analysis, elementary functional analysis,
and a first graduate course in probability theory
and stochastic processes. As a refresher textbook for this mathematical 
background, the reader is referred, e.g., to \cite{Rosenthal_book}.

Chapter~\ref{chapter: martingales} on the martingale approach is structured as follows:
Section~\ref{section: Discrete-Time Martingales} lists key definitions and basic facts 
pertaining to discrete-time martingales, and Section~\ref{section: Two Basic Concentration Inequalities}
presents basic inequalities that form the basis of the martingale 
approach to concentration of measure. The concentration inequalities in 
Section~\ref{section: Two Basic Concentration Inequalities} include the 
celebrated Azuma--Hoeffding and McDiarmid inequalities, and 
Section~\ref{section: Refined Versions of the Azuma--Hoeffding Inequality}
is focused on the derivation of refined versions of the Azuma--Hoeffding inequality.
Section~\ref{section: Relations to Results in Probability Theory}
discusses the connections of the concentration inequalities introduced in
Section~\ref{section: Refined Versions of the Azuma--Hoeffding Inequality} to classical
limit theorems of probability theory, including the central limit theorem for martingales,
the moderate deviations principle for i.i.d.\ real-valued random variables, and the suitability
of the concentration inequalities derived in Chapter~\ref{chapter: martingales} for some structured
functions of discrete-time Markov chains.
Section~\ref{section: Applications} forms the second part of Chapter~\ref{chapter: martingales},
applying the concentration inequalities from Sections~\ref{section: Two Basic Concentration Inequalities}
and~\ref{section: Refined Versions of the Azuma--Hoeffding Inequality} to information theory, 
communications and coding theory. 
Section~\ref{Section: Summary} concludes with a summary of the chapter.

Several nice surveys on concentration inequalities via the martingale
approach are available, including \cite{McDiarmid_tutorial},
\cite[Chapter~7]{AlonS_tpm3}, \cite[Chapter~2]{Chung_LU2006},
\cite{survey2006} and \cite[Chapters~1 and~2]{Steele_book}.
The main focus of Chapter~\ref{chapter: martingales} is on the
presentation of several concentration inequalities that
form the basis of the martingale approach, with an emphasis on a sample
of their potential applications in information and communication-theoretic
aspects.

Chapter~\ref{chapter: entropy method} on the entropy method is structured as
follows: Section~\ref{sec:ingredients} introduces the main ingredients of the
entropy method, and it sets up the major themes that recur throughout the chapter.
Section~\ref{sec:Gaussian_LSI} focuses on the logarithmic Sobolev inequality
for Gaussian measures, as well as on its numerous links to information-theoretic
ideas. The general scheme of logarithmic Sobolev inequalities is introduced in
Section~\ref{sec:LSI}, and then applied to a variety of continuous and discrete
examples, including an alternative derivation of McDiarmid's inequality that does
not rely on martingale methods.
Thus, Sections~\ref{sec:Gaussian_LSI} and \ref{sec:LSI} present an approach to
deriving concentration bounds based on {\em functional} inequalities. In
Section~\ref{sec:transportation}, concentration is examined through the lens of
geometry in probability spaces equipped with a metric. This viewpoint centers around
intrinsic properties of probability measures, and has received a great deal of attention
since the pioneering work of Marton \cite{Marton_dbar,Marton_blowup} on transportation-cost
inequalities.  Although the focus in Chapter~\ref{chapter: entropy method} is mainly on
concentration for product measures, Section~\ref{sec:nonproduct} contains a brief summary
of a few results on concentration for functions of dependent random variables, and discusses
the connection between these results and the information-theoretic machinery that has been
the subject of the chapter. Several applications of concentration to problems in information
theory are surveyed in Section~\ref{sec:applications}. Section~\ref{section: Entropy Chapter Summary}
concludes with a brief summary.

%% file: chapter2.tex
\chapter{Concentration Inequalities via the Martingale Approach}
\label{chapter: martingales}
\chaptermark{Concentration Inequalities via the Martingale Approach}

This chapter introduces concentration inequalities for discrete-time
martingales with bounded differences, and it provides several of their
potential applications in information theory, digital communications and
coding. It starts by introducing the basic concentration inequalities of
Azuma--Hoeffding and McDiarmid, as well as various refinements.
It then moves to applications, which include concentration for
random binary linear block codes, concentration for random regular
bipartite graphs, concentration for low-density parity-check (LDPC) codes,
and concentration for orthogonal-frequency-division-multiplexing (OFDM) signals.

\section{Discrete-time martingales}
\label{section: Discrete-Time Martingales}

We start with a brief review of martingales to set definitions and notation.

\begin{definition}[Discrete-time martingales] Let $(\Omega, \mathcal{F},
\pr)$ be a probability space. A sequence $\{X_i, \mathcal{F}_i\}_{i=0}^n$,
$n \in \naturals$, where the $X_i$'s are
random variables and the $\mathcal{F}_i$'s are $\sigma$-algebras,
is a martingale if the following conditions are satisfied:

\begin{enumerate}
\item The $\mathcal{F}_i$'s form a {\em filtration}, i.e., $\mathcal{F}_0
\subseteq \mathcal{F}_1 \subseteq \ldots
\subseteq \mathcal{F}_n \subseteq \mathcal{F}$; usually, $\mathcal{F}_0$ is the trivial
$\sigma$-algebra $\{\emptyset, \Omega\}$ and $\mathcal{F}_n$ is the
full $\sigma$-algebra $\mathcal{F}$.
\item $X_i \in \LL^1(\Omega, \mathcal{F}_i, \pr)$
for every $i \in \{0, \ldots, n\}$; this means that each $X_i$ is
defined on the same sample space $\Omega$, it is $\mathcal{F}_i$-measurable,
and $\expectation [|X_i|] = \int_{\Omega} |X_i(\omega)| \pr(\d\omega) < \infty.$
\item For all $i \in \{1, \ldots, n\}$, $X_{i-1}
= \expectation[ X_i | \mathcal{F}_{i-1}]$ holds almost surely.
\end{enumerate}
\label{definition: Doob's martingales}
\end{definition}
\noindent In general, relations between random variables such as
$X=Y$, $X \le Y$ or $X \ge Y$ are assumed to hold almost surely (a.s.).

\bigskip Here are some useful facts about martingales.

\begin{fact}Since $\{\mathcal{F}_i\}_{i=0}^n$ is a
filtration, it follows from the tower property for
conditional expectations that
\begin{equation}
X_j = \expectation[X_i | \mathcal{F}_j],  \quad \forall \, i>j.
\label{eq:martingales}
\end{equation}
Also
$
\expectation[X_i] = \expectation \bigl[ \expectation[X_i |
\mathcal{F}_{i-1}] \bigr] = \expectation[X_{i-1}],
$
so, it follows from \eqref{eq:martingales} that the expectations
of every term $X_i$ of a martingale sequence are all equal to
$\expectation[X_0]$. Note that, since $X_i$ is $\mathcal{F}_i$-measurable,
\eqref{eq:martingales} also holds for $i=j$.
\end{fact}

\begin{fact} One can generate martingale sequences by the
following procedure: Given a random variable $X \in \LL^1(\Omega, \mathcal{F},
\pr)$ and an arbitrary filtration $\{\mathcal{F}_i\}_{i=0}^n$, let
\begin{equation*}
X_i = \expectation[X | \mathcal{F}_i], \quad \, \forall \, i \in
\{0, 1, \ldots\, n\}.
\end{equation*}
Then, the sequence $X_0, X_1, \ldots, X_n$ forms a martingale (with respect to
the above filtration) since
\begin{enumerate}

\item The random variable $X_i = \expectation[X |
\mathcal{F}_i]$ is $\mathcal{F}_i$-measurable, and
$\expectation[|X_i|] \leq \expectation[|X|] < \infty$.

\item By assumption, $\{\mathcal{F}_i\}_{i=0}^n$ is a filtration.

\item For every $i \in \{1, \ldots, n\}$
\begin{eqnarray*}
&& \hspace*{-1.8cm} \expectation[X_i | \mathcal{F}_{i-1}]
= \expectation \bigl[ \expectation[ X | \mathcal{F}_i] |
\mathcal{F}_{i-1} \bigr] \\
&& = \expectation[X | \mathcal{F}_{i-1}]  \quad (\text{since} \,
\mathcal{F}_{i-1} \subseteq \mathcal{F}_i) \\
&& = X_{i-1}.
\end{eqnarray*}
\end{enumerate}

In the particular case where
$\mathcal{F}_0 = \{\emptyset, \Omega\}$ and $\mathcal{F}_n = \mathcal{F}$,
we see that $X_0, X_1, \ldots, X_n$ is a martingale sequence with
\begin{eqnarray*}
&& \hspace*{-0.8cm} X_0 = \expectation[X | \mathcal{F}_0] =
\expectation[X], \quad X_n = \expectation[X | \mathcal{F}_n] = X.
\end{eqnarray*}
That is, we get a martingale sequence where the first element
is the expected value of $X$ and the last element
is $X$ itself (a.s.). This has the following interpretation: at the
beginning, we don't know anything about $X$, so we estimate it by its
expected value. At each step, more and more information about the random
variable $X$ is revealed, until its value is known almost surely.
\label{fact: construction of martingales}
\end{fact}

\begin{example}
Let $\{U_k\}_{k=1}^n$ be independent random variables on a common probability
space $(\Omega, \mathcal{F}, \pr)$, and assume that $\expectation[U_k] = 0$
and $\expectation[|U_k|] < \infty$ for every $k$. Let us define
$$ X_k = \sum_{j=1}^k U_j, \quad \forall \, k \in \{1, \ldots, n\}$$ with $X_0 = 0$.
Define the natural filtration where $\mathcal{F}_0 = \{\emptyset, \Omega\}$, and
\begin{align*}
\mathcal{F}_k &= \sigma(X_1, \ldots, X_k) \\
&= \sigma(U_1, \ldots, U_k), \quad \forall \, k \in \{1, \ldots, n\}.
\end{align*}
Note that $\mathcal{F}_k = \sigma(X_1, \ldots, X_k)$ denotes the minimal
$\sigma$-algebra that includes all the sets of the form $$\bigl\{\omega \in
\Omega \colon X_1(\omega) \leq \alpha_1, \ldots, X_k(\omega) \leq \alpha_k \bigr\}$$
where $\alpha_j \in \reals \cup \{-\infty, +\infty\}$ for $j \in \{1, \ldots, k\}$.
It is easy to verify that $\{X_k, \mathcal{F}_k\}_{k=0}^n$ is a martingale
sequence; this implies that all the concentration inequalities that apply
to discrete-time martingales (like those introduced in this chapter) can be
particularized to concentration inequalities for sums of independent random variables.
\label{example: sum of zero-mean and independent Rvs}
\end{example}

If we relax the equality in the definition of a martingale to an inequality,
we obtain sub- and super-martingales. More precisely, to define sub- and
super-martingales, we keep the first two conditions in
Definition~\ref{definition: Doob's martingales}, and the equality
in the third condition is replaced by one of the following:
\begin{itemize}
\item $\expectation[ X_i | \mathcal{F}_{i-1}] \geq X_{i-1}$
holds a.s.\ for sub-martingales.
\item $\expectation[ X_i | \mathcal{F}_{i-1}] \leq X_{i-1}$
holds a.s.\ for super-martingales.
\end{itemize}
From the tower property for conditional expectations, it follows that
\begin{equation}
\expectation[X_i | \mathcal{F}_j] \ge X_j,  \quad \forall \, i>j
\label{eq:sub-martingales}
\end{equation}
for sub-martingales, and
\begin{equation}
\expectation[X_i | \mathcal{F}_j] \le X_j,  \quad \forall \, i>j
\label{eq:super-martingales}
\end{equation}
for super-martingales. By taking expectations on both sides of
\eqref{eq:sub-martingales} and \eqref{eq:super-martingales}, it
follows that the expectations of the terms of a sub-martingale
(respectively, super-martingale) sequence form a monotonic increasing
(respectively, decreasing) sequence.
Clearly, every random process that is both a sub- and super-martingale is
a martingale, and vice versa. Furthermore, $\{X_i, \mathcal{F}_i\}$ is a
sub-martingale if and only if $\{-X_i, \mathcal{F}_i\}$ is a
super-martingale. The following properties are direct consequences
of Jensen's inequality for conditional expectations:

\begin{theorem} \label{theorem: mappings of martingales or sub/ super martingales}
The following holds for mappings of martingales or sub/ super martingales:
\begin{itemize}
\item If $\{X_i, \mathcal{F}_i\}$ is a martingale, $h$ is
a convex (concave) function and $\expectation\bigl[|h(X_i)|\bigr]<\infty$,
then $\{h(X_i), \mathcal{F}_i\}$ is a sub- (super-) martingale.
\item If $\{X_i, \mathcal{F}_i\}$ is a super-martingale, $h$ is monotonic
increasing and concave, and $\expectation\bigl[|h(X_i)|\bigr]<\infty$,
then $\{h(X_i), \mathcal{F}_i\}$ is a super-martingale.
Similarly, if $\{X_i, \mathcal{F}_i\}$ is a
sub-martingale, $h$ is monotonic increasing and convex, and
$\expectation\bigl[|h(X_i)|\bigr]<\infty$, then
$\{h(X_i), \mathcal{F}_i\}$ is a sub-martingale.
\end{itemize}
\end{theorem}

\begin{example}
The following are special cases of
Theorem~\ref{theorem: mappings of martingales or sub/ super martingales}:
\begin{itemize}
\item If $\{X_i, \mathcal{F}_i\}$ is a martingale, then $\{|X_i|, \mathcal{F}_i\}$
is a sub-martingale.
\item If $\{X_i, \mathcal{F}_i\}$ is a martingale and
$X_i \in \LL^2(\Omega, \mathcal{F}_i, \pr)$,
then $\{X_i^2, \mathcal{F}_i\}$ is a sub-martingale.
\item If $\{X_i, \mathcal{F}_i\}$ is a non-negative sub-martingale and
$X_i \in \LL^2(\Omega, \mathcal{F}_i, \pr)$ (i.e., for every $i$, the random variable
$X_i$ is defined on the same sample space $\Omega$, it is $\mathcal{F}_i$-measurable,
and $\expectation [X_i^2] < \infty$) then also $\{X_i^2, \mathcal{F}_i\}$ is a sub-martingale.
\end{itemize}
\end{example}

\section{Basic concentration inequalities}
\label{section: Two Basic Concentration Inequalities}

We now turn to the main topic of the chapter, namely the martingale approach
to proving concentration inequalities, i.e., sharp bounds on the {\em deviation
probabilities} $\pr\left(|U - \expectation U| \ge r\right)$ for all $r \ge 0$,
where $U$ is a real-valued random variable with some additional ``structure"
--- for instance, $U$ may be a function of a large number $n$ of independent
or weakly dependent random variables $X_1,\ldots,X_n$.  In a nutshell, the
martingale approach has two basic ingredients:
\begin{enumerate}
	\item {\bf The martingale decomposition} --- we first construct a suitable
filtration $\{\cF_i\}^n_{i=0}$ on the probability space $(\Omega,\cF,\pr)$ that
carries $U$, where $\cF_0 = \{\emptyset,\Omega\}$ is the trivial $\sigma$-algebra,
and $\cF_n = \cF$. Then we decompose the difference $U - \expectation U$ as
\begin{align}  \label{eq:martingale_decomposition}
U - \expectation U &= \expectation[U|\cF_n] - \expectation[U|\cF_0] \nonumber\\
&= \sum^n_{i=1} \left( \expectation[U|\cF_i] - \expectation[U|\cF_{i-1}]\right).
\end{align}
The idea is to choose the $\sigma$-algebras $\{\cF_i\}$ in such a way that the differences
$\xi_i = \expectation[U|\cF_i] - \expectation[U|\cF_{i-1}]$ are bounded in some sense,
e.g., almost surely.
\item {\bf The Chernoff bounding technique} --- using Markov's inequality, the problem of
bounding the deviation probability $\pr (|U - \expectation U| \ge r)$ is reduced to the
analysis of the {\em logarithmic moment-generating function}
$\Lambda(t) \deq \ln \expectation[\exp(t U)]$, $t \in \reals$. Moreover, exploiting the
martingale decomposition \eqref{eq:martingale_decomposition}, we may write
\begin{align*}
\Lambda(t) = t \expectation[U] + \ln \expectation\left[ \prod^n_{i=1}\exp(t\xi_i)\right],
\end{align*}
which allows us to focus on the behavior of individual terms $\exp(t\xi_i)$, $i = 1,\ldots,n$.
Now, the logarithmic moment-generating function plays a key role in the theory of large
deviations \cite{Dembo_Zeitouni}, which can be thought of as a (mainly) asymptotic analysis
of the concentration of measure phenomenon. Thus, its prominent appearance here is not
entirely unexpected.
\end{enumerate}
There are more sophisticated variants of the martingale approach, some of which we will
have occasion to see later on, but the above two ingredients are a good starting point.
In the remainder of this section, we will elaborate on these ideas and examine their basic
consequences.

\subsection{The Chernoff bounding technique and the Hoeffding lemma}
\label{subsection: Chernoff}

The first ingredient of the martingale method is the well-known Chernoff bounding
technique\footnote{The name of H.~Chernoff is associated with this technique because
of his 1952 paper \cite{Chernoff}; however, its roots go back to S.N.~Bernstein's
1927 textbook on the theory of probability \cite{Bernstein}.}: Using Markov's inequality,
for every $t > 0$,
\begin{align*}
	\pr (U  \ge r) &= \pr \big(\exp(t U) \ge \exp(t r)\big)\\
	&\le \exp(-t r) \expectation[\exp(t U)].
\end{align*}
Equivalently, if we define the {\em logarithmic moment generating function}
$\Lambda(t) \deq \ln \expectation[\exp(t U)]$, $t \in \reals$, we can write
\begin{align}\label{eq:upper_tail}
\pr (U \ge r) & \le \exp\big(\Lambda(t)-t r \big), \qquad \forall \, t > 0.
\end{align}
To bound the probability of the lower tail, $\pr (U \le - r)$, we follow the
same steps, but with $-U$ instead of $U$. Now the success of the whole enterprize
hinges on our ability to obtain tight upper bounds on $\Lambda(t)$. One of the
basic tools available for that purpose is the following lemma due to Hoeffding
\cite{Hoeffding}:
\begin{lemma}[Hoeffding]\label{lm:Hoeffding} Let $U \in \reals$ be a random
variable, such that $U \in [a,b]$ a.s.\ for some finite $a < b$. Then, for every
$t \in \reals$,
\begin{align}
\label{eq:Hoeffding}
\expectation\left[ \exp\big(t(U- \expectation U)\big)\right]
\le \exp \left(\frac{t^2(b-a)^2}{8}\right).
\end{align}
\end{lemma}
\begin{proof}
For every $p \in [0,1]$ and $\lambda \in \reals$, let us define the function
\begin{align}\label{eq:Hoeffding_function}
H_p(\lambda) \deq \ln \left( pe^{\lambda (1-p)} + (1-p)e^{-\lambda p}\right).
\end{align}	
Let $\xi = U - \expectation U$, where $\xi \in [a-\expectation U, b - \expectation U]$.
Using the convexity of the exponential function, we can write
\begin{align*}
\exp(t\xi)
&= \exp\left( \frac{U-a}{b-a} \cdot t(b-\expectation U) + \frac{b-U}{b-a}
\cdot t(a - \expectation U)\right) \\
&\le \left(\frac{U-a}{b-a}\right) \exp\big(t(b-\expectation U)\big)
+ \left(\frac{b-U}{b-a}\right)\exp\big(t(a-\expectation U)\big).
\end{align*}
Taking expectations of both sides, we get
\begin{align}
\expectation[\exp(t\xi)] &\le \left(\frac{\expectation U - a}{b-a}\right)
\exp\big(t(b-\expectation U)\big) + \left(\frac{b - \expectation U}{b-a}\right)
\exp\big(t(a-\expectation U)\big) \nonumber\\
&= \exp\big( H_p(\lambda)\big) \label{eq:Hoeffding_bound_0}
\end{align}
where we have let
\begin{align*}
p = \frac{\expectation U - a}{b-a} \qquad \text{and} \qquad \lambda =  t(b-a).
\end{align*}
In the following, we show that for every $\lambda \in \reals$
\begin{equation} \label{eq:bound on Hoeffding_function}
H_p(\lambda) \le \frac{\lambda^2}{8}, \quad \forall \, p \in [0,1].
\end{equation}
From \eqref{eq:Hoeffding_function}, we have
\begin{align}
& H_p(\lambda) = -\lambda p + \ln\bigl(p e^\lambda + (1-p)\bigr),
\label{eq:Hoeffding_function2}\\
& H_p'(\lambda) = -p + \frac{p e^\lambda}{p e^\lambda + 1-p},
\label{eq:derivative of Hoeffding_function} \\
& H_p''(\lambda) = \frac{p(1-p) e^{\lambda}}{\bigl(p e^{\lambda} + (1-p) \bigr)^2}.
\label{eq:2nd derivative of Hoeffding_function}
\end{align}
From \eqref{eq:Hoeffding_function2}--\eqref{eq:2nd derivative of Hoeffding_function},
we have $H_p(0) = H_p'(0)=0$, and
\begin{align*}
H_p''(\lambda) &= \frac{1}{4} \, \frac{p e^{\lambda} \cdot (1-p)}{\left(\frac{p
e^{\lambda} + (1-p)}{2} \right)^2} \\
& \le \frac{1}{4}, \quad \forall \, \lambda \in \reals, \; p \in [0,1]
\end{align*}
where the last inequality holds since the geometric mean is less than or equal
to the arithmetic mean. Using a Taylor's series expansion, there exists an
intermediate value $\theta \in [0, \lambda]$ (or $\theta \in [\lambda, 0]$ if
$t < 0$) such that
\begin{align*}
H_p(\lambda) = H_p(0) + H_p'(0) \lambda + \frac{1}{2} \, H_p''(\theta) \, \lambda^2
\end{align*}
so, consequently, \eqref{eq:bound on Hoeffding_function} holds.
Substituting this bound into \eqref{eq:Hoeffding_bound_0} and using the above definitions
of $p$ and $\lambda$, we get \eqref{eq:Hoeffding}.
\end{proof}

\subsection{The Azuma--Hoeffding inequality}
\label{subsection: Azuma--Hoeffding inequality}
The Azuma--Hoeffding inequality, stated in
Theorem~\ref{theorem: Azuma--Hoeffding inequality} below, is a useful concentration
inequality for bounded-difference martingales. It was proved by Hoeffding \cite{Hoeffding}
for sums of independent and bounded random variables, followed by a discussion on sums of
dependent random variables. This inequality was later generalized by Azuma \cite{Azuma} to
the more general setting of bounded-difference martingales. The proof of the Azuma--Hoeffding
inequality that we present below is a nice concrete illustration of the general approach outlined
in the beginning of this section. Moreover, we will have many occasions to revisit this
proof in order to obtain various refinements of the Azuma--Hoeffding inequality.

\begin{theorem}[The Azuma--Hoeffding inequality]
Let $\{X_k, \mathcal{F}_k\}_{k=0}^n$
be a real-valued martingale sequence.
Suppose that there exist nonnegative reals $d_1,\ldots,d_n$, such that
$ |X_k - X_{k-1}| \leq d_k$ a.s.\ for all $k \in \{1,\ldots,n\}$. Then, for every $r > 0$,
\begin{equation}
\pr( | X_n - X_0 | \geq r) \leq 2 \exp\left(-\frac{r^2}{2
\sum_{k=1}^n d_k^2}\right).
\label{eq: Azuma--Hoeffding concentration inequality - general case}
\end{equation}
\label{theorem: Azuma--Hoeffding inequality}
\end{theorem}
\begin{proof}
For an arbitrary $r > 0$,
\begin{equation}
\pr(|X_n - X_0| \geq r) = \pr(X_n - X_0 \geq r) + \pr(X_n - X_0 \leq -r).
\label{eq: union of disjoint events}
\end{equation}
Let $\xi_k \triangleq X_k - X_{k-1}$ for $k \in \{1, \ldots, n\}$ denote
the differences of the martingale sequence. By hypothesis,
$|\xi_k| \leq d_k$ and $\expectation[\xi_k \, | \, \mathcal{F}_{k-1}] = 0$
a.s.\ for every $k \in \{1, \ldots, n\}$.

We now apply the Chernoff bounding technique:
\begin{eqnarray}
&& \pr(X_n - X_0 \geq r) \nonumber \\
&& = \pr \Biggl(\sum_{k=1}^n \xi_k \geq r \Biggr) \nonumber \\
&& \leq \exp(-tr) \, \expectation\left[\exp \left(t \sum_{k=1}^n \xi_k \right) \right],
\quad \forall \, t \geq 0.
\label{eq: Chernoff's inequality}
\end{eqnarray}
By the law of iterated expectations, the expectation on the right-hand side of
\eqref{eq: Chernoff's inequality} is equal to
\begin{eqnarray}
&& \expectation \biggl[ \exp \biggl(t \sum_{k=1}^n \xi_k \biggr)
\biggr] \nonumber \\
&& = \expectation \Biggl[ \expectation \biggl[ \exp \biggl(t
\sum_{k=1}^n \xi_k \biggr) \, \bigg| \, \mathcal{F}_{n-1} \biggr] \Biggr]
\nonumber \\
&& = \expectation \Biggl[ \exp \biggl(t \sum_{k=1}^{n-1} \xi_k
\biggr) \, \expectation \bigl[ \exp(t \xi_n) \, | \,
\mathcal{F}_{n-1} \bigr] \Biggr]
\label{eq: smoothing theorem}
\end{eqnarray}
where the last equality holds since $Y_n \triangleq \exp \bigl(t
\sum_{k=1}^{n-1} \xi_k \bigr)$ is $\mathcal{F}_{n-1}$-measurable.
We now apply the Hoeffding lemma with the conditioning on $\cF_{n-1}$.
Indeed, we know that $\expectation[\xi_n|\cF_{n-1}] = 0$ and that
$\xi_n \in [-d_n,d_n]$ a.s., so Lemma~\ref{lm:Hoeffding} gives that
\begin{align}
\label{eq:bound on the conditional moment generating function of xi_n}
\expectation \bigl[\exp(t \xi_n) \, | \, \mathcal{F}_{n-1}\bigr]
\leq \exp\left(\frac{t^2 \, d_n^2}{2}\right).
\end{align}
Continuing recursively in a similar manner, we can bound the quantity
in \eqref{eq: smoothing theorem} by
\begin{equation}  \label{eq:bound on the moment generating function of Xn-X0}
\expectation \biggl[ \exp \biggl(t \sum_{k=1}^n \xi_k \biggr)
\biggr] \leq \prod_{k=1}^n \exp\left(\frac{t^2 \, d_k^2}{2}\right)
= \exp \left(\frac{t^2}{2} \, \sum_{k=1}^n d_k^2 \right).
\end{equation}
Substituting this bound into \eqref{eq: Chernoff's inequality}, we obtain
\begin{align}\label{eq:parametric_Azuma}
\pr(X_n - X_0 \geq r) \leq
\exp\left(-t r + \frac{t^2}{2} \, \sum_{k=1}^n d_k^2 \right), \quad \forall \, t \geq 0.
\end{align}
Finally, choosing
$t = r \left(\sum_{k=1}^n d_k^2\right)^{-1}$ to minimize the right-hand side of
\eqref{eq:parametric_Azuma}, we get
\begin{equation}
\pr(X_n - X_0 \geq r) \leq \exp \left(-\frac{r^2}{2 \sum_{k=1}^n d_k^2} \right).
\label{eq: one-sided Azuma--Hoeffding inequality}
\end{equation}
Since $\{X_k, \mathcal{F}_k \}$ is a martingale with bounded differences,
so is $\{-X_k, \mathcal{F}_{k}\}$ (with the same bounds on its differences).
This implies that the same bound is also valid for the probability
$\pr(X_n - X_0 \leq -r)$. Using these bounds in
\eqref{eq: union of disjoint events}, we complete the proof of
Theorem~\ref{theorem: Azuma--Hoeffding inequality}.
\end{proof}

\begin{remark}
In \cite[Theorem~3.13]{McDiarmid_tutorial}, the Azuma--Hoeffding inequality
is stated as follows: Let $\{Y_k, \mathcal{F}_k\}_{k=0}^n$ be a
martingale-difference sequence with $Y_0=0$ (i.e., $Y_k$ is
$\mathcal{F}_k$-measurable, $\expectation[|Y_k|] < \infty$ and
$\expectation[Y_k| \mathcal{F}_{k-1}]=0$ a.s.\ for every $k \in
\{1, \ldots, n\}$). Assume that, for every $k$, there
exist some numbers $a_k, b_k \in \reals$ such that, a.s.,
$a_k \leq Y_k \leq b_k$. Then, for every $r \geq 0$,
\begin{equation}
\pr \left( \bigg|\sum_{k=1}^n Y_k\bigg| \geq r \right) \leq 2
\exp\left(-\frac{2r^2}{\sum_{k=1}^n (b_k-a_k)^2} \right).
\label{eq: concentration inequality for a martingale-difference
sequence (McDiarmid's tutorial)}
\end{equation}
Consider a real-valued martingale sequence $\{X_k, \mathcal{F}_k\}_{k=0}^n$,
where $a_k \leq X_k - X_{k-1} \leq b_k$ a.s.\ for every $k$. Let $Y_k
\triangleq X_k - X_{k-1}$ for every $k \in \{1, \ldots, n\}$. Then it
is easy to see that $\{Y_k, \mathcal{F}_k\}_{k=0}^n$ is a
martingale-difference sequence. Since $\sum_{k=1}^n Y_k = X_n - X_0$, it
follows from
\eqref{eq: concentration inequality for a martingale-difference sequence (McDiarmid's tutorial)}
that
$$\pr \left( |X_n - X_0| \geq r \right) \leq 2
\exp\left(-\frac{2r^2}{\sum_{k=1}^n (b_k-a_k)^2} \right),
\quad \forall \, r > 0.$$
\end{remark}

\begin{example}
Let $\{Y_i\}_{i=0}^{\infty}$ be i.i.d.\ binary random variables
which take values $\pm d$ with equal
probability, where $d > 0$ is some constant. Let $X_k = \sum_{i=0}^k Y_i$ for $k \in \{0, 1,
\ldots, \}$, and define the natural filtration $\mathcal{F}_0
\subseteq \mathcal{F}_1 \subseteq \mathcal{F}_2 \ldots $ where
$$\mathcal{F}_k = \sigma(Y_0, \ldots, Y_k) \, , \quad \forall
\, k \in \{0, 1, \ldots, \}$$ is the $\sigma$-algebra
generated by  $Y_0, \ldots, Y_k$. Note that
$\{X_k, \mathcal{F}_k\}_{k=0}^{\infty}$ is a martingale sequence, and
(a.s.) $ |X_k - X_{k-1}| = |Y_k| = d, \, \forall \, k \in
\naturals$. It therefore follows from the Azuma--Hoeffding inequality that
\begin{equation}
\pr( | X_n - X_0 | \geq \alpha \sqrt{n}) \leq 2
\exp\left(-\frac{\alpha^2}{2d^2}\right) \label{eq: Azuma--Hoeffding
inequality for example1}
\end{equation}
for every $\alpha \geq 0$ and $n \in \naturals$.  Since the random
variables $\{Y_i\}_{i=0}^{\infty}$ are i.i.d.\ with zero mean and
variance~$d^2$, the Central Limit Theorem (CLT) says that
$\frac{1}{\sqrt{n}} (X_n - X_0) = \frac{1}{\sqrt{n}} \sum_{k=1}^n Y_k$
converges in distribution to
$\mathcal{N}(0,d^2)$. Therefore, for every $\alpha \geq 0$,
\begin{equation}
\lim_{n \rightarrow \infty} \pr( | X_n - X_0 | \geq \alpha
\sqrt{n})= 2 \, Q\Bigl(\frac{\alpha}{d}\Bigr)
\label{CLT1 - i.i.d. RVs}
\end{equation}
where
\begin{equation}
Q(x) \triangleq \frac{1}{\sqrt{2\pi}} \, \int_{x}^{\infty}
\exp\left(-\frac{t^2}{2}\right) \mathrm{d}t, \quad \forall
\, x \in \reals \label{eq: Q function}
\end{equation}
is the complementary standard Gaussian CDF (also known as the
$Q$-function), for which we have the following exponential
upper and lower bounds (see, e.g., \cite[Section~3.3]{MUD_book_SV}):
\begin{equation}
\frac{1}{\sqrt{2\pi}} \, \frac{x}{1+x^2} \cdot \exp\left(-\frac{x^2}{2}\right)
< Q(x) < \frac{1}{\sqrt{2\pi} \, x} \cdot \exp\left(-\frac{x^2}{2}\right), \;
\; \forall \, x>0. \label{eq: upper and lower bounds for the Q
function}
\end{equation}
From \eqref{CLT1 - i.i.d. RVs} and \eqref{eq: upper and lower bounds for the Q
function}, it follows that the exponent on the right-hand side of
\eqref{eq: Azuma--Hoeffding inequality for example1} is exact.
\label{example1}
\end{example}

\begin{example} Fix some $\gamma \in (0,1]$. Let us generalize
Example~\ref{example1} above by considering the case where the
i.i.d.\ binary random variables $\{Y_i\}_{i=0}^{\infty}$ have the probability law
$$ \pr(Y_i = +d) = \frac{\gamma}{1+\gamma}, \quad \pr(Y_i =
-\gamma d) = \frac{1}{1+\gamma} \; .$$
Therefore, each $Y_i$ has zero mean and variance $\sigma^2 = \gamma d^2$.
Define the martingale sequence $\{X_k, \mathcal{F}_k\}_{k=0}^{\infty}$ as
in Example~\ref{example1}. By the CLT, $\frac{1}{\sqrt{n}} \, (X_n - X_0) =
\frac{1}{\sqrt{n}} \, \sum_{k=1}^n Y_k$ converges weakly to
$\mathcal{N}(0, \gamma d^2)$, so for every $\alpha \geq 0$
\begin{equation}
\lim_{n \rightarrow \infty} \pr( | X_n - X_0 | \geq \alpha
\sqrt{n})= 2 \, Q\biggl(\frac{\alpha}{\sqrt{\gamma} \, d}\biggr).
\label{CLT2 - i.i.d. RVs}
\end{equation}
From the bounds on the $Q$-function given in
\eqref{eq: upper and lower bounds for the Q function}, it follows that the
right-hand side of \eqref{CLT2 - i.i.d. RVs} scales exponentially like
$e^{-\frac{\alpha^2}{2 \gamma d^2}}$. Hence, the exponent in this example
is improved by a factor of $\frac{1}{\gamma}$ in comparison to the Azuma--Hoeffding inequality
(which gives the same bound as in Example~\ref{example1} since
$|X_k - X_{k-1}| \leq d$ for every $k \in \naturals$). This indicates that
a refinement of the Azuma--Hoeffding inequality is possible if additional information on
the variance is available. Refinements of this sort were studied extensively
in the probability literature, and they are the focus of
Section~\ref{subsection: Martingales with a uniform bound on the differences}.
\label{example2}
\end{example}

\subsection{McDiarmid's inequality}
\label{subsection: McDiarmid's inequality}

A prominent application of the martingale approach is the derivation of a
powerful inequality due to McDiarmid (see \cite[Theorem~3.1]{McDiarmid_1997}
or \cite{McDiarmid_bounded_differences_Martingales_1989}), also known as the
{\em bounded-difference inequality}. Let $\cX$ be a set, and let
$f \colon \cX^n \to \reals$ be a function that satisfies the
{\em bounded difference assumption}
\begin{align}\label{eq:bounded_differences_ch2}
\sup_{x_1,\ldots,x_n,x'_i \in \cX} \Big| & f(x_1,\ldots,x_{i-1},x_i,x_{i+1}\ldots,x_n) \nonumber\\
& - f(x_1,\ldots,x_{i-1},x'_i,x_{i+1},\ldots,x_n) \Big| \le d_i
\end{align}
for every $1 \le i \le n$, where $d_1,\ldots,d_n$ are arbitrary nonnegative
real constants. This is equivalent to saying that, for every given $i$,
the variation of the function $f$ with respect to\ its $i$th coordinate is upper
bounded by $d_i$. (We assume that each argument of $f$ takes values in
the same set $\cX$ mainly for simplicity of presentation; an extension
to different domains for each variable is easy.)

\begin{theorem}[McDiarmid's inequality]
Let $\{X_k\}_{k=1}^n$ be independent (not necessarily identically distributed)
random variables taking values in a measurable space $\cX$. Consider a
random variable $U = f(X^n)$ where $f \colon \cX^n \to \reals$ is a measurable
function satisfying the bounded difference assumption \eqref{eq:bounded_differences_ch2},
and $X^n \triangleq (X_1, \ldots, X_n)$.
Then, for every $r \geq 0$,
\begin{equation}
\pr\left( \bigl| U - \expectation U \bigr|
\geq r\right) \leq 2 \exp \left(-\frac{2 r^2}{\sum_{k=1}^n d_k^2} \right).
\label{eq: McDiarmid's inequality}
\end{equation}
\label{theorem: McDiarmid's inequality}
\end{theorem}

\begin{remark}
One can use the Azuma--Hoeffding inequality for a derivation of a
concentration inequality in the considered setting. However, the following
proof provides an improvement by a factor of~4 in the exponent
of the bound.
\end{remark}

\begin{proof}
Let $\cF_0 = \{\emptyset, \Omega\}$ be the trivial $\sigma$-algebra,
and for $k \in \{1, \ldots, n\}$ let $\cF_k = \sigma(X_1, \ldots, X_k)$
be the $\sigma$-algebra  generated by $X_1, \ldots, X_k$.
For every $k \in \{1,\ldots,n\}$, define
\begin{equation}
\xi_k \triangleq \expectation \bigl[ f(X^n) \, | \, \mathcal{F}_k \bigr]
- \expectation \bigl[ f(X^n) \, | \, \mathcal{F}_{k-1} \bigr].
\label{eq: martingale difference}
\end{equation}
Note that $\mathcal{F}_0 \subseteq \mathcal{F}_1 \ldots \subseteq \mathcal{F}_n$
is a filtration, and
\begin{eqnarray}
&& \expectation \bigl[ f(X^n) \, | \, \mathcal{F}_0 \bigr] =
\expectation \bigl[ f(X^n) \bigr], \nonumber \\
&& \expectation \bigl[ f(X^n) \, | \, \mathcal{F}_n \bigr] = f(X^n).
\end{eqnarray}
From the last three equalities, it follows that
\begin{equation*}
f(X^n) - \expectation \bigl[ f(X^n) \bigr] = \sum_{k=1}^n \xi_k.
\end{equation*}
In the following, we need a lemma:
\begin{lemma} \label{lemma: McDiarmid}
For every $k \in \{1, \ldots, n\}$, the following properties hold a.s.:
\medskip

\noindent 1. $\expectation[\xi_k \, | \, \mathcal{F}_{k-1}] = 0$ and $\xi_k$ is
$\mathcal{F}_k$-measurable, so $\{\xi_k, \mathcal{F}_k\}$ is a martingale-difference.
\medskip

\noindent 2. $|\xi_k| \leq d_k$.
\medskip

\noindent 3. $\xi_k \in [A_k, A_k + d_k]$ where $A_k$ is a non-positive and
$\mathcal{F}_{k-1}$-measurable random variable.
\end{lemma}
\begin{proof}
The random variable $\xi_k$, defined in \eqref{eq: martingale difference}, is
$\mathcal{F}_k$-measurable since $\mathcal{F}_{k-1} \subseteq \mathcal{F}_k$,
and $\xi_k$ is a difference of two functions where one is
$\mathcal{F}_k$-measurable and the other is $\mathcal{F}_{k-1}$-measurable.
Furthermore, since $\{\cF_i\}$ is a filtration, it follows from
\eqref{eq: martingale difference} and the tower principle for conditional
expectations that $\expectation[\xi_k \, | \, \mathcal{F}_{k-1}] = 0.$
This proves the first item.
The second item follows from the first and third items since the latter two
items imply that
\begin{align} \label{eq:sub-interval of third item}
A_k &= \expectation[A_k | \cF_{k-1}] \nonumber \\
    &\leq \expectation[\xi_k | \cF_{k-1}] = 0 \nonumber \\
    &\leq \expectation[A_k + d_k | \cF_{k-1}] \nonumber \\
    &=A_k + d_k
\end{align}
where the first and last equalities hold since $A_k$ is $\mathcal{F}_{k-1}$-measurable.
Hence, $0 \in [A_k, A_k + d_k]$ which implies that $[A_k, A_k + d_k] \subseteq [-d_k, d_k]$;
consequently, it follows from the third item that $|\xi_k| \le d_k$.

To prove the third item, note that $\xi_k = f_k(X_1, \ldots, X_k)$ holds
a.s. for the $\cF_k$-measurable function $f_k \colon \cX^k \rightarrow \reals$
which is given by
\begin{align} \label{eq:f_k in McDiarmid's proof}
& f_k(x_1, \ldots, x_k) \nonumber \\
& = \expectation\bigl[f(x_1, \ldots, x_k, X_{k+1}, \ldots, X_n)\bigr]
- \expectation\bigl[f(x_1, \ldots, x_{k-1}, X_k, \ldots, X_n)\bigr].
\end{align}
Equality \eqref{eq:f_k in McDiarmid's proof} holds due to
the definition of $\{\xi_k\}$ in \eqref{eq: martingale difference}
with $\cF_k = \sigma(X_1, \ldots, X_k)$ for $k \in \{1, \ldots, n\}$,
and the independence of the random variables $\{X_k\}_{k=1}^n$.
Let us define, for every $k \in \{1, \ldots, n\}$,
\begin{align*}
& A_k \triangleq \inf_{x \in \mathcal{X}} f_k(X_1, \ldots, X_{k-1}, x), \nonumber \\
& B_k \triangleq \sup_{x \in \mathcal{X}} f_k(X_1, \ldots, X_{k-1}, x)
\end{align*}
which are $\mathcal{F}_{k-1}$-measurable\footnote{This is certainly the case
if ${\cal X}$ is countably infinite. For uncountable spaces, one needs to
introduce some regularity conditions to guarantee measurability of infima
and suprema. We choose not to dwell on these technicalities here to keep
things simple; the book by van der Vaart and Wellner \cite{Vaart_Wellner_book}
contains a thorough treatment of these issues.}, and by definition
$\xi_k \in [A_k, B_k]$ holds almost surely. Furthermore, for every point
$(x_1, \ldots, x_{k-1}) \in \mathcal{X}^{k-1}$,
we obtain from \eqref{eq:f_k in McDiarmid's proof} that
\begin{align}
& \sup_{x \in \mathcal{X}} f_k(x_1, \ldots, x_{k-1}, x) -
\inf_{x' \in \mathcal{X}} f_k(x_1, \ldots, x_{k-1}, x') \nonumber \\
& = \sup_{x, x' \in \mathcal{X}} \bigl\{ f_k(x_1, \ldots, x_{k-1}, x) -
f_k(x_1, \ldots, x_{k-1}, x') \bigr\} \nonumber \\
& = \sup_{x, x' \in \mathcal{X}} \Bigl\{
\expectation \bigl[ f(x_1, \ldots, x_{k-1}, x, X_{k+1}, \ldots, X_n)]
\nonumber\\
& \qquad \qquad - \expectation \bigl[ f(x_1, \ldots, x_{k-1}, x', X_{k+1}, \ldots, X_n)] \Bigr\}
\label{eq: independence} \\
& = \sup_{x, x' \in \mathcal{X}} \Bigl\{
\expectation \bigl[ f(x_1, \ldots, x_{k-1}, x, X_{k+1}, \ldots, X_n)
\nonumber\\
& \qquad \qquad - f(x_1, \ldots, x_{k-1}, x', X_{k+1}, \ldots, X_n)] \Bigr\} \nonumber \\
& \leq d_k \label{eq: bound on the variation of g with respect to k-th coordinate}
\end{align}
where \eqref{eq: independence} follows from \eqref{eq:f_k in McDiarmid's proof},
and \eqref{eq: bound on the variation of g with respect to k-th coordinate}
follows from the bounded-difference condition in \eqref{eq:bounded_differences_ch2}.
Hence, $B_k - A_k \leq d_k$ a.s., which implies that
$\xi_k \in [A_k, A_k+d_k]$. Note that
the third item of the lemma gives better control on the range of $\xi_k$ than what we
had in the proof of the Azuma--Hoeffding inequality (i.e., item~2 asserts that $\xi_k$ is
contained in the interval $[-d_k, d_k]$ which is twice longer than the sub-interval
$[A_k,A_k+d_k]$ in the third item, see \eqref{eq:sub-interval of third item}).
\end{proof}
We now proceed in the same manner as in the proof of the Azuma--Hoeffding inequality.
Specifically, for $k \in \{1,\ldots,n\}$, $\xi_k \in [A_k,A_k+d_k]$ a.s.,
where $A_k$ is $\cF_{k-1}$-measurable, and $\expectation[\xi_k|\cF_{k-1}] = 0$. Thus,
we may apply the Hoeffding lemma (see Lemma~\ref{lm:Hoeffding}) with a conditioning
on $\cF_{k-1}$ to get
\begin{align} \label{eq:bound on exponent in McDiarmid's proof}
\expectation\Big[e^{t\xi_k}\Big|\cF_{k-1}\Big] \le \exp\left( \frac{t^2d^2_k}{8}\right).
\end{align}
Similarly to the proof of the Azuma--Hoeffding inequality, by repeatedly using
the recursion in \eqref{eq: smoothing theorem}, the last inequality implies that
\begin{equation}
\expectation \biggl[ \exp \biggl(t \sum_{k=1}^n \xi_k \biggr)
\biggr] \leq \exp \left(\frac{t^2}{8} \, \sum_{k=1}^n d^2_k \right)
\end{equation}
and, from \eqref{eq: Chernoff's inequality},
\begin{align}
& \pr(f(X^n) - \expectation[f(X^n)] \geq r) \nonumber \\
& = \pr\left( \sum_{k=1}^n \xi_k \geq r \right) \nonumber \\
& \leq \exp\left(-t r + \frac{t^2}{8} \, \sum_{k=1}^n d_k^2 \right),
\quad \forall \, t \geq 0.
\label{eq:parametric_McDiarmid}
\end{align}
The choice $t = 4r \left(\sum_{k=1}^n d_k^2\right)^{-1}$ minimizes the expression
in \eqref{eq:parametric_McDiarmid}, so
\begin{equation}
\pr\Big(f(X^n) - \expectation[f(X^n)] \geq r\Big)
\leq \exp \left(-\frac{2 r^2}{\sum_{k=1}^n d_k^2} \right).
\label{eq: one-sided McDiarmid's inequality}
\end{equation}
By replacing $f$ with $-f$, it follows that this bound is also valid for the probability
$\pr\bigl(f(X^n) - \expectation[f(X^n)] \leq -r \bigr)$, so
\begin{align*}
&\Pr\Big(\Big|f(X^n) - \expectation[f(X^n)]\Big| \geq r\Big) \nonumber \\
&= \Pr\Big(f(X^n) - \expectation[f(X^n)] \geq r\Big) +
\pr\Big(f(X^n) - \expectation[f(X^n)] \leq -r \Big) \nonumber \\
& \le 2 \exp \left(-\frac{2 r^2}{\sum_{k=1}^n d_k^2} \right)
\end{align*}
which gives the bound in \eqref{eq: McDiarmid's inequality}.
\end{proof}

\begin{example}  \label{example: concentration of L - Alon and Spencer's book}
A nice example from \cite[Section~7.5]{AlonS_tpm3}
is revisited in the following. The concentration inequality that was
obtained in \cite[Theorem~7.5.1]{AlonS_tpm3}, via the Azuma-Hoeffding
inequality, is improved in this example by applying McDiarmid's inequality
(Theorem~\ref{theorem: McDiarmid's inequality}).

Let $g \colon \{1, \ldots, n\} \to \{1, \ldots, n\}$ be chosen uniformly at
random from all $n^n$ such possible functions. Let $L(g)$ denote the number of values
$y \in \{1, \ldots, n\}$ for which the equation $g(x)=y$ has no solution (i.e.,
$g(x) \neq y$ for every $x \in \{1, \ldots, n\}$).
By the linearity of the expectation, we have
$\expectation[L(g)] = n \left(1-\frac{1}{n}\right)^n.$
Consequently, for every $n \in \naturals$,
\begin{align} \label{eq: tight bounds on the expected value of L}
\frac{n-1}{e} < \expectation[L(g)] < \frac{n}{e}.
\end{align}
The right-hand side of \eqref{eq: tight bounds on the expected value of L}
holds due to the fact that the sequence
$\{\bigl(1-\frac{1}{n}\bigr)^n\}_{n \in \naturals}$
is monotonic increasing, and it converges to $\frac{1}{e}$;
the left-hand side of \eqref{eq: tight bounds on the expected value of L}
can be verified as follows:
\begin{align*}
\expectation[L(g)] & = n \left(1-\frac{1}{n}\right)
\cdot \left(1-\frac{1}{n}\right)^{n-1} \\[0.1cm]
& = \frac{n-1}{\left( 1 + \frac{1}{n-1} \right)^{n-1}} \\
& > \frac{n-1}{e}
\end{align*}
where the last inequality holds since the sequence
$\{\bigl(1+\frac{1}{n}\bigr)^n\}_{n \in \naturals}$
is monotonic increasing, and it converges to $e$.
Hence, \eqref{eq: tight bounds on the expected value of L}
provides tight bounds on $\expectation[L(g)]$, which scale
linearly with $n$.

In \cite[Section~7.5]{AlonS_tpm3}, the following approach
implies a concentration inequality for $L(g)$ around its
expected value. Let us construct a martingale sequence
$\{X_k, \cF_k\}_{k=0}^n$ (see Fact~\ref{fact: construction of martingales})
by $$X_k = \expectation[L(g) \, | \, \cF_k], \quad \forall
\, k \in \{0, \ldots, n\}$$ with the natural filtration
$\cF_k = \sigma\bigl(g(1), \ldots, g(k)\bigr)$
which denotes the $\sigma$-algebra that is generated by the first
$k$ values of the random function $g$, for $k \in \{1, \ldots, n\}$,
and $\cF_0 = \{\emptyset, \{1, \ldots, n\}\}$
is the minimal $\sigma$-algebra that only includes the empty set and the
probability space. By construction, $X_0 = \expectation[L(g)]$ and $X_n = L(g)$.
Since a modification of one value of $g$ cannot change $L(g)$ by more than~1,
it follows that $|X_k - X_{k-1}| \leq 1$ for every $k \in \{1, \ldots, n\}$.
From the Azuma-Hoeffding inequality and
\eqref{eq: tight bounds on the expected value of L},
it follows that
\begin{align} \label{AH inequality for L}
\pr \left( \Bigl|L(g) - \frac{n}{e}\Bigr| > \alpha \sqrt{n} + 1 \right)
\leq 2 \exp\left(-\frac{\alpha^2}{2}\right), \quad \forall \, \alpha > 0.
\end{align}
This concentration result, as stated in \cite[Theorem~7.5.1]{AlonS_tpm3},
can be improved as follows:
let $f \colon \{1, \ldots, n\}^n \to \{1, \ldots, n\}$ be defined by
$L(g) \triangleq f\bigl( g(1), \ldots, g(n) \bigr)$
so, the function $f$ maps the $n$-length vector $(g(1), \ldots, g(n))$
to the number of elements $y \in \{1, \ldots, n\}$ where
$g(x) \neq y$ for every $x \in \{1, \ldots, n\}$.
Since by assumption $g(1), \ldots, g(n)$ are independent
random variables, the variation of $f$ with respect to each of its arguments
(while all the other $n-1$ arguments of $f$ are kept fixed) is no more than~1.
Consequently, from McDiarmid's inequality,
\begin{align} \label{McDiarmid's inequality for L}
\pr \left( \Bigl|L(g) - \frac{n}{e}\Bigr| > \alpha \sqrt{n} + 1 \right)
\leq 2 \exp\bigl(-2 \alpha^2\bigr), \quad \forall \, \alpha > 0,
\end{align}
which implies that the exponent of the concentration inequality
\eqref{AH inequality for L} is improved by a factor of~4.
\end{example}

\begin{example}
Let $B$ be a normed space, and $\{\underline{v}_k\}_{k=1}^n$  be
$n$ vectors in $B$. Let $\{\Theta_k\}_{k=1}^n$ be independent
${\rm Bernoulli}\bigl(\frac{1}{2}\bigr)$ random variables with
$\pr(\Theta_k = 1) = \pr(\Theta_k = -1) = \frac{1}{2}$, and let
$X = \Bigl\| \sum_{k=1}^n \Theta_k \, \underline{v}_k \Bigr\|.$
By setting
$$f(\theta_1, \ldots, \theta_n) = \left\| \sum_{k=1}^n \theta_k \,
\underline{v}_k \right\|, \quad  \forall \, \theta_k \in \{-1, +1\},
\; k \in \{1, \ldots, n\}$$
the variation of $f$ with respect to its $k$-th argument is upper
bounded by $2 \| \underline{v}_k\|$. Consequently, since $\{\Theta_k\}$
are independent, it follows from McDiarmid's inequality that
$$\pr\bigl( | X - \expectation[X] | \ge \alpha \bigr)
\leq 2 \exp\left(-\frac{\alpha^2}{2 \sum_{k=1}^n \|\underline{v}_k\|^2}
\right), \quad \forall \, \alpha > 0.$$
\end{example}

\begin{remark}
Due to the large
applicability of McDiarmid's inequality, there is an interest to improve
this inequality for sub-classes of Lipschitz functions of independent random
variables. An improvement of this inequality for separately Lipschitz functions
of independent random variables has been recently derived in \cite{Rio_ecp13}
(see also a recent follow-up paper in \cite{DedeckerF14}).
\end{remark}

\subsection{Hoeffding's inequality and its improved versions}

The following concentration inequality for sums of independent and bounded random variables,
originally due to Hoeffding \cite[Theorem~2]{Hoeffding}, can be viewed as a special
case of McDiarmid's inequality:

\begin{theorem}[Hoeffding's inequality]
Let $\{U_k\}_{k=1}^n$ be a sequence of independent and bounded random variables where,
for $k \in \{1, \ldots, n\}$, $U_k \in [a_k, b_k]$ holds a.s. for some finite
constants $a_k, b_k \in \reals$ ($a_k < b_k$).
Let $\mu_n \triangleq \sum_{k=1}^n \expectation[U_k]$. Then,
\begin{equation}
\pr\left( \left| \sum_{k=1}^n U_k - \mu_n \right| \geq r \right)
\leq 2 \exp\left(-\frac{2 r^2}{\sum_{k=1}^n (b_k-a_k)^2} \right), \quad \forall \, r \geq 0.
\label{eq: Hoeffding inequality}
\end{equation}
\label{theorem: Hoeffding inequality}
\end{theorem}
\begin{proof}
Apply Theorem~\ref{theorem: McDiarmid's inequality} to the function
$$f(u^n) \triangleq \sum_{k=1}^n u_k, \qquad
\forall \, u^n \in \prod_{k=1}^n [a_k, b_k].$$

An alternative elementary proof combines the Chernoff bound with Lemma~\ref{lm:Hoeffding}
to get
\begin{align} \label{eq: Hoeffding - step1}
& \pr\left( \sum_{k=1}^n U_k - \mu_n \geq r \right) \nonumber \\
& = \pr \left( \sum_{k=1}^n \bigl( U_k - \expectation[U_k] \bigr) \geq r \right) \nonumber \\
& \leq \exp(-tr) \, \expectation \left[ \exp\Biggl( t \sum_{k=1}^n \bigl( U_k - \expectation[U_k] \bigr)
\Biggr) \right]  \qquad \forall \, t \geq 0 \nonumber \\
& = \exp(-tr) \, \prod_{k=1}^n \expectation \left[ \exp\Bigl( t \bigl( U_k - \expectation[U_k] \bigr)
\Bigr) \right] \nonumber \\
& \leq \exp(-tr) \, \prod_{k=1}^n \exp\left( \frac{t^2 (b_k-a_k)^2}{8} \right) \nonumber \\
& = \exp\Bigl(-tr + \frac{t^2}{8} \sum_{k=1}^n (b_k - a_k)^2 \Bigr).
\end{align}
Optimization of the right-hand side of \eqref{eq: Hoeffding - step1}
with respect to $t$ gives $$t = \frac{4r}{\sum_{k=1}^n (b_k - a_k)^2}$$ and
its substitution into \eqref{eq: Hoeffding - step1} yields that, for every
$r \geq 0$,
$$ \pr\left( \sum_{k=1}^n U_k - \mu_n \geq r \right) \leq
\exp\left(-\frac{2 r^2}{\sum_{k=1}^n (b_k-a_k)^2} \right).$$
The same bound holds for $ \pr\left( \sum_{k=1}^n U_k - \mu_n \leq -r \right)$,
which leads to the inequality in \eqref{eq: Hoeffding inequality}.
\end{proof}

Recall that a key step in the proof of McDiarmid's inequality is to invoke Hoeffding's
lemma (Lemma~\ref{lm:Hoeffding}). However, a careful look at the proof of Lemma~\ref{lm:Hoeffding}
reveals a potential source of slack in the bound
$$
\ln \expectation\Bigl[ \exp\Bigl(t(U-\expectation[U])\Bigr) \Bigr] \le \frac{t^2(b-a)^2}{8}
$$
--- namely, that this bound is the same regardless of the {\em location} of the mean
$\expectation[U]$ relative to the endpoints of the interval $[a,b]$. As it turns out,
one does indeed obtain an improved version of Hoeffding's inequality by making use of
this information. An improved version of Hoeffding's inequality was derived by
Kearns and Saul \cite{KS_1998}, and it has been recently further improved by
Berend and Kontorovich \cite{Berend_Kontorovich_missing_mass_2012}. The following
improvement of Hoeffding's inequality (Lemma~\ref{lm:Hoeffding}) is obtained in
\cite{Berend_Kontorovich_missing_mass_2012}:

\begin{lemma}[Berend and Kontorovich]\label{lm:Berend and Kontorovich}
Let $U$ be a real-valued random variable, such that $U \in [a,b]$ a.s.\ for
finite $a < b$. Then, for every $t \ge 0$,
\begin{align}
\label{eq:Berend and Kontorovich lemma}
\expectation\left[ \exp\big(t(U- \expectation U)\big)\right]
\le \exp \left(c_{\text{BK}}(p) \, t^2(b-a)^2\right)
\end{align}
where
\begin{align}  \label{eq:c of Berend and Kontorovich inequality}
c_{\text{BK}}(p) &= \begin{cases}
0, & \text{if $p = 0$} \\[0.2cm]
\dfrac{1-2p}{4\ln\left(\frac{1-p}{p}\right)}, & \text{if $0 < p < \dfrac{1}{2}$} \\[0.6cm]
\dfrac{p(1-p)}{2}, & \text{if $\dfrac{1}{2} \leq p \leq 1$}
\end{cases}
\end{align}
with
\begin{align}  \label{eq:p of Berend and Kontorovich inequality}
p = \frac{\expectation[U] - a}{b-a}.
\end{align}
\end{lemma}

\begin{proof}
Recall the definition of $H_p(\lambda)$ in \eqref{eq:Hoeffding_function}. We
deviate from the proof of Lemma~\ref{lm:Hoeffding} at the point where the bound
$H_p(\lambda) \le \frac{\lambda^2}{8}$ in \eqref{eq:bound on Hoeffding_function}
is replaced by the improved bound
\begin{align}  \label{eq:improved bound on Hoeffding_function}
H_p(\lambda) \leq c_{\text{BK}}(p) \, \lambda^2, \qquad \forall \, \lambda \ge 0,
\; p \in [0,1].
\end{align}
where $c_{\text{BK}}(p)$ is introduced in
\eqref{eq:c of Berend and Kontorovich inequality}; for a proof of
\eqref{eq:improved bound on Hoeffding_function}, the reader is referred
to the proofs of \cite[Theorem~3.2]{Berend_Kontorovich_missing_mass_2012} and
\cite[Lemma~3.3]{Berend_Kontorovich_missing_mass_2012}.
\end{proof}

\begin{remark}
The bound on the right-hand side of \eqref{eq:Berend and Kontorovich lemma}
depends on the location of $\expectation[U]$ in the interval $[a,b]$, and it
therefore refines Hoeffding's inequality in Lemma~\ref{lm:Hoeffding}. The
worst case where $p=\frac{1}{2}$ (i.e., if $\expectation[U] = \frac{a+b}{2}$
is in the middle of the interval $[a,b]$) coincides however with Hoeffding's
inequality (since, from \eqref{eq:c of Berend and Kontorovich inequality},
$c_{\text{BK}}(p) = \frac{1}{8}$ if $p = \frac{1}{2}$). The bound on
$H_p(\lambda)$ in \eqref{eq:improved bound on Hoeffding_function} can be
weakened to
\begin{align}  \label{eq:Kearns-Saul bound on Hoeffding_function}
H_p(\lambda) \leq c_{\text{KS}}(p) \, \lambda^2, \qquad
\forall \, \lambda \in \reals, \; p \in [0,1]
\end{align}
where the abbreviation 'KS' on the right-hand side of
\eqref{eq:Kearns-Saul bound on Hoeffding_function}
stands for the Kearns-Saul inequality in \cite{KS_1998},
and it is given by
\begin{align}  \label{eq:c of Kearns-Saul inequality}
c_{\text{KS}}(p) &= \begin{cases}
0, & \text{if $p = 0, 1$} \\[0.2cm]
\dfrac{1}{8}, & \text{if $p = \frac{1}{2}$} \\[0.2cm]
\dfrac{1-2p}{4\ln\left(\frac{1-p}{p}\right)}, & \text{if $p \in (0,1) \setminus \{\dfrac{1}{2}\}$}.
\end{cases}
\end{align}
From \eqref{eq:c of Berend and Kontorovich inequality} and
\eqref{eq:c of Kearns-Saul inequality}, we have
\begin{align*}
c_{\text{BK}}(p) = c_{\text{KS}}(p), \qquad & \forall \, p \in \Bigl[0, \frac{1}{2}\Bigr]  \\
0 \le c_{\text{BK}}(p) \le c_{\text{KS}}(p) \le \frac{1}{8}, \qquad & \forall \, p \in [0, 1]
\end{align*}
where the equality $c_{\text{BK}}(p) = c_{\text{KS}}(p) = \frac{1}{8}$ holds if and only if $p=\frac{1}{2}$
(see Figure~\ref{Figure:c}). Note that
$$\lim_{p \rightarrow \frac{1}{2}} c_{\text{BK}}(p) = \lim_{p \rightarrow \frac{1}{2}} c_{\text{KS}}(p) = \frac{1}{8}$$
which implies the continuity of $c_{\text{BK}}(\cdot)$ and $c_{\text{KS}}(\cdot)$ over the interval $[0,1]$.
\begin{figure}[here!]  
\begin{center}
\epsfig{file=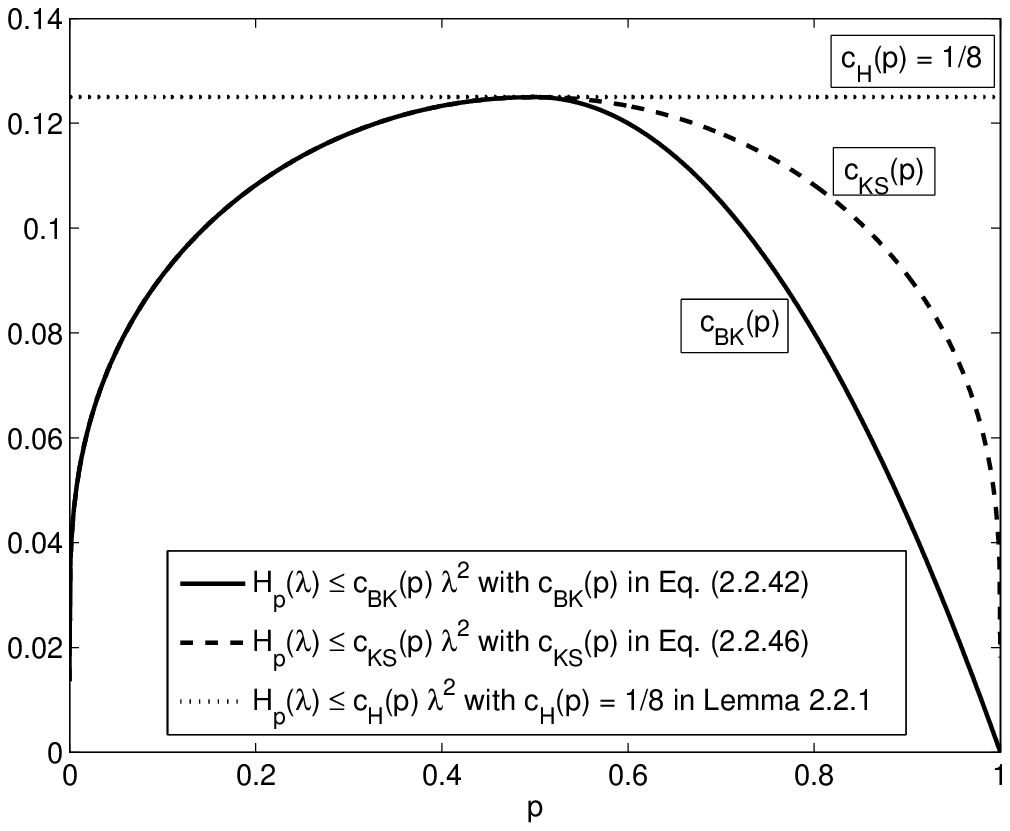,scale=0.5}
\end{center}
\caption{\label{Figure:c} A comparison between upper bounds on the Hoeffding
function $H_p(\lambda)$ in \eqref{eq:Hoeffding_function}; these bounds are of
the type $H_p(\lambda) \le c(p) \, \lambda^2$ for every $p \in [0,1]$ and
$\lambda \ge 0$ (see Eqs.~\eqref{eq:bound on Hoeffding_function},
\eqref{eq:improved bound on Hoeffding_function} and
\eqref{eq:Kearns-Saul bound on Hoeffding_function} with $c(p) = \frac{1}{8}$
or $c(p)$ in \eqref{eq:c of Berend and Kontorovich inequality} and
\eqref{eq:c of Kearns-Saul inequality}, respectively; these values of $c(p)$ correspond
to the dotted, solid and dashed lines, respectively, as a function of $p \in [0,1]$.)}
\end{figure}
\end{remark}

The improved bound in Lemma~\ref{lm:Berend and Kontorovich} (cf. Lemma~\ref{lm:Hoeffding})
leads to the following improvement of Hoeffding's inequality (Theorem~\ref{theorem: Hoeffding inequality}):

\begin{theorem}[Berend and Kontorovich inequality]
Let $\{U_k\}_{k=1}^n$ be a sequence of independent and bounded random variables such that,
for every $k \in \{1, \ldots, n\}$, $U_k \in [a_k, b_k]$ holds a.s. for some constants
$a_k, b_k \in \reals$. Let $\mu_n \triangleq \sum_{k=1}^n \expectation[U_k]$.
Then,
\begin{equation}
\pr\left( \left| \sum_{k=1}^n U_k - \mu_n \right| \geq r \right)
\leq 2 \exp\left(-\frac{r^2}{4 \sum_{k=1}^n c_k \, (b_k-a_k)^2} \right), \quad \forall \, r \geq 0
\label{eq: Berend and Kontorovich inequality}
\end{equation}
\label{theorem: Berend and Kontorovich inequality}
where $c_k \triangleq c_{\text{BK}}(p_k)$ (see \eqref{eq:c of Berend and Kontorovich inequality})
with
\begin{equation} \label{eq:p_k}
p_k = \frac{\expectation[U_k] - a_k}{b_k - a_k}.
\end{equation}
\end{theorem}
\begin{proof}
Inequality \eqref{eq: Berend and Kontorovich inequality} follows from a combination of
the Chernoff bound and Lemma~\ref{lm:Berend and Kontorovich} (similarly to the proof of
Theorem~\ref{theorem: Hoeffding inequality} that relies on the Chernoff bound and
Lemma~\ref{lm:Hoeffding}).
\end{proof}

A loosening of the bound in Theorem~\ref{eq: Berend and Kontorovich inequality},
by a replacement of $c_k \triangleq c_{\text{BK}}(p_k)$ with
$\widetilde{c}_k \triangleq c_{\text{KS}}(p_k)$
(see \eqref{eq:c of Berend and Kontorovich inequality},
\eqref{eq:p of Berend and Kontorovich inequality} and
\eqref{eq:c of Kearns-Saul inequality}), gives the Kearns-Saul
inequality in \cite{KS_1998}:
\begin{theorem}[Kearns--Saul inequality]
Let $\{U_k\}_{k=1}^n$ be a sequence of independent and bounded random variables such that,
for every $k \in \{1, \ldots, n\}$, $U_k \in [a_k, b_k]$ holds a.s.\ for some constants
$a_k, b_k \in \reals$. Let $\mu_n \triangleq \sum_{k=1}^n \expectation[U_k]$.
Then, for every $r \geq 0$,
\begin{equation}
\pr\left( \left| \sum_{k=1}^n U_k - \mu_n \right| \geq r \right)
\leq 2 \exp\left(-\frac{r^2}{4 \sum_{k=1}^n c_k \, (b_k-a_k)^2} \right)
\label{eq: Kearns-Saul inequality}
\end{equation}
where $c_k = c_{\text{KS}}(p_k)$ with $c_{\text{KS}}(\cdot)$ and $p_k$
in \eqref{eq:c of Kearns-Saul inequality} and \eqref{eq:p_k},
respectively. The bound in \eqref{eq: Kearns-Saul inequality}
improves Hoeffding's inequality in \eqref{eq: Hoeffding inequality} unless
$p_k = \frac{1}{2}$ (i.e., if $\expectation[U_k] = \frac{a_k + b_k}{2}$)
for every $k \in \{1, \ldots, n\}$; in the latter case, both bounds coincide.
\label{theorem: Kearns-Saul inequality}
\end{theorem}
An information-theoretic proof of the basic inequality
that leads to the Kearns-Saul inequality is given in
Section~\ref{subsection: Gaussian concentration and $T_1$ inequalities} of
Chapter~\ref{chapter: entropy method}.

The reader is referred to \cite{From_and_Swift_2011} for another refinement
of Hoeffding's inequality that is not covered in this section.

\section{Refined versions of the Azuma--Hoeffding inequality}
\label{section: Refined Versions of the Azuma--Hoeffding Inequality}

The following section considers generalized and refined versions of
the Azuma-Hoeffding inequality (see
Sections~\ref{subsection: A generalization of the Azuma--Hoeffding inequality}
and~\ref{subsection: Martingales with a uniform bound on the differences}).
A derivation of one-sided inequalities for sub and super martingales is
considered as well (see
Section~\ref{subsection: concentration inequalities for sub and super martingales}).

\subsection{A generalization of the Azuma--Hoeffding inequality}
\label{subsection: A generalization of the Azuma--Hoeffding inequality}
The following theorem generalizes the Azuma-Hoeffding inequality for
real-valued martingale sequences $\{X_k, \mathcal{F}_k\}_{k=0}^n$
with bounded differences in the case where the differences
$\xi_k \triangleq X_k - X_{k-1}$ are bounded between the endpoints
of {\em asymmetric} intervals around zero. Furthermore, it states
that the same bound holds not only for the probability of the event
where $|X_n - X_0| \ge r$, for some $r \ge 0$, but also for the
probability of the more likely event where there exists an index
$k \in \{1, \ldots, n\}$ such that $|X_k - X_0| \ge r$; the idea
that strengthens the bound to hold for the latter event applies to
all the concentration inequalities derived in this chapter.

\begin{theorem}[A generalization of the Azuma-Hoeffding Inequality]
Let $\{X_k, \mathcal{F}_k\}_{k=0}^n$ be a real-valued martingale
sequence. Suppose that $a_1, b_1, \ldots, a_n, b_n$ are constants
such that $a_k \le X_k - X_{k-1} \le b_k$ holds a.s. for every
$k \in \{1, \ldots, n\}$. Then, for every $r \ge 0$,
\begin{equation}
\label{eq: my generalization of the Azuma-Hoeffding Inequality}
\pr \biggl( \max_{k \in \{1, \ldots, n\}} |X_k - X_0| \ge r \biggl)
\le 2 \exp \left(-\frac{r^2}{4 \sum_{k=1}^n c_k \, (b_k - a_k)^2} \right)
\end{equation}
where $c_k = c(p_k)$ with $$p_k = -\frac{a_k}{b_k-a_k} \in [0,1], \quad \forall
\, k \in \{1, \ldots, n\},$$ and $c(\cdot) = c_{\text{BK}}(\cdot)$ is introduced
in \eqref{eq:c of Berend and Kontorovich inequality} over the interval $[0,1]$.
\label{theorem: my generalization of the Azuma-Hoeffding Inequality}
\end{theorem}

\begin{remark}
In the following, it is shown that the Azuma-Hoeffding inequality
(Theorem~\ref{theorem: Azuma--Hoeffding inequality}) is a special case of
Theorem~\ref{theorem: my generalization of the Azuma-Hoeffding Inequality}.
Consider the setting in the Azuma-Hoeffding inequality where the intervals
$[a_k, b_k]$ in
Theorem~\ref{theorem: my generalization of the Azuma-Hoeffding Inequality}
are symmetric around zero, i.e., $b_k = -a_k = d_k$ for every
$k \in \{1, \ldots, n\}$, and for some non-negative reals $d_1, \ldots, d_n$.
In this special case, it follows from
Theorem~\ref{theorem: my generalization of the Azuma-Hoeffding Inequality}
that $p_k = \frac{1}{2}$, and $c(p_k) = \frac{1}{8}$ for every $k$. Hence,
from \eqref{eq: my generalization of the Azuma-Hoeffding Inequality}, we have
\begin{align*}
\pr(|X_n - X_0| \geq r)
& \le \pr \biggl( \max_{k \in \{1, \ldots, n\}} |X_k - X_0| \ge r \biggl) \\
& \le  2 \exp\left(-\frac{r^2}{2 \sum_{k=1}^n d_k^2}\right), \quad \forall \, r \ge 0,
\end{align*}
which gives the Azuma-Hoeffding inequality in
\eqref{eq: Azuma--Hoeffding concentration inequality - general case}.
\end{remark}

\begin{proof}
In the following, the proof of the Azuma-Hoeffding inequality is modified
for a derivation of the generalized inequality in
\eqref{eq: my generalization of the Azuma-Hoeffding Inequality}. As a first
step, the equality in \eqref{eq: union of disjoint events} is replaced by
the equality
\begin{align}
& \pr \biggl( \max_{k \in \{1, \ldots, n\}} |X_k - X_0| \ge r \biggl) \nonumber \\
& = \pr \biggl( \max_{k \in \{1, \ldots, n\}} (X_k - X_0) \ge r \biggl) +
\pr \biggl( \max_{k \in \{1, \ldots, n\}} (X_0 - X_k) \ge r \biggl).
\label{eq: 2nd union of disjoint events}
\end{align}
Let $\xi_k = X_k - X_{k-1}$ be the differences of the martingale sequence,
then $\expectation[\xi_k | \mathcal{F}_{k-1}]=0$ and $a_k \le \xi_k \le b_k$
hold a.s. for every $k \in \{1, \ldots, n\}$.

Recall that a composition of a convex function with a martingale gives a
sub-martingale with respect to the same filtration
(see Theorem~\ref{theorem: mappings of martingales or sub/ super martingales}).
Since $\{X_k - X_0, \mathcal{F}_k\}_{k=1}^n$ is a martingale and
$f_t(x)=\exp(tx)$ is a convex function over $\reals$ for every
$t \in \reals$, it follows that
$\left\{\exp\bigl(t(X_k-X_0)\bigr), \mathcal{F}_k\right\}_{k=1}^n$ is a
sub-martingale for every $t \in \reals$. From the maximal inequality for
sub-martingales (a.k.a. the Doob-Kolmogorov inequality), which states
that if $\{Y_k, \mathcal{F}_k\}_{k=1}^n$ is a sub-martingale then
$$\pr\Bigl( \max_{1 \leq k \leq n} Y_k \geq \lambda \Bigr) \leq
\frac{\expectation[|Y_n|]}{\lambda}, \quad \forall \, \lambda > 0$$
(see, e.g., \cite[Theorem~14.3.1]{Rosenthal_book}), it follows that for
every $t \ge 0$
\begin{align}
& \pr \biggl( \max_{k \in \{1, \ldots, n\}} (X_k - X_0) \ge r \biggl) \nonumber \\
& = \pr \biggl( \max_{k \in \{1, \ldots, n\}} \exp\bigl(t(X_k - X_0)\bigr)
\ge \exp(tr) \biggl) \nonumber \\
& \le \exp(-tr) \, \expectation\left[\exp \bigl(t (X_k - X_0) \bigr) \right] \nonumber \\
& = \exp(-tr) \, \expectation\left[\exp \left(t \sum_{k=1}^n \xi_k \right) \right].
\label{eq: Doob's maximal inequality}
\end{align}
Hence, by applying the maximal inequality for sub-martingales instead of the
Chernoff bound, inequality \eqref{eq: Chernoff's inequality} is replaced with
the stronger result in \eqref{eq: Doob's maximal inequality}.
Similarly to the proof of the Azuma-Hoeffding inequality, by the law
of iterated expectations, we have from \eqref{eq: smoothing theorem}
\begin{equation*}
\expectation \biggl[ \exp \biggl(t \sum_{k=1}^n \xi_k \biggr) \biggr]
= \expectation \Biggl[ \exp \biggl(t \sum_{k=1}^{n-1} \xi_k \biggr) \,
\expectation \bigl[ \exp(t \xi_n) \, | \, \mathcal{F}_{n-1} \bigr] \Biggr].
\end{equation*}
In the following, Lemma~\ref{lm:Berend and Kontorovich} is applied
with the conditioning on $\cF_{n-1}$. Based on the information that
$\expectation[\xi_n|\cF_{n-1}] = 0$ and $\xi_n \in [a_n,b_n]$ a.s.,
it follows that
\begin{equation}
\label{eq:an improved bound on the conditional moment generating function of xi_n}
\expectation \bigl[ \exp(t \xi_n) \, | \, \mathcal{F}_{n-1} \bigr] \le
\exp \Bigl(c_n (b_n-a_n)^2 t^2 \Bigr)
\end{equation}
where $c_n \triangleq c_{\text{BK}}(p_n)$ is given in
\eqref{eq:c of Berend and Kontorovich inequality} with
(see \eqref{eq:p of Berend and Kontorovich inequality})
\begin{equation*}
p_n = \frac{\expectation[\xi_n | \mathcal{F}_{n-1}] - a_n}{b_n - a_n} = -\frac{a_n}{b_n-a_n}.
\end{equation*}
(If $b_n = -a_n \triangleq d_n$ for a non-negative real number $d_n$ then $p_n = \frac{1}{2}$
and $c_n = c_{\text{BK}}(p_n) = \frac{1}{8}$, and
inequality~\eqref{eq:an improved bound on the conditional moment generating function of xi_n}
is particularized to \eqref{eq:bound on the conditional moment generating function of xi_n};
the latter inequality can be obtained by applying Hoeffding's lemma, as in the proof of the
Azuma-Hoeffding lemma.) Continuing recursively in a similar manner, in parallel to
\eqref{eq:bound on the moment generating function of Xn-X0}, the quantity in
\eqref{eq: smoothing theorem} is upper-bounded by
\begin{equation*}
\expectation \biggl[ \exp \biggl(t \sum_{k=1}^n \xi_k \biggr)
\biggr] \leq \exp \left(t^2 \, \sum_{k=1}^n c_k (b_k-a_k)^2 \right).
\end{equation*}
The combination of this bound with \eqref{eq: Doob's maximal inequality}
gives that, for every $r \ge 0$,
\begin{align}\label{eq:parametric_generalized inequality}
& \pr \biggl( \max_{k \in \{1, \ldots, n\}} (X_k - X_0) \ge r \biggl) \nonumber \\
& \leq \exp\left(-tr + t^2 \, \sum_{k=1}^n c_k (b_k-a_k)^2 \right), \quad \forall \, t \geq 0.
\end{align}
An optimization with respect to the non-negative parameter $t$ gives
$$t = \frac{r}{2 \sum_{k=1}^n c_k (b_k-a_k)^2}$$ and the substitution of this
optimized value into \eqref{eq:parametric_generalized inequality} yields that,
for every $r \ge 0$,
\begin{equation}
\pr \biggl( \max_{k \in \{1, \ldots, n\}} (X_k - X_0) \ge r \biggl) \leq
\exp \left(-\frac{r^2}{4 \sum_{k=1}^n c_k (b_k-a_k)^2} \right).
\label{eq: one-sided generalization of the Azuma-Hoeffding inequality}
\end{equation}
The same bound as in
\eqref{eq: one-sided generalization of the Azuma-Hoeffding inequality}
holds for $\pr \biggl( \max_{k \in \{1, \ldots, n\}} (X_0 - X_k) \ge r \biggl)$.
Using these two bounds on the right-hand side of
\eqref{eq: 2nd union of disjoint events} completes the proof of
Theorem~\ref{theorem: my generalization of the Azuma-Hoeffding Inequality}.
\end{proof}

\begin{example} \label{example:feedback scheme}
The advantage of the inequality in
Theorem~\ref{theorem: my generalization of the Azuma-Hoeffding Inequality}
over the Azuma-Hoeffding inequality is exemplified in the following.

Let $\{X_k\}$ be a real-valued sequence of random variables, defined on a
probability space $(\Omega, \mathcal{F}, \pr)$, that is generated
by the recursion
\begin{equation} \label{eq: feedback scheme for this example}
X_k = X_{k-1} + \xi_k, \quad \forall \, k \ge 1
\end{equation}
where $X_k = 0$ for $k \le 0$. The
differences $\xi_k = X_k - X_{k-1}$ are defined as follows: Let $g \colon
\reals^m \rightarrow [0,1]$ be an arbitrary measurable function,
for $m \ge 1$, and let $\{\Theta_k\}$ be i.i.d.\ random variables where
for some $\alpha \in (0,1]$
\begin{equation} \label{eq: distribution of Theta}
\pr(\Theta_k = 1) = \frac{1}{1+\alpha}, \qquad
\pr\left(\Theta_k = -\frac{1}{\alpha}\right) = \frac{\alpha}{1+\alpha}
\end{equation}
and $\Theta_k$ is independent of $X_{k-1}, X_{k-2}, \ldots$ for every
$k \ge 1$. Let us define
\begin{equation} \label{eq: definition of xi_k in this example}
\xi_k = \frac{\Theta_k \, g(X_{k-1}, \ldots, X_{k-m})}{k^2}, \quad \forall \, k \ge 1.
\end{equation}
The sequence $\{X_k\}$ is generated by the following feedback scheme:
\begin{figure}[here!]  
\begin{center}
\epsfig{file=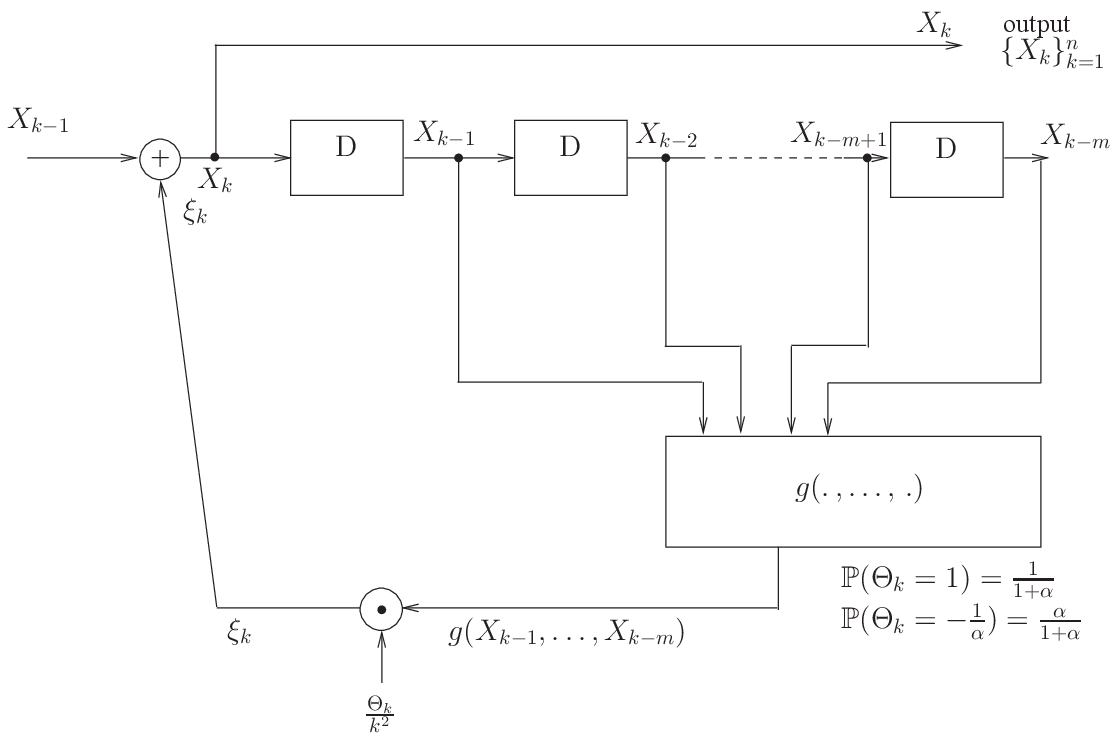,scale=0.63}
\end{center}
\caption{\label{Figure:feedback_scheme_example}
The feedback scheme in Example~\ref{example:feedback scheme} (see
\eqref{eq: feedback scheme for this example}--\eqref{eq: definition of xi_k in this example}).}
\end{figure}

Let $\cF_k = \sigma(X_0, X_1, \ldots, X_k)$, for $k \ge 0$, be the $\sigma$-algebra
that is generated by the random variables $X_0, X_1, \ldots, X_k$ (recall its
definition in Example~\ref{example: sum of zero-mean and independent Rvs}), so $\{\cF_k\}$
is a filtration. The random variable $X_k$ is $\cF_k$-measurable for every $k \ge 0$,
so $\xi_k = X_k - X_{k-1}$ is also $\cF_k$-measurable (since $\cF_{k-1} \subseteq \cF_k$).
We have $X_k \in L^1(\Omega, \cF_k, \pr)$ since
\begin{align*}
\expectation[|X_k|] & \le \sum_{i=1}^k \expectation[|\xi_i|]
\le \sum_{i=1}^k \frac{\expectation[|\Theta_i|]}{i^2}
= \frac{2}{1+\alpha} \sum_{i=1}^k \frac{1}{i^2} \\
& < \frac{2}{1+\alpha} \sum_{i=1}^{\infty} \frac{1}{i^2}
= \frac{\pi^2}{3(1+\alpha)} < \infty
\end{align*}
where the last equality holds since $\zeta(2) \triangleq \sum_{k=1}^{\infty} \frac{1}{k^2} = \frac{\pi^2}{6}$.
Furthermore,
\begin{align*}
\expectation[\xi_k | \cF_{k-1}]
& = \frac{1}{k^2} \, \expectation[\Theta_k \, g(X_{k-1}, \ldots, X_{k-m}) \, | \, \cF_{k-1}] \\
& = \frac{g(X_{k-1}, \ldots, X_{k-m}) \; \expectation[\Theta_k | \cF_{k-1}]}{k^2} \\
& = \frac{g(X_{k-1}, \ldots, X_{k-m}) \; \expectation[\Theta_k]}{k^2} = 0
\end{align*}
where the third equality holds since $\Theta_k$ is independent of the past inputs $X_{k-1}, X_{k-2}, \ldots$
for $k \ge 1$. It therefore follows that $\{X_k, \cF_k\}$ is a martingale. From
\eqref{eq: feedback scheme for this example}--\eqref{eq: definition of xi_k in this example}, together with
the assumption that $0 \le g \le 1$, it follows that the differences of the martingale sequence
(i.e., $\xi_k = X_k - X_{k-1}$ for $k \ge 1$) satisfy the inequality
\begin{equation} \label{eq:bounds on the martingale differences in this example}
-\frac{1}{\alpha k^2} \triangleq a_k \le \xi_k \le b_k \triangleq \frac{1}{k^2}, \quad \forall \, k \ge 1.
\end{equation}

From the Azuma-Hoeffding inequality, since $|\xi_k| \le \frac{1}{\alpha k^2}$ for $\alpha \in (0,1]$
(see \eqref{eq:bounds on the martingale differences in this example}), it follows that
\begin{align}
\pr \left( \max_{1 \le k \le n} | X_k| \ge r \right)
& \le 2 \exp\left(-\frac{\alpha^2 r^2}{2 \sum_{k=1}^n \frac{1}{k^4}} \right) \nonumber \\[0.1cm]
& < 2 \exp\left( -\frac{45 \alpha^2 r^2}{\pi^4} \right), \quad \forall \, r>0
\label{eq: AH inequality for this example}
\end{align}
where the last inequality holds since $\zeta(4) \triangleq \sum_{k=1}^{\infty} \frac{1}{k^4} = \frac{\pi^4}{90}$.
On the other hand, from Theorem~\ref{theorem: my generalization of the Azuma-Hoeffding Inequality} and
\eqref{eq:bounds on the martingale differences in this example}, we have for every $k \in \naturals$
\begin{align*}
& a_k = -\frac{1}{\alpha k^2}, \quad b_k = \frac{1}{k^2}, \\
& p_k = -\frac{a_k}{b_k - a_k} = \frac{1}{1+\alpha} \in \Bigl[\frac{1}{2}, 1\Bigr),
\quad \forall \, \alpha \in (0,1], \\
& c_k = c(p_k) = \frac{p_k (1-p_k)}{2} = \frac{\alpha}{2(1+\alpha)^2}, \\
\Rightarrow & \sum_{k=1}^n c_k (b_k - a_k)^2 = \frac{1}{2 \alpha} \sum_{k=1}^n \frac{1}{k^4}
< \frac{\zeta(4)}{2 \alpha} = \frac{\pi^4}{180 \alpha}.
\end{align*}
Consequently, it follows from \eqref{eq: my generalization of the Azuma-Hoeffding Inequality} that
\begin{align}
\pr \left( \max_{1 \le k \le n} | X_k| \ge r \right)
\le 2 \exp\left( -\frac{45 \alpha r^2}{\pi^4} \right), \quad \forall \, r>0.
\label{eq: generalized AH inequality for this example}
\end{align}
A comparison of the bounds in \eqref{eq: AH inequality for this example} and
\eqref{eq: generalized AH inequality for this example} shows an improvement by
a factor of $\frac{1}{\alpha}$ in the exponent of the latter bound. This shows
the advantage of the concentration inequality in
Theorem~\ref{theorem: my generalization of the Azuma-Hoeffding Inequality}
over the Azuma-Hoeffding inequality, where this improvement is more pronounced
as the value of $\alpha$ in \eqref{eq: distribution of Theta} becomes closer to
zero (which then enhances the asymmetry of the distribution of $\{\Theta_k\}$
in \eqref{eq: distribution of Theta}).

In the following, we examine numerically the bounds in \eqref{eq: AH inequality for this example}
and \eqref{eq: generalized AH inequality for this example}. Let us require
that the peak-to-average power ratio of $\{\Theta_k\}$ does not exceed a certain level,
e.g., 20~dB = 100. For $\alpha \in (0, 1]$,
$$ \| \Theta_k \|_{\infty} = \expectation[\Theta_k^2] = \frac{1}{\alpha}
\quad \Longrightarrow \quad \frac{\| \Theta_k \|_{\infty}^2}{\expectation[\Theta_k^2]}
= \frac{1}{\alpha}$$ hence, in order to satisfy this requirement, let
$\alpha = \alpha_{\min} = \frac{1}{100}$. Let us find the minimal value of $r$
such that each of the bounds in \eqref{eq: AH inequality for this example}
and \eqref{eq: generalized AH inequality for this example} assures that, irrespectively of $n$,
$$ \pr \left( \max_{1 \le k \le n} | X_k| \ge r \right) \leq \varepsilon \triangleq 10^{-10}.$$
The Azuma-Hoeffding inequality in \eqref{eq: AH inequality for this example}
gives $r = r_{\min} = 716.54$, whereas the improved bound in
\eqref{eq: generalized AH inequality for this example} implies that
$r = r_{\min} = 71.654$. The improved value of $r$ is reduced by a factor
of $\sqrt{\frac{1}{\alpha}} = 10$, so the concentration result for the
sequence $\{X_k\}$ is significantly strengthened by the use of
Theorem~\ref{theorem: my generalization of the Azuma-Hoeffding Inequality}.
\end{example}

\subsection{On martingales with uniformly bounded differences}
\label{subsection: Martingales with a uniform bound on the differences}

Example~\ref{example2} in the preceding section serves to
motivate a derivation of another improvement of the Azuma-Hoeffding
inequality with an additional constraint on the conditional variance of
the martingale sequence. In the following, assume that
$|X_k - X_{k-1}| \leq d$ holds a.s. for every $k$ (note that
$d$ does not depend on $k$, so it is a global bound on the differences
of the martingale). A new condition is added for the derivation
of the next concentration inequality, where it is assumed that a.s.
$$ \var(X_k \, | \, \mathcal{F}_{k-1}) =
\expectation\bigl[(X_k-X_{k-1})^2 \, | \, \mathcal{F}_{k-1}\bigr]
\leq \gamma d^2$$ for some constant $\gamma \in (0,1]$.

One of the disadvantages of the Azuma--Hoeffding inequality
(Theorem~\ref{theorem: Azuma--Hoeffding inequality}) and
McDiarmid's inequality (Theorem~\ref{theorem: McDiarmid's inequality})
is their insensitivity to the variance, which leads to suboptimal
exponents compared to the central limit theorem (CLT) and moderate
deviation principle (MDP).
The following theorem, which appears in
\cite{McDiarmid_bounded_differences_Martingales_1989}
(see also \cite[Corollary~2.4.7]{Dembo_Zeitouni}), makes use of
the variance:

\begin{theorem}
Let $\{X_k, \mathcal{F}_k\}_{k=0}^n$ be a discrete-time real-valued
martingale. Assume that, for some constants $d, \sigma > 0$, the
following two requirements are satisfied a.s.\ for every $k \in \{1,\ldots,n\}$:
\begin{eqnarray*}
&& | X_k - X_{k-1} | \leq d, \\
&& \var (X_k | \mathcal{F}_{k-1}) = \expectation \bigl[(X_k
- X_{k-1})^2 \, | \, \mathcal{F}_{k-1} \bigr] \leq \sigma^2
\end{eqnarray*}
Then, for every $\alpha \geq 0$,
\begin{equation}
\hspace*{-0.2cm} \pr(|X_n-X_0| \geq \alpha n) \leq 2 \exp\left(-n
\, H\bigg(\frac{\delta+\gamma}{1+\gamma} \bigg\|
\frac{\gamma}{1+\gamma}\bigg) \right) \label{eq: first refined
concentration inequality}
\end{equation}
where
\begin{equation}
\gamma \triangleq \frac{\sigma^2}{d^2}, \quad \delta \triangleq
\frac{\alpha}{d}  \label{eq: notation}
\end{equation}
and
\begin{equation}
H(p \| q) \triangleq p \ln\Bigl(\frac{p}{q}\Bigr) + (1-p)
\ln\Bigl(\frac{1-p}{1-q}\Bigr), \quad \forall \, p, q \in [0,1]
\label{eq: divergence}
\end{equation}
is the divergence between the ${\rm Bernoulli}(p)$ and ${\rm Bernoulli}(q)$
probability measures.
If $\delta>1$, the probability on the left-hand side of
\eqref{eq: first refined concentration inequality} is equal to zero.
\label{theorem: first refined concentration inequality}
\end{theorem}

\begin{proof}
The proof of this bound goes along the same lines as the proof of the
Azuma--Hoeffding inequality, up to \eqref{eq: smoothing theorem}.
The new ingredient in this proof is the use of the so-called Bennett's
inequality (see, e.g., \cite[Lemma~2.4.1]{Dembo_Zeitouni}), which
improves upon Lemma~\ref{lm:Hoeffding} by incorporating a bound on
the variance: Let $X$ be a real-valued random variable with
$\overline{x} = \expectation(X)$ and
$\expectation[(X-\overline{x})^2] \leq \sigma^2$ for some
$\sigma>0$. Furthermore, suppose that $X \leq b$ a.s. for
some $b \in \reals$. Then, for every $\lambda \geq 0$, Bennett's
inequality states that
\begin{equation}
\expectation\bigl[e^{\lambda X}\bigr] \leq
\frac{e^{\lambda \overline{x}} \left[(b-\overline{x})^2 e^{-\frac{\lambda
\sigma^2}{b-\overline{x}}}+\sigma^2
e^{\lambda(b-\overline{x})}\right]}{(b-\overline{x})^2+\sigma^2}.
\label{eq: Bennett's inequality for unconditional expectation}
\end{equation}
The proof of \eqref{eq: Bennett's inequality for unconditional expectation}
is provided in Appendix~\ref{appendix: proof of Bennett's inequality} for
completeness.

We now apply Bennett's inequality
\eqref{eq: Bennett's inequality for unconditional expectation} to the
conditional law of $\xi_k$ given the $\sigma$-algebra
$\mathcal{F}_{k-1}$. Since
$\expectation[\xi_k | \mathcal{F}_{k-1}] = 0$, $\var[\xi_k |
\mathcal{F}_{k-1}] \leq \sigma^2$ and $\xi_k \leq d$ a.s.\ for $k
\in \naturals$, we have
\begin{equation}
\expectation \left[ \exp(t \xi_k) \, | \, \mathcal{F}_{k-1}
\right] \leq \frac{\sigma^2 \exp(td) + d^2 \exp\left(-\frac{t
\sigma^2}{d}\right)}{d^2 + \sigma^2}, \qquad \text{a.s.}. \label{eq: Bennett's
inequality for the conditional law of xi_k}
\end{equation}
From \eqref{eq: smoothing theorem} and
\eqref{eq: Bennett's inequality for the conditional law of xi_k} it follows
that, for every $t \geq 0$,
\begin{eqnarray*}
\expectation \biggl[ \exp \biggl(t
\sum_{k=1}^n \xi_k \biggr) \biggr]
\leq \left(\frac{\sigma^2 \exp(td) + d^2
\exp\left(-\frac{t \sigma^2}{d}\right)}{d^2 + \sigma^2}\right)
\expectation \biggl[ \exp \biggl(t \sum_{k=1}^{n-1} \xi_k \biggr)
\biggr].
\end{eqnarray*}
Repeating this argument recursively, we conclude that, for every $t \geq 0$,
\begin{equation*}
\expectation \biggl[ \exp \biggl(t \sum_{k=1}^n \xi_k \biggr)
\biggr] \leq \left(\frac{\sigma^2 \exp(td)
+ d^2 \exp\left(-\frac{t \sigma^2}{d}\right)}{d^2 + \sigma^2}\right)^n.
\end{equation*}
Using the definition of $\gamma$ in \eqref{eq: notation}, we can rewrite this inequality as
\begin{equation}
\expectation \biggl[ \exp \biggl(t \sum_{k=1}^n
\xi_k \biggr) \biggr] \leq \left(\frac{\gamma \exp(td) +
\exp(-\gamma td)}{1+\gamma}\right)^n, \quad \forall \, t \geq 0.
\label{eq: important inequality used for the derivation of Theorem 2}
\end{equation}
Let $x \triangleq td$ (so $x \geq 0$). We can now use
\eqref{eq: important inequality used for the derivation of Theorem 2}
with the Chernoff bounding technique to get that for every $\alpha \geq 0$
(from the definition of $\delta$ in \eqref{eq: notation},
$\alpha t = \delta x$)
\begin{align}
& \pr(X_n-X_0 \geq \alpha n) \nonumber \\
& \leq \exp(-\alpha n t) \, \expectation \biggl[ \exp \biggl(t \sum_{k=1}^n
\xi_k \biggr) \biggr] \nonumber \\
& \leq \left(\frac{\gamma \exp\bigl((1-\delta)
x\bigr) + \exp\bigl(-(\gamma+\delta) x\bigr)}{1+\gamma}\right)^n,
\quad \forall \, x \geq 0. \label{eq: first concentration inequality
before the optimization over the non-negative parameter x}
\end{align}
Consider first the case where $\delta=1$ (i.e., $\alpha = d$).
Then \eqref{eq: first concentration inequality before the
optimization over the non-negative parameter x} becomes
\begin{equation*}
\pr(X_n-X_0 \geq d n) \leq \left( \frac{\gamma +
\exp\bigl(-(\gamma+1)x\bigr)}{1+\gamma} \right)^n, \quad \forall
\, x \geq 0
\end{equation*}
and the expression on the right-hand side is minimized in the limit
as $x \rightarrow \infty$. This gives the inequality
\begin{equation}
\pr(X_n-X_0 \geq d n) \leq \left( \frac{\gamma}{1+\gamma}
\right)^n. \label{eq: first concentration inequality for delta=1}
\end{equation}
Otherwise, if $\delta \in [0, 1)$, we minimize the base of
the exponent on the right-hand side of \eqref{eq: first
concentration inequality before the optimization over the
non-negative parameter x} with respect to\ the free parameter
$x \ge 0$. Setting the derivative of this exponent to zero
yields that the optimal value of $x$ is given by
\begin{equation}
x = \left(\frac{1}{1+\gamma}\right) \ln
\left(\frac{\gamma+\delta}{\gamma (1-\delta)}\right). \label{eq:
optimized non-negative x for the first refined concentration
inequality}
\end{equation}
Substituting \eqref{eq:
optimized non-negative x for the first refined concentration
inequality} into the right-hand side of
\eqref{eq: first concentration inequality before
the optimization over the non-negative parameter x} gives that,
for every $\alpha \ge 0$,
\begin{align}
 \pr(X_n-X_0 \geq \alpha n) \nonumber & \leq \left[ \left(\frac{\gamma+\delta}
{\gamma}\right)^{-\frac{\gamma+\delta}{1+\gamma}}
(1-\delta)^{-\frac{1-\delta}{1+\gamma}} \right]^n \nonumber \\
&  = \exp \left( -n \, H\bigg(\frac{\delta+\gamma}
{1+\gamma} \bigg\| \frac{\gamma}{1+\gamma}\bigg) \right)
\label{eq: one-sided concentration inequality of the first refined bound}
\end{align}
where $H(\cdot \| \cdot)$ is introduced in \eqref{eq: divergence}.
Finally, if $\delta > 1$ (i.e., $\alpha > d$), the exponent is equal to
$+\infty$. The application of inequality~\eqref{eq: one-sided
concentration inequality of the first refined bound} to the
martingale $\{-X_k, \mathcal{F}_k\}_{k=0}^{\infty}$ gives the same
upper bound for the other tail probability $\pr(X_n-X_0 \leq
-\alpha n)$. Overall, we get the bound \eqref{eq: first refined
concentration inequality}, which completes the proof of
Theorem~\ref{theorem: first refined concentration inequality}.
\end{proof}

\begin{remark}
The divergence (a.k.a. Kullback-Leibler distance or relative entropy)
between two probability measures $P$ and $Q$ is denoted, throughout this
monograph, by $D(P\|Q)$. The notation $H(p\|q)$ is used in \eqref{eq: divergence}
for the divergence in the special case where $P$ and $Q$ are
${\rm Bernoulli}(p)$ and ${\rm Bernoulli}(q)$, respectively. In this
case, where $P = {\rm Bernoulli}(p)$ and $Q = {\rm Bernoulli}(q)$, we
have $D(P\|Q) \triangleq H(p\|q)$.
\end{remark}

\noindent Here is an illustration of how one can use Theorem~\ref{theorem: first
refined concentration inequality} for getting better bounds in comparison to
the Azuma--Hoeffding inequality:
\begin{example}
Let $d>0$ and $\varepsilon \in (0, \frac{1}{2}]$ be some
constants. Consider a discrete-time real-valued martingale $\{X_k,
\mathcal{F}_k\}_{k=0}^{\infty}$ where a.s.\ $X_0 = 0$, and for
every $m \in \naturals$
\begin{eqnarray*}
&& \pr(X_m - X_{m-1} = d
\, | \, \mathcal{F}_{m-1}) = \varepsilon \, , \\[0.1cm]
&& \pr\left(X_m - X_{m-1} = -\frac{\varepsilon d}{1-\varepsilon}
\, \Big| \, \mathcal{F}_{m-1}\right) = 1-\varepsilon \, .
\end{eqnarray*}
This implies that
$\expectation[X_m - X_{m-1} \, | \, \mathcal{F}_{m-1}] = 0$
a.s.\ for every $m \in \naturals$, and, since $X_{m-1}$ is
$\mathcal{F}_{m-1}$-measurable, we have
$\expectation[X_m \, | \, \mathcal{F}_{m-1}] = X_{m-1}$ almost surely.
Moreover, since $\varepsilon \in (0, \frac{1}{2}]$,
$$|X_m - X_{m-1}| \leq \max \left\{d,
\frac{\varepsilon d}{1-\varepsilon} \right\} = d \qquad \text{a.s.}$$
so the Azuma--Hoeffding inequality gives
\begin{equation}
\pr(X_k \geq kx) \leq \exp\left(-\frac{kx^2}{2d^2}\right), \qquad \forall \, x \ge 0
\label{example: Azuma--Hoeffding inequality}
\end{equation}
independently of the value of $\varepsilon$ (note that $X_0=0$
a.s.). However, we can use Theorem~\ref{theorem: first refined concentration inequality}
to get a better bound; since for every $m \in \naturals$
\begin{eqnarray*}
\expectation \bigl[ (X_m
- X_{m-1})^2 \, | \, \mathcal{F}_{m-1} \bigr] = \frac{d^2 \varepsilon}{1-\varepsilon}, \qquad \text{a.s.}
\end{eqnarray*}
it follows from
\eqref{eq: one-sided concentration
inequality of the first refined bound} that
\begin{equation}
\pr(X_k \geq kx) \leq
\exp\Bigg(-k \, H\bigg(\frac{x(1-\varepsilon)}{d} + \varepsilon
\; \Big\| \; \varepsilon \bigg) \Bigg), \qquad \forall \, x \geq 0.
\label{example: one-sided concentration inequality in Theorem2}
\end{equation}
Consider the case where $\varepsilon \rightarrow 0$. Then,
for arbitrary $x>0$ and $k \in \naturals$, the Azuma--Hoeffding inequality in
\eqref{example: Azuma--Hoeffding inequality} provides an upper bound that
is strictly positive independently of $\varepsilon$,
whereas the one-sided concentration inequality of
Theorem~\ref{theorem: first refined concentration inequality}
implies a bound in \eqref{example: one-sided concentration
inequality in Theorem2} that tends to zero.
\end{example}

\begin{corollary}
Let $\{X_k, \mathcal{F}_k\}_{k=0}^n$ be a discrete-time
real-valued martingale, and assume that $| X_k - X_{k-1} | \leq d$
holds a.s. for some constant $d > 0$ and for every
$k \in \{1, \ldots, n\}$. Then, for every $\alpha \geq 0$,
\begin{equation}
\pr(|X_n-X_0| \geq \alpha n) \leq 2 \exp \left(-n f(\delta)
\right) \label{eq: the first refined concentration inequality with
no constraint on the conditional variance}
\end{equation}
where $\delta \triangleq \frac{\alpha}{d}$,
\begin{equation}
f(\delta) = \left\{
\begin{array}{ll}
\ln(2) \Bigl[1 - h_2\left(\frac{1-\delta}{2} \right) \Bigr], \quad &0
\leq \delta \leq 1
\\[0.2cm]
+\infty, \quad & \delta > 1
\end{array}
\right. \label{eq: f}
\end{equation}
and $h_2(x) \triangleq -x \log_2(x) - (1-x) \log_2(1-x)$ for $0
\leq x \leq 1$ is the binary entropy function (base $2$).
\label{corollary: a tightened version of the Azuma--Hoeffding inequality}
\end{corollary}
\begin{proof}
By substituting $\gamma=1$ in Theorem~\ref{theorem: first refined
concentration inequality} (since
there is no constraint on the conditional variance, one can
take $\sigma^2=d^2$), the corresponding exponent in
\eqref{eq: first refined concentration inequality} is equal to
\begin{equation}
H\bigg(\frac{1+\delta}{2} \Big\| \frac{1}{2}\bigg) = f(\delta),
\label{eq: f as divergence}
\end{equation}
since, from \eqref{eq: divergence}, it is easy to verify that
$H(p \| \frac{1}{2}) = \ln 2 \, \bigl[1-h_2(p)\bigr]$ for every
$p \in [0,1]$.
\end{proof}

An alternative proof of
Corollary~\ref{corollary: a tightened version of the Azuma--Hoeffding inequality},
which provides some further insight, is suggested in the following.
\begin{proof}
As a first step, a refined version of Hoeffding's lemma is provided (cf. Lemma~\ref{lm:Hoeffding}).
\begin{lemma}\label{lm:refined Hoeffding} Let $U \in \reals$ be a random
variable, such that $U \in [a,b]$ a.s.\ for some finite $a < b$, and $\expectation U = \frac{a+b}{2}$.
Then, for every $t \ge 0$,
\begin{align}
\label{eq:refined Hoeffding}
\expectation\left[ \exp\big(t(U- \expectation U)\big)\right]
\le \cosh\left(\frac{t(b-a)}{2}\right).
\end{align}
\end{lemma}
\begin{proof}
This refinement of \eqref{eq:Hoeffding}, if $\expectation U = \frac{a+b}{2}$,
follows from \eqref{eq:Hoeffding_bound_0}.
\end{proof}
The proof of Corollary~\ref{corollary: a tightened version of the Azuma--Hoeffding inequality}
continues by following the proof of the Azuma--Hoeffding inequality.
Recall that $\xi_k = X_k - X_{k-1}$, for all $k \in \naturals$, form the differences
of the martingale sequence with $|\xi_k| \le d$ (in the case where $d_k = d$, independently of $k$)
and $\expectation[\xi_k | \mathcal{F}_{k-1}] = 0$.
Using a conditional version of Lemma~\ref{lm:refined Hoeffding}, the bound in
\eqref{eq:bound on the conditional moment generating function of xi_n}
is improved to
\begin{align} \label{eq:improved bound on the conditional moment generating function of xi_n}
\expectation \bigl[\exp(t \xi_n) \, | \, \mathcal{F}_{n-1}\bigr]
\leq \cosh(td), \quad \forall \, t \geq 0
\end{align}
and continuing recursively, the quantity in \eqref{eq: smoothing theorem}
is upper bounded by
\begin{equation*}
\expectation \biggl[ \exp \biggl(t \sum_{k=1}^n \xi_k \biggr)
\biggr] \leq \cosh^n(td), \quad \forall \, t \ge 0.
\end{equation*}
Based on Chernoff's inequality, the following refinement of
\eqref{eq:parametric_Azuma} holds
\begin{align} \label{eq:refined parametric_Azuma}
\pr(X_n - X_0 \ge \alpha n) & \le \exp(-\alpha n t) \, \cosh^n(td) \nonumber \\
& = \exp\Bigl(-n \bigl[\alpha t - \ln \cosh(td)\bigr]\Bigr), \quad \forall \, t \ge 0.
\end{align}
Due to the bounded differences assumption, we have (a.s.)
$$|X_n - X_0| \le \sum_{k=1}^n |X_k - X_{k-1}| \le nd$$
so, if $\alpha > d$, we have
$\pr(X_n - X_0 \ge \alpha n)=0$.
If $0 \le \alpha < d$, an optimization of the free parameter $t$ on the
right-hand side of \eqref{eq:refined parametric_Azuma} gives
$ t = \frac{1}{d} \, \tanh^{-1}\left(\frac{\alpha}{d}\right).$
Substituting this optimized value of $t$ into \eqref{eq:refined parametric_Azuma},
combined with the use of the following two identities for hyperbolic functions:
\begin{align*}
& \tanh^{-1}(x) = \frac{1}{2} \, \ln \left(\frac{1+x}{1-x}\right), \quad \forall \, |x|<1, \\
& \cosh(x) = \frac{1}{\sqrt{1-\tanh^2(x)}}, \quad \forall \, x \in \reals,
\end{align*}
yield that the exponent on the right-hand side of \eqref{eq:refined parametric_Azuma}
is equal to
\begin{align*}
& \alpha t - \ln \cosh(td) \\
& = \frac{\alpha}{2d} \, \ln \left(\frac{1+\frac{\alpha}{d}}{1-\frac{\alpha}{d}}\right)
+ \frac{1}{2} \, \ln \left(1-\frac{\alpha^2}{d^2}\right)\\
& = \frac{1}{2} \left(1+\frac{\alpha}{d}\right) \ln\left(1+\frac{\alpha}{d}\right) +
\frac{1}{2} \left(1-\frac{\alpha}{d}\right) \ln\left(1-\frac{\alpha}{d}\right) \\
&= \ln 2 \, \left[ 1 - h_2\left(\frac{1}{2} \, \left(1-\frac{\alpha}{d}\right)\right) \right] \\
&= f(\delta)
\end{align*}
where the last equality follows from \eqref{eq: notation} and \eqref{eq: f}.
This gives the exponential bound in
Corollary~\ref{corollary: a tightened version of the Azuma--Hoeffding inequality}
for $\alpha \in [0,d)$. Finally, the result of this corollary
for $\alpha = d$ is obtained by letting $t$ tend to infinity in
the exponential bound on the right-hand side of \eqref{eq:refined parametric_Azuma}.
This gives
$$\lim_{t \rightarrow \infty} \bigl( td - \ln \cosh(td) \bigr)
= \ln 2, \quad \forall \, d>0$$
and, consequently, $$\pr(X_n - X_0 \ge d n) \leq 2^{-n}$$
which proves
Corollary~\ref{corollary: a tightened version of the Azuma--Hoeffding inequality}
for $\alpha = d$. Note that the factor~2 in the bound of
\eqref{eq: the first refined concentration inequality with
no constraint on the conditional variance} was justified in
the proof of Theorem~\ref{theorem: Azuma--Hoeffding inequality}.
\end{proof}

\begin{remark}
Corollary~\ref{corollary: a tightened version of the Azuma--Hoeffding
inequality}, which is a special case of
Theorem~\ref{theorem: first refined concentration inequality}
with $\gamma = 1$, forms a tightening of the Azuma--Hoeffding
inequality for the case where $d_k = d$ (independently of $k$).
This follows from Pinsker's inequality, which implies that
$f(\delta) > \frac{\delta^2}{2}$ for $\delta > 0$.
Figure~\ref{Figure: compare_exponents_theorem2} plots the two
exponents of the Azuma--Hoeffding inequality and its improvement
in Corollary~\ref{corollary: a tightened version of the Azuma--Hoeffding
inequality}, and they nearly coincide for $\delta \leq 0.4$.
The exponential bound of
Theorem~\ref{theorem: first refined concentration inequality} is
improved as the value of $\gamma \in (0,1)$ is reduced (see
Figure~\ref{Figure: compare_exponents_theorem2});
this holds since the additional constraint on the conditional
variance in Theorem~\ref{theorem: first refined concentration inequality}
has a growing effect by reducing the value of $\gamma$.
\begin{figure}[here!]  
\begin{center}
\epsfig{file=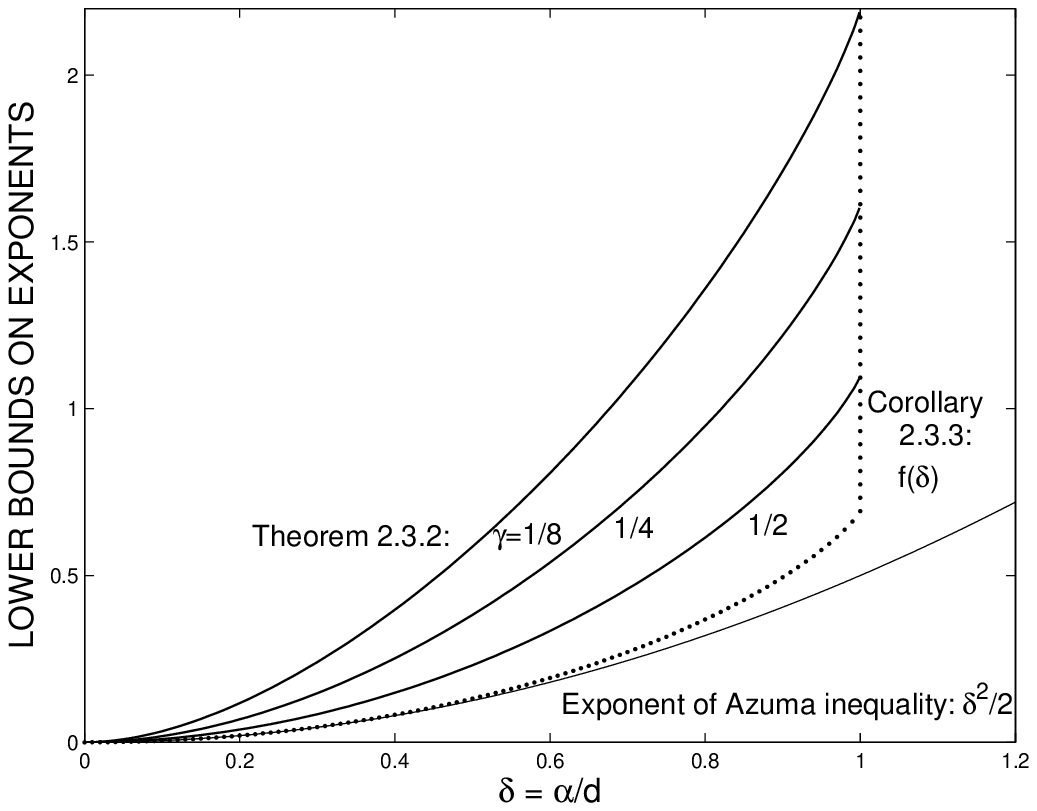,scale=0.5}
\end{center}
\caption{\label{Figure: compare_exponents_theorem2} Plot of the
lower bounds on the exponents in the Azuma--Hoeffding inequality
and the improved bounds in
Theorem~\ref{theorem: first refined concentration inequality} and
Corollary~\ref{corollary: a tightened version of the Azuma--Hoeffding
inequality}. The pointed line
refers to the exponent in Corollary~\ref{corollary: a tightened
version of the Azuma--Hoeffding inequality},
and the three solid lines for $\gamma = \frac{1}{8}, \frac{1}{4}$
and $\frac{1}{2}$ refer to the exponents in Theorem~\ref{theorem:
first refined concentration inequality}.}
\end{figure}
\label{remark: comparison of exponents}
\end{remark}

Theorem~\ref{theorem: first refined concentration inequality} can
also be used to analyze the probabilities of {\em small deviations},
i.e., events of the form $\{ |X_n - X_0| \geq \alpha \sqrt{n}\}$
for $\alpha \geq 0$ (in contrast to
{\em large-deviation} events of the form $\{ |X_n - X_0| \ge \alpha n \}$):

\begin{proposition}
Let $\{X_k, \mathcal{F}_k\}$ be a discrete-time real-valued martingale
that satisfies the conditions of
Theorem~\ref{theorem: first refined concentration inequality}. Then, for
every $\alpha \geq 0$,
\begin{equation}
\pr(|X_n-X_0| \geq \alpha \sqrt{n}) \leq 2
\exp\Bigl(-\frac{\delta^2}{2\gamma}\Bigr)
\Bigl(1+ O\bigl(n^{-\frac{1}{2}}\bigr)\Bigr).
\label{eq: concentration1}
\end{equation}
\label{proposition: a similar scaling of the
concentration inequalities}
\end{proposition}

\begin{remark}
From Proposition~\ref{proposition: a similar scaling of the
concentration inequalities}, for an arbitrary $\alpha
\geq 0$, the upper bound on $\pr(|X_n-X_0| \geq \alpha \sqrt{n})$
improves the exponent of the Azuma--Hoeffding inequality by a factor
of $\frac{1}{\gamma}$.
\end{remark}

\begin{proof}
Let $\{X_k, \mathcal{F}_k\}_{k=0}^{\infty}$ be a discrete-time martingale
that satisfies the conditions in Theorem~\ref{theorem: first
refined concentration inequality}. From \eqref{eq: first refined
concentration inequality}, for every $\alpha \geq 0$ and $n \in \naturals$,
\begin{equation}
\pr(|X_n-X_0| \geq \alpha \sqrt{n}) \leq 2 \exp\left(-n \,
H\bigg(\frac{\delta_n+\gamma}{1+\gamma} \bigg\|
\frac{\gamma}{1+\gamma}\bigg) \right)
\label{eq: exponent from first refined version with delta'}
\end{equation}
where, following \eqref{eq: notation},
\begin{equation}
\gamma \triangleq \frac{\sigma^2}{d^2}, \quad
\delta_n \triangleq \frac{\frac{\alpha}{\sqrt{n}}}{d} =
\frac{\delta}{\sqrt{n}} \, . \label{eq: new delta}
\end{equation}
With these definitions, we have
\begin{align}
H\bigg(\frac{\delta_n+\gamma}{1+\gamma}
\bigg\| \frac{\gamma}{1+\gamma}\bigg)
& = \frac{\gamma}{1+\gamma}
\Bigg[\left(1+\frac{\delta}{\gamma \sqrt{n}} \right)
\ln\left(1+\frac{\delta}{\gamma \sqrt{n}} \right)
\nonumber\\
& \qquad \qquad + \frac{1}{\gamma}
\left(1-\frac{\delta}{\sqrt{n}} \right)
\ln\left(1-\frac{\delta}{\sqrt{n}} \right)
\Bigg]. \label{eq: intermediate step1}
\end{align}
Using the power series expansion
\begin{equation*}
(1+u) \ln(1+u) = u + \sum_{k=2}^{\infty}
\frac{(-u)^k}{k(k-1)} \, , \quad -1 < u \leq 1
\end{equation*}
in \eqref{eq: intermediate step1}, it follows that for every
$n > \frac{\delta^2}{\gamma^2}$
\begin{eqnarray*}
&& n H\bigg(\frac{\delta_n+\gamma}{1+\gamma} \bigg\|
\frac{\gamma}{1+\gamma}\bigg)
= \frac{\delta^2}{2 \gamma} - \frac{\delta^3 (1-\gamma)}{6
\gamma^2} \, \frac{1}{\sqrt{n}} + \ldots \\
&& \hspace*{3.5cm} = \frac{\delta^2}{2 \gamma} +
O\left(\frac{1}{\sqrt{n}}\right).
\end{eqnarray*}
Substituting this into the exponent on the right-hand side of
\eqref{eq: exponent from first refined version with delta'}
gives \eqref{eq: concentration1}.
\end{proof}

\subsection{Inequalities for sub- and super-martingales}
\label{subsection: concentration inequalities for sub and super
martingales} Upper bounds on the probability $\pr( X_n - X_0 \geq
r)$ for $r \geq 0$, derived earlier in this section for
martingales, can be adapted to super-martingales (similarly to,
e.g., \cite[Chapter~2]{Chung_LU2006} or
\cite[Section~2.7]{survey2006}). Alternatively, by replacing $\{X_k,
\mathcal{F}_k\}_{k=0}^n$ with $\{-X_k, \mathcal{F}_k\}_{k=0}^n$, we may
obtain upper bounds on the probability $\pr( X_n - X_0 \leq -r)$
for sub-martingales. For example, the adaptation of
Theorem~\ref{theorem: first refined concentration inequality} to
sub- and super-martingales gives the following inequality:

\begin{corollary}
Let $\{X_k, \mathcal{F}_k\}_{k=0}^{\infty}$ be a
discrete-time real-valued super-martingale. Assume that, for
some constants $d, \sigma > 0$, the following two requirements are
satisfied a.s.:
\begin{eqnarray*}
&& X_k - \expectation[X_k \, | \, \mathcal{F}_{k-1}] \leq d, \\
&& \var (X_k | \mathcal{F}_{k-1}) \triangleq \expectation
\Bigl[\bigl(X_k - \expectation[X_k \, | \,
\mathcal{F}_{k-1}]\bigr)^2 \, | \, \mathcal{F}_{k-1} \Bigr] \leq
\sigma^2
\end{eqnarray*}
for every $k \in \{1, \ldots, n\}$. Then, for every $\alpha \geq
0$,
\begin{equation}
\hspace*{-0.2cm} \pr(X_n-X_0 \geq \alpha n) \leq \exp\left(-n \,
H\bigg(\frac{\delta+\gamma}{1+\gamma} \bigg\|
\frac{\gamma}{1+\gamma}\bigg) \right) \label{eq: first refined
concentration inequality for super-martingales}
\end{equation}
where $\gamma$ and $\delta$ are defined in \eqref{eq:
notation}, and the binary divergence $H(p\|q)$ is introduced in
\eqref{eq: divergence}. Alternatively, if $\{X_k,
\mathcal{F}_k\}_{k=0}^{\infty}$ is a sub-martingale, the same
upper bound in \eqref{eq: first refined concentration inequality
for super-martingales} holds for the probability $\pr(X_n-X_0 \leq
-\alpha n)$. If $\delta>1$, these two probabilities are zero.
\label{corollary: first refined concentration inequality for sub and super martingales}
\end{corollary}
\begin{proof}
It is similar to the proof of
Theorem~\ref{theorem: first refined concentration inequality};
the only difference is that, for a super-martingale,
$X_n - X_0 = \sum_{k=1}^n (X_k-X_{k-1}) \leq \sum_{k=1}^n
\xi_k$ a.s., where $\xi_k \triangleq X_k-\expectation[X_k \, | \,
\mathcal{F}_{k-1}]$ is $\mathcal{F}_k$-measurable. Therefore, we
have $\pr(X_n-X_0 \geq \alpha n) \leq \pr\bigl(\sum_{k=1}^n \xi_k \geq
\alpha n\bigr)$ where, a.s., $\xi_k \leq d$, $\expectation[\xi_k \,
| \, \mathcal{F}_{k-1}]=0$, and $\var(\xi_k \, | \,
\mathcal{F}_{k-1}) \leq \sigma^2$. The rest of the proof
coincides with the proof of Theorem~\ref{theorem: first refined
concentration inequality} (starting from \eqref{eq: Chernoff's inequality}).
The other inequality for sub-martingales holds due to the fact
that if $\{X_k, \mathcal{F}_k\}$ is a sub-martingale then $\{-X_k,
\mathcal{F}_k\}$ is a super-martingale.
\end{proof}

The reader is referred to \cite{FanGL_paper1} for an extension of Hoeffding's
inequality to super-martingales with differences bounded from above (or sub-martingales
with differences bounded from below), and to \cite{FanGL_paper2} for large deviation
exponential inequalities for super-martingales.

\section{Relations to classical results in probability theory}
\label{section: Relations to Results in Probability Theory}

\subsection{The martingale central limit theorem}
\label{subsection: relation between the martingale CLT and Proposition 4.3}
A relation between
Proposition~\ref{proposition: a similar scaling of the concentration inequalities}
and the martingale central limit theorem (CLT) is considered in the following.

Let $(\Omega, \mathcal{F}, \pr)$ be a probability space. Given a
filtration $\{\mathcal{F}_k\}$, we say that $\{Y_k, \mathcal{F}_k\}_{k=0}^{\infty}$
is a martingale-difference sequence if, for every $k$,
\begin{enumerate}
\item $Y_k$ is $\mathcal{F}_k$-measurable,
\item $ \expectation [|Y_k|] < \infty$,
\item $ \expectation \bigl[Y_k \, | \,
\mathcal{F}_{k-1}\bigr] = 0.$
\end{enumerate}
Let $$S_n = \sum_{k=1}^n Y_k, \quad \forall \, n \in \naturals$$
and $S_0 = 0$; then $\{S_k, \mathcal{F}_k\}_{k=0}^{\infty}$ is a
martingale. Assume that the sequence of random variables $\{Y_k\}$
is bounded, i.e., there exists a constant $d$ such that
$|Y_k| \leq d$ a.s., and furthermore, assume that the limit
$$ \sigma^2 \triangleq \lim_{n \rightarrow \infty} \frac{1}{n}
\sum_{k=1}^n \expectation\bigl[Y_k^2 \, | \, \mathcal{F}_{k-1}\bigr]$$
exists in probability and is positive. The martingale CLT asserts that,
under the above conditions, $\left\{\frac{S_n}{\sqrt{n}}\right\}$ converges
in distribution (or weakly) to the Gaussian distribution
$\mathcal{N}(0, \sigma^2)$; we denote this convergence by
$\frac{S_n}{\sqrt{n}} \Rightarrow \mathcal{N}(0, \sigma^2)$.
(There exist more general versions of this statement --- see, e.g.,
\cite[pp.~475--478]{Billingsley}).

Let $\{X_k, \mathcal{F}_k\}_{k=0}^{\infty}$ be a
real-valued martingale with bounded differences where there exists a constant
$d$ such that a.s.
$$|X_k - X_{k-1}| \leq d, \quad \forall \, k \in \naturals.$$
Define, for every $k \in \naturals$,
$$Y_k \triangleq X_k - X_{k-1}$$ and $Y_0 \triangleq 0$.
Then $\{Y_k, \mathcal{F}_k\}_{k=0}^{\infty}$ is a
martingale-difference sequence, and $|Y_k| \leq d$ a.s.\ for every
$k \in \naturals \cup \{0\}$.
Assume also that there exists a constant $\sigma > 0$, such that, for all $k$,
$$\expectation[Y_k^2 \, | \, \mathcal{F}_{k-1}] =
\expectation[(X_k - X_{k-1})^2 \, | \, \mathcal{F}_{k-1}] =
\sigma^2, \qquad \text{a.s.}$$
Consequently, from the martingale CLT, it follows that
$$ \frac{X_n - X_0}{\sqrt{n}} \Longrightarrow \mathcal{N}(0,
\sigma^2),$$ so, for every $\alpha \geq 0$,
\begin{equation*}
\lim_{n \rightarrow \infty} \pr\Big(|X_n - X_0| \geq \alpha \sqrt{n}\Big)
= 2 \, Q\Bigl(\frac{\alpha}{\sigma}\Bigr)
\end{equation*}
where the $Q$-function is defined in \eqref{eq: Q function}. In terms of
the notation in \eqref{eq: notation}, we have
$\frac{\alpha}{\sigma} = \frac{\delta}{\sqrt{\gamma}}$, so that
\begin{equation}
\lim_{n \rightarrow \infty} \pr\Big(|X_n - X_0| \geq \alpha \sqrt{n}\Big)
= 2 \, Q\biggl(\frac{\delta}{\sqrt{\gamma}}\biggr). \label{eq:
consequence of the martingale CLT}
\end{equation}
From the fact that
$$
Q(x) \leq \frac{1}{2} \,
\exp\left(-\frac{x^2}{2}\right), \qquad \forall \, x \ge 0$$
it follows that, for every
$\alpha \geq 0$,
\begin{equation*}
\lim_{n \rightarrow \infty} \pr(|X_n - X_0| \geq \alpha \sqrt{n})
\leq \exp\left(-\frac{\delta^2}{2\gamma}\right).
\end{equation*}
This inequality coincides with the large-$n$ limit of the
inequality in Proposition~\ref{proposition: a similar scaling of
the concentration inequalities}, except for the additional factor
of~$2$ in the pre-exponent (see the right-hand side of
\eqref{eq: concentration1}). Note also that the proof of
Proposition~\ref{proposition: a similar scaling of the
concentration inequalities} is applicable for
finite $n$, and not only in the asymptotic regime $n \to \infty$.
Furthermore, from the exponential upper and lower
bounds on the Q-function in \eqref{eq: upper and lower bounds for
the Q function} and from \eqref{eq: consequence of
the martingale CLT}, it follows that the exponent in the concentration
inequality \eqref{eq: concentration1} cannot be improved without imposing
additional conditions on the martingale sequence.

\subsection{The moderate deviations principle}
\label{subsection: MDP for real-valued i.i.d. RVs}

The moderate deviations principle (MDP) on the real line (see, e.g.,
\cite[Theorem~3.7.1]{Dembo_Zeitouni}) states the following: Let
$\{X_i\}_{i=1}^n$ be a sequence of real-valued i.i.d.\ random variables such
that $\Lambda_X(\lambda) \triangleq \ln \expectation[e^{\lambda X_i}] < \infty$
in some neighborhood of zero, and also assume that
$\expectation[X_i] = 0$ and $\sigma^2 = \var(X_i) > 0$. Let
$\{a_n\}_{n=1}^{\infty}$ be a non-negative sequence such that $a_n
\rightarrow 0$ and $n a_n \rightarrow \infty$ as $n \rightarrow
\infty$, and let
\begin{equation}
Z_n \triangleq \sqrt{\frac{a_n}{n}} \sum_{i=1}^n X_i,
\quad \forall \, n \in \naturals.
\label{eq: Z sequence}
\end{equation}
Then, for every measurable set $\Gamma \subseteq \reals$,
\begin{align}
-\frac{1}{2 \sigma^2} \inf_{x \in \Gamma^0} x^2
&\leq \liminf_{n \rightarrow \infty} a_n \ln \pr(Z_n \in \Gamma)  \nonumber\\
& \leq \limsup_{n \rightarrow \infty} a_n \ln \pr(Z_n \in \Gamma) \nonumber \\
& \leq -\frac{1}{2 \sigma^2} \inf_{x \in \overline{\Gamma}} x^2
\end{align}
where $\Gamma^0$ and $\overline{\Gamma}$ denote, respectively,
the interior and the closure of $\Gamma$.

Let $\eta \in (\frac{1}{2}, 1)$ be an arbitrary fixed number, and
let $\{a_n\}_{n=1}^{\infty}$ be the non-negative sequence
$$a_n = n^{1-2\eta}, \quad \forall \, n \in \naturals$$ so that
$a_n \rightarrow 0$ and $n a_n \rightarrow \infty$ as $n
\rightarrow \infty$. Let $\alpha \in \reals^+$, and $\Gamma
\triangleq (-\infty, -\alpha] \cup [\alpha, \infty)$. Note that,
from \eqref{eq: Z sequence},
\begin{equation} \label{eq:probability of a sum of i.i.d. RVs}
\pr\left( \left|\sum_{i=1}^n X_i \right| \geq \alpha n^{\eta} \right)
= \pr(Z_n \in \Gamma)
\end{equation}
so, by the MDP,
\begin{equation}
\hspace*{-0.4cm} \lim_{n \rightarrow \infty} n^{1-2\eta} \;
\ln \pr\left( \left|\sum_{i=1}^n X_i \right| \geq \alpha n^{\eta} \right)
= -\frac{\alpha^2}{2 \sigma^2}, \qquad \forall \, \alpha \geq 0.
\label{eq: MDP for i.i.d. real-valued RVs}
\end{equation}
We show in Appendix~\ref{appendix: MDP} that, in
contrast to the Azuma--Hoeffding inequality, Theorem~\ref{theorem: first
refined concentration inequality} provides an upper bound on the left-hand
side of \eqref{eq:probability of a sum of i.i.d. RVs} which
coincides with the asymptotic limit in \eqref{eq: MDP for i.i.d.
real-valued RVs}. The analysis in Appendix~\ref{appendix: MDP}
provides another interesting link between Theorem~\ref{theorem:
first refined concentration inequality} and a classical result
in probability theory, and thus emphasizes the significance
of the refinements of the Azuma--Hoeffding inequality.

\subsection{Functions of discrete-time Markov chains}

An interesting relation between discrete-time Markov
chains and martingales is the following (see, e.g.,
\cite[p.~473]{GrimmettS_book}): Let $\{X_n\}^\infty_{n=0}$ be a
discrete-time Markov chain taking values in a countable state space
$\mathcal{S}$ with transition matrix ${\bf{P}}$. Let
$\psi \colon \mathcal{S} \to \reals$  be a {\em harmonic function}
of the Markov chain, i.e.,
\begin{equation}
\sum_{s \in \mathcal{S}} p_{s',s} \psi(s) = \psi(s'), \quad
\forall \, s' \in \mathcal{S}
\label{eq: harmonic function}
\end{equation}
and assume also that $\psi$ is a measurable and bounded
function. Let $Y_n \deq \psi(X_n)$ for every $n \ge 0$,
and let $\{\cF_n\}$ be the natural filtration where
$\cF_n = \sigma(X_0, \ldots, X_n)$. It is a remarkable
fact that $\{Y_n, \cF_n\}$ is a martingale; this
property holds since $Y_n$ is $\cF_n$-measurable,
$\expectation[|Y_n|] < \infty$ (due to the requirement that
$\psi$ is bounded), and from \eqref{eq: harmonic function}
\begin{align} \label{eq:martingality of the composed Markov chain}
\expectation[Y_n \, | \, \cF_{n-1}] = \sum_{s \in \mathcal{S}}
p_{X_{n-1},s} \, \psi(s) = \psi(X_{n-1}) = Y_{n-1}, \quad
\forall \, n \in \naturals.
\end{align}
This relation between Markov chains and martingales enables
to apply the concentration inequalities of this chapter
to the composition of a bounded harmonic function and a Markov
chain; note that the boundedness of $\psi$ implies
that the differences of the martingale sequence are uniformly
bounded (this holds since, for every $n$, we have
$|Y_n - Y_{n-1}| \le 2 \|\psi\|_{\infty} < \infty$).

More generally, let $\underline{\psi}$ be a right eigenvector
of the transition matrix $P$ such that
$\|\underline{\psi}\|_{\infty} < \infty$, and let $\lambda$ be
its corresponding eigenvalue such that $|\lambda| \ge 1$. Let
$\mathcal{S} = \{s_1, s_2, \ldots\}$ be the countable state space
of the Markov chain, and
let $\psi \colon \mathcal{S} \to \reals$ be a real-valued
function such that $\psi(s_i)$ is equal to the $i$-th entry
of the vector $\underline{\psi}$. Then, the following equality holds:
$$\sum_{s \in \mathcal{S}} p_{s',s} \psi(s) = \lambda \, \psi(s'),
\quad \forall \, s' \in \mathcal{S}$$
which generalizes \eqref{eq: harmonic function} (i.e., if $\lambda=1$,
the function $\psi$ is harmonic).
Similarly to \eqref{eq:martingality of the composed Markov chain},
for every $n \ge 1$,
$$\expectation[\psi(X_n) \, | \, \cF_{n-1}] = \lambda \, \psi(X_{n-1}).$$
Defining $Y_n = \lambda^{-n} \, \psi(X_n)$, for $n \ge 0$,
implies that $\expectation[Y_n | \cF_{n-1}] = Y_{n-1}$. Since
$|\lambda| \ge 1$ and $\|\underline{\psi}\|_{\infty} < \infty$
then $\expectation[|Y_n|] < \infty$. Consequently,
$\{Y_n, \cF_n\}$ is a martingale sequence, and its differences are
uniformly bounded. The latter property holds since, for every $n \ge 1$,
\begin{align*}
& |Y_n - Y_{n-1}| \\[0.1cm]
& \le |\lambda|^{-n} \, |\psi(X_n)|
+ |\lambda|^{-(n-1)} \, |\psi(X_{n-1})| \\[0.1cm]
& \le |\psi(X_n)| + |\psi(X_{n-1})| \\[0.1cm]
& \le 2 \| \underline{\psi} \|_{\infty} < \infty.
\end{align*}
Since $\{Y_n, \cF_n\}$ is demonstrated to be a discrete-time
martingale with uniformly bounded differences, the concentration
inequalities of this chapter are applicable here as well.

Exponential deviation bounds for an important class of Markov
chains, so-called Doeblin chains, were derived by Kontoyiannis
\cite{Kontoyiannis_ISIT05}. These bounds are essentially identical
to the Hoeffding inequality in the special case of i.i.d.\ random
variables (see \cite[Remark~1]{Kontoyiannis_ISIT05}).

\section{Applications in information theory and coding}
\label{section: Applications}

This section is focused on applications of the concentration inequalities,
derived in this chapter via the martingale approach, in information theory,
communications and coding.

\subsection{Minimum distance of binary linear block codes}
Consider the ensemble of binary linear block codes of length $n$ and
rate $R$, where the codes are chosen uniformly at random. The
asymptotic average value of the normalized minimum distance is
equal to (see \cite[Section~2.C]{BargF_IT2002})
\begin{equation*}
\lim_{n \rightarrow \infty}
\frac{ \expectation[d_{\min}(\mathcal{C})]}{n} = h_2^{-1}(1-R)
\end{equation*}
where $h_2^{-1} \colon [0,1] \to [0, \frac{1}{2}]$ denotes the inverse
of the binary entropy function to the base $2$.

Let ${\bf{H}}$ denote an $n(1-R) \times n$ parity-check matrix of a
linear block code $\mathcal{C}$ from this ensemble. The minimum
distance of the code is equal to the minimal number of columns in
${\bf{H}}$ that are linearly dependent. Note that the minimum distance
is a property of the code, and it does not depend on the choice of the
particular parity-check matrix which represents the code.

Let us construct a sequence of integer-valued random variables
$\{X_i\}_{i=0}^n$ where $X_i$ is defined to be the minimal number
of linearly dependent columns of a parity-check matrix ${\bf{H}}$,
chosen uniformly at random from the ensemble, given that the first
$i$ columns of ${\bf{H}}$ are already revealed; this refers to
a random process where sequentially, at every time instant, a new
column of the parity-check matrix ${\bf{H}}$ is revealed.

Recalling Fact~\ref{fact: construction of martingales} from
Section~\ref{section: Discrete-Time Martingales}, we see that this
is a martingale sequence with the natural filtration $\{\cF_i\}_{i=0}^n$
where $\mathcal{F}_i$ is the $\sigma$-algebra that is generated by all
subsets of $n(1-R) \times n$ binary parity-check matrices whose first
$i$~columns are fixed. This martingale sequence has bounded differences,
and it satisfies $|X_i - X_{i-1}| \leq 1$ for $i \in \{1, \ldots, n\}$;
this can be verified by noticing that the observation of a new column
of ${\bf{H}}$ can change the minimal number of linearly dependent columns
by at most~$1$. Note that the random variable $X_0$ is the expected minimum
Hamming distance of the ensemble, and $X_n$ is the minimum distance of a
particular code from the ensemble (since once all the $n$ columns of
${\bf{H}}$ are revealed, the code is known exactly). Hence, by the
Azuma--Hoeffding inequality,
\begin{equation*}
\pr\Big( | d_{\min}(\mathcal{C}) - \expectation [d_{\min}(\mathcal{C})]
| \geq \alpha \sqrt{n} \Big) \leq 2
\exp\left(-\frac{\alpha^2}{2}\right), \; \forall \, \alpha
> 0.
\end{equation*}

This leads to the following concentration theorem of the minimum
distance around the expected value:
\begin{theorem}
Let $\mathcal{C}$ be chosen uniformly at random from the ensemble of
binary linear block codes of length $n$ and rate $R$. Then for every
$\alpha > 0$, with probability at least $1-2
\exp\left(-\frac{\alpha^2}{2}\right)$, the minimum distance of
$\mathcal{C}$ lies in the interval
$[n \, h_2^{-1}(1-R) - \alpha \sqrt{n}, \; n \, h_2^{-1}(1-R) +
\alpha \sqrt{n}].$
\end{theorem}

\begin{remark}
Note that some well-known capacity-approaching families
of binary linear block codes have a minimum Hamming distance
that grows sublinearly with the block length $n$. For example,
the class of parallel concatenated convolutional (turbo) codes was
proved to have minimum distance that grows at most as the
logarithm of the interleaver length \cite{Breiling04}.
\end{remark}

\subsection{Expansion properties of random regular bipartite graphs}
\label{section: Expansion in Random Regular Bipartite Graphs}
The Azuma--Hoeffding inequality is useful for analyzing the expansion
properties of random bipartite graphs. The following theorem was
proved by Sipser and Spielman \cite[Theorem~25]{SipserS96} in the context
of bit-flipping decoding algorithms for expander codes.
It is stated, in the following, in a more precise form that captures the
relation between the deviation from the expected value and the
exponential convergence rate of the resulting probability:
\begin{theorem}
Let $\mathcal{G}$ be a bipartite graph that is chosen uniformly at random from
the ensemble of bipartite graphs with $n$ vertices on the left, a left degree $l$,
and a right degree $r$.
Let $\alpha \in (0,1)$ and $\delta > 0$ be fixed numbers.
Then, with probability at least $1-\exp(-\delta n)$, all sets of
$\alpha n$ vertices on the left side of $\mathcal{G}$
are connected to at least
\begin{equation}
n \left[ \frac{l \bigl(1-(1-\alpha)^r\bigr)}{r} - \sqrt{2l \alpha \,
\bigl(h(\alpha)+\delta\bigr)} \, \right] \label{eq: number of
neighbors}
\end{equation}
vertices (neighbors) on the right side of $\mathcal{G}$,
where $h$ is the binary entropy function to base~$e$
(i.e., $h(x) = -x \ln(x) - (1-x) \ln(1-x)$ for $x \in [0,1]$).
\label{theorem: expansion}
\end{theorem}
\begin{proof} The proof starts by looking at the expected number of
neighbors, and then exposing one neighbor at a time to bound the
probability that the number of neighbors deviates significantly from
this mean.

Let $\mathcal{V}$ denote a given set of $n \alpha$ vertices on the
left side of the selected bipartite graph $\mathcal{G}$.
The set $\mathcal{V}$ has $n \alpha l$ outgoing edges in $\mathcal{G}$.
Let $X(\mathcal{G})$ be a random variable which denotes the number of
neighbors of $\mathcal{V}$ on the right side of $\mathcal{G}$, and let
$\expectation[X(\mathcal{G})]$ be the expected value of neighbors of $\mathcal{V}$
where all the bipartite graphs are chosen uniformly at random from the ensemble.
This expected number is equal to
\begin{align}
\label{eq:expected number of ineighbors}
\expectation[X(\mathcal{G})] = \frac{n l \bigl(1-(1-\alpha)^r\bigr)}{r}
\end{align}
since, for each of the $\frac{nl}{r}$ vertices on the right side of $\mathcal{G}$,
the probability that it has at least one edge in the subset of $n \alpha$ chosen
vertices on the left side of $\mathcal{G}$ is $1-(1-\alpha)^r$.

Let us form a martingale sequence to estimate, via the Azuma--Hoeffding
inequality, the probability that the actual number of neighbors
deviates by a certain amount from the expected value in
\eqref{eq:expected number of ineighbors}.

The set of $n \alpha$ vertices in $\mathcal{V}$ has
$n \alpha l$ outgoing edges. Let us reveal the destination of each
of these edges one at a time. More precisely, let $S_i$ be the random variable
denoting the vertex on the right side of $\mathcal{G}$ which the
$i$-th edge is connected to, where $i \in \left\{1, \ldots, n \alpha l\right\}$.
Let us define, for $i \in \{0, \ldots, n \alpha l\}$,
\begin{equation*}
X_i = \expectation [X(\mathcal{G}) | S_1, \ldots, S_{i-1}].
\end{equation*}
Note that this forms a martingale sequence where $X_0 =
\expectation[X(\mathcal{G})]$ and $X_{n \alpha l} =
X(\mathcal{G})$. For every $i \in \left\{1, \ldots, n \alpha l\right\}$,
we have $|X_i - X_{i-1}| \leq 1$ since every time only one
connected vertex on the right side of $\mathcal{G}$ is revealed,
so the number of neighbors of the
chosen set $\mathcal{V}$ cannot change by more than~1 at every
single time. Hence, from the one-sided Azuma--Hoeffding inequality in
Section~\ref{subsection: Azuma--Hoeffding inequality},
\begin{equation} \label{eq: number of neigbors with free parameter lambda}
\pr \Big( \expectation[X(\mathcal{G})] - X(\mathcal{G}) \geq
\lambda \sqrt{l \alpha n} \Big) \leq
\exp\left(-\frac{\lambda^2}{2}\right), \quad \forall \, \lambda
> 0.
\end{equation}
Since there are ${n}\choose{n\alpha}$ choices for the set
$\mathcal{V}$, the event that there
exists a set of size $n \alpha$ with less than
$\expectation[X(\mathcal{G})] - \lambda \sqrt{l \alpha n}$ neighbors
occurs with probability at most $ {{n}\choose{n\alpha}} \,
\exp\bigl(-\frac{\lambda^2}{2}\bigr)$, by the union bound. Based on the
inequality ${{n}\choose{n\alpha}} \leq e^{n h(\alpha)}$, we get the
exponential upper bound $\exp\bigl(n h(\alpha) -\frac{\lambda^2}{2} \bigr)$.
Finally, choosing $\lambda = \sqrt{2n \bigl( h(\alpha) + \delta
\bigr)}$ in \eqref{eq: number of neigbors with free parameter lambda}
gives the bound in \eqref{eq: number of neighbors}.
\end{proof}

\subsection{Concentration of the crest factor for OFDM signals}
\label{subsection: Concentration of the Crest-Factor for OFDM Signals}

Orthogonal-frequency-division-multiplexing (OFDM) is a widely used
modulation scheme that converts a high-rate data stream into a large
number of closely spaced orthogonal sub-carrier signals. These sub-carriers
are used to transmit data steams over parallel narrow-band channels.
OFDM signals are used in various international standards for digital
television and audio broadcasting, DSL internet access, wireless networks,
and the fourth generation (4G) mobile communications. For a textbook
treatment of OFDM, the reader is referred to, e.g., \cite[Chapter~19]{Molisch_book}.

The primary advantage of OFDM signals over single-carrier modulation
schemes is in their immunity to severe channel conditions (e.g.,
attenuation of high frequencies in a long copper wire, narrowband
interference and frequency-selective fading due to multipath
propagation) without using complex equalization filters. This
important advantage arises from the fact that channel equalization
is significantly simplified due to the fact that the OFDM modulation
scheme can be viewed as using many slowly-varying modulated narrowband
signals rather than one rapidly-varying modulated wideband signal.
Nevertheless, one of the significant problems of OFDM signals is that
the peak amplitude of such a signal is typically much larger than its
average amplitude. The high  peak-to-average power ratio (PAPR) of OFDM
signals makes their transmission sensitive to non-linear devices in the
communication path, such as digital-to-analog converters, mixers and
high-power amplifiers. As a result of this drawback, linear transmitter
circuitry is required for OFDM signals, which suffers from a poor power
efficiency. For a recent comprehensive tutorial that considers this
long-lasting problem of the high PAPR, and some related issues, the reader
is referred to \cite{WunderFBLN_OFDM_tutorial_2012}.

Given an $n$-length codeword $\{X_i\}_{i=0}^{n-1}$, a single OFDM
baseband symbol is described by \vspace*{-0.2cm}
\begin{equation}
s(t) = \frac{1}{\sqrt{n}} \sum_{i=0}^{n-1} X_i \exp\Bigl(\frac{j
\, 2\pi i t}{T}\Bigr), \quad 0 \leq t \leq T. \label{eq: OFDM
signal}
\end{equation}
Let us assume that $X_0, \ldots, X_{n-1}$ are complex random variables, and
$|X_i|=1$ a.s. (for the moment, these random variables may be dependent; however,
later in this section, some concentration inequalities are derived for the case where
these random variables are independent). Since the sub-carriers are orthonormal
over $[0,T]$,  the signal power over the interval $[0,T]$ is~1 a.s.:
\begin{equation}
\frac{1}{T} \int_0^T |s(t)|^2 dt = 1.  \label{eq: energy of OFDM
signal}
\end{equation}
The crest factor (CF) of the signal $s$, composed of $n$ sub-carriers,
is defined as
\begin{equation}
\text{CF}_n(s) \triangleq \max_{0 \leq t \leq T} | s(t) |.
\label{eq: CF}
\end{equation}
Commonly, the impact of nonlinearities is described by the distribution
of the CF of the transmitted signal \cite{LitsynW06},
but its calculation involves time-consuming simulations even for a small
number of sub-carriers.
From \cite[Section~4]{SalemZ} and \cite{WunderB_IT}, it follows
that the CF scales with high probability like $\sqrt{\ln n}$ for
large $n$. In \cite[Theorem~3 and Corollary~5]{LitsynW06}, a
concentration inequality was derived for the CF of OFDM signals.
It states that, for an arbitrary $c \geq 2.5$,
\begin{equation*} \pr \biggl(\Bigl| \text{CF}_n(s) -
\sqrt{\ln n} \Bigr| < \frac{c \ln \ln n}{\sqrt{\ln n}} \biggr) = 1
- O\Biggl( \frac{1}{\bigl(\ln n \bigr)^4} \Biggr).
\end{equation*}

\begin{remark}
The analysis used to derive this rather strong concentration
inequality (see \cite[Appendix~C]{LitsynW06}) requires some
assumptions on the distribution of the $X_i$'s (see the two
conditions in \cite[Theorem~3]{LitsynW06} followed by
\cite[Corollary~5]{LitsynW06}). These requirements are not needed
in the following analysis, and the derivation of concentration
inequalities that are introduced in this subsection is much simpler
and provides some insight into the problem, although the resulting
concentration result is weaker than the one in \cite[Theorem~3]{LitsynW06}.
\end{remark}

\vspace*{0.2cm} In the following, the concentration of the crest factor
of OFDM signals is studied via the Azuma--Hoeffding inequality,
its refinement in Proposition~\ref{proposition: a similar scaling of the
concentration inequalities}, and McDiarmid's inequality. It is assumed
in the following that the symbols $\{X_j\}_{j=0}^{n-1}$ are independent
complex-valued random variables with magnitude~1, attaining the $M$ points
of an $M$-ary PSK constellation with equal probability. The material in
this section presents in part the work in \cite{Sason_ISWCS_OFDM}.

\bigskip
{\em Concentration via the Azuma--Hoeffding inequality}:
Let us define the random variables
\begin{equation}
Y_i = \expectation[ \, \text{CF}_n(s) \, | \, X_0, \ldots,
X_{i-1}], \quad i =0, \ldots, n.  \label{eq: martingale sequence
for OFDM}
\end{equation}
Based on a standard construction of martingales, $\{Y_i,
\mathcal{F}_i\}_{i=0}^n$ is a martingale, where $\mathcal{F}_i$ is
the $\sigma$-algebra generated by the first $i$ symbols
$(X_0, \ldots, X_{i-1})$ in \eqref{eq: OFDM signal}. Hence,
$\mathcal{F}_0 \subseteq \mathcal{F}_1 \subseteq \ldots \subseteq
\mathcal{F}_n$ is a filtration. This martingale also has bounded
differences:
$$|Y_i - Y_{i-1}| \leq \frac{2}{\sqrt{n}}, \qquad i \in \{1,\ldots,n\}$$
since revealing the additional $i$th
coordinate $X_i$ affects the CF, as defined in \eqref{eq: CF},
by at most $\frac{2}{\sqrt{n}}$ (see the first part of
Appendix~\ref{appendix: OFDM}). It therefore follows from the
Azuma--Hoeffding inequality that, for every $\alpha > 0$,
\begin{equation}
\hspace*{-0.2cm} \pr( | \text{CF}_n(s) -
\expectation [\text{CF}_n(s)] | \geq \alpha) \leq 2
\exp\left(-\frac{\alpha^2}{8}\right),
\label{eq: Azuma--Hoeffding inequality for OFDM}
\end{equation}
which demonstrates concentration around the expected value.

\bigskip
{\em Concentration of the crest factor via
Proposition~\ref{proposition: a similar scaling of the
concentration inequalities}}: We will now use
Proposition~\ref{proposition: a similar scaling of the
concentration inequalities} to derive an improved concentration
result. For the martingale sequence $\{Y_i\}_{i=0}^n$ in
\eqref{eq: martingale sequence for OFDM}, Appendix~\ref{appendix:
OFDM} gives that a.s.
\begin{equation}
|Y_i - Y_{i-1}| \leq \frac{2}{\sqrt{n}} \, , \quad
\expectation\bigl[(Y_i-Y_{i-1})^2  | \mathcal{F}_{i-1}\bigr] \leq
\frac{2}{n} \label{eq: properties for OFDM signals}
\end{equation}
for every $i \in \{1, \ldots, n\}$. Note that the conditioning on
the $\sigma$-algebra $\mathcal{F}_{i-1}$ is equivalent to
conditioning on the symbols $X_0, \ldots, X_{i-2}$, and there is
no conditioning for $i=1$. Further, let $Z_i = \sqrt{n} Y_i$ for
$0 \leq i \leq n$. Proposition~\ref{proposition: a similar scaling
of the concentration inequalities} therefore implies that, for an
arbitrary $\alpha
> 0$,
\begin{align}
& \pr( | \text{CF}_n(s) - \expectation [\text{CF}_n(s)] | \geq \alpha ) \nonumber \\
& = \pr( |Y_n - Y_0| \geq \alpha) \nonumber\\
& = \pr( |Z_n - Z_0| \geq \alpha \sqrt{n}) \nonumber\\
& \leq 2 \exp \left( -\frac{\alpha^2}{4} \, \Biggl(1 +
O\biggl(\frac{1}{\sqrt{n}}\biggr) \Biggr) \right) \label{eq:
OFDM-inequality1}
\end{align}
(since $\delta = \frac{\alpha}{2}$ and $\gamma = \frac{1}{2}$ in
the setting of Proposition~\ref{proposition: a similar scaling of
the concentration inequalities}). Note that the exponent in the
last inequality is doubled as compared to the bound that was
obtained in \eqref{eq: Azuma--Hoeffding inequality for OFDM} via the
Azuma--Hoeffding inequality, and the term that scales like
$O\Bigl(\frac{1}{\sqrt{n}}\Bigr)$ on the right-hand side of
\eqref{eq: OFDM-inequality1} is expressed explicitly for finite
$n$ (see the proof of Proposition~\ref{proposition: a similar scaling of the
concentration inequalities}).

\bigskip
{\em Establishing concentration via McDiarmid's inequality}:
We use in the following McDiarmid's inequality (see
Theorem~\ref{theorem: McDiarmid's inequality}) in order to prove a
concentration inequality for the crest factor of OFDM signals. To
this end, let us define
\begin{eqnarray*}
&& U \triangleq \max_{0 \leq t \leq T} \bigl|s(t; X_0, \ldots,
X_{i-1},
X_i, \ldots, X_{n-1})\bigr| \\
&& V \triangleq \max_{0 \leq t \leq T} \bigl|s(t; X_0, \ldots,
X'_{i-1}, X_i, \ldots, X_{n-1})\bigr|
\end{eqnarray*}
where the two vectors $(X_0, \ldots, X_{i-1}, X_i, \ldots,
X_{n-1})$ and $(X_0, \ldots, X'_{i-1}, X_i, \ldots, X_{n-1})$ may
only differ in their $i$-th coordinate. This then implies that
\begin{eqnarray*}
&& |U-V| \leq \max_{0 \leq t \leq T} \bigl|s(t; X_0, \ldots,
X_{i-1}, X_i, \ldots, X_{n-1}) \\
&& \hspace*{2.5cm} - s(t; X_0, \ldots, X'_{i-1}, X_i, \ldots,
X_{n-1})\bigr| \nonumber \\[0.15cm]
&& \hspace*{1.3cm} = \max_{0 \leq t \leq T} \frac{1}{\sqrt{n}} \,
\Bigr|\bigl(X_{i-1} - X'_{i-1}\bigr) \exp\Bigl(\frac{j \, 2\pi i
t}{T}\Bigr)\Bigr| \\[0.15cm]
&& \hspace*{1.3cm} = \frac{|X_{i-1} - X'_{i-1}|}{\sqrt{n}} \leq
\frac{2}{\sqrt{n}}
\end{eqnarray*}
where the last inequality holds since $|X_{i-1}|=|X'_{i-1}|=1$.
Hence, McDiarmid's inequality in
Theorem~\ref{theorem: McDiarmid's inequality}
implies that, for every $\alpha \geq 0$,
\begin{eqnarray}
\pr( | \text{CF}_n(s) - \expectation [ \text{CF}_n(s)] | \geq
\alpha ) \leq 2 \exp\Bigl(-\frac{\alpha^2}{2}\Bigr) \label{eq:
McDiarmid's inequality for OFDM}
\end{eqnarray}
which demonstrates concentration of the CF around its
expected value. The improvement
of McDiarmid's inequality is by a factor of~2 in comparison
to the refined version of the Azuma--Hoeffding inequality in
Proposition~\ref{proposition: a similar scaling of the
concentration inequalities}. As will be seen in
Chapter~\ref{chapter: entropy method}, there are some deep connections
between McDiarmid's inequality and information-theoretic aspects;
McDiarmid's inequality will be proved in Chapter~\ref{chapter: entropy method}
by the use of the entropy method and information-theoretic tools, and it
will be proved useful in information-theoretic problems.

\bigskip
To conclude, three concentration inequalities
for the crest factor (CF) of OFDM signals have been derived in this section
under the assumption that the symbols are independent. The first
two concentration inequalities rely on the Azuma--Hoeffding inequality
and its refinement in Proposition~\ref{proposition: a similar scaling of the
concentration inequalities},
whereas the third bound is based on McDiarmid's inequality. Although these
concentration results are weaker than some existing results in the literature
(see \cite{LitsynW06} and \cite{WunderB_IT}), they establish concentration
in a rather simple way and provide some additional insight to the problem.
McDiarmid's inequality improves the exponent of the Azuma--Hoeffding inequality
by a factor of~$4$, and the exponent of the refined version of the Azuma--Hoeffding
inequality from Proposition~\ref{proposition: a similar scaling of
the concentration inequalities} by a factor of~$2$. Note, however,
that Proposition~\ref{proposition: a similar scaling of the
concentration inequalities} may, in general, be tighter than
McDiarmid's inequality (this happens to be the case if $\gamma < \frac{1}{4}$
in the setting of Proposition~\ref{proposition: a similar scaling of the
concentration inequalities}).

\subsection{Concentration of the cardinality of the fundamental
system of cycles for LDPC code ensembles} \label{subsection:
Concentration of the Cardinality of the Fundamental System of
Cycles}
Low-density parity-check (LDPC) codes are linear block
codes that are represented by sparse parity-check matrices
\cite{Gallager_1962}. A sparse parity-check matrix allows one to
represent the corresponding linear block code by a sparse
bipartite graph, and to use this graphical representation for
implementing low-complexity iterative message-passing decoding.
The low-complexity decoding algorithms used for LDPC codes and
some of their variants are remarkable in that they achieve rates
close to the Shannon capacity limit for properly designed code
ensembles (see, e.g., \cite{RiU_book}). As a result of their
remarkable performance under practical decoding algorithms, these
coding techniques have revolutionized the field of channel coding,
and have been incorporated in various digital communication
standards during the last decade.

In the following, we consider ensembles of binary LDPC codes. The
codes are represented by bipartite graphs, where the variable nodes
are located on the left side of the graph and the parity-check
nodes are on the right. The parity-check equations that define the
linear code are represented by edges connecting each check node
with the variable nodes that are involved in the corresponding
parity-check equation. The bipartite graphs representing these
codes are sparse in the sense that the number of edges in the
graph scales linearly with the block length $n$ of the code.
Following standard notation, let $\lambda_i$ and $\rho_i$ denote
the fraction of edges attached, respectively, to variable and
parity-check nodes of degree~$i$. The LDPC code ensemble is
denoted by $\text{LDPC}(n,\lambda,\rho)$, where $n$ is the block
length of the codes, and the pair $\lambda(x) \triangleq \sum_i
\lambda_i x^{i-1}$ and $\rho(x) \triangleq \sum_i \rho_i x^{i-1}$
represents, respectively, the left and right degree distributions
of the ensemble from the edge perspective. It is well-known that
linear block codes that can be represented by cycle-free bipartite
(Tanner) graphs have poor performance even under ML decoding
\cite{Cycle_free_codes}. The bipartite graphs of capacity-approaching
LDPC codes should therefore have cycles. Thus, we need to examine
the cardinality of the {\em fundamental system of cycles} of a bipartite
graph. For preliminary material, the reader is referred to
Sections~II-A and II-E of \cite{Sason_LDPC09}. In \cite{Sason_LDPC09}
and \cite{Sason_Eshel_ISIT11}, the following question is addressed:

\begin{quote} Consider an LDPC ensemble whose transmission takes
place over a memoryless binary-input output-symmetric channel, and
refer to the bipartite graphs which represent codes from this
ensemble, where every code is chosen uniformly at random from the
ensemble. How does the average cardinality of the fundamental
system of cycles of these bipartite graphs scale as a function of
the achievable gap to capacity? \end{quote}

An information-theoretic lower bound on
the average cardinality of the fundamental system of cycles was
derived in \cite[Corollary~1]{Sason_LDPC09}. This bound was expressed
in terms of the achievable gap to capacity (even under ML decoding)
when the communication takes place over a memoryless binary-input
output-symmetric channel. More explicitly, it was shown that the
number of fundamental cycles should grow at least like $\log
\frac{1}{\varepsilon}$, where
$\varepsilon$ denotes the gap in rate to capacity. This lower bound
diverges as the gap to capacity tends to zero, which is consistent
with the findings in \cite{Cycle_free_codes} on cycle-free codes,
and expresses quantitatively the necessity of cycles in bipartite
graphs that represent good LDPC code ensembles.
As a continuation of this work, we will now provide a
large-deviations analysis of the cardinality of
the fundamental system of cycles for LDPC code ensembles.

Let the triplet $(n, \lambda, \rho)$ represent an LDPC code
ensemble, and let $\mathcal{G}$ be a bipartite graph that
corresponds to a code from this ensemble. Then the cardinality of
the fundamental system of cycles of $\mathcal{G}$, denoted by
$\beta(\mathcal{G})$, is equal to
\begin{equation*}
\beta(\mathcal{G}) = |E(\mathcal{G})| - |V(\mathcal{G})| +
c(\mathcal{G})
\end{equation*}
where $E(\mathcal{G})$ and $V(\mathcal{G})$ are the edge and the
vertex sets of ${\cal G}$, and $c(\mathcal{G})$ denotes the number
of connected components of $\mathcal{G}$, and $|A|$ denotes the
cardinality of a set $A$. Let $R_{\rm d} \in [0,1)$ denote the
{\em design rate} of the ensemble. Then, in every bipartite graph
${\mathcal G}$ drawn from the ensemble, there are $n$ variable nodes
and $m = n(1-R_{\text{d}})$ parity-check nodes, for a total of
$|V(\mathcal{G})| = n(2-R_{\text{d}})$ nodes. If we let $a_{\text{R}}$
designate the average right degree (i.e., the average degree of
the parity-check nodes), then the number of edges in $\mathcal{G}$
is given by $|E(\mathcal{G})| = m a_{\text{R}}$. Therefore, for a
code from the $(n, \lambda, \rho)$ LDPC code ensemble, the
cardinality of the fundamental system of cycles satisfies the
equality
\begin{equation}
\beta(\mathcal{G}) = n \bigl[(1-R_{\text{d}}) a_{\text{R}} -
(2-R_{\text{d}}) \bigr] + c(\mathcal{G}) \label{eq: cardinality}
\end{equation}
where the design rate and the average right degree can be computed from
the degree distributions $\lambda$ and $\rho$ as
\begin{equation*}
R_{\text{d}} = 1 - \frac{\int_0^1 \rho(x) \; \mathrm{d}x}{\int_0^1
\lambda(x) \; \mathrm{d}x}, \quad a_{\text{R}} = \frac{1}{\int_0^1
\rho(x) \; \mathrm{d}x}.
\end{equation*}

Let
\begin{equation}
E \triangleq |E(\mathcal{G})| = n(1-R_{\text{d}}) a_{\text{R}}
\label{eq: number of edges}
\end{equation}
denote the number of edges of an arbitrary bipartite graph
$\mathcal{G}$ from the ensemble (for a fixed ensemble, we
will use the terms ``code" and ``bipartite graph" interchangeably).
Let us arbitrarily assign numbers $1, \ldots, E$ to the $E$ edges
of $\mathcal{G}$. Based on
Fact~\ref{fact: construction of martingales}, let us construct a martingale
sequence $X_0, \ldots, X_E$, where $X_i$ (for $i=0, 1, \ldots, E$)
is a random variable that denotes the conditional expected number of components
of a bipartite graph $\mathcal{G}$ chosen uniformly at random
from the ensemble, given that the first $i$ edges of the graph
$\mathcal{G}$ have been revealed. Note that the corresponding filtration
$\mathcal{F}_0 \subseteq \mathcal{F}_1 \subseteq \ldots \subseteq
\mathcal{F}_E$ in this case is defined so that $\mathcal{F}_i$ is
the $\sigma$-algebra generated by all the sets of
bipartite graphs from the considered ensemble whose first $i$
edges are fixed. For this martingale sequence, $$ X_0 =
\expectation_{\text{LDPC}(n,\lambda,\rho)}[\beta(\mathcal{G})],
\quad X_E = \beta(\mathcal{G})$$ and (a.s.) $|X_k - X_{k-1}| \leq
1$ for $k=1, \ldots, E$ (since revealing a new edge of
$\mathcal{G}$ can change the number of components in the graph
by at most~$1$). By Corollary~\ref{corollary: a tightened version
of the Azuma--Hoeffding inequality},
it follows that for every $\alpha \geq 0$
\begin{eqnarray}
&& \hspace*{-0.9cm} \pr \left( |c(\mathcal{G}) -
\expectation_{\text{LDPC}(n,\lambda,\rho)}[c(\mathcal{G})]| \geq
\alpha E \right) \leq 2 e^{-f(\alpha) E} \nonumber \\
&& \hspace*{-1.4cm} \Rightarrow \pr \left( |\beta(\mathcal{G}) -
\expectation_{\text{LDPC}(n,\lambda,\rho)}[\beta(\mathcal{G})]|
\geq \alpha E \right) \leq 2 e^{-f(\alpha) E} \label{eq:
Corollary 2 for the cardinality of the set of fundamental system
of cycles}
\end{eqnarray}
where the implication is a consequence of \eqref{eq: cardinality},
and the function $f$ was defined in \eqref{eq: f}. Hence, for
$\alpha > 1$, this probability is zero (since $f(\alpha) =
+\infty$ for $\alpha > 1$).
Note that, from \eqref{eq: cardinality},
$\expectation_{\text{LDPC}(n,\lambda,\rho)}[\beta(\mathcal{G})]$
scales linearly with $n$.
The combination of Eqs.~\eqref{eq: f}, \eqref{eq: number of edges},
\eqref{eq: Corollary 2 for the cardinality of the set of fundamental
system of cycles} gives the following statement:

\begin{theorem}
Let $\text{LDPC}(n,\lambda,\rho)$ be the LDPC code ensemble with block
length $n$ and a pair $(\lambda,\rho)$ of left and right degree
distributions (from the edge perspective).
Let $\mathcal{G}$ be a bipartite graph chosen uniformly at random
from this ensemble. Then, for every $\alpha \geq 0$, the
cardinality of the fundamental system of cycles of $\mathcal{G}$,
denoted by $\beta(\mathcal{G})$, satisfies the following inequality:
\begin{equation}\label{eq: LDPC concentration bound}
\pr \left( \bigl|\beta(\mathcal{G}) -
\expectation_{\text{LDPC}(n,\lambda,\rho)}[\beta(\mathcal{G})]\bigr|
\geq \alpha n \right) \leq 2 \cdot 2^{- \left[ 1 -
h_2\left(\frac{1-\eta}{2}\right) \right] \, \frac{\alpha n}{\eta}}
\end{equation}
where $h_2$ is the binary entropy function to the base~2,
$\eta \triangleq \frac{\alpha}{(1-R_{\text{d}}) \,
a_{\text{R}}}$, and $R_{\text{d}}$ and $a_{\text{R}}$ are,
respectively, the design rate and average right degree of the
ensemble. Consequently, if $\eta > 1$,
this probability is zero. \label{theorem: Concentration
result for the cardinality of the fundamental system of cycles}
\end{theorem}

\begin{remark}
We can obtain the following weakened version of \eqref{eq: LDPC concentration bound}
from the Azuma--Hoeffding inequality: for every $\alpha \geq 0$,
\begin{eqnarray*}
&& \pr \left( |\beta(\mathcal{G}) -
\expectation_{\text{LDPC}(n,\lambda,\rho)}[\beta(\mathcal{G})]|
\geq \alpha n \right) \leq 2 e^{-\frac{\alpha \eta n}{2}}
\end{eqnarray*}
where $\eta$ is defined in
Theorem~\ref{theorem: Concentration result for the cardinality
of the fundamental system of cycles} (note that
$\frac{\alpha}{\eta} = \frac{E}{n}$ is equal to the average degree
of the variable nodes). The exponential decay of the last two bounds
is similar for values of $\alpha$ close to zero (see the exponents of
the Azuma--Hoeffding inequality
and Corollary~\ref{corollary: a tightened version of the Azuma--Hoeffding
inequality} in Figure~\ref{Figure: compare_exponents_theorem2}).
\end{remark}

\begin{remark}
For various capacity-achieving sequences of LDPC code ensembles on
the binary erasure channel, the average right degree scales like
$\log \frac{1}{\varepsilon}$ where $\varepsilon$ denotes the
fractional gap to capacity under belief-propagation decoding
(i.e., $R_{\text{d}} = (1-\varepsilon)C$) \cite{LubyMSS_IT01}.
Therefore, for small values of $\alpha$, the exponential decay
rate in the inequality of Theorem~\ref{theorem: Concentration
result for the cardinality of the fundamental system of cycles}
scales like $ \left(\log \frac{1}{\varepsilon} \right)^{-2}$. This
large-deviations result complements the result in
\cite[Corollary~1]{Sason_LDPC09}, which provides a lower bound on the
average cardinality of the fundamental system of cycles that
scales like $\log \frac{1}{\varepsilon}$.
\end{remark}

\begin{remark}
Consider small deviations from the expected value that scale like
$\sqrt{n}$. Note that Corollary~\ref{corollary: a tightened
version of the Azuma--Hoeffding inequality}
is a special case of Theorem~\ref{theorem: first refined
concentration inequality} when $\gamma = 1$ (i.e., when only an
upper bound on the differences of the martingale sequence is available,
but there is no non-trivial upper bound on the conditional
variance). Hence, it follows from Proposition~\ref{proposition: a
similar scaling of the concentration inequalities} that, in this case,
Corollary~\ref{corollary: a tightened version of the Azuma--Hoeffding inequality}
does not provide any improvement in the exponent of the concentration
inequality (in comparison to the Azuma--Hoeffding inequality) when small
deviations are considered.
\end{remark}

\subsection{Concentration theorems for LDPC code ensembles over ISI channels}
\label{section: Theorems- message errors with ISI}
Concentration analysis of the number of erroneous variable-to-check messages
for random ensembles of LDPC codes was introduced in \cite{RichardsonU2001}
and \cite{LubyMSS_IT01_paper2} for memoryless channels. It was shown that the
performance of an individual code from the ensemble concentrates around the
expected (average) value over this ensemble  when the length of the block length
of the code tends to infinity, and that this average performance converges asymptotically
to the performance in the cycle-free case (when the bipartite graph that represents
a linear code contains no cycles, the messages that are delivered by the message-passing
decoder through the edges of the graph are statistically independent \cite{RiU_book}).
These concentration results were later generalized in~\cite{KavcicMM_IT2003}
for intersymbol-interference (ISI) channels. The proofs of
\cite[Theorems~1 and 2]{KavcicMM_IT2003}, which refer to regular LDPC code ensembles,
are revisited in the following in order to derive an explicit expression for the
exponential rate of the concentration inequality. It is then shown that particularizing
the expression for memoryless channels provides a tightened concentration inequality
in comparison to \cite{RichardsonU2001} and \cite{LubyMSS_IT01_paper2}. The presentation
in the following is based on \cite{Eshel_MSc2011}.

\subsubsection{The ISI channel and its message-passing decoding}
\label{subsection: channel and its decoding }
We start by briefly describing the ISI channel and the graph used
for its message-passing decoding. For a detailed description, the
reader is referred to \cite{KavcicMM_IT2003}.
Consider a binary discrete-time ISI channel with a finite memory length,
denoted by~$I$. The channel output $Y_j$ at time instant $j$ is given by
\begin{equation*}
Y_j = \sum_{i=0}^I h_i X_{j-i} \, + N_j, \quad \forall \, j \in \integers
\end{equation*}
where $\{X_j\}$ is a sequence of $\{-1,+1\}$-valued binary inputs,
$\{h_i\}_{i=0}^I$ is the input response of the ISI channel, and
$\{N_j\}$ is a sequence of i.i.d.\ Gaussian random variables with zero
mean and variance $\sigma^2$. It is assumed that an information block of
length $k$ is encoded by using a regular $(n,d_{\text{v}},d_{\text{c}})$
LDPC code, and the resulting $n$ coded bits are converted into a channel
input sequence before its transmission over the channel.
For decoding, we consider the windowed version of the sum-product algorithm
when applied to ISI channels (for specific details about this decoding algorithm,
the reader is referred to \cite{KavcicMM_IT2003} and \cite{Douillard_sept95};
in general, it is an iterative message-passing decoding algorithm). The
variable-to-check and check-to-variable messages are computed as in the
sum-product algorithm for the memoryless case with the difference that a
message that is received from the channel at a variable node is not only
a function of the channel output that corresponds to the considered symbol,
but it is also a function of the $2W$ neighboring channel outputs and $2W$ neighboring
variables nodes (as is illustrated in Fig.~\ref{Fig :ISI message flow- depth 1}).

\begin{figure}
\begin{center}
\epsfig{file=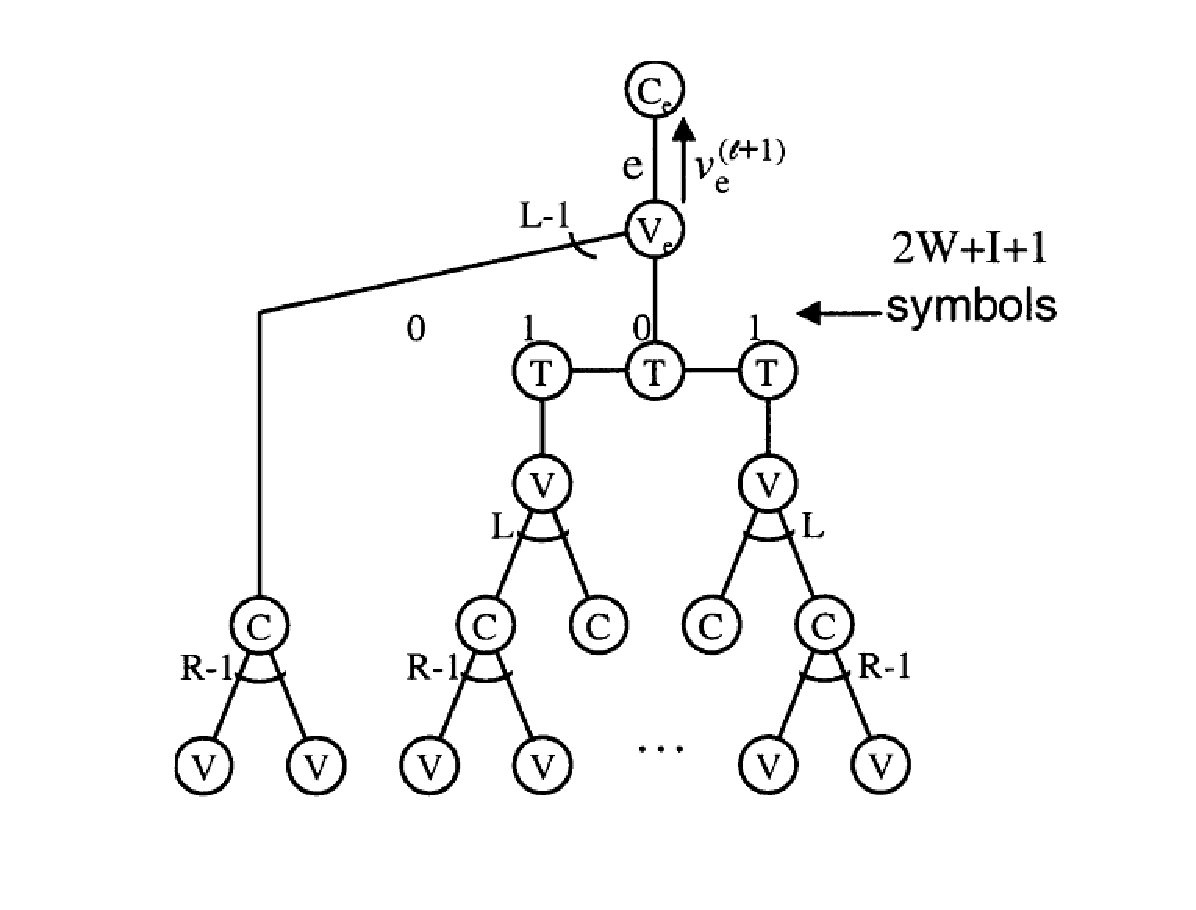,scale=0.4}
\end{center}
\caption{Message flow neighborhood of depth 1.
This figure corresponds to the parameters
$(I,W,d_{\text{v}}=L,d_{\text{c}}=R)=(1,1,2,3).$}
\label{Fig :ISI message flow- depth 1}
\end{figure}

\subsubsection{Concentration}
We prove that, for a large $n$, a neighborhood of depth $\ell$
of a variable-to-check node message is tree-like with high probability. Using this
result in conjunction with the Azuma--Hoeffding inequality, we will then show that,
for most graphs and channel realizations, if $\underline{s}$ is the transmitted
codeword, then the probability of a variable-to-check message being erroneous after
$\ell$ rounds of message-passing decoding is highly concentrated around its expected
value. This expected value is shown to converge to the value of
$p^{(\ell)}(\underline{s})$ that corresponds to the cycle-free case.

In the following theorems, we consider an ISI channel and windowed message-passing
decoding algorithm, where the code graph is chosen uniformly at random from the
ensemble of graphs with variable and check node degrees $d_{\text{v}}$ and
$d_{\text{c}}$, respectively. Let $\mathcal{N}^{(\ell)}_{\vec{e}}$ denote the
neighborhood of depth $\ell$ of an edge $\vec{e}=(\text{v},\text{c})$ between a
variable-to-check node. Let ${N}_{\text{c}}^{(\ell)}$, ${N}_{\text{v}}^{(\ell)}$
and ${N}_{\text{e}}^{(\ell)}$ denote, respectively, the total number of check nodes, variable
nodes and code-related edges in this neighborhood. Similarly, let ${N}_Y^{(\ell)}$
denote the number of variable-to-check node messages in the directed neighborhood
of depth $\ell$ of a received symbol of the channel (explicit expressions are given in
Appendix~\ref{appendix: proof of concentration results on the performance of LDPC code ensembles}).

\begin{theorem}
Let $P^{(\ell)}_{\overline{\text{t}}}\equiv \Pr\left\{ \mathcal{N}^{(\ell)}_{\vec{e}}
\text{ not a tree}\right\}$ denote the probability that the sub-graph $\mathcal{N}^{(\ell)}_{\vec{e}} $
is not a tree (i.e., it contains cycles).
Then, there exists a positive constant $\gamma \triangleq \gamma(d_{\text{v}},d_{\text{c}},\ell)$
that does not depend on the block-length $n$, such that $P^{(\ell)}_{\overline{\text{t}}} \leq \frac{\gamma}{n}$.
More explicitly, one can choose
$\gamma(d_{\text{v}},d_{\text{c}},\ell) \triangleq \bigl(N_{\text{v}}^{(\ell)}\bigr)^2
+ \bigl(\frac{d_{\text{c}}}{d_{\text{v}}} \cdot N_{\text{c}}^{(\ell)} \bigr)^2$.
\label{theorem: probability of a subgraph that is not a tree}
\end{theorem}
\begin{proof}
This proof is a straightforward generalization of the proof in~\cite{RichardsonU2001}
(for binary-input output-symmetric memoryless channels) to binary-input ISI channels.
A detailed proof is available in \cite{Eshel_MSc2011}.
\end{proof}

The following concentration inequalities follow from
Theorem~\ref{theorem: probability of a subgraph that is not a tree}
and the Azuma--Hoeffding inequality:
\begin{theorem}
Let $\underline{s}$ be the transmitted codeword, and let $Z^{(\ell)}(\underline{s})$
be the number of erroneous variable-to-check messages after $\ell$ rounds of the
windowed message-passing decoding algorithm.
Let $p^{(\ell)}(\underline{s})$ be the expected fraction of incorrect messages passed
through an edge with a tree-like directed neighborhood of depth $\ell$.
Then there exist some positive constants $\beta$ and $\gamma$ that do not depend on the
block-length $n$, such that the following statements hold:
\medskip

\noindent\textbf{Concentration around the expected value.}
For any $\eps>0$,
\begin{equation} \label{eq : th1-Concentration around expectation}
\pr \left(\left|\frac{Z^{(\ell)}(\underline{s})}{nd_{\text{v}}}-
\frac{\mathbb{E}[Z^{(\ell)}(\underline{s})]}{nd_{\text{v}}}\right|>\eps/2\right)
\leq 2 e^{-\beta{\eps^2}n}.
\end{equation}

\noindent\textbf{Convergence of the expected value to the cycle-free case.}
For any $\eps>0$ and $n>\frac{2\gamma}{\eps}$, we have a.s.
\begin{equation} \label{eq : th1-Convergence to cycle-free case}
\left|\frac{\mathbb{E}[Z^{(\ell)}(\underline{s})]}{nd_{\text{v}}}
-p^{(\ell)}(\underline{s})\right| \leq \eps/2.
\end{equation}

\noindent\textbf{Concentration around the cycle-free case.}
For any $\eps>0$ and $n>\frac{2\gamma}{\eps}$,
\begin{equation} \label{eq : th1-Concentration around cycle-free case}
\pr \left(\left|\frac{Z^{(\ell)}(\underline{s})}{nd_{\text{v}}}
-p^{(\ell)}(\underline{s})\right| > \eps \right) \leq 2 e^{-\beta{\eps^2}n}.
   \end{equation}

More explicitly, the above statements hold for $$\beta \triangleq
\beta(d_{\text{v}},d_{\text{c}},\ell)=
\frac{d_{\text{v}}^2}{8\left(4d_{\text{v}}{({N}_{\text{e}}^{(\ell)})^2}
+{({N}_Y^{(\ell)})^2}\right)},$$ and
$$\gamma \triangleq \gamma(d_{\text{v}},d_{\text{c}},\ell)=
\bigl(N_{\text{v}}^{(\ell)}\bigr)^2+ \left(\frac{d_{\text{c}}}{d_{\text{v}}}
\cdot {N_{\text{c}}^{(\ell)}}\right)^2.$$
\label{theorem: concentration results on the performance of LDPC code ensembles}
\end{theorem}

\begin{proof}
See Appendix~\ref{appendix: proof of concentration results on the performance of LDPC code ensembles}.
\end{proof}

The concentration inequalities in
Theorem~\ref{theorem: concentration results on the performance of LDPC code ensembles}
extend the results in~\cite{RichardsonU2001} from the special setting of memoryless binary-input
output-symmetric (MBIOS) channels to ISI channels. One can particularize the above
expression for $\beta$ to MBIOS channels by setting $W=0$ and $I=0$. Since the proof
of Theorem~\ref{theorem: concentration results on the performance of LDPC code ensembles}
uses exact expressions for ${N}_{\text{e}}^{(\ell)}$ and ${N}_Y^{(\ell)}$,
one would expect a tighter bound in comparison to the value of $\beta$ in~\cite{RichardsonU2001},
which is given by $\frac{1}{\beta}=544d_{\text{v}}^{2\ell-1}d_{\text{c}}^{2\ell}$.
As an example, for $(d_{\text{v}},d_{\text{c}},\ell)=(3,4,10)$,
one gets an improvement by a factor of about $1$~million. However, even with
this improvement, the required size of $n$ according to the analysis in this section
can be absurdly large. This is because the proof is very pessimistic in the sense that it assumes
that any change in an edge or the decoder's input introduces an error in every message
it affects. This is especially pessimistic if a large $\ell$ is considered, because the
neighborhood grows with $\ell$, so each message is a function of many edges and received
output symbols from the channel.

The same concentration phenomena that are established above for regular LDPC
code ensembles can be extended to irregular LDPC code ensembles as well. In the
special case of MBIOS channels, the
following theorem was proved by Richardson and Urbanke in \cite[pp.~487--490]{RiU_book},
based on the Azuma--Hoeffding inequality (we use here the same notation for LDPC
code ensembles as in Section~\ref{subsection: Concentration of the Cardinality of
the Fundamental System of Cycles}):
\begin{theorem}
Let $\mathcal{C}$, a code chosen uniformly at
random from the ensemble $\text{LDPC}(n, \lambda, \rho)$, be used
for transmission over an MBIOS channel characterized by its L-density
$a_{\text{MBIOS}}$ (this denotes the conditional {\em pdf} of the
log-likelihood ratio $L \triangleq l(Y) =
\ln \left(\frac{p_{Y|X}(Y|1)}{p_{Y|X}(Y|-1)}\right)$,
given that $X=1$ is the transmitted symbol).
Assume that the decoder performs $l$ iterations of message-passing
decoding, and let $P_{\text{b}}(\mathcal{C}, a_{\text{MBIOS}}, l)$
denote the resulting bit error probability. Then, for every $\delta
> 0$, there exists a positive $\alpha$ where $\alpha = \alpha(\lambda,
\rho, \delta, l)$ is {\em independent of the block length $n$}, such that
the following concentration inequality holds:
\begin{equation*}
\pr\left( |P_{\text{b}}(\mathcal{C},
a_{\text{MBIOS}}, l) - \expectation_{\text{LDPC}(n,\lambda,\rho)}
[P_{\text{b}}(\mathcal{C}, a_{\text{MBIOS}}, l)] | \geq \delta
\right) \leq \exp(-\alpha n).
\end{equation*}
\end{theorem}
This theorem asserts that the performance of all codes, except for
a fraction which is exponentially small in the block length $n$, is
with high probability arbitrarily close to the ensemble average.
Hence, assuming a sufficiently large block length, the ensemble
average is a good indicator for the performance of individual codes;
it is therefore reasonable to focus on the design and analysis of
capacity-approaching ensembles (via the density evolution
technique \cite{RichardsonU2001}). This forms a fundamental result
in the theory of codes on graphs and iterative decoding.

\subsection{On the concentration of the conditional entropy for LDPC code ensembles}
\label{section: A Tightened Concentration Result for the Conditional Entropy}
A large-deviation analysis of the conditional entropy for random
ensembles of LDPC codes was introduced by M\'easson, Montanari and
Urbanke in \cite[Theorem~4]{MeassonMU08} and \cite[Theorem~1]{Montanari05}.
The following theorem is proved in \cite[Appendix~I]{MeassonMU08},
based on the Azuma--Hoeffding inequality (although here we rephrase
it to consider small deviations of order $\sqrt{n}$, instead of large
deviations of order $n$):
\begin{theorem}
Let $\mathcal{C}$ be chosen uniformly at random from the ensemble
$\text{LDPC}(n,\lambda,\rho)$. Assume that the transmission of the
code $\mathcal{C}$ takes place over an MBIOS channel. Let
$H({\bf{X}}|{\bf{Y}})$ denote the conditional entropy of the
transmitted codeword ${\bf{X}}$ given the received sequence
${\bf{Y}}$ from the channel. Then, for every $\xi > 0$,
\begin{equation*}
\pr\bigl( \big|H({\bf{X}}|{\bf{Y}}) -
\expectation_{\text{LDPC}(n,\lambda,\rho)}[H({\bf{X}}|{\bf{Y}})] \big|
\geq \xi \, \sqrt{n} \, \bigr) \leq 2 \exp(-B \xi^2)
\end{equation*}
where $B \triangleq \frac{1}{2(d_{\text{c}}^{\max}+1)^2 (1-R_d)}$,
$d_{\text{c}}^{\max}$ is the maximal check-node degree, and
$R_{\text{d}}$ is the design rate of the ensemble. \label{theorem:
Concentration of Conditional Entropy}
\end{theorem}

In this section, we revisit the proof of Theorem~\ref{theorem:
Concentration of Conditional Entropy}, originally given in
\cite[Appendix~I]{MeassonMU08}, in order to derive a tightened
version of this bound. To that end, let $\mathcal{G}$ be a
bipartite graph that represents a code chosen uniformly at random
from the ensemble LDPC$(n, \lambda, \rho)$. Define the random variable
\begin{equation*}
Z = H_{\mathcal{G}}({\bf{X}}|{\bf{Y}}),
\end{equation*}
i.e., the conditional entropy when the transmission is
over an MBIOS channel with transition probabilities
$P_{{\bf{Y}}|{\bf{X}}}({\bf{y}} | {\bf{x}}) = \prod_{i=1}^n
p_{Y|X}(y_i | x_i)$, where (by output symmetry)
$p_{Y|X}(y|1) = p_{Y|X}(-y|0)$. Fix an arbitrary order for
the $m = n(1-R_{\text{d}})$ parity-check nodes,
where $R_{\text{d}}$ is the design rate of the LDPC code
ensemble. Let $\left\{\mathcal{F}_t\right\}_{t \in \left\{0, 1, \ldots, m\right\}}$
form a filtration of $\sigma$-algebras $\mathcal{F}_0 \subseteq
\mathcal{F}_1 \subseteq \ldots \subseteq \mathcal{F}_m$ where
$\mathcal{F}_t$ (for $t=0,1,\ldots,m$) is the $\sigma$-algebra
generated by all the subsets of $m \times n$ parity-check
matrices that are characterized by the pair of degree
distributions $(\lambda, \rho)$, and whose first $t$ parity-check
equations are fixed (for $t=0$ nothing is fixed, and therefore
$\mathcal{F}_0 = \left\{\emptyset, \Omega\right\}$ where $\emptyset$ denotes
the empty set, and $\Omega$ is the whole sample space of $m \times n$ binary
parity-check matrices that are characterized by the pair of degree
distributions $(\lambda, \rho)$). Accordingly, based on
Fact~\ref{fact: construction of martingales} in
Section~\ref{section: Discrete-Time Martingales}, let us define the following
martingale sequence:
\begin{equation*}
Z_t = \expectation[Z | \mathcal{F}_t] \, \quad t \in \left\{0, 1, \ldots,
m\right\}.
\end{equation*}
By construction, $Z_0 =
\expectation[H_{\mathcal{G}}({\bf{X}}|{\bf{Y}})]$ is the expected
value of the conditional entropy with respect to\ the LDPC code ensemble,
and $Z_m$ is the random variable that is equal a.s.\ to the conditional entropy of
the particular code from the ensemble. Similarly to
\cite[Appendix~I]{MeassonMU08}, we obtain upper bounds on the
differences $|Z_{t+1} - Z_{t}|$ and then rely on the Azuma--Hoeffding
inequality in Theorem~\ref{theorem: Azuma--Hoeffding inequality}.

Without loss of generality, we can order the parity-check nodes
by increasing degree, as done in \cite[Appendix~I]{MeassonMU08}.
Let ${\bf{r}} = (r_1, r_2, \ldots)$ be the set of parity-check
degrees in ascending order, and $\Gamma_i$ be the fraction of
parity-check nodes of degree~$i$. Hence, the first
$m_1 = n(1-R_{\text{d}}) \Gamma_{r_1}$ parity-check nodes are
of degree $r_1$, the successive
$m_2 = n(1-R_{\text{d}}) \Gamma_{r_2}$ parity-check
nodes are of degree $r_2$, and so on. The $(t+1)$th parity-check
will therefore have a well-defined degree, which we denote by $r$.
From the proof in \cite[Appendix~I]{MeassonMU08},
\begin{equation}
|Z_{t+1} - Z_{t}| \leq (r+1) \,
H_{\mathcal{G}}(\tilde{X}|{\bf{Y}}) \label{eq: result from the
original proof}
\end{equation}
where $H_{\mathcal{G}}(\tilde{X}|{\bf{Y}})$ is a random variable that is
equal to the conditional entropy of a parity-bit $\tilde{X} = X_{i_1} \oplus
\ldots \oplus X_{i_r}$ (i.e., $\tilde{X}$ is equal to the modulo-2
sum of some $r$ bits in the codeword ${\bf{X}}$) given the received
sequence ${\bf{Y}}$ at the channel output. The proof in
\cite[Appendix~I]{MeassonMU08} was then completed by upper-bounding
the parity-check degree $r$ by the maximal parity-check degree
$d_{\text{c}}^{\max}$, and also by upper-bounding the conditional
entropy of the parity-bit $\tilde{X}$ by~1. This gives
\begin{equation}
|Z_{t+1} - Z_{t}| \leq d_{\text{c}}^{\max}+1 \quad t = 0, 1,
\ldots, m-1 \label{eq: bound on the differences of Z_t}
\end{equation}
which, together with the Azuma--Hoeffding inequality, completes the
proof of Theorem~\ref{theorem: Concentration of Conditional Entropy}.
Note that the $d_i$'s in Theorem~\ref{theorem: Azuma--Hoeffding inequality}
are equal to $d_{\text{c}}^{\max}+1$, and $n$ in
Theorem~\ref{theorem: Azuma--Hoeffding inequality} is
replaced with the length $m=n(1-R_{\text{d}})$ of the martingale
sequence $\left\{Z_t\right\}$ (that is equal to the number of the
parity-check nodes in the graph).

Based on \cite{Sason_Eshel_ISIT11}, a refined analysis is provided;
it departs from the analysis in \cite[Appendix~I]{MeassonMU08} in
two respects:
\begin{itemize}
\item The first difference is related to the upper bound on the
conditional entropy $H_{\mathcal{G}}(\tilde{X}|{\bf{Y}})$ in
\eqref{eq: result from the original proof}, where $\tilde{X}$ is the
modulo-2 sum of some $r$ bits of the transmitted codeword ${\bf{X}}$
given the channel output ${\bf{Y}}$. Instead of taking the most
trivial upper bound that is equal to~$1$, as was done in
\cite[Appendix~I]{MeassonMU08}, we derive a simple upper bound that
depends on the parity-check degree $r$ and the channel capacity $C$
(see Proposition~\ref{proposition: upper bounds on the conditional
entropy}).
\item The second difference is minor, but it proves to be helpful for
tightening the concentration inequality for LDPC code ensembles
that are not right-regular (i.e., the case where the degrees of
the parity-check nodes are not fixed to a certain value). Instead
of upper-bounding the term $r+1$ on the right-hand side of
\eqref{eq: result from the original proof} with
$d_{\text{c}}^{\max}+1$, we propose to leave it as is, since the
Azuma--Hoeffding inequality applies to the case when the bounded
differences of the martingale sequence are not fixed (see
Theorem~\ref{theorem: Azuma--Hoeffding inequality}), and
since the number of the parity-check nodes of degree $r$ is equal
to $n (1-R_{\text{d}}) \Gamma_{r}$. The effect of this simple
modification will be shown in Example~\ref{example: tornado
codes}.
\end{itemize}

The following upper bound is related to the first item above:
\begin{proposition}
Let $\mathcal{G}$ be a bipartite graph which corresponds to a
binary linear block code used for transmission over an
MBIOS channel. Let ${\bf{X}}$ and ${\bf{Y}}$ designate the
transmitted codeword and received sequence at the channel output.
Let $\tilde{X} = X_{i_1} \oplus \ldots \oplus X_{i_r}$ be a
parity-bit of some $r$ code bits of ${\bf{X}}$. Then, the
conditional entropy of $\tilde{X}$ given ${\bf{Y}}$ satisfies
\begin{equation}
H_{\mathcal{G}}(\tilde{X}|{\bf{Y}}) \leq h_2
\left(\frac{1-C^{\frac{r}{2}}}{2} \right). \label{eq: tightened
bound on the conditional entropy of a parity-bit for an MBIOS
channel}
\end{equation}
Furthermore, for a binary symmetric channel (BSC) or a binary erasure
channel (BEC), this bound can be improved to
\begin{equation}
H_{\mathcal{G}}(\tilde{X}|{\bf{Y}}) \leq h_2 \left(\frac{1 -
\bigl[1-2 h_2^{-1}(1-C) \bigr]^r}{2} \right)
\label{eq: tightened bound on the conditional entropy of a parity-bit for a BSC}
\end{equation}
and
\begin{equation}
H_{\mathcal{G}}(\tilde{X}|{\bf{Y}}) \leq 1-C^r
\label{eq: tightened bound on the conditional entropy of a parity-bit for a BEC}
\end{equation}
respectively, where $h_2^{-1}$ in \eqref{eq: tightened bound on the
conditional entropy of a parity-bit for a BSC} denotes the
inverse of the binary entropy function to base~2.
\label{proposition: upper bounds on the conditional entropy}
\end{proposition}

Note that if the MBIOS channel is perfect (i.e., its capacity is
$C=1$ bit per channel use), then \eqref{eq: tightened bound on the
conditional entropy of a parity-bit for an MBIOS channel} holds
with equality (where both sides of \eqref{eq: tightened bound on
the conditional entropy of a parity-bit for an MBIOS channel} are
zero), whereas the trivial upper bound is 1.
\begin{proof}
Since conditioning reduces the entropy, we have $$H(\tilde{X} \big{|}
{\bf{Y}} ) \leq H(\tilde{X} \big{|} Y_{i_1}, \ldots, Y_{i_r} ).$$
Note that $Y_{i_1}, \ldots, Y_{i_r}$ are the channel outputs that correspond
to the channel inputs $X_{i_1}, \ldots X_{i_r}$, where these $r$ bits are
used to calculate the parity-bit $\tilde{X}$. Hence, by combining the
last inequality with \cite[Eq.~(17) and Appendix~I]{Sason_LDPC09}, we can show that
\begin{equation}
H(\tilde{X} \big{|} {\bf{Y}}) \leq 1 - \frac{1}{2 \ln 2}
\sum_{k=1}^{\infty} \frac{(g_k)^r}{k(2k-1)} \label{eq: bound on
conditional entropy}
\end{equation}
where (see \cite[Eq.~(19)]{Sason_LDPC09})
\begin{equation}
g_k \triangleq \int_{0}^{\infty} a(l) (1+e^{-l})
\tanh^{2k}\left(\frac{l}{2}\right) {\mathrm d}l, \quad \forall \, k \in
\naturals \label{eq: g_p}
\end{equation}
and $a(\cdot)$ denotes the symmetric pdf of the log-likelihood
ratio at the output of the MBIOS channel, given that the channel
input is equal to zero. From \cite[Lemmas~4 and 5]{Sason_LDPC09}, it
follows that $g_k \geq C^k$ for every $k \in \naturals$.
Substituting this inequality in \eqref{eq: bound on conditional
entropy} gives
\begin{eqnarray}
H(\tilde{X} \big{|} {\bf{Y}}) &\leq& 1 - \frac{1}{2 \ln 2}
\sum_{k=1}^{\infty} \frac{C^{kr}}{k(2k-1)} \nonumber\\
&=& h_2 \left(\frac{1-C^{\frac{r}{2}}}{2} \right) \label{eq: h_2
bound}
\end{eqnarray}
where the last equality follows from the power series expansion
of the binary entropy function:
\begin{equation}
h_2(x) = 1 - \frac{1}{2 \ln 2} \sum_{k=1}^{\infty}
\frac{(1-2x)^{2k}}{k(2k-1)}, \quad 0 \leq x \leq 1. \label{eq:
power series expansion of the binary entropy function}
\end{equation}
This proves the result in \eqref{eq: tightened
bound on the conditional entropy of a parity-bit for an MBIOS
channel}.

The tightened bound on the conditional entropy for the BSC is
obtained from \eqref{eq: bound on conditional entropy} and the
equality
\begin{equation*}
g_k = \bigl(1-2 h_2^{-1}(1-C)\bigr)^{2k}, \quad \forall \, k \in
\naturals
\end{equation*}
that holds for the BSC (see \cite[Eq.~(97)]{Sason_LDPC09}). This replaces
$C$ on the right-hand side of \eqref{eq: h_2 bound} with $\bigl(1-2
h_2^{-1}(1-C)\bigr)^{2}$, thus leading to the tightened bound in
\eqref{eq: tightened bound on the conditional entropy of a
parity-bit for a BSC}.

The tightened result for the BEC follows from
\eqref{eq: bound on conditional entropy} where,
from \eqref{eq: g_p},
\begin{equation*}
g_k = C, \quad \forall \, k \in \naturals
\end{equation*}
(see \cite[Appendix~II]{Sason_LDPC09}). Substituting $g_k$
into the right-hand side of \eqref{eq: bound on conditional entropy}
gives \eqref{eq: tightened bound on the conditional entropy of a
parity-bit for a BSC} (note that $\sum_{k=1}^{\infty}
\frac{1}{k(2k-1)}=2 \ln 2$). This completes the proof of
Proposition~\ref{proposition: upper bounds on the conditional
entropy}.
\end{proof}

From Proposition~\ref{proposition: upper bounds on the conditional
entropy} and \eqref{eq: result from the original proof}, we get
\begin{equation}
|Z_{t+1}-Z_t| \leq (r+1) \, h_2 \left(\frac{1-C^{\frac{r}{2}}}{2}
\right), \label{eq: second tightened inequality}
\end{equation}
where the two improvements for the BSC and BEC are obtained by replacing
the second term, $h_2(\cdot)$, on the right-hand side of \eqref{eq: second
tightened inequality} by \eqref{eq: tightened bound on
the conditional entropy of a parity-bit for a BSC} and \eqref{eq:
tightened bound on the conditional entropy of a parity-bit for a
BEC}, respectively. This improves upon the earlier bound of
$(d_{\text{c}}^{\max}+1)$ in \cite[Appendix~I]{MeassonMU08}.
From \eqref{eq: second tightened inequality} and
Theorem~\ref{theorem: Azuma--Hoeffding inequality}, we obtain
the following tightened version of the concentration inequality
in Theorem~\ref{theorem: Concentration of Conditional Entropy}:

\begin{theorem}
Let $\mathcal{C}$ be chosen uniformly at random from the ensemble
$\text{LDPC}(n,\lambda,\rho)$. Assume that the transmission of the
code $\mathcal{C}$ takes place over an MBIOS channel.
Let $H({\bf{X}}|{\bf{Y}})$ designate the conditional entropy of
the transmitted codeword ${\bf{X}}$ given the received sequence
${\bf{Y}}$ at the channel output. Then, for every $\xi > 0$,
\begin{equation}
\pr \big( \big|H({\bf{X}}|{\bf{Y}}) -
\expectation_{\text{LDPC}(n,\lambda,\rho)}[H({\bf{X}}|{\bf{Y}})] \big|
\geq \xi \sqrt{n} \, \big) \leq 2 \exp(-B \xi^2),
\label{eq: large deviations result for the conditional entropy}
\end{equation}
where
\begin{equation}
B \triangleq \frac{1}{2(1-R_{\text{d}})
\dsum_{i=1}^{d_{\text{c}}^{\max}} \left\{ (i+1)^2 \, \Gamma_i \; \biggl[
h_2\left(\frac{1-C^{\frac{i}{2}}}{2} \right) \biggr]^2 \right\}} \, , \label{eq: B}
\end{equation}
$d_{\text{c}}^{\text{max}}$ is the maximal check-node
degree, $R_{\text{d}}$ is the design rate of the ensemble, and $C$
is the channel capacity (in bits per channel use).
Furthermore, for a binary symmetric channel (BSC) or a binary erasure
channel (BEC), the parameter $B$ on the right-hand side of
\eqref{eq: large deviations result for the conditional entropy}
can be improved (i.e., increased), respectively, to
\begin{equation*}
B_{\text{BSC}} \triangleq \frac{1}{2(1-R_{\text{d}})
\dsum_{i=1}^{d_{\text{c}}^{\max}} \left\{ (i+1)^2 \, \Gamma_i \; \biggl[
h_2\left(\frac{1-[1-2h_2^{-1}(1-C)]^i}{2} \right) \biggr]^2 \right\} }
\end{equation*}
and
\begin{equation}
B_{\text{BEC}} \triangleq \frac{1}{2(1-R_{\text{d}})
\dsum_{i=1}^{d_{\text{c}}^{\max}} \left\{ (i+1)^2 \, \Gamma_i \; (1-C^i)^2 \right\}} \, .
\label{eq: B for BEC}
\end{equation}
\label{theorem: Tightened Concentration of Conditional Entropy}
\end{theorem}

\begin{remark}
From \eqref{eq: B}, Theorem~\ref{theorem: Tightened Concentration of
Conditional Entropy} indeed yields a stronger concentration
inequality than the one in Theorem~\ref{theorem: Concentration of Conditional
Entropy}.
\end{remark}

\begin{remark}
In the limit where $C \rightarrow 1$ bit per channel use, it follows
from \eqref{eq: B} that, if $d_{\text{c}}^{\max} < \infty$, then $B
\rightarrow \infty$. This is in contrast to the value of $B$ in
Theorem~\ref{theorem: Concentration of Conditional Entropy}, which
does not depend on the channel capacity and is finite. Note that $B$
should indeed be infinity for a perfect channel, and therefore
Theorem~\ref{theorem: Tightened Concentration of Conditional
Entropy} is tight in this case. Moreover, in the case where
$d_{\text{c}}^{\max}$ is not finite, we prove the following:
\begin{lemma}
If $d_{\text{c}}^{\max} = \infty$ and $\rho'(1) < \infty$, then $B
\rightarrow \infty$ in the limit where $C \rightarrow 1$.
\label{lemma: lemma for the case of an infinite maximal degree}
\end{lemma}
\begin{proof}
See Appendix~\ref{appendix: proof of the lemma for the case of an
infinite maximal degree}.
\end{proof}
This is in contrast to the value of $B$ in Theorem~\ref{theorem:
Concentration of Conditional Entropy}, which vanishes when
$d_{\text{c}}^{\max} = \infty$, making it useless in this
case (see Example~\ref{example: tornado codes}).
\end{remark}

\begin{example}[Comparison of Theorems~\ref{theorem: Concentration of Conditional
Entropy} and~\ref{theorem: Tightened Concentration of Conditional
Entropy} for right-regular LDPC code ensembles] Let us examine
the improvement resulting from the tighter bounds in
Theorem~\ref{theorem: Tightened Concentration of Conditional
Entropy} for right-regular LDPC code ensembles. Consider the case
where the communication takes place over a binary-input additive
white Gaussian noise channel (BIAWGNC) or a BEC. Let us consider
the $(2, 20)$ regular LDPC code ensemble whose design rate is
equal to $0.900$ bits per channel use. For a BEC, the threshold of
the channel bit erasure probability under belief-propagation (BP)
decoding is given by
\begin{equation*}
p_{\text{BP}} = \inf_{x \in (0,1]} \frac{x}{1-(1-x)^{19}} = 0.0531,
\end{equation*}
which corresponds to a channel capacity of $C=0.9469$ bits per
channel use (note that the above calculation of $p_{\text{BP}}$
for the BEC follows from the fixed-point characterization of the
threshold in \cite[Theorem~3.59]{RiU_book} with the pair of degree
distributions $\lambda(x)=x$ and $\rho(x)=x^{19}$).
For the BIAWGNC, the threshold under BP decoding is
equal to $\sigma_{\text{BP}} = 0.4156590$ (this numerical
result is based on a computation that follows from
\cite[Example~11]{RichardsonU2001}). From
\cite[Example~4.38]{RiU_book} that expresses the capacity of the
BIAWGNC in terms of the standard deviation $\sigma$ of the Gaussian
noise, the minimum capacity of a BIAWGNC over which it is possible
to communicate with vanishing bit error probability under BP
decoding is~$C = 0.9685$ bits per channel use. Accordingly, let us
assume that, for reliable communications over both channels, the
capacity of the BEC and BIAWGNC is set to $0.98$ bits per channel use.
Since the considered code ensemble is right-regular with
$d_{\text{c}}=20$, the value of $B$ in
Theorem~\ref{theorem: Tightened Concentration of Conditional
Entropy} is improved by a factor of
\begin{equation*}
\left[h_2\left(\frac{1-C^{\frac{d_{\text{c}}}{2}}}{2}
\right)\right]^{-2} = 5.134.
\end{equation*}
For the BEC, the result is improved by a factor of
$\bigl(1-C^{d_{\text{c}}}\bigr)^{-2} = 9.051$;
this follows from the tightened value of $B$ in
\eqref{eq: B for BEC}, which improves the concentration
inequality in Theorem~\ref{theorem: Concentration of Conditional Entropy}.
\label{example 1}
\end{example}

\begin{example}[Comparison of Theorems~\ref{theorem: Concentration of Conditional
Entropy} and~\ref{theorem: Tightened Concentration of Conditional
Entropy} for a heavy-tail Poisson distribution (Tornado codes)]
In this example, we compare Theorems~\ref{theorem: Concentration of
Conditional Entropy} and~\ref{theorem: Tightened Concentration of
Conditional Entropy} for Tornado codes. This
capacity-achieving sequence for the BEC refers to the heavy-tail
Poisson distribution, and it was introduced in
\cite[Section~IV]{LubyMSS_IT01}, \cite{Shokrollahi-IMA2000}
(see also \cite[Problem~3.20]{RiU_book}). We rely in the following
on the analysis in \cite[Appendix~VI]{Sason_LDPC09}.

Suppose that we wish to design Tornado code ensembles that achieve a
fraction $1-\varepsilon$ of the capacity of a BEC under iterative
message-passing decoding (where $\varepsilon$ can be set arbitrarily
small). Let $p$ denote the bit erasure probability of the
channel. The parity-check degree is Poisson-distributed, and
therefore the maximal degree of the parity-check nodes is infinity.
Hence, $B=0$ according to Theorem~\ref{theorem: Concentration of
Conditional Entropy}, which renders this theorem useless for the
considered code ensemble. On the other hand, from
Theorem~\ref{theorem: Tightened Concentration of Conditional
Entropy},
\begin{align*}
&\sum_i (i+1)^2 \Gamma_i \left[h_2
\left(\frac{1-C^{\frac{i}{2}}}{2} \right)\right]^2 \\[0.1cm]
& \stackrel{\mathrm{(a)}}{\leq} \sum_i (i+1)^2 \Gamma_i \\
& \stackrel{\mathrm{(b)}}{=} \frac{\sum_i \rho_i (i+2)}{\dint_0^1
\rho(x) \; d{\text{x}}} + 1 \\
& \stackrel{\mathrm{(c)}}{=} (\rho'(1)+3) d_{\text{c}}^{\text{avg}} + 1\\
& \stackrel{\mathrm{(d)}}{=} \left(\frac{\lambda'(0)
\rho'(1)}{\lambda_2} + 3 \right)
d_{\text{c}}^{\text{avg}} + 1\\
& \stackrel{\mathrm{(e)}}{\leq} \left(\frac{1}{p \lambda_2} + 3 \right)
d_{\text{c}}^{\text{avg}} + 1\\
& \stackrel{\mathrm{(f)}}{=} O\left(\log^2
\biggl(\frac{1}{\varepsilon}\biggr) \right)
\end{align*}
with the following justification:
{\begin{itemize}
\item inequality~(a) holds since the binary entropy function to
base~$2$ is bounded between zero and one;
\item equality~(b) holds since
\begin{equation*}
\Gamma_i = \frac{\frac{\rho_i}{i}}{\dint_0^1 \rho(x)\,
\mathrm{d}x},
\end{equation*}
where $\Gamma_i$ and $\rho_i$ denote the fraction of parity-check
nodes and the fraction of edges that are connected to parity-check
nodes of degree~$i$ respectively (and also since $\sum_i \Gamma_i =
1$);
\item equality~(c) holds since
\begin{equation*}
d_{\text{c}}^{\text{avg}} = \frac{1}{\dint_0^1 \rho(x) \,
\mathrm{d}x},
\end{equation*}
where $d_{\text{c}}^{\text{avg}}$ denotes the average parity-check
node degree;
\item equality~(d) holds since $\lambda'(0) = \lambda_2$;
\item inequality~(e) is due to the stability condition for a BEC 
with an erasure probability $p$, which states that satisfying the 
inequality $p \lambda'(0) \rho'(1) < 1$ is a necessary condition 
for reliable communication under BP decoding (see
\cite[Theorem~3.65]{RiU_book});
\item equality~(f) follows from the analysis in \cite[Appendix~VI]{Sason_LDPC09}
(an upper bound on $\lambda_2$ is derived in
\cite[Eq.~(120)]{Sason_LDPC09}, and the average parity-check node degree
scales like $ \log \frac{1}{\varepsilon} $).
\end{itemize}}
It therefore follows from the above
chain of inequalities and \eqref{eq: B} that, for a small
gap to capacity, the parameter $B$ in Theorem~\ref{theorem:
Tightened Concentration of Conditional Entropy} scales (at least)
like
\begin{equation*}
B = O \left( \frac{1}{\log^2 \bigl(\frac{1}{\varepsilon}\bigr)}
\right). \\[0.1cm]
\end{equation*}
Theorem~\ref{theorem: Tightened Concentration of Conditional
Entropy} is therefore useful for the analysis of
this LDPC code ensemble. As is shown above, the parameter $B$ in
\eqref{eq: B} tends to zero rather slowly as we let the fractional
gap $\varepsilon$ tend to zero (which therefore demonstrates a
rather fast concentration in Theorem~\ref{theorem: Tightened
Concentration of Conditional Entropy}). \label{example: tornado
codes}
\end{example}

\begin{example} Here, we continue with the setting of Example~\ref{example 1}
on the $(n,d_{\text{v}},d_{\text{c}})$ regular LDPC code ensemble,
where $d_{\text{v}}=2$ and $d_{\text{c}}=20$. With the setting of
this example, Theorem~\ref{theorem: Concentration of Conditional
Entropy} gives
\begin{align}
& \pr\bigl( \big|H({\bf{X}}|{\bf{Y}}) -
\expectation_{\text{LDPC}(n,\lambda,\rho)}[H({\bf{X}}|{\bf{Y}})] \big|
\geq \xi \sqrt{n} \, \bigr) \nonumber\\[0.1cm]
& \leq 2 \exp(-0.0113 \, \xi^2), \quad \forall \, \xi > 0.
\label{eq: first numerical exponent}
\end{align}
As was mentioned already in Example~\ref{example 1}, the exponential
inequalities in Theorem~\ref{theorem: Tightened
Concentration of Conditional Entropy} achieve an improvement in the
exponent of Theorem~\ref{theorem: Concentration of Conditional
Entropy} by factors of~5.134 and 9.051 for the BIAWGNC and BEC,
respectively. One therefore obtains from the concentration inequalities in
Theorem~\ref{theorem: Tightened Concentration of Conditional
Entropy} that, for every $\xi > 0$,
\begin{align}
&\pr\bigl( \big|H({\bf{X}}|{\bf{Y}}) -
\expectation_{\text{LDPC}(n,\lambda,\rho)}[H({\bf{X}}|{\bf{Y}})] \big|
\geq \xi \sqrt{n} \, \bigr) \nonumber\\[0.1cm]
&\leq \left\{ \begin{array}{ll} 2 \exp(-0.0580 \, \xi^2), \quad &
\text{(BIAWGNC)}
\\[0.25cm]
2 \exp(-0.1023 \, \xi^2), \quad & \text{(BEC)}
\end{array} \right.  .
\label{eq: second numerical exponent}
\end{align}
\end{example}

\section{Summary}
\label{Section: Summary}
This chapter introduces several classical concentration
inequalities for discrete-time martingales with bounded
differences, and some of their applications in information
theory, communications and coding.

The exposition starts with the martingale decomposition
of Doob, the Chernoff bound, and the Hoeffding Lemma (see
Section~\ref{section: Two Basic Concentration Inequalities});
these form basic ingredients for the derivation of concentration
inequalities via the martingale approach.
This chapter derives the Azuma--Hoeffding inequality for discrete-time
martingales with bounded differences (\cite{Azuma}, \cite{Hoeffding}),
and some of its refined versions (see
Sections~\ref{subsection: Azuma--Hoeffding inequality}
and~\ref{section: Refined Versions of the Azuma--Hoeffding Inequality}).
The martingale approach also serves as a useful tool for establishing
concentration of a function $f \colon \reals^n \rightarrow \reals$
whose value changes by a bounded amount whenever any of its
$n$ input variables is changed arbitrarily while the other variables
are held fixed. A common method for proving concentration of such a
function of $n$ independent random variables around its expected
value $\expectation[f]$ revolves around McDiarmid's inequality or the
``independent bounded-differences inequality'' \cite{McDiarmid_tutorial}.
McDiarmid's inequality was originally proved via the martingale approach, as it
is derived in Section~\ref{subsection: McDiarmid's inequality}.
Although the proof of this inequality has some similarity to the
proof of the well-known Azuma--Hoeffding inequality, the bounded-difference
assumption on $f$ yields an improvement by a factor of~$4$ in the exponent.

The presentation of the concentration inequalities in this chapter is followed
by a short discussion on their relations to some selected classical results
in probability theory (see Section~\ref{section: Relations to Results in Probability Theory});
these include the central limit theorem for discrete-time martingales, the
moderate deviations principle, and the suitability of the concentration
inequalities derived in this chapter for harmonic and bounded functions of
discrete-time Markov chains.

Section~\ref{section: Applications} is focused on the applications of the
concentration inequalities in information theory,
communication, and coding theory. These include the establishment of
concentration results for the minimum distance of random binary linear codes,
expansion properties of random bipartite graphs, the crest factor (or peak to average
power ratio) of OFDM signals, and concentration results for LDPC code ensembles.
Additional concentration results have been established by Richardson and Urbanke
for LDPC code ensembles under MAP and iterative message-passing decoding
\cite{RichardsonU2001}. These martingale inequalities also prove to be useful
for the derivation of achievable rates and random coding error exponents, under
ML decoding, when transmission takes place over linear or nonlinear additive white
Gaussian noise channels with or without memory (\cite{Volterra_ITW09}--\cite{ISIT2012_Voltera}).
Nice and interesting applications of these concentration inequalities to
discrete mathematics and random graphs were provided, e.g., in
\cite[Section~3]{McDiarmid_tutorial}, \cite[Chapter~7]{AlonS_tpm3} and
\cite[Chapters~1 and~2]{Steele_book}.

A recent interesting avenue that follows from the inequalities that are introduced in
this chapter is their generalization to random matrices (see, e.g., \cite{Tropp_FoCM_2011}
and \cite{Tropp_ECP_2011}). The interested reader is also referred to \cite{GodboleH_1998}
for a derivation of concentration inequalities that refer to martingales whose differences
are not necessarily bounded, followed by some applications to graph theory.

\begin{subappendices}	
	
\section{Proof of Bennett's inequality}
\label{appendix: proof of Bennett's inequality}

The inequality in \eqref{eq: Bennett's inequality for unconditional expectation}
is trivial for $\lambda=0$, so we prove it for $\lambda > 0$.
Let $Y \triangleq \lambda (X-\overline{x})$ for $\lambda > 0$. Then,
by assumption, $Y \leq \lambda (b-\overline{x}) \triangleq b_{Y}$ a.s.\ and
$\var(Y) \leq \lambda^2 \sigma^2 \triangleq \sigma_Y^2$. It is
therefore required to show that, if $\expectation[Y] = 0$, $Y \leq b_{Y}$,
and $\var(Y) \leq \sigma_Y^2$, then
\begin{equation}
\expectation[e^Y] \leq \left(\frac{b_Y^2}{b_Y^2 + \sigma_Y^2} \right)
\, e^{-\frac{\sigma_Y^2}{b_Y}} + \left(\frac{\sigma_Y^2}{b_Y^2 + \sigma_Y^2}
\right) \, e^{b_Y}.
\label{eq: Bennett's inequality for a zero-mean RV}
\end{equation}
Let $Y_0$ be a random variable that takes two possible values $-\frac{\sigma_Y^2}{b_Y}$
and $b_Y$ with probabilities
\begin{equation}
\pr\left(Y_0 = -\frac{\sigma_Y^2}{b_Y}\right) = \frac{b_Y^2}{b_Y^2 + \sigma_Y^2}, \qquad
\pr(Y_0 = b_Y) = \frac{\sigma_Y^2}{b_Y^2 + \sigma_Y^2}.
\label{eq: probability distribution to achieve equality in Bennett's inequality}
\end{equation}
Then inequality \eqref{eq: Bennett's inequality for a zero-mean RV} is equivalent to
\begin{equation}
\expectation[e^Y] \leq \expectation[e^{Y_0}],
\label{eq: simplified form of Bennett's inequality}
\end{equation}
which is what we will prove.
To that end, let $\phi$ be the unique parabola such that the function
$$f(y) \triangleq \phi(y) - e^y, \quad \forall \, y \in \reals$$
is zero at $y = b_Y$, and has $f(y) = f'(y) = 0$ at $y = -\frac{\sigma_Y^2}{b_Y}$.
Since $\phi''$ is constant, $f''(y)=0$ at exactly one value of $y$, say,
 $y_0$. Furthermore, since $f(-\frac{\sigma_Y^2}{b_Y}) = f(b_Y)$ (both are
equal to zero), we must have $f'(y)=0$ for some
$y_1 \in \bigl(-\frac{\sigma_Y^2}{b_Y}, b_Y \bigr)$. By the same argument
applied to $f'$ on $\bigl[-\frac{\sigma_Y^2}{b_Y}, y_1 \bigr]$, it follows
that $y_0 \in \bigl(-\frac{\sigma_Y^2}{b_Y}, y_1 \bigr)$. The function $f$
is convex on $(-\infty, y_0]$ (since, on this interval, $f''(y) =
\phi''(y) - e^y \geq \phi''(y) - e^{y_0} = \phi''(y_0) - e^{y_0} = f''(y_0) = 0$),
and its minimal value on this interval is attained at $y = -\frac{\sigma_Y^2}{b_Y}$
(since at this point $f'$ is zero); this minimal value is zero. Furthermore,
$f$ is concave on $[y_0, \infty)$ (since its second derivative is non-positive on this interval)
and it attains its maximal value on this interval at $y=y_1$. By construction, $f(b_Y)=0$;
this implies that $f \geq 0$ on the interval $(-\infty, b_Y]$, so $\expectation[f(Y)] \geq 0$
for an arbitrary random variable $Y$ such that $Y \leq b_Y$ a.s., which therefore gives
$$\expectation[e^Y] \leq \expectation[\phi(Y)],$$ with equality if
$\pr(Y \in \{-\frac{\sigma_Y^2}{b_Y}, b_Y\}) = 1$. Since $f''(y) \geq 0$
for $y < y_0$, it must be the case that $\phi''(y) - e^y = f''(y) \geq 0$ for $y < y_0$,
so $\phi''(0) = \phi''(y) > 0$
(recall that $\phi''$ is constant since $\phi$ is a parabola). Hence, for every random
variable $Y$ of zero mean, $\expectation[\phi(Y)]$, which only depends on $\expectation[Y^2]$,
is a non-decreasing function of $\expectation[Y^2]$. The random variable $Y_0$ that
takes values in $\{-\frac{\sigma_Y^2}{b_Y}, b_Y\}$, and whose distribution is given in
\eqref{eq: probability distribution to achieve equality in Bennett's inequality}, is of
zero mean and variance $\expectation[Y_0^2] = \sigma_Y^2$, so
$$\expectation[\phi(Y)] \leq \expectation[\phi(Y_0)].$$ Note also that
$$\expectation[\phi(Y_0)] = \expectation[e^{Y_0}]$$ since $f(y)=0$ (i.e., $\phi(y) = e^y$)
if $y = -\frac{\sigma_Y^2}{b_Y}$ or $b_Y$, and $Y_0$ only takes these two values.
Combining the last two inequalities with the last equality gives
inequality~\eqref{eq: simplified form of Bennett's inequality}, which therefore
completes the proof of Bennett's inequality in \eqref{eq: Bennett's inequality for unconditional expectation}.

\section{On the moderate deviations principle
in Section~\ref{subsection: MDP for real-valued i.i.d. RVs}}
\label{appendix: MDP}

Here we show that, in contrast to the Azuma--Hoeffding
inequality, Theorem~\ref{theorem: first refined
concentration inequality}
provides an upper bound on $$\pr\left( \Big|\sum_{i=1}^n X_i \Big|
\geq \alpha n^{\eta} \right), \quad \forall \, \alpha \geq 0$$
which coincides with the exact asymptotic limit in \eqref{eq: MDP
for i.i.d. real-valued RVs} under an extra
assumption that there exists some constant $d > 0$ such that
$|X_k| \leq d$ a.s.\ for every $k \in \naturals$. Let us define the
martingale sequence $\{S_k, \mathcal{F}_k\}_{k=0}^n$ where
\begin{eqnarray*}
S_k \triangleq \sum_{i=1}^k X_i, \quad \mathcal{F}_k \triangleq
\sigma(X_1, \ldots, X_k)
\end{eqnarray*}
for every $k \in \{1, \ldots, n\}$ with $S_0 = 0$ and
$\mathcal{F}_0 = \{\emptyset, \mathcal{F}\}$. This martingale sequence has
uniformly bounded differences: $|S_k - S_{k-1}| = |X_k| \leq d$
a.s. for every $k \in \{1, \ldots, n\}$. Hence, it follows from
the Azuma--Hoeffding inequality that, for every $\alpha \geq 0$,
\begin{equation*}
\pr\left( |S_n| \geq \alpha n^{\eta} \right)
\leq 2 \exp\left(-\frac{\alpha^2 n^{2\eta-1}}{2d^2}\right)
\end{equation*}
and therefore
\begin{equation}
\lim_{n \rightarrow \infty} n^{1-2 \eta} \; \ln \pr\bigl(
|S_n| \geq \alpha n^{\eta} \bigr) \leq
-\frac{\alpha^2}{2d^2}. \label{eq: MDP scaling from Azuma--Hoeffding
inequality for the sum of i.i.d. real-valued RVs}
\end{equation}
This differs from the limit in \eqref{eq: MDP for i.i.d.
real-valued RVs} where $\sigma^2$ is replaced by $d^2$, so the
Azuma--Hoeffding inequality does not provide the asymptotic limit
in \eqref{eq: MDP for i.i.d. real-valued RVs} (unless
$\sigma^2 = d^2$, i.e., $|X_k|=d$ a.s. for every $k$).

{\em An analysis that follows from
Theorem~\ref{theorem: first refined concentration inequality}}:
The following analysis is a slight modification of the analysis in
the proof of Proposition~\ref{proposition: a similar scaling of the
concentration inequalities}, with the
required adaptation of the calculations for $\eta \in
(\frac{1}{2}, 1)$. It follows from Theorem~\ref{theorem: first
refined concentration inequality} that, for every $\alpha \geq 0$,
\begin{equation*}
\pr(|S_n| \geq \alpha n^{\eta}) \leq 2 \exp\left(-n \,
H\left(\frac{\delta_n+\gamma}{1+\gamma} \Big\|
\frac{\gamma}{1+\gamma}\right) \right)
\end{equation*}
where $\gamma$ is introduced in \eqref{eq: notation},
$H(p \| q)$ is the divergence in \eqref{eq: divergence}
between the ${\rm Bernoulli}(p)$ and ${\rm Bernoulli}(q)$
probability measures, and $\delta_n$ in
\eqref{eq: new delta} is replaced with
\begin{equation}
\delta_n \triangleq \frac{\frac{\alpha}{n^{1-\eta}}}{d} =
\delta n^{-(1-\eta)}
\label{eq: new delta'}
\end{equation}
due to the definition of $\delta$ in \eqref{eq: notation}.
Following the same analysis as in the proof of
Proposition~\ref{proposition: a similar scaling of the
concentration inequalities}, it follows that for every
$n \in \naturals$
\begin{equation*}
\pr(|S_n| \geq \alpha n^{\eta})
\leq 2 \exp\left( -\frac{\delta^2 n^{2\eta-1}}{2\gamma} \left[
1 + \frac{\alpha (1-\gamma)}{3 \gamma d} \cdot n^{-(1-\eta)} +
\ldots \right] \right)
\end{equation*}
and therefore (since, from \eqref{eq: notation},
$\frac{\delta^2}{\gamma} = \frac{\alpha^2}{\sigma^2})$
\begin{equation*}
\lim_{n \rightarrow \infty} n^{1-2 \eta} \;
\ln \pr\bigl(|S_n| \geq \alpha n^{\eta} \bigr) \leq
-\frac{\alpha^2}{2 \sigma^2}.
\end{equation*}
Hence, this upper bound coincides with the exact asymptotic
result in \eqref{eq: MDP for i.i.d. real-valued RVs}.

\section{Proof of the properties in \eqref{eq: properties for OFDM signals}
for OFDM signals}
\label{appendix: OFDM} Consider an OFDM signal from
Section~\ref{subsection: Concentration of the Crest-Factor for
OFDM Signals}. The sequence in \eqref{eq: martingale sequence for
OFDM} is a martingale. From
\eqref{eq: CF}, for every $i \in \{0, \ldots, n\}$,
\begin{eqnarray*}
&& \hspace*{-0.7cm} Y_i = \expectation \Bigl[ \, \max_{0 \leq t
\leq T} \bigl|s(t; X_0, \ldots, X_{n-1}) \bigr| \Big| \, X_0,
\ldots, X_{i-1}\Bigr].
\end{eqnarray*}
The conditional expectation for the random variable $Y_{i-1}$ refers to the
case where only $X_0, \ldots, X_{i-2}$ are revealed. Let
$X'_{i-1}$ and $X_{i-1}$ be independent copies, which are also
independent of $X_0, \ldots, X_{i-2}, X_i, \ldots, X_{n-1}$. Then,
for every $1 \leq i \leq n$,
\begin{eqnarray*} && \hspace*{-0.7cm} Y_{i-1} =
\expectation \Bigl[ \, \max_{0 \leq t \leq T} \bigl|s(t; X_0,
\ldots, X'_{i-1}, X_i, \ldots,
X_{n-1})\bigr| \, \Big| \, X_0, \ldots, X_{i-2} \Bigr] \\
&& \hspace*{0.2cm} = \expectation \Bigl[ \, \max_{0 \leq t \leq T} \bigl|s(t; X_0,
\ldots, X'_{i-1}, X_i, \ldots,
X_{n-1})\bigr| \, \Big| \, X_0, \ldots, X_{i-2}, X_{i-1} \Bigr].
\end{eqnarray*}
Since $|\expectation(Z)| \leq \expectation(|Z|)$, then for $i \in \{1,
\ldots, n\}$
\begin{equation}
\hspace*{-0.2cm} |Y_i - Y_{i-1}| \leq \expectation_{X'_{i-1}, X_i,
\ldots, X_{n-1}} \Bigl[|U-V| \; \Big| \; X_0, \ldots, X_{i-1}
\Bigr] \label{eq: bounded differences for the sequence Y}
\end{equation}
where
\begin{eqnarray*}
&& U \triangleq \max_{0 \leq t \leq T} \bigl|s(t; X_0, \ldots,
X_{i-1},
X_i, \ldots, X_{n-1})\bigr| \\
&& V \triangleq \max_{0 \leq t \leq T} \bigl|s(t; X_0, \ldots,
X'_{i-1}, X_i, \ldots, X_{n-1})\bigr|.
\end{eqnarray*}
From \eqref{eq: OFDM signal}
\begin{align}
& \hspace*{-1.5cm} |U-V| \leq \max_{0 \leq t \leq T} \bigl|s(t;
X_0, \ldots,
X_{i-1}, X_i, \ldots, X_{n-1}) \nonumber\\
& \qquad \qquad \qquad - s(t; X_0, \ldots, X'_{i-1}, X_i, \ldots,
X_{n-1})\bigr| \nonumber \\
&  = \max_{0 \leq t \leq T} \frac{1}{\sqrt{n}} \,
\Bigr|\bigl(X_{i-1} - X'_{i-1}\bigr) \exp\Bigl(\frac{j \, 2\pi i
t}{T}\Bigr)\Bigr| \nonumber \\
& = \frac{|X_{i-1} - X'_{i-1}|}{\sqrt{n}}.
\label{eq: bound on |U-V|}
\end{align}
By assumption, $|X_{i-1}| = |X'_{i-1}| = 1$, and therefore a.s.
$$|X_{i-1} - X'_{i-1}| \leq 2 \Longrightarrow |Y_i - Y_{i-1}| \leq \frac{2}{\sqrt{n}}.$$
We now obtain an upper bound on the conditional variance
$\var(Y_i \, | \, \mathcal{F}_{i-1}) = \expectation \bigl[
(Y_i - Y_{i-1})^2 \, | \, \mathcal{F}_{i-1} \bigr]$.
Since $\bigl(\expectation(Z)\bigr)^2 \leq \expectation(Z^2)$ for a
real-valued random variable $Z$,  from \eqref{eq: bounded differences for
the sequence Y}, \eqref{eq: bound on |U-V|} and the tower property
for conditional expectations, it follows that
\begin{equation*}
\vspace*{-0.1cm} \expectation\bigl[(Y_i - Y_{i-1})^2 \, |
\mathcal{F}_{i-1}\bigr] \leq \frac{1}{n} \cdot
\expectation_{X'_{i-1}} \bigl[|X_{i-1} - X'_{i-1}|^2 \, | \,
\mathcal{F}_{i-1} \bigr]
\end{equation*}
where $\mathcal{F}_{i-1}$ is the $\sigma$-algebra generated by
$X_0, \ldots, X_{i-2}$. Due to the symmetry in the PSK constellation, and
the independence of $X_{i-1}, X'_{i-1}$ in $X_0, \ldots, X_{i-2}$, we have
\begin{align*}
\expectation \bigl[ (Y_i - Y_{i-1})^2 \, |
\, \mathcal{F}_{i-1} \bigr] & \leq \frac{1}{n} \, \expectation \bigl[
|X_{i-1} - X'_{i-1}|^2 \, | \, X_0, \ldots, X_{i-2} \bigr] \\
& = \frac{1}{n} \, \expectation \bigl[
|X_{i-1} - X'_{i-1}|^2 \bigr] \\
& = \frac{1}{n} \, \expectation \Bigl[ |X_{i-1}
- X'_{i-1}|^2 \, | \, X_{i-1} = e^{\frac{j \pi}{M}} \Bigr] \\
& = \frac{1}{nM} \sum_{l=0}^{M-1}
\Big| \, e^{\frac{j \pi}{M}}-e^{\frac{j (2l+1)\pi}{M}}\Big|^2 \\
& = \frac{4}{nM} \sum_{l=1}^{M-1} \sin^2 \Bigl(
\frac{\pi l}{M} \Bigr) = \frac{2}{n}.
\end{align*}
The last equality holds since
\begin{align*}
 \sum_{l=1}^{M-1} \sin^2 \Bigl(\frac{\pi l}{M} \Bigr)
&= \frac{1}{2} \sum_{l=0}^{M-1} \Bigl(1-\cos\Bigl(\frac{2\pi l}{M} \Bigr)\Bigr) \\
& = \frac{M}{2} - \frac{1}{2} \, \text{Re}
\biggl\{ \sum_{l=0}^{M-1} e^{j 2l\pi/M} \biggr\} \\
& = \frac{M}{2} - \frac{1}{2} \, \text{Re}
\biggl\{\frac{1-e^{2j \pi}}{1-e^{j2\pi/M}} \biggr\} = \frac{M}{2}.
\end{align*}

\section{Proof of Theorem~\ref{theorem: concentration results on the performance of LDPC code ensembles}}
\label{appendix: proof of concentration results on the performance of LDPC code ensembles}

From the triangle inequality, we have
\begin{align}
\label{eq : th1 proof- connection between sub-theorems}
& \pr\left(\left|\frac{Z^{(\ell)}(\underline{s})}{nd_{\text{v}}}-p^{(\ell)}(
\underline{s})\right|>\eps\right) \\[0.1cm]
& \leq \pr\left(\left|\frac{Z^{(\ell)}(\underline{s})}{nd_{\text{v}}}-
\frac{\mathbb{E}[Z^{(\ell)}(\underline{s})]}{nd_{\text{v}}}\right|>\eps/2\right)
+ \pr\left(\left|\frac{\mathbb{E}[Z^{(\ell)}(\underline{s})]}{nd_{\text{v}}}
- p^{(\ell)}(\underline{s})\right|>\eps/2\right). \nonumber
\end{align}
If inequality \eqref{eq : th1-Convergence to cycle-free case} holds a.s.,
then $\pr\left(\left|\frac{Z^{(\ell)}(\underline{s})}{nd_{\text{v}}}
- p^{(\ell)}(\underline{s})\right|>\eps/2\right)=0$; therefore,
using~\eqref{eq : th1 proof- connection between sub-theorems}, we deduce that
\eqref{eq : th1-Concentration around cycle-free case} follows from
\eqref{eq : th1-Concentration around expectation} and
\eqref{eq : th1-Convergence to cycle-free case}
for any $\eps>0$ and $n>\frac{2\gamma}{\eps}$.
We start by proving~\eqref{eq : th1-Concentration around expectation}.
For an arbitrary sequence $\underline{s}$, the random variable
$Z^{(\ell)}(\underline{s})$ denotes the number of incorrect variable-to-check
node messages among all $nd_{\text{v}}$ variable-to-check node messages passed
in the $\ell$th iteration for a particular graph $\mathcal{G}$, and decoder-input
$\underline{Y}$. Let us form a martingale by first exposing the $nd_{\text{v}}$
edges of the graph one by one, and then exposing the $n$ received symbols $Y_i$
one by one.
Let $\underline{a}$ denote the sequence of the $nd_{\text{v}}$ variable-to-check
node edges of the graph, followed by the sequence of the $n$ received symbols at
the channel output. For $i=0,...n(d_{\text{v}}+1)$, let the random variable
$\widetilde{Z}_{i} \triangleq \expectation[Z^{(\ell)}(\underline{s})|a_1,\ldots,a_i]$
be defined as the conditional expectation of $Z^{(\ell)}(\underline{s})$ given the
first $i$ elements of the sequence $\underline{a}$. Note that it forms a martingale
sequence (see Fact~\ref{fact: construction of martingales} in
Section~\ref{section: Discrete-Time Martingales}),
where $\widetilde{Z}_{0} = \expectation[Z^{(\ell)}(\underline{s})]$ and
$\widetilde{Z}_{n(d_{\text{v}}+1)}=Z^{(\ell)}(\underline{s})$. Hence, getting an
upper bound on the sequence of differences $|\widetilde{Z}_{i+1}-\widetilde{Z}_{i}|$
enables to apply the Azuma--Hoeffding inequality for proving concentration around the expected
value $\widetilde{Z}_{0}$.
To this end, let us consider the effect of exposing an edge of the graph. Consider two graphs
$\mathcal{G}$ and $\mathcal{\widetilde{G}}$ whose edges are identical except for an exchange
of an endpoint of two edges. A variable-to-check message is affected by this change if at
least one of these edges is included in its directed neighborhood of depth $\ell$.

Consider a neighborhood of depth $\ell$ of a variable-to-check node message. Since at each
level, the graph expands by a factor of
$$\alpha \triangleq (d_{\text{v}}-1+2Wd_{\text{v}})(d_{\text{c}}-1),$$
there are a total of
\begin{equation*}
\label{eq : th1 proof-number of V to C edges }
{N}_{\text{e}}^{(\ell)}=1+d_{\text{c}}(d_{\text{v}}-1+2Wd_{\text{v}})\sum\limits_{i=0}^{\ell-1}{\alpha^i}
\end{equation*}
edges related to the code structure (variable-to-check node edges or vice versa) in the neighborhood
$\mathcal{N}^{(\ell)}_{\vec{e}}$. By symmetry, the two edges can affect at most $2{N}_{\text{e}}^{(\ell)}$
neighbors (alternatively, we could directly sum the number of variable-to-check node edges in
a neighborhood of a variable-to-check node edge, and in a neighborhood of a check-to-variable node edge).
The change in the number of incorrect variable-to-check node messages is bounded by the extreme case,
where each change in the neighborhood of a message introduces an error. In a similar manner, when we
reveal a received output symbol, the variable-to-check node messages whose directed neighborhood includes
that channel input can be affected. We consider a neighborhood of depth $\ell$ of a received output symbol.
By counting, it can be shown that this neighborhood includes
\begin{equation*}
\label{eq : th1 proof-number of channel inputs }
{N}_Y^{(\ell)}=(2W+1) \, d_{\text{v}} \, \sum\limits_{i=0}^{\ell-1}{\alpha^i}
\end{equation*}
variable-to-check node edges.
Therefore, a change of a received output symbol can affect up to $N_Y^{(\ell)}$
variable-to-check node messages.
We conclude that $|\widetilde{Z}_{i+1}-\widetilde{Z}_{i}| \leq 2{N}_{\text{e}}^{(\ell)}$ for
the first $nd_{\text{v}}$ exposures, and
$|\widetilde{Z}_{i+1}-\widetilde{Z}_{i}| \leq {N}_Y^{(\ell)}$ for the last $n$ exposures.
Applying the Azuma--Hoeffding inequality, we get
\begin{align*}
\label{eq : th1 proof-Azuma's inequality}
&\pr\left( \left| \frac{Z^{(\ell)}(\underline{s})}{nd_{\text{v}}}
-\frac{\expectation[Z^{(\ell)}(\underline{s})]}{nd_{\text{v}}} \right| > \frac{\eps}{2} \right) \nonumber \\
&\leq 2 \exp \left(-\frac{\left(nd_{\text{v}}\eps/2\right)^2}{2 \Bigl(nd_{\text{v}}{\big(2{N}_{\text{e}}^{(\ell)}\big)}^2
+n{\big({N}_Y^{(\ell)}\big)}^2 \Bigr)} \right)
\end{align*}
and a comparison of this concentration inequality with~\eqref{eq : th1-Concentration around expectation} gives that
\begin{equation}
\label{eq : th1 proof-beta expression}
\frac{1}{\beta} = \frac{8\left(4d_{\text{v}}{({N}_{\text{e}}^{(\ell)})^2}+{({N}_Y^{(\ell)})^2}\right)}{d_{\text{v}}^2}.
\end{equation}

Next, proving inequality~\eqref{eq : th1-Convergence to cycle-free case} relies on concepts from
\cite{RichardsonU2001} and \cite{KavcicMM_IT2003}. Let $\expectation[Z_i^{(\ell)}(\underline{s})]$,
for $i \in \{1, \ldots, nd_{\text{v}}\}$, be the expected number of incorrect messages passed along
edge $\overrightarrow{e_i}$ after $\ell$ rounds, where the average is with respect to all realizations
of graphs and all output symbols from the channel. Then, by symmetry in the graph construction and by
linearity of  expectation, it follows that
\begin{equation}
\label{eq : th1 proof- E[Z] expension}
\mathbb{E}[Z^{(\ell)}(\underline{s})]=
\sum_{i=1}^{nd_{\text{v}}} {\mathbb{E}[Z_i^{(\ell)}(\underline{s})]}
=nd_{\text{v}}\mathbb{E}[Z_1^{(\ell)}(\underline{s})],
\end{equation}
and
\begin{align*}
&\mathbb{E}[Z_1^{(\ell)}(\underline{s})] \nonumber\\
&=\mathbb{E}[Z_1^{(\ell)}(\underline{s}) \, | \,
\mathcal{N}^{(\ell)}_{\vec{e}}\text{\, is a tree}] \, P^{(\ell)}_{{\text{t}}} +
\mathbb{E}[Z_1^{(\ell)}(\underline{s}) \, | \, \mathcal{N}^{(\ell)}_{\vec{e}} \text{\, not a tree}] \,
P^{(\ell)}_{\overline{\text{t}}}
\end{align*}
where $P^{(\ell)}_{{\text{t}}}$ and $P^{(\ell)}_{\overline{\text{t}}} \triangleq 1 - P^{(\ell)}_{{\text{t}}}$
denote the probabilities that the sub-graph $\mathcal{N}^{(\ell)}_{\vec{e}} $ is or, respectively, is not a tree.
From Theorem~\ref{theorem: probability of a subgraph that is not a tree},
we have $P^{(\ell)}_{\overline{\text{t}}}\le\frac{\gamma}{n}$,
where $\gamma$ is a positive constant which is independent of $n$.
Furthermore,
$\mathbb{E}[Z_1^{(\ell)}(\underline{s}) \, | \,
\text{neighborhood is a tree}]=p^{(\ell)}(\underline{s})$, so
\begin{eqnarray}
\label{eq : EZ inequalities set}
  \nonumber \mathbb{E}[Z_1^{(\ell)}(\underline{s})] &\le&
  (1-P^{(\ell)}_{\overline{\text{t}}})p^{(\ell)}(\underline{s})
  +P^{(\ell)}_{\overline{\text{t}}}\le p^{(\ell)}(\underline{s})
  +P^{(\ell)}_{\overline{\text{t}}}\\
   \mathbb{E}[Z_1^{(\ell)}(\underline{s})] &\ge&
   (1-P^{(\ell)}_{\overline{\text{t}}})p^{(\ell)}(\underline{s})
   \ge p^{(\ell)}(\underline{s})-P^{(\ell)}_{\overline{\text{t}}}.
\end{eqnarray}
Using \eqref{eq : th1 proof- E[Z] expension}, \eqref{eq : EZ inequalities set}
and the inequality $P^{(\ell)}_{\overline{\text{t}}}\le\frac{\gamma}{n}$ gives that
\begin{equation*}
\left|\frac{\mathbb{E}[Z^{(\ell)}(\underline{s})]}{nd_{\text{v}}}
-p^{(\ell)}(\underline{s}) \right|\le
P^{(\ell)}_{\overline{\text{t}}} \le \frac{\gamma}{n}.
\end{equation*}
Hence, if $n>\frac{2\gamma}{\eps}$, then~\eqref{eq : th1-Convergence to cycle-free case} holds.

\section{Proof of Lemma~\ref{lemma: lemma for the case of an infinite maximal degree}}
\label{appendix: proof of the lemma for the case of an infinite
maximal degree} For proving Lemma~\ref{lemma: lemma for the
case of an infinite maximal degree}, one needs to show that, if
$\rho'(1) < \infty$, then
\begin{equation}
\lim_{C \rightarrow 1} \, \sum_{i=1}^{\infty} (i+1)^2 \Gamma_i \,
\left[h_2 \left(\frac{1-C^{\frac{i}{2}}}{2} \right)\right]^2 = 0
\label{eq: limit in appendix 1}
\end{equation}
which, from \eqref{eq: B}, yields that $B \rightarrow \infty$ in
the limit where $C \rightarrow 1$.

By the assumption in Lemma~\ref{lemma: lemma for the case of an
infinite maximal degree}, since $\rho'(1) < \infty$,
$$\sum_{i=1}^{\infty} i \rho_i = \rho'(1)+1 < \infty,$$ and it follows
from the Cauchy--Schwarz inequality that
\begin{equation*}
\sum_{i=1}^{\infty} \frac{\rho_i}{i} \geq
\left( \, \dsum_{i=1}^{\infty} i \rho_i\right)^{-1} > 0.
\end{equation*}
Hence, the {\em average} degree of the parity-check nodes is finite:
\begin{equation*}
d_{\text{c}}^{\text{avg}} = \left( \, \dsum_{i=1}^{\infty}
\frac{\rho_i}{i} \right)^{-1} < \infty
\end{equation*}
and
\begin{align*}
 \sum_{i=1}^{\infty} (i+1)^2 \Gamma_i & = \sum_{i=1}^{\infty} i^2
 \Gamma_i + 2 \sum_{i=1}^{\infty} i \Gamma_i + \sum_i \Gamma_i \\
& = d_{\text{c}}^{\text{avg}} \left( \, \sum_{i=1}^{\infty} i \rho_i +
2 \right) + 1 < \infty
\end{align*}
where $\Gamma_i$ denotes the fraction of parity-check nodes of degree~$i$
and $\rho_i$ denotes the fraction of edges that are connected to parity-check
nodes of degree~$i$, and the last equality holds since
\begin{align*}
\Gamma_i = \frac{\frac{\rho_i}{i}}{\dsum_{j=1}^{\infty} \frac{\rho_j}{j}}
= \frac{\frac{\rho_i}{i}}{\dint_0^1 \rho(x) \,
\mathrm{d}x}
= d_{\text{c}}^{\text{avg}} \, \left( \frac{\rho_i}{i}\right),
\quad \forall \, i \in \naturals.
\end{align*}
This therefore implies that the infinite series in \eqref{eq: limit in appendix 1}
converges uniformly for $C \in [0,1]$, so the order of the limit
and the infinite sum can be exchanged. Every term of the infinite
series in \eqref{eq: limit in appendix 1} converges to zero in the
limit where $C \rightarrow 1$, so the limit in \eqref{eq: limit
in appendix 1} is zero. This completes the proof of
Lemma~\ref{lemma: lemma for the case of an infinite maximal degree}.

\end{subappendices} 

%% file: chapter3.tex
\chapter{The Entropy method, Log-Sobolev and Transportation-Cost Inequalities}
\label{chapter: entropy method}
\chaptermark{The Entropy Method, LSI and TC Inequalities}

This chapter introduces the entropy method for deriving concentration inequalities for functions of a large number of independent random variables, and exhibits its multiple connections to information theory. The chapter is divided into four parts. Sections~\ref{sec:ingredients}--\ref{sec:LSI} introduce the basic ingredients of the entropy method and closely related topics, such as logarithmic Sobolev inequalities. These topics underlie the so-called functional approach to deriving concentration inequalities. Section~\ref{sec:transportation} is devoted to a related viewpoint based on probability in metric spaces. This viewpoint centers around the so-called transportation-cost inequalities, which have been introduced into the study of concentration by Marton. Section~\ref{sec:nonproduct} gives a brief summary of some results on concentration for dependent random variables, emphasizing the connections to information-theoretic ideas. Section~\ref{sec:applications} lists several applications of concentration inequalities and the entropy method to problems in information theory, including strong converses for several source and channel coding problems, empirical distributions of good channel codes with non-vanishing error probability, and an information-theoretic converse for concentration of measure.

\section{The main ingredients of the entropy method}
\label{sec:ingredients}

As a reminder, we are interested in the following question. Let $X_1,\ldots,X_n$ be $n$ independent random variables, each taking values in a set $\cX$. Given a function $f \colon \cX^n \to \reals$, we would like to find tight upper bounds on the {\em deviation probabilities} for the random variable $U = f(X^n)$, i.e., we wish to bound from above the probability $\pr (|U - \expectation U| \ge r)$ for each $r > 0$. Of course, if $U$ has finite variance, then Chebyshev's inequality already gives
\begin{align}\label{eq:Chebyshev_deviation_bound}
	\pr (|U - \expectation U| \ge r) \le \frac{\var(U)}{r^2}, \quad \forall \, r > 0.
\end{align}
However, in many instances a bound like \eqref{eq:Chebyshev_deviation_bound} is not nearly as tight as one would like, so ideally we aim for Gaussian-type bounds
\begin{align}\label{eq:Gaussian_deviation_bound}
	\pr (|U - \expectation U| \ge r) \le K \exp \left(-\kappa r^2 \right), \quad \forall \, r > 0
\end{align}
for some constants $K, \kappa > 0$. Whenever such a bound is available, $K$ is typically a small constant (usually, $K=2$), while $\kappa$ depends on the sensitivity of the function $f$ to variations in its arguments.

In the preceding chapter, we have demonstrated the martingale method for deriving Gaussian concentration bounds of the form \eqref{eq:Gaussian_deviation_bound}, such as the inequalities of Azuma--Hoeffding (Theorem~\ref{theorem: Azuma--Hoeffding inequality}) and McDiarmid (Theorem~\ref{theorem: McDiarmid's inequality}). In this chapter, our focus is on the so-called ``entropy method,'' an information-theoretic technique that has become increasingly popular starting with the work of Ledoux \cite{Ledoux_paper} (see also \cite{Ledoux}). In the following, we will always assume (unless specified otherwise) that the function $f \colon \cX^n \to \reals$ and the probability distribution $P$ of $X^n$ are such that
\begin{itemize}
	\item $U = f(X^n)$ has zero mean: $\expectation U = \expectation f(X^n) = 0$
	\item $U$ is {\em exponentially integrable}:
	\begin{align}\label{eq:U_expint}
		\expectation[\exp(\lambda U)] = \expectation\left[\exp\big(\lambda f(X^n)\big)\right] < \infty, \qquad \forall \, \lambda \in \reals
	\end{align}
	[another way of writing this is $\exp(\lambda f) \in L^1(P)$ for all $\lambda \in \reals$].
\end{itemize}
In a nutshell, the entropy method has three basic ingredients:
\begin{enumerate}
	\item {\bf The Chernoff bound} --- using Markov's inequality, the problem of bounding the deviation probability $\pr (|U - \expectation U| \ge r)$ is reduced to the analysis of the logarithmic moment-generating function $\Lambda(\lambda) \deq \ln \expectation[\exp(\lambda U)]$, $\lambda \in \reals$. (This is also the starting point of the martingale approach, see Chapter~\ref{chapter: martingales}.)
	\item {\bf The Herbst argument} --- the function $\Lambda(\lambda)$ is related through a simple first-order differential equation to the relative entropy (information divergence)
$$D(P^{(\lambda f)} \| P) \triangleq \expectation_{P^{(\lambda f)}} \left[ \ln \frac{\d P^{(\lambda f)}}{\d P} \right]
= \expectation_{P} \left[ \frac{\d P^{(\lambda f)}}{\d P} \, \ln \frac{\d P^{(\lambda f)}}{\d P} \right],$$
where $P = P_{X^n}$ is the probability distribution of $X^n$, and $P^{(\lambda f)}$ is the {\em tilted probability distribution} defined by
	\begin{align}\label{eq:tilted_distribution}
		\frac{\d P^{(\lambda f)}}{\d P} = \frac{\exp(\lambda f)}{\expectation[\exp(\lambda f)]} = \exp\big(\lambda f - \Lambda(\lambda)\big).
		\end{align}
If the function $f$ and the probability distribution $P$ are such that
\begin{align}\label{eq:divergence_quadratic_bound}
D(P^{(\lambda f)} \| P) \le \frac{c\lambda^2}{2}
\end{align}
for some $c > 0$, then the Gaussian bound \eqref{eq:Gaussian_deviation_bound} holds with $K = 2$ and $\kappa = \frac{1}{2c}$. The standard way to establish  \eqref{eq:divergence_quadratic_bound} is through the so-called {\em logarithmic Sobolev inequalities}.
\item {\bf Tensorization of the entropy} --- with few exceptions, it is difficult to derive a bound like \eqref{eq:divergence_quadratic_bound} directly. Instead, one typically takes a divide-and-conquer approach: Using the fact that $P_{X^n}$ is a product distribution (by the assumed independence of the $X_i$'s), the divergence $D(P^{(\lambda f)} \| P)$ is bounded from above by a sum of ``one-dimensional'' (or ``local'') conditional divergence\footnote{Recall the usual definition of the conditional divergence:
$$
D(P_{V|U} \| Q_{V|U} | P_U) \deq \int P_U(\d u) D(P_{V|U=u} \| Q_{V|U=u}).
$$
} terms
\begin{align}\label{eq:local_divergence_term}
D\Big(P^{(\lambda f)}_{X_i|\bar{X}^i} \big\| P_{X_i} \big| P^{(\lambda f)}_{\bar{X}^i} \Big),
\qquad i = 1,\ldots, n
\end{align}
where, for each $i$, $\bar{X}^i \in \cX^{n-1}$ denotes the $(n-1)$-tuple obtained from $X^n$ by removing the $i$th coordinate, i.e., $\bar{X}^i = (X_1,\ldots,X_{i-1},X_{i+1},\ldots,X_n)$. Despite their formidable appearance, the conditional divergences in \eqref{eq:local_divergence_term} are easier to handle because, for each given realization $\bar{X}^i = \bar{x}^i$, the $i$th such term involves a single-variable function $f_i(\cdot|\bar{x}^i) \colon \cX \to \reals$ defined by $f_i(y|\bar{x}^i) \deq f(x_1,\ldots,x_{i-1},y,x_{i+1},\ldots,x_n)$ and the corresponding tilted distribution $P^{(\lambda f)}_{X_i|\bar{X}^i = \bar{x}^i}$, where
\begin{align}\label{eq:local_tilting}
	\frac{\d P^{(\lambda f)}_{X_i|\bar{X}^i = \bar{x}^i}}{\d P_{X_i}} = \frac{\exp\big(\lambda f_i(\cdot|\bar{x}^i)\big)}{\expectation \left[\exp\big(\lambda f_i(X_i|\bar{x}^i)\big)\right]}, \qquad \forall \, \bar{x}^i \in \cX^{n-1}.
\end{align}
In fact, from \eqref{eq:tilted_distribution} and \eqref{eq:local_tilting}, it is easy to see that the conditional distribution $P^{(\lambda f)}_{X_i|\bar{X}^i = \bar{x}^i}$ is nothing but the tilted distribution $P^{(\lambda f_i(\cdot|\bar{x}^i))}_{X_i}$. This simple observation translates into the following: If the function $f$ and the probability distribution $P = P_{X^n}$ are such that there exist constants $c_1,\ldots,c_n > 0$ so that
\begin{align}\label{eq:local_divergence_quadratic_bound}
	D\Big(P^{(\lambda f_i(\cdot|\bar{x}^i))}_{X_i} \Big\| P_{X_i} \Big) \le \frac{c_i \lambda^2}{2}, \,\, \forall \, i \in \{1,\ldots,n\}, \bar{x}^i \in \cX^{n-1}
\end{align}
then \eqref{eq:divergence_quadratic_bound} holds with $c = \sum^n_{i=1}c_i$ (to be shown explicitly later),
which in turn gives the bound
\begin{align}
	\pr \Big( \left|f(X^n) - \expectation f(X^n) \right| \ge r \Big) \le 2 \exp\left( - \frac{r^2}{2 \sum^n_{i=1}c_i} \right),
\end{align}
for all $r > 0$. Again, one would typically use logarithmic Sobolev inequalities to verify \eqref{eq:local_divergence_quadratic_bound}. Conceptually, the tensorization step is similar to ``single-letter'' techniques common in information theory.
\end{enumerate}
In the remainder of this section, we shall elaborate on these three ingredients. Logarithimic Sobolev inequalities and their applications to concentration bounds are described in detail in Sections~\ref{sec:Gaussian_LSI} and \ref{sec:LSI}.

\subsection{The Chernoff bounding technique revisited}

We start by recalling the Chernoff bounding technique (see Section~\ref{subsection: Chernoff}), which reduces the problem of bounding the deviation probability $\pr (U \ge r)$ to the analysis of the logarithmic moment-generating function $\Lambda(\lambda) = \ln \expectation[\exp(\lambda U)]$:
\begin{align*}
	\pr(U \ge r) \le \exp\Big(\Lambda(\lambda) - \lambda r\Big), \qquad \forall \, \lambda > 0.
\end{align*}
The following properties of $\Lambda(\lambda)$ will be useful later on:
\begin{itemize}
	\item $\Lambda(0) = 0$
	\item Because of the exponential integrability of $U$ [cf.~\eqref{eq:U_expint}], $\Lambda(\lambda)$ is infinitely differentiable, and  one can interchange derivative and expectation. In particular,
	\begin{align}
		\Lambda'(\lambda) &= \frac{\expectation[U \exp(\lambda U)]}{\expectation[\exp(\lambda U)]}, \nonumber\\
		\Lambda''(\lambda) &= \frac{\expectation[U^2 \exp(\lambda U)]}{\expectation[\exp(\lambda U)]} - \left( \frac{\expectation[U \exp(\lambda U)]}{\expectation[\exp(\lambda U)]} \right)^2. \label{eq:LMGF_derivatives}
	\end{align}
	Since we have assumed that $\expectation U = 0$, we have $\Lambda'(0) = 0$ and $\Lambda''(0) = \var (U)$.
	\item Since $\Lambda(0) = \Lambda'(0) = 0$, we get
	\begin{align}\label{eq:zero_lambda_limit}
		\lim_{\lambda \to 0} \frac{\Lambda(\lambda)}{\lambda} = 0.
	\end{align}
\end{itemize}

\subsection{The Herbst argument}
\label{ssec:Herbst}

The second ingredient of the entropy method consists in relating the logarithmic moment-generating function to a certain relative entropy. The underlying technique is often referred to as the {\em Herbst argument} because its basic idea  had been described in an unpublished 1975 letter from I.~Herbst to L.~Gross (the first explicit mention of this letter appears in a paper by Davies and Simon \cite{DaviesSimon}).

Given an arbitrary function $g \colon \cX^n \to \reals$ which is exponentially integrable with respect to $P$, i.e., $\expectation[\exp(g(X^n))] < \infty$, let us denote by $P^{(g)}$ the {\em $g$-tilting} of $P$:
\begin{align*}
	\frac{\d P^{(g)}}{\d P} = \frac{\exp(g)}{\expectation[\exp(g)]}.
\end{align*}
Then
\begin{align*}
	D\big( P^{(g)} \big\| P \big) &= \int_{\cX^n} \ln \biggl(\frac{\d P^{(g)}}{\d P} \biggr) \, \d P^{(g)} \\[0.1cm]
    &= \int_{\cX^n} \frac{\d P^{(g)}}{\d P} \, \ln \biggl(\frac{\d P^{(g)}}{\d P} \biggr) \, \d P \\[0.1cm]
	&= \frac{\expectation[g \exp(g)]}{\expectation[\exp(g)]} - \ln \expectation[\exp(g)].
\end{align*}
In particular, if we let $g = t f$ for some $t \neq 0$, then
\begin{align}
	D \big (P^{(t f)} \big\| P \big) &= \frac{t \cdot \expectation[f \exp(t f)]}{\expectation[\exp(t f)]} - \ln \expectation[\exp(t f)] \nonumber \\
	&= t \Lambda'(t) - \Lambda(t) \nonumber \\
	&= t^2 \frac{\d}{\d t}\left(\frac{\Lambda(t)}{t}\right), \label{eq: differential form of Herbst's bound}
\end{align}
where in the second line we have used \eqref{eq:LMGF_derivatives}. Integrating from $t=0$ to
$t = \lambda$ and using \eqref{eq:zero_lambda_limit}, we get
\begin{align}\label{eq:entropy_LMGF}
	\Lambda(\lambda) = \lambda \int^\lambda_0 \frac{D\big(P^{(t f)} \big \| P \big)}{t^2} \; \d t.
\end{align}
Combining \eqref{eq:entropy_LMGF} with \eqref{eq:upper_tail}, we have proved the following:
\begin{proposition}\label{prop:herbst_bound} Let $U = f(X^n)$ be a zero-mean random variable that is
exponentially integrable. Then, for every $r \ge 0$,
	\begin{align}\label{eq:herbst_bound}
	   \pr \big( U \ge r \big) \le \exp\left( \lambda \int^\lambda_0 \frac{D(P^{(tf)} \| P)}{t^2} \; \d t
        - \lambda r \right), \qquad \forall \, \lambda > 0.
	\end{align}
\end{proposition}
Thus, we have reduced the problem of bounding the deviation probabilities $\pr (U \ge r)$ to the problem of bounding the relative entropies $D(P^{(tf)} \| P)$. In particular, we have
\begin{corollary}\label{cor:Herbst} Suppose that the function $f$ and the probability distribution $P$ of $X^n$ are such that
	\begin{align}\label{eq:quadratic_divergence_bound}
		D \big( P^{(t f)} \big\| P \big) \le \frac{ct^2}{2}, \qquad \, \forall \, t > 0
	\end{align}
	for some constant $c > 0$. Then
	\begin{align}\label{eq:one_sided_Gaussian_bound}
		\pr \big( U \ge r \big) \le \exp\left( - \frac{r^2}{2c} \right), \quad \forall \, r \geq 0.
	\end{align}
\end{corollary}
\begin{proof} Using \eqref{eq:quadratic_divergence_bound} to upper-bound the integrand on the right-hand side of \eqref{eq:herbst_bound}, we get
\begin{align}\label{eq:proto_Gaussian_tail}
\pr \big(U \ge r \big) \le \exp\left(\frac{c\lambda^2}{2} - \lambda r \right), \qquad \forall \, \lambda > 0.
\end{align}
Optimizing over $\lambda > 0$ to get the tightest bound gives $\lambda = \frac{r}{c}$, and its substitution
in \eqref{eq:proto_Gaussian_tail} gives the bound in \eqref{eq:one_sided_Gaussian_bound}.
\end{proof}

\subsection{Tensorization of the (relative) entropy}
\label{ssec:tensorization}

The relative entropy $D(P^{(tf)} \| P)$ involves two probability measures on the Cartesian product space $\cX^n$, so bounding this quantity directly is generally very difficult. This is where the third ingredient of the entropy method, the so-called {\em tensorization step}, comes in. The name ``tensorization'' reflects the fact that this step involves bounding $D(P^{(tf)} \| P)$ by a sum of ``one-dimensional'' relative entropy terms, each involving a conditional distribution of one of the variables given the rest.
The tensorization step hinges on the following simple bound:
\begin{proposition}\label{prop:erasure_entropy_bound} Let $P$ and $Q$ be two probability measures on the product space $\cX^n$, where $P$ is a product measure. For every $i \in \{1,\ldots,n\}$, let $\bar{X}^i$ denote the $(n-1)$-tuple $(X_1,\ldots,X_{i-1},X_{i+1},\ldots,X_n)$ obtained by removing $X_i$ from $X^n$. Then
	\begin{align}\label{eq:erasure_entropy_bound}
		D(Q \| P) \le \sum^n_{i=1} D\big( Q_{X_i|\bar{X}^i} \big\| P_{X_i} \big| Q_{\bar{X}^i} \big).
	\end{align}
\end{proposition}
\begin{proof}
From the relative entropy chain rule
\begin{eqnarray}
&& D(Q||P) = \sum_{i=1}^n D\bigl(Q_{X_i \, | \, X^{i-1}} \, \| \, P_{X_i | X^{i-1}} \, | \, Q_{X^{i-1}}\bigr) \nonumber \\
&& \hspace*{1.5cm} = \sum_{i=1}^n D\bigl(Q_{X_i \, | \, X^{i-1}} \, \| \, P_{X_i} \, | \, Q_{X^{i-1}}\bigr)
\label{eq: chain rule for relative entropy}
\end{eqnarray}
where the last equality holds since $X_1, \ldots, X_n$ are independent random variables under $P$ (which implies that
$P_{X_i|X^{i-1}} = P_{X_i|\bar{X}^i} = P_{X_i}$). Furthermore, for every $i \in \{1, \ldots, n\}$,
\begin{align}
&	D\big( Q_{X_i|\bar{X}^i} \big\| P_{X_i} \big| Q_{\bar{X}^i} \big) - D\big( Q_{X_i|X^{i-1}} \big\| P_{X_i} \big| Q_{X^{i-1}} \big) \nonumber \\
& \qquad = \expectation_Q \left[ \ln \frac{\d Q_{X_i|\bar{X}^i}}{\d P_{X_i}}\right] - \expectation_Q \left[ \ln \frac{\d Q_{X_i|X^{i-1}}}{\d P_{X_i}}\right] \nonumber \\
& \qquad = \expectation_Q \left[ \ln \frac{\d Q_{X_i|\bar{X}^i}}{\d Q_{X_i|X^{i-1}}}\right] \nonumber \\
& \qquad = D\big(Q_{X_i|\bar{X}^i} \big\| Q_{X_i|X^{i-1}} \big| Q_{\bar{X}^i}\big) \ge 0.
\label{eq: difference between conditional relative entropies}
\end{align}
Hence, by combining \eqref{eq: chain rule for relative entropy} and \eqref{eq: difference between conditional relative entropies}, we get the inequality in \eqref{eq:erasure_entropy_bound}.
\end{proof}

\begin{remark} The quantity on the right-hand side of \eqref{eq:erasure_entropy_bound} is actually the so-called {\em erasure divergence} $D^-(Q\|P)$ between $Q$ and $P$ (see \cite[Definition~4]{Verdu_Weissman_erasures}), which in the case of arbitrary $Q$ and $P$ is defined by
	\begin{equation}
		D^-(Q\|P) \deq \sum^n_{i=1}D(Q_{X_i|\bar{X}^i} \| P_{X_i|\bar{X}^i} | Q_{\bar{X}^i}).
    \label{eq: erasure divergence}
	\end{equation}
	Because in the inequality \eqref{eq:erasure_entropy_bound} $P$ is assumed to be a product measure, we can replace $P_{X_i|\bar{X}^i}$ by $P_{X_i}$. For a general (non-product) measure $P$, the erasure divergence $D^-(Q \| P)$ may be strictly larger or smaller than the ordinary divergence $D(Q \| P)$. For example, if $n=2$, $P_{X_1} = Q_{X_1}$, $P_{X_2} = Q_{X_2}$, then
\begin{align*}
\frac{\d Q_{X_1|X_2}}{\d P_{X_1|X_2}} = \frac{\d Q_{X_2|X_1}}{\d P_{X_2|X_1}} =  \frac{\d Q_{X_1,X_2}}{\d P_{X_1,X_2}},
\end{align*}
so, from \eqref{eq: erasure divergence},
\begin{align*}
	&D^-(Q_{X_1, X_2} \| P_{X_1, X_2}) \\
	&\qquad = D(Q_{X_1|X_2} \| P_{X_1|X_2} | Q_{X_2}) + D(Q_{X_2|X_1} \| P_{X_2|X_1} | Q_{X_1}) \\
	&\qquad  = 2 D(Q_{X_1, X_2} \| P_{X_1, X_2}).
\end{align*}
On the other hand, if $X_1 = X_2$ under both $P$ and $Q$, then $D^-(Q \| P) = 0$, but $D(Q \| P) > 0$ whenever $P \neq Q$, so $D(Q \| P) > D^-(Q \| P)$ in this case.
\end{remark}

Applying Proposition~\ref{prop:erasure_entropy_bound} with $Q = P^{(tf)}$ to bound the divergence in the integrand in \eqref{eq:herbst_bound}, we obtain from Proposition~\ref{prop:herbst_bound} the following:
\begin{proposition}\label{prop:herbst_bound_tensorized} For every $r \ge 0$, we have
	\begin{align}\label{eq:herbst_bound_tensorized}
		\pr\big(U \ge r) &\le \exp\left( \lambda \sum^n_{i=1}\int^\lambda_0 \frac{D\big(P^{(tf)}_{X_i|\bar{X}^i} \big\| P_{X_i} \big| P^{(tf)}_{\bar{X}^i}\big)}{t^2} \; \d t - \lambda r\right), \,\, \forall \, \lambda > 0.
	\end{align}
\end{proposition}
The conditional divergences in the integrand in \eqref{eq:herbst_bound_tensorized} may look formidable, but the remarkable thing is that, for each $i$ and a given $\bar{X}^i = \bar{x}^i$, the corresponding term involves a tilting of the marginal distribution $P_{X_i}$. Indeed, let us fix some $i \in \{1,\ldots,n\}$, and for each choice of $\bar{x}^i \in \cX^{n-1}$ let us define a function $f_i(\cdot|\bar{x}^i) \colon \cX \to \reals$ by setting
\begin{align}\label{eq:f_i}
	f_i(y|\bar{x}^i) \deq f(x_1,\ldots,x_{i-1},y,x_{i+1},\ldots,x_n), \qquad \forall \, y \in \cX.
\end{align}
Then
\begin{align}\label{eq:tilting_at_i}
\frac{\d P^{(f)}_{X_i|\bar{X}^i = \bar{x}^i}}{\d P_{X_i}} = \frac{\exp\big(  f_i(\cdot|\bar{x}^i)\big)}{\expectation\left[\exp\big( f_i(X_i|\bar{x}^i)\big)\right]}.
\end{align}
In other words, $P^{(f)}_{X_i|\bar{X}^i=\bar{x}^i}$ is the $f_i(\cdot|\bar{x}^i)$-tilting of $P_{X_i}$, the marginal distribution of $X_i$. This is the essence of tensorization: we have effectively decomposed the $n$-dimensional problem of bounding $D(P^{(tf)}\| P)$ into $n$ one-dimensional problems, where the $i$th problem involves the tilting of the marginal distribution $P_{X_i}$ by functions of the form $f_i(\cdot|\bar{x}^i), \forall \, \bar{x}^i$. In particular, we get the following:

\begin{corollary} \label{corollary: from an inequality on the divergence to a concentration inequality}
Suppose that the function $f$ and the probability distribution $P$ of $X^n$ are such that there exist some constants $c_1,\ldots,c_n > 0$, so that, for every $t > 0$,
\begin{align}\label{eq:local_quadratic_bound}
D \Big( P^{(t f_i(\cdot|\bar{x}^i))}_{X_i} \big\| P_{X_i} \Big) \le \frac{c_i t^2}{2},
\qquad \forall \, i \in \{1,\ldots,n\}, \; \bar{x}^i \in \cX^{n-1}.
	\end{align}
	Then
	\begin{align}\label{eq:deviation_bound_tensorized}
		\pr \Big( f(X^n) - \expectation f(X^n) \ge r \Big)
        \le \exp\left( - \frac{r^2}{2 \sum^n_{i=1}c_i} \right), \quad \forall \, r > 0.
	\end{align}
\end{corollary}

\begin{remark} Note the obvious similarity between the bound \eqref{eq:deviation_bound_tensorized} and McDiarmid's inequality \eqref{eq: McDiarmid's inequality}. Indeed, as we will show later on in Section~\ref{ssec:bounded_differences_revisited}, it is possible to derive McDiarmid's inequality using the entropy method.
\end{remark}

\begin{proof}
For every $t > 0$
\begin{align}
D(P^{(t f)} \| P) & \leq \sum_{i=1}^n D\Bigl(P^{(t f)}_{X_i|\bar{X}^i} \big\| P_{X_i} \, | \, P^{(t f)}_{\bar{X}^i} \Bigr)
\label{eq1} \\
& = \sum_{i=1}^n \int_{\cX^{n-1}} D\Bigl(P^{(t f)}_{X_i|\bar{X}^i = \bar{x}^i} \big\| P_{X_i} \Bigr) \, P^{(t f)}_{\bar{X}^i}(\d \bar{x}^i) \label{eq2} \\
& = \sum_{i=1}^n \int_{\cX^{n-1}} D\Bigl(P_{X_i}^{(t f_i(\cdot | \bar{x}^i))} \big\| P_{X_i} \Bigr) \,
P^{(t f)}_{\bar{X}^i}(\d \bar{x}^i) \label{eq3} \\
& \leq  \frac{t^2}{2} \cdot \sum_{i=1}^n c_i \label{eq4}
\end{align}
where \eqref{eq1} follows from the tensorization of the relative entropy, \eqref{eq2} holds since
$P$ is a product measure (so $P_{X_i} = P_{X_i | \bar{X}^i}$) and by the definition of
the conditional relative entropy, \eqref{eq3} follows from \eqref{eq:f_i} and \eqref{eq:tilting_at_i}
which implies that $P^{(t f)}_{X_i|\bar{X}^i = \bar{x}^i} = P_{X_i}^{(t f_i(\cdot | \bar{x}^i))}$,
and inequality \eqref{eq4} holds by the assumption in \eqref{eq:local_quadratic_bound}.
Finally, the inequality in \eqref{eq:deviation_bound_tensorized} follows from \eqref{eq4} and Corollary~\ref{cor:Herbst}.
\end{proof}

\subsection{Preview: logarithmic Sobolev inequalities}
\label{ssec:LSI_preview}

Ultimately, the success of the entropy method hinges on demonstrating that the bounds in
\eqref{eq:local_quadratic_bound} hold for the function $f \colon \cX^n \to \reals$ and the probability distribution
$P = P_{X^n}$ of interest. In the next two sections, we will show how to derive such bounds using the so-called
{\em logarithmic Sobolev inequalities}. Here, we give a quick preview of this technique.

Let $\mu$ be a probability measure on $\cX$, and let $\cA$ be a family of real-valued functions $g \colon \cX \to \reals$, such that for every $a \ge 0$ and $g \in \cA$, we also have $a g \in \cA$. Let $E \colon \cA \to \reals^+$ be a non-negative
functional that is homogeneous of degree~2, i.e., for every $a \ge 0$ and $g \in \cA$, we have $E(ag) = a^2 E(g)$. We are interested in the case when there exists a constant $c > 0$, such that the inequality
\begin{align}\label{eq:PM_LSI}
	D(\mu^{(g)} \| \mu) \le \frac{c E(g)}{2}
\end{align}
holds for every $g \in \cA$. Now suppose that, for each $i \in \{1,\ldots,n\}$, inequality~\eqref{eq:PM_LSI}
holds with $\mu = P_{X_i}$ and some constant $c_i > 0$. Let $f \colon \cX^n \to \reals$ be a function
such that, for every $\bar{x}^i \in \cX^{n-1}$ and $i \in \{1, \ldots, n\}$,
\begin{enumerate}
	\item $f_i(\cdot|\bar{x}^i) \in \cA$
	\item $E\big( f_i(\cdot|\bar{x}^i) \big) \le 1$
\end{enumerate}
where $f_i \colon \cX \to \reals$ is defined in \eqref{eq:f_i}.
Then, the bounds in \eqref{eq:local_quadratic_bound} hold, since from \eqref{eq:PM_LSI} and the above
properties of the functional $E$ it follows that for every $t>0$ and $\bar{x}^i \in \cX^{n-1}$
\begin{align*}
D\Bigl(P^{(t f)}_{X_i|\bar{X}^i = \bar{x}^i} \big\| P_{X_i} \Bigr)
& = D\Bigl( P_{X_i}^{(t f_i(\cdot | \bar{x}^i))} \big\| P_{X_i} \Bigr) \\
& \leq \frac{c_i \, E\bigl(t \, f_i(\cdot|\bar{x}^i) \bigr)}{2} \\
& = \frac{c_i t^2 \, E\bigl(f_i(\cdot|\bar{x}^i) \bigr)}{2} \\
& \leq \frac{c_i t^2}{2}, \quad \forall \, i \in \{1, \ldots, n\}.
\end{align*}
Consequently, the Gaussian concentration inequality in \eqref{eq:deviation_bound_tensorized} follows from
Corollary~\ref{corollary: from an inequality on the divergence to a concentration inequality}.

\section{The Gaussian logarithmic Sobolev inequality}
\label{sec:Gaussian_LSI}

Before turning to the general scheme of logarithmic Sobolev inequalities in the next section, we will
illustrate the basic ideas in the particular case when $X_1,\ldots,X_n$ are i.i.d.\ standard Gaussian
random variables. The relevant log-Sobolev inequality in this instance comes from a seminal paper of
Gross \cite{Gross}, and it connects two key information-theoretic quantities, namely the relative entropy
and the relative Fisher information. In addition, there are deep links between Gross's log-Sobolev inequality
and other fundamental information-theoretic inequalities, such as Stam's inequality and the entropy power inequality.
Some of these fundamental links are considered in this section.

For every $n \in \naturals$ and every positive semidefinite matrix $K \in \reals^{n \times n}$, we will
denote by $G^n_K$ the Gaussian distribution with zero mean and covariance matrix $K$. When $K = sI_n$
for some $s \ge 0$ (where $I_n$ denotes the $n \times n$ identity matrix), we will write $G^n_s$;
it will be written $G_s$ for $n=1$.
We will also write $G^n$ for $G^n_1$ when $n \ge 2$, and $G$ for $G^1_1$. We will denote
by $\gamma^n_K$, $\gamma^n_s$, $\gamma_s$, $\gamma^n$, and $\gamma$ the corresponding densities.

We first state Gross's inequality in its (more or less) original form:

\begin{theorem}[Log-Sobolev inequality for the Gaussian measure]\label{thm:Gross_LSI} For $Z \sim G^n$ and for every smooth\footnote{Here and elsewhere, we will use the term ``smooth'' somewhat loosely to mean ``satisfying enough regularity conditions to make sure that all relevant quantities are well-defined.'' In the present context, smooth means that both $\phi$ and $\nabla \phi$ should be square-integrable with respect to the standard Gaussian measure $G^n$.} function $\phi \colon \reals^n \to \reals$,
we have
	\begin{align}\label{eq:Gross_LSI}
		\expectation[\phi^2(Z) \ln \phi^2(Z)] - \expectation[\phi^2(Z)] \ln \expectation[ \phi^2(Z)] \le 2\, \expectation\left[ \| \nabla \phi(Z) \|^2\right],
	\end{align}
	where $\| \cdot \|$ denotes the usual Euclidean norm on $\reals^n$.
\end{theorem}
\begin{remark} As shown by Carlen \cite{Carlen_LSI}, equality in \eqref{eq:Gross_LSI} holds if and only if $\phi$ is of the form $\phi(z) = \exp{\langle a,z \rangle}$ for some $a \in \reals^n$, where $\langle \cdot, \cdot \rangle$ denotes the standard Euclidean inner product.\end{remark}
\begin{remark} There is no loss of generality in assuming that $\expectation[\phi^2(Z)] = 1$. Then \eqref{eq:Gross_LSI} can be rewritten as
	\begin{align}\label{eq:Gross_LSI_a}
			\expectation[\phi^2(Z) \ln \phi^2(Z)] \le 2\, \expectation\left[ \| \nabla \phi(Z) \|^2\right] \text{ if } \expectation[\phi^2(Z)] = 1, \; Z \sim G^n.
		\end{align}
Moreover, a simple rescaling argument shows that, for $Z \sim G^n_s$ and an arbitrary smooth function $\phi$ with $\expectation[\phi^2(Z)]=1$,
	\begin{align}\label{eq:Gross_LSI_rescaled}
		\expectation[\phi^2(Z)\ln \phi^2(Z)] \le 2s\, \expectation\left[ \| \nabla \phi(Z) \|^2\right].
	\end{align}
\end{remark}

We give an information-theoretic proof of the Gaussian LSI  (Theorem~\ref{thm:Gross_LSI}) later in this section; we refer the reader to \cite{Adams_Clarke} as an example of a typical proof using techniques from functional analysis.

\sloppypar From an information-theoretic point of view, the Gaussian LSI \eqref{eq:Gross_LSI} relates two measures of (dis)similarity between probability measures --- the relative entropy (or divergence) and the {\em relative
Fisher information} (or {\em Fisher information distance}). The latter is defined as follows. Let $P_1$ and
$P_2$ be two Borel probability measures on $\reals^n$ with differentiable densities $p_1$ and $p_2$, and suppose that the Radon--Nikodym derivative $\d P_1/\d P_2 \equiv p_1/p_2$ is differentiable $P_2$-a.e. Then
the {\em relative Fisher information} (or {\em Fisher information distance}) between $P_1$ and $P_2$ is
defined as (see \cite[Eq.~(6.4.12)]{Blower})
\begin{align}\label{eq:Fisher_dist}
	I(P_1 \| P_2) \deq \int_{\reals^n} \left\| \nabla \ln \frac{p_1(z)}{p_2(z)}\right\|^2 p_1(z) \d z = \expectation_{P_1}\left[\left\| \nabla \ln \frac{ \d P_1}{\d P_2} \right\|^2 \right],
	\end{align}
whenever the above integral converges. Under suitable regularity conditions, $I(P_1 \| P_2)$ admits the
equivalent form (see \cite[Eq.~(1.108)]{Johnson_IT_CLT})
\begin{align}
	I(P_1 \| P_2) &= 4 \int_{\reals^n} p_2(z) \left\| \nabla \sqrt{\frac{p_1(z)}{p_2(z)}}\right\|^2 \d z \nonumber\\
	&= 4\, \expectation_{P_2}\left[ \left\| \nabla \sqrt{\frac{\d P_1}{\d P_2}} \right\|^2\right].\label{eq:Fisher_dist_a}
\end{align}
\begin{remark} One condition under which \eqref{eq:Fisher_dist_a} holds is as follows. Let $\xi \colon \reals^n \to \reals^n$ be the {\em distributional} (or {\em weak}) {\em gradient} of $\sqrt{\d P_1/\d P_2} = \sqrt{p_1/p_2}$, so that the equality
	$$
	\int_{-\infty}^{\infty} \sqrt{\frac{p_1(z)}{p_2(z)}} \partial_i \psi(z) \d z =
    - \int_{-\infty}^{\infty} \xi_i(z) \psi(z) \d z
	$$
	holds for all $i = 1,\ldots,n$ and all test functions $\psi \in C^\infty_c(\reals^n)$ \cite[Sec.~6.6]{Lieb_Loss_Analysis}. (Here, $\partial_i \psi$ denotes the $i$th coordinate of $\nabla \psi$.) Then \eqref{eq:Fisher_dist_a} holds, provided $\xi \in L^2(P_2)$.
\end{remark}
	
Now let us fix a smooth function $\phi \colon \reals^n \to \reals$ satisfying the normalization condition $\int_{\reals^n} \phi^2\, \d G^n = 1$; we can assume w.l.o.g.\ that $\phi \ge 0$. Let $Z$ be a standard $n$-dimensional Gaussian random variable, i.e., $P_Z = G^n$, and let $Y \in \reals^n$ be a random vector with distribution $P_Y$ satisfying
\begin{align} \label{eq: a setting for phi}
	\frac{\d P_Y}{\d P_Z} = \frac{\d P_Y}{\d G^n} = \phi^2.
\end{align}
Then, on the one hand, we have
\begin{align}
	\expectation\left[ \phi^2(Z) \ln \phi^2(Z) \right] &= \expectation\left[ \left(\frac{\d P_Y}{\d P_Z}(Z)\right) \ln \left(\frac{\d P_Y}{\d P_Z}(Z)\right)\right] \nonumber\\
	& = D(P_Y \| P_Z),\label{eq: equality for the relative entropy}
\end{align}
and on the other, from \eqref{eq:Fisher_dist_a},
\begin{align} \label{eq: equality for the relative Fisher information}
\expectation \left[ \| \nabla \phi(Z) \|^2\right] = \expectation
\left[ \left\| \nabla \sqrt{\frac{\d P_Y}{\d P_Z}(Z)}\right\|^2\right]
= \frac{1}{4} I(P_Y \| P_Z).
\end{align}
Substituting \eqref{eq: equality for the relative entropy} and
\eqref{eq: equality for the relative Fisher information} into
\eqref{eq:Gross_LSI_a}, we obtain the inequality
\begin{align}\label{eq:Gross_LSI_1}
	D(P_Y \| P_Z) \le \frac{1}{2} I(P_Y \| P_Z), \qquad P_Z = G^n
\end{align}
which holds for every $P_Y$ such that $P_Y \ll G^n$ and $\nabla \sqrt{\d P_Y/\d G^n} \in L^2(G^n)$.
Conversely, for every $P_Y \ll G^n$ satisfying \eqref{eq:Gross_LSI_1}, we can derive
\eqref{eq:Gross_LSI_a} by letting $\phi = \sqrt{\d P_Y / \d G^n}$, provided
$\nabla \phi$ exists (e.g., in the distributional sense). Similarly, for every
$s > 0$, \eqref{eq:Gross_LSI_rescaled} can be written as
\begin{align} \label{eq: a rescaled inequality}
	D(P_Y \| P_Z) \le \frac{s}{2} I(P_Y \| P_Z), \qquad P_Z = G^n_s.
\end{align}
Now let us apply the Gaussian LSI \eqref{eq:Gross_LSI} to functions of the form $\phi = \exp( g/2 )$ for all suitably well-behaved $g \colon \reals^n \to \reals$. Then we obtain
	\begin{align}\label{eq:Gross_LSI_2a}
		\expectation\left[ \exp\big(g(Z)\big) \ln \frac{\exp\big(g(Z)\big)}{\expectation[\exp\big(g(Z)\big)]}\right] \le \frac{1}{2} \expectation \left[ \| \nabla g(Z) \|^2 \exp\big(g(Z)\big) \right],
	\end{align}
where $Z \sim G^n$. If we let $P = G^n$ and denote by $P^{(g)}$ the $g$-tilting of $P$,
the left-hand side of \eqref{eq:Gross_LSI_2a} is recognized as $\expectation[\exp\big(g(Z)\big)] \cdot D(P^{(g)} \| P)$.
Similarly, the right-hand side is equal to $\expectation[\exp\big(g(Z)\big)] \cdot \expectation^{(g)}_{P}[ \| \nabla g \|^2]$ with $\expectation^{(g)}_P[\cdot]$ denoting expectation with respect to $P^{(g)}$.  We therefore obtain the so-called {\em modified log-Sobolev inequality} for the standard Gaussian measure:
	\begin{align}\label{eq:Gross_LSI_2}
		D(P^{(g)} \| P) \le \frac{1}{2}\expectation^{(g)}_{P}\left[ \| \nabla g \|^2 \right], \qquad P = G^n
	\end{align}
which holds for all smooth functions $g \colon \reals^n \to \reals$ that are exponentially integrable with respect to $G^n$.
Observe that \eqref{eq:Gross_LSI_2} implies \eqref{eq:PM_LSI} with $\mu = G^n$, $c=1$, and
$E(g) = \| \nabla g \|^2_\infty$.

In the remainder of this section, we first present a proof of Theorem~\ref{thm:Gross_LSI}, and then discuss several applications of the modified log-Sobolev inequality \eqref{eq:Gross_LSI_2} to derivation of Gaussian concentration inequalities via the Herbst argument.

\subsection{An information-theoretic proof of Gross's log-Sobolev inequality}
\label{ssec:Gaussian_LSI_proof}

In accordance with our general theme, we will prove Theorem~\ref{thm:Gross_LSI} via tensorization:
We first show that the satisfiability of the theorem for $n=1$ yields that it holds for all $n \geq 2$
by scaling up to general $n$ using suitable (sub)additivity properties, and then establish the $n=1$ case.
Indeed, suppose that \eqref{eq:Gross_LSI} holds in dimension $1$. For $n \ge 2$, let $X = (X_1,\ldots,X_n)$ be an $n$-tuple of i.i.d.\ $\cN(0,1)$ variables and consider a smooth function $\phi \colon \reals^n \to \reals$, such that $\expectation_P[\phi^2(X)] = 1$, where $P = P_X = G^n$ is the product of $n$ copies of the standard Gaussian distribution $G$. If we define a probability measure $Q = Q_X$ with $\d Q_X/\d P_X = \phi^2$, then using Proposition~\ref{prop:erasure_entropy_bound} we can write
\begin{align}
	\expectation_P\left[\phi^2(X)\ln \phi^2(X)\right] &= \expectation_P \left[\frac{\d Q}{\d P} \ln \frac{\d Q}{\d P}\right]\nonumber \\
	&= D(Q \| P) \nonumber\\
	&\le \sum^n_{i=1} D\big(Q_{X_i|\bar{X}^i} \big\| P_{X_i} \big| Q_{\bar{X}^i} \big). \label{eq:Gaussian_erasure_entropy_bound}
\end{align}
Following the same steps as the ones that led to \eqref{eq:f_i}, we can define for each $i=1,\ldots,n$ and each $\bar{x}^i = (x_1,\ldots,x_{i-1},x_{i+1},\ldots,x_n) \in \reals^{n-1}$ the function $\phi_i(\cdot|\bar{x}^i) \colon \reals \to \reals$ via
$$\phi_i(y|\bar{x}^i) \deq \phi(x_1,\ldots,x_{i-1},y,x_{i+1},\ldots,x_n), \qquad \forall \,
\bar{x}^i \in \reals^{n-1}, \, y \in \reals.$$
Then
\begin{align*}
	\frac{\d Q_{X_i|\bar{X}^i=\bar{x}^i}}{\d P_{X_i}} &= \frac{\phi^2_i(\cdot|\bar{x}^i)}{\expectation_P[\phi^2_i(X_i|\bar{x}^i)]}
\end{align*}
for all $i \in \{1,\ldots,n\}, \bar{x}^i \in \reals^{n-1}$. With this, we can write
\begin{align}
	& D\big(Q_{X_i|\bar{X}^i} \big\| P_{X_i} \big| Q_{\bar{X}^i}\big) \nonumber\\
	&\qquad= \expectation_Q\left[ \ln \frac{\d Q_{X_i|\bar{X}^i}}{\d P_{X_i}}\right]  \nonumber\\
	&\qquad= \expectation_P\left[\frac{\d Q}{\d P}\ln \frac{\d Q_{X_i|\bar{X}^i}}{\d P_{X_i}}\right] \nonumber\\
	&\qquad= \expectation_P\left[\phi^2(X) \ln \frac{\phi^2_i(X_i|\bar{X}^i)}{\expectation_P[\phi^2_i(X_i|\bar{X}^i)|\bar{X}^i]}\right] \nonumber\\
	&\qquad= \expectation_P\left[\phi^2_i(X_i|\bar{X}^i)  \ln \frac{\phi^2_i(X_i|\bar{X}^i)}{\expectation_P[\phi^2_i(X_i|\bar{X}^i)|\bar{X}^i]}\right] \nonumber\\
	&\qquad= \int_{\reals^{n-1}} \expectation_P\left[ \phi^2_i(X_i|\bar{x}^i)\ln \frac{\phi^2_i(X_i|\bar{x}^i)}{\expectation_P[\phi^2_i(X_i|\bar{x}^i)]}\right] \,
P_{\bar{X}^i}(\d\bar{x}^i).
\label{eq:erasure_change_of_measure}
\end{align}
Since each $X_i \sim G$, we can apply the Gaussian LSI \eqref{eq:Gross_LSI} to the univariate functions $\phi_i(\cdot|\bar{x}^i)$
(note that we currently assume that the Gaussian LSI holds for $n=1$) to get
\begin{align}
	\expectation_P\left[ \phi^2_i(X_i|\bar{x}^i) \ln\frac{\phi^2_i(X_i|\bar{x}^i)}{\expectation_P[\phi^2_i(X_i|\bar{x}^i)]}\right] &\le 2\,\expectation_P \left[ \left(\phi'_i(X_i|\bar{x}^i)\right)^2\right]
	\label{eq:LSI_i}
\end{align}
for all $i = 1,\ldots,n$ and all $\bar{x}^i \in \reals^{n-1}$, where the prime denotes the derivative of $\phi_i(y|\bar{x}^i)$ with respect to $y$:
$$
\phi'_i(y|\bar{x}^i) = \frac{\d \phi_i(y|\bar{x}^i)}{\d y} = \frac{\partial \phi(\bar{x})}{\partial x_i}\Big|_{x_i = y}.
$$
Since $X_1,\ldots,X_n$ are i.i.d.\ under $P$, we can express \eqref{eq:LSI_i} as
\begin{align*}
	\expectation_P\left[ \phi^2_i(X_i|\bar{x}^i) \ln\frac{\phi^2_i(X_i|\bar{x}^i)}{\expectation_P[\phi^2_i(X_i|\bar{x}^i)]}\right]	 &\le 2\, \expectation_P \left[ \big( \partial_i \phi(X)\big)^2 \Big| \bar{X}^i = \bar{x}^i \right],
\end{align*}
where $\partial_i \phi$ denotes the $i$th coordinate of the gradient $\nabla \phi$. Substituting this bound into \eqref{eq:erasure_change_of_measure}, we have
\begin{align*}
	D\big(Q_{X_i|\bar{X}^i} \big\| P_{X_i} \big| Q_{\bar{X}^i}\big) \le 2\, \expectation_P\left[ \big(\partial_i \phi(X)\big)^2 \right].
\end{align*}
Using this to bound each term in the sum on the right-hand side of \eqref{eq:Gaussian_erasure_entropy_bound} together with the equality $\sum^n_{i=1}\big( \partial_i \phi(x)\big)^2 = \left\| \nabla \phi(x) \right\|^2$, we get
\begin{align}
	\expectation_P\left[ \phi^2(X)\ln \phi^2(X)\right] \le
    2\, \expectation_P\left[ \left\| \nabla \phi(X) \right\|^2\right],
\end{align}
which is precisely the Gaussian LSI \eqref{eq:Gross_LSI_a} in $\reals^n$. Thus, if the Gaussian LSI holds for $n=1$, it also holds for all $n \ge 2$.

Based on the above argument, we will now focus on proving the Gaussian LSI  for $n=1$. To that end, it will be convenient to express it in a different but equivalent form that relates the Fisher information and the entropy power of a real-valued random variable with a sufficiently regular density. In this form, the Gaussian LSI was first derived by Stam \cite{Stam}, and the equivalence between Stam's inequality and \eqref{eq:Gross_LSI} was only noted much later by Carlen \cite{Carlen_LSI}. We will first establish this equivalence following Carlen's argument, and then give a new information-theoretic proof of Stam's inequality that, unlike existing proofs \cite{Dembo_Cover_Thomas,Costa_Cover_EPI}, does not directly rely on de Bruijn's identity or on the entropy-power inequality.

First, let us start with some definitions. Let $Y$ be a real-valued random variable with density $p_Y$.
The {\em differential entropy} of $Y$ is given by
\begin{align}
	h(Y) = h(p_Y) \deq - \int_{-\infty}^{\infty} p_Y(y)\ln p_Y(y) \d y,
\end{align}
provided the integral exists. If it does, the {\em entropy power} of $Y$ is given by
\begin{align}
	N(Y) \deq \frac{\exp(2h(Y))}{2\pi e}.
\end{align}
Moreover, if the density $p_Y$ is differentiable, the {\em Fisher information} is given by
\begin{align}\label{eq:Fisher_info}
J(Y) = J(p_Y) = \int_{-\infty}^{\infty} \left( \frac{\d}{\d y}\ln p_Y(y)\right)^2 p_Y(y) \d y = \expectation [\rho^2_Y(Y)],
\end{align}
where $\rho_Y(y) \deq (\d/\d y)\ln p_Y(y) = \frac{p_Y'(y)}{p_Y(y)}$ is known as the {\em score function}.

\begin{remark} An alternative definition of the Fisher information of a real-valued random variable $Y$ is (see \cite[Definition~4.1]{Huber})
	\begin{align}\label{eq:Fisher_Huber}
		J(Y) \deq \sup \left\{ \left(\expectation \psi'(Y)\right)^2 \colon \psi \in C^1, \expectation[\psi^2(Y)] = 1 \right\}
	\end{align}
where the supremum is taken over the set of all continuously differentiable functions $\psi$ with compact support, such that $\expectation[\psi^2(Y)] = 1$.
	Note that this definition does not involve derivatives of any functions of the density of $Y$ (nor assumes that such a density even exists). It can be shown that the quantity defined in \eqref{eq:Fisher_Huber} exists and is finite if and only if $Y$ has an absolutely continuous density $p_Y$, in which case $J(Y)$ is equal to \eqref{eq:Fisher_info} (see \cite[Theorem~4.2]{Huber}).
\end{remark}
We will need the following facts:
\begin{enumerate}
	\item If $D(P_Y \| G_s) < \infty$, then
	\begin{align}\label{eq:divergence_Gs}
		D(P_Y \| G_s) =  \frac{1}{2}\ln \frac{1}{N(Y)} + \frac{1}{2} \ln s - \frac{1}{2} + \frac{1}{2s} \; \expectation Y^2.
	\end{align}
	This is proved by direct calculation: Since $D(P_Y \| G_s) < \infty$, we have $P_Y \ll G_s$ and $\d P_Y / \d G_s = p_Y/\gamma_s$.
    Consequently,
	\begin{align*}
		D(P_Y \| G_s) &= \int_{-\infty}^{\infty} p_Y(y) \ln \frac{p_Y(y)}{\gamma_s(y)} \d y \\
		&= -h(Y) + \frac{1}{2}\ln (2\pi s) + \frac{1}{2s} \; \expectation Y^2 \\
		&= -\frac{1}{2} \left(2 h(Y) - \ln (2\pi e)\right) + \frac{1}{2}\ln s - \frac{1}{2} + \frac{1}{2s} \; \expectation Y^2 \\
		&=  \frac{1}{2}\ln \frac{1}{N(Y)} + \frac{1}{2}\ln s - \frac{1}{2} + \frac{1}{2s} \; \expectation Y^2,
	\end{align*}
	which is \eqref{eq:divergence_Gs}.
	\item If $J(Y) < \infty$ and $\expectation Y^2 < \infty$, then for every $s > 0$
\begin{align}\label{eq:Fisher_dist_Gs}
	I(P_Y \| G_s) = J(Y) + \frac{1}{s^2} \; \expectation Y^2 - \frac{2}{s} < \infty,
\end{align}
where $I( \cdot \| \cdot)$ is the relative Fisher information, cf.~\eqref{eq:Fisher_dist}. This equality is verified as follows:
\begin{align}\label{eq:Fisher_dist_Gs_step1}
I(P_Y \| G_s) &= \int_{-\infty}^{\infty} p_Y(y) \left( \frac{\d}{\d y}\ln p_Y(y)
- \frac{\d}{\d y}\ln \gamma_s(y)\right)^2 \d y \nonumber \\
&= \int_{-\infty}^{\infty} p_Y(y)\left( \rho_Y(y) + \frac{y}{s}\right)^2 \d y \nonumber \\
&= \expectation [\rho^2_Y(Y)] + \frac{2}{s} \; \expectation[Y \rho_Y(Y)] + \frac{1}{s^2} \; \expectation Y^2 \nonumber\\
&= J(Y) + \frac{2}{s} \; \expectation[Y \rho_Y(Y)] + \frac{1}{s^2} \, \expectation Y^2.
\end{align}
	Since $\expectation Y^2 < \infty$, we have $\expectation |Y| < \infty$, so
    $\lim_{y \rightarrow \pm \infty} y \, p_Y(y) = 0$. Furthermore, integration by parts gives
    \begin{align*}
 \expectation[Y \rho_Y(Y)] & = \int_{-\infty}^{\infty} y \, \rho_Y(y) \, p_Y(y) \, \d y \\
    & = \int_{-\infty}^{\infty} y \, p_Y'(y) \, \d y \\
    & = \left(\lim_{y \rightarrow \infty} y \, p_Y(y) - \lim_{y \rightarrow -\infty} y \, p_Y(y) \right) - \int_{-\infty}^{\infty} p_Y(y) \, \d y \\
    & = -1
    \end{align*}
(see \cite[Lemma~A1]{Johnson_Barron_CLT} for another proof). Substituting this into \eqref{eq:Fisher_dist_Gs_step1}, we get \eqref{eq:Fisher_dist_Gs}.
\end{enumerate}
We are now in a position to prove the following result of Carlen \cite{Carlen_LSI}:

\begin{proposition} The following statements are equivalent to hold for the class of real-valued random variables $Y$
with a smooth density $p_Y$, such that $J(Y) < \infty$ and $\expectation Y^2 < \infty$:
\begin{enumerate}
	\item Gaussian log-Sobolev inequality, $D(P_Y \| G) \le \frac{1}{2} \; I(P_Y \| G)$.
	\item Stam's inequality, $N(Y)J(Y) \ge 1$.
\end{enumerate}
\label{proposition: equivalence of Gaussian LSI and Stam's inequality}
\end{proposition}
\begin{remark} Carlen's original derivation in \cite{Carlen_LSI} requires $p_Y$ to be in the Schwartz space $\cS(\reals)$ of infinitely differentiable functions, all of whose derivatives vanish sufficiently rapidly at infinity. In comparison, the regularity conditions of the above proposition are much weaker, requiring only that
$P_Y$ has a differentiable and absolutely continuous density, as well as a finite second moment.
\end{remark}

\begin{proof} We first show the implication $1) \Rightarrow 2)$. If $1)$ holds for every real-valued random variable $Y$
as in Proposition~\ref{proposition: equivalence of Gaussian LSI and Stam's inequality}, it follows that
	\begin{align}\label{eq:Gross_LSI_rescaled_1d}
		D(P_Y \| G_s) \le \frac{s}{2} \, I(P_Y \| G_s), \qquad \forall \, s > 0.
	\end{align}
Inequality~\eqref{eq:Gross_LSI_rescaled_1d} can be verified from equalities
\eqref{eq: a setting for phi}--\eqref{eq: equality for the relative Fisher information},
together with the equivalence of \eqref{eq:Gross_LSI_a} and \eqref{eq:Gross_LSI_rescaled},
which gives \eqref{eq: a rescaled inequality} (or \eqref{eq:Gross_LSI_rescaled_1d}).
Since $J(Y)$ and $\expectation Y^2$ are finite by assumption, the right-hand side of
\eqref{eq:Gross_LSI_rescaled_1d} is finite and equal to \eqref{eq:Fisher_dist_Gs}.
Therefore, $D(P_Y \| G_s)$ is also finite, and it is equal to \eqref{eq:divergence_Gs}.
Hence, we can rewrite \eqref{eq:Gross_LSI_rescaled_1d} as
	\begin{align*}
		\frac{1}{2}\ln \frac{1}{N(Y)} + \frac{1}{2}\ln s - \frac{1}{2} + \frac{1}{2s} \; \expectation Y^2 \le  \frac{s}{2}J(Y) + \frac{1}{2s} \; \expectation Y^2 - 1.
	\end{align*}
	Since $\expectation Y^2 < \infty$, we can cancel the corresponding term from both sides and, upon rearranging, obtain
	\begin{align*}
		\ln \frac{1}{N(Y)} \le  s J(Y) - \ln s - 1.
	\end{align*}
Importantly, this bound holds for {\em every} $s > 0$. Therefore, using the fact that
\begin{align*}
	1 + \ln a = \inf_{s > 0}(as -  \ln s), \qquad \forall \, a > 0
\end{align*}
we obtain Stam's inequality $N(Y)J(Y) \ge 1$.

To establish the converse implication $2) \Rightarrow 1)$, we simply run the above proof backwards.
Note that it is first required to show that $D(P_Y \| G_s) < \infty$.
Since by assumption $J(Y)$ is finite and $2)$ holds, also $\frac{1}{N(Y)}$
is finite; since both $\expectation[Y^2]$ and $\frac{1}{N(Y)}$ are finite, it follows
from \eqref{eq:divergence_Gs} that $D(P_Y \| G_s)$ is finite.
\end{proof}

We now turn to the proof of Stam's inequality. Without loss of generality, we may assume that $\expectation Y = 0$ and $\expectation Y^2 = 1$. Our proof will exploit the formula, due to Verd\'u \cite{Verdu_mismatched_estimation}, that expresses the divergence between two probability distributions in terms of an integral of the excess mean squared error (MSE) in a certain estimation problem with additive Gaussian noise. Specifically, consider the problem of estimating a real-valued random variable $Y$ on the basis of a noisy observation $\sqrt{s} Y + Z$, where $s > 0$ is the signal-to-noise ratio (SNR) and the additive standard Gaussian noise $Z \sim G$ is independent of $Y$. If $Y$ has distribution $P$, the minimum MSE (MMSE) at SNR $s$ is defined as
\begin{align}\label{eq:MMSE}
	\mmse(Y,s) \deq \inf_{\varphi}\expectation[(Y - \varphi(\sqrt{s}Y + Z))^2],
\end{align}
where the infimum is taken over all measurable functions (estimators) $\varphi \colon \reals \to \reals$. It is well-known that the infimum in \eqref{eq:MMSE} is achieved by the conditional expectation $u \mapsto \expectation[Y|\sqrt{s}Y+Z = u]$, so
\begin{align*}
	\mmse(Y,s) = \expectation\left[ \left(Y - \expectation[Y|\sqrt{s}Y + Z] \right)^2\right].
\end{align*}
On the other hand, suppose we assume that $Y$ has distribution $Q$ and therefore use the
{\em mismatched estimator} $u \mapsto \expectation_Q[Y|\sqrt{s}Y + Z = u]$, where the conditional
expectation is computed assuming that $Y \sim Q$. Then, the resulting {\em mismatched} MSE is given by
\begin{align} \label{eq:mismatched MSE}
	\mse_Q(Y,s) = \expectation\left[ \left( Y - \expectation_Q[Y|\sqrt{s}Y + Z]\right)^2 \right],
\end{align}
where the outer expectation on the right-hand side is computed using the correct distribution $P$ of $Y$. Then, the following relation holds for the divergence between $P$ and $Q$ (see \cite[Theorem~1]{Verdu_mismatched_estimation}):
\begin{align}\label{eq:Verdu_mismatched_estimation}
	D(P \| Q) = \frac{1}{2}\int^\infty_0 \left[\mse_Q(Y,s) - \mmse(Y,s)\right] \d s.
\end{align}
We will apply the formula \eqref{eq:Verdu_mismatched_estimation} to $P = P_Y$ and $Q = G$, where $P_Y$ satisfies $\expectation Y = 0$ and $\expectation Y^2 = 1$. In that case it can be shown that, for every $s>0$,
\begin{align} \label{eq:optimal estimator in the Gaussian setting is linear}
	\mse_Q(Y,s) = \mse_G(Y,s) = \lmmse(Y,s),
\end{align}
where $\lmmse(Y,s)$ is the {\em linear} MMSE, i.e., the MMSE attainable by an arbitrary {\em affine} estimator $u \mapsto a u + b$, $a, b \in \reals$:
\begin{align}\label{eq:LMMSE}
	\lmmse(Y,s) = \inf_{a,b \in \reals}\expectation\left[\left(Y - a(\sqrt{s} Y + Z) - b\right)^2\right].
\end{align}
The infimum in \eqref{eq:LMMSE} is achieved by $a^* = \sqrt{s}/(1+s)$ and $b=0$, giving
\begin{align}\label{eq:Gaussian_LMMSE}
	\lmmse(Y,s) = \frac{1}{1+s}.
\end{align}
Moreover, $\mmse(Y,s)$ can be bounded from below using the so-called {\em van Trees inequality} \cite{VanTrees}
(see also Appendix~\ref{app:VanTrees}):
\begin{align}\label{eq:VanTrees}
	\mmse(Y,s) \ge \frac{1}{J(Y)+s}.
\end{align}
Then
\begin{align}
	D(P_Y \| G) &= \frac{1}{2}\int^\infty_0 \left(\lmmse(Y,s) - \mmse(Y,s)\right) \d s \nonumber \\
	&\le \frac{1}{2}\int^\infty_0 \left( \frac{1}{1+s} - \frac{1}{J(Y)+s}\right)\d s \nonumber \\
    &= \frac{1}{2} \lim_{\lambda \rightarrow \infty} \int^\lambda_0 \left( \frac{1}{1+s}
    - \frac{1}{J(Y)+s}\right)\d s \nonumber \\
    &= \frac{1}{2} \lim_{\lambda \rightarrow \infty} \ln \left(\frac{J(Y) \, (1+\lambda)}{J(Y)+\lambda} \right)
    \nonumber \\
	&=\frac{1}{2}\ln J(Y), \label{eq:Stam_from_mismatch}
\end{align}
where the second step uses \eqref{eq:Gaussian_LMMSE} and \eqref{eq:VanTrees}. On the other hand, using \eqref{eq:divergence_Gs} with $s=\expectation Y^2=1$, we get $D(P_Y \| G) = \frac{1}{2}\ln \frac{1}{N(Y)}$. Combining this equality with \eqref{eq:Stam_from_mismatch}, we recover Stam's inequality $N(Y)J(Y) \ge 1$. Moreover, the van Trees bound \eqref{eq:VanTrees} is achieved with equality if and only if $Y$ is a standard Gaussian random variable.

\subsection{From Gaussian log-Sobolev inequality to Gaussian concentration inequalities}

We are now ready to apply the log-Sobolev machinery to establish Gaussian concentration for random variables of the form $U = f(X^n)$, where $X_1,\ldots,X_n$ are i.i.d.\ standard normal random variables and $f \colon \reals^n \to \reals$ is an arbitrary Lipschitz function. We start by considering the special case when $f$ is also differentiable.

\begin{proposition}\label{prop:Gauss_smooth_Lipschitz} Let $X_1,\ldots,X_n$ be i.i.d.\ $\cN(0,1)$ random variables. Then, for every differentiable function $f \colon \reals^n \to \reals$ such that $\| \nabla f(X^n) \| \le 1$ almost surely, we have
\begin{align}\label{eq:Gauss_smooth_Lipschitz}
\pr \Big( f(X^n) \ge \expectation f(X^n) + r\Big) \le \exp\left( - \frac{r^2}{2} \right), \quad \forall \, r \ge 0
\end{align}
\end{proposition}
\begin{proof} Let $P = G^n$ denote the distribution of $X^n$. If $Q$ is an arbitrary probability measure such that
$P$ and $Q$ are mutually absolutely continuous (i.e., $Q \ll P$ and $P \ll Q$), then every event that has $P$-probability $1$ will also have $Q$-probability $1$ and vice versa. Since the function $f$ is differentiable,
it is everywhere finite, so $P^{(f)}$ and $P$ are mutually absolutely continuous. Hence, every event that occurs
$P$-a.s.\ also occurs $P^{(tf)}$-a.s.\ for all $t \in \reals$. In particular, $\|\nabla f(X^n) \| \le 1$
$P^{(tf)}$-a.s.\ for all $t>0$. Therefore, applying the modified log-Sobolev inequality \eqref{eq:Gross_LSI_2}
to $g = tf$ for some $t>0$, we get
\begin{align}
	D(P^{(tf)} \| P) \le \left( \frac{t^2}{2} \right) \expectation^{(tf)}_P\left[ \| \nabla f(X^n) \|^2 \right] \le \frac{t^2}{2}.
\end{align}
Now for the Herbst argument: using Corollary~\ref{cor:Herbst} with $U = f(X^n) - \expectation f(X^n)$, we get \eqref{eq:Gauss_smooth_Lipschitz}.
\end{proof}
\begin{remark} Corollary~\ref{cor:Herbst} and inequality \eqref{eq:Gross_LSI_2} with $g = tf$ imply that,
for every smooth function $f$ with $\| \nabla f(X^n) \|^2 \le L$ a.s.,
\begin{align}
\pr \Big( f(X^n) \ge \expectation f(X^n) + r \Big) \le \exp\left(- \frac{r^2}{2L} \right), \quad \forall \, r \ge 0.
\end{align}
Thus, the constant $\kappa$ in the corresponding Gaussian concentration bound \eqref{eq:Gaussian_deviation_bound} is controlled by the sensitivity of $f$ to modifications of its coordinates.
\end{remark}

Having established concentration for smooth $f$, we can now proceed to the general case:

\begin{theorem}\label{thm:Gauss_Lipschitz} Let $X^n$ be as before, and let $f \colon \reals^n \to \reals$ be a
$1$-Lipschitz function, i.e.,
$$
|f(x^n) - f(y^n)| \le \| x^n - y^n \|, \qquad \forall \, x^n, y^n \in \reals^n.
$$
Then
\begin{align}\label{eq:Gauss_Lipschitz}
\pr \Big( f(X^n) \ge \expectation f(X^n) + r\Big) \le \exp\left( - \frac{r^2}{2} \right), \quad \forall \, r \ge 0.
\end{align}
\end{theorem}
\begin{proof}
By Rademacher's theorem (see, e.g., \cite[Section~3.1.2]{Evans_Gariepy}), the assumption that $f$ is $1$-Lipschitz
implies that it is differentiable almost everywhere and $\| \nabla f \| \le 1$ almost everywhere. This further
implies that $\| \nabla f(X^n) \| \le 1$ almost surely ($X_1, \ldots, X_n$ are i.i.d.\
standard Gaussian random variables). The result of this theorem follows from Proposition~\ref{prop:Gauss_smooth_Lipschitz}.

\end{proof}

\subsection{Hypercontractivity, Gaussian log-Sobolev inequality, and R\'enyi divergence}

We close our treatment of the Gaussian log-Sobolev inequality with a striking result, proved by Gross in his original paper \cite{Gross}, that this inequality is equivalent to a very strong contraction property (dubbed {\em hypercontractivity}) of a certain class of stochastic transformations. The original motivation behind the work of Gross \cite{Gross} came from problems in quantum field theory. However, we will take an information-theoretic point of view and relate it to data processing inequalities for a certain class of channels with additive Gaussian noise, as well as to the rate of convergence in the second law of thermodynamics for Markov processes \cite{Mackey_book}.

Consider a pair $(X,Y)$ of real-valued random variables that are related through the stochastic transformation
\begin{align}\label{eq:OU_channel}
Y = e^{-t} X + \sqrt{1-e^{-2t}} Z
\end{align}
for some $t \ge 0$, where the additive noise $Z \sim G$ is independent of $X$. For reasons that will become clear shortly, we will refer to the channel that implements the transformation \eqref{eq:OU_channel} for a given $t \ge 0$ as the {\em Ornstein--Uhlenbeck channel with noise parameter $t$} and denote it by $\OU(t)$. Similarly, we will refer to the collection of channels $\{ \OU(t) \}^\infty_{t=0}$ indexed by all $t \ge 0$ as the {\em Ornstein--Uhlenbeck channel family}. We immediately note the following properties:
\begin{enumerate}
\item $\OU(0)$ is the ideal channel, $Y = X$.
\item If $X \sim G$, then $Y \sim G$ as well, for every $t$.
\item Using the terminology of \cite[Chapter~4]{RiU_book}, the channel family $\{ \OU(t) \}^\infty_{t=0}$ is {\em ordered by degradation}: for every $t_1,t_2 \ge 0$ we have
\begin{align}\label{eq:OU_degradation}
\OU(t_1 + t_2) &= \OU(t_2) \circ \OU(t_1) = \OU(t_1) \circ \OU(t_2),
\end{align}
which is shorthand for the following statement: for every input random variable $X$, every standard Gaussian random variable $Z$ independent of $X$, and every $t_1,t_2 \ge 0$, we can always find independent standard Gaussian random variables $Z_1,Z_2$ that are also independent of $X$, such that
\begin{align}
& e^{-(t_1 + t_2)} X + \sqrt{1-e^{-2(t_1 + t_2)}} Z \nonumber\\
&\qquad\eqdist e^{-t_2} \left[e^{-t_1} X + \sqrt{1-e^{-2t_1}} Z_1 \right] + \sqrt{1-e^{-2t_2}} Z_2 \nonumber \\
&\qquad\eqdist e^{-t_1} \left[e^{-t_2} X + \sqrt{1-e^{-2t_2}} Z_1 \right] + \sqrt{1-e^{-2t_1}} Z_2
\label{eq: ordering by degradation of the OU channel}
\end{align}
where $\eqdist$ denotes equality of distributions. In other words, we can always define real-valued random variables $X,Y_1,Y_2,Z_1,Z_2$ on a common probability space $(\Omega,\cF,\pr)$, such that $Z_1,Z_2 \sim G$, $(X,Z_1,Z_2)$ are mutually independent,
\begin{align*}
Y_1 &\eqdist e^{-t_1} X + \sqrt{1-e^{-2t_1}} Z_1 \\
Y_2 &\eqdist e^{-(t_1 + t_2)} X + \sqrt{1-e^{-2(t_1+t_2)}} Z_2
\end{align*}
and $X \longrightarrow Y_1 \longrightarrow Y_2$ is a Markov chain. Even more generally, given an arbitrary real-valued random variable $X$, we can construct a continuous-time Markov process $\{Y_t \}^\infty_{t=0}$ with $Y_0 \eqdist X$ and $Y_t \eqdist e^{-t} X + \sqrt{1-e^{-2t}} \cN(0,1)$ for all $t \ge 0$. One way to do this is to let $\{Y_t\}^\infty_{t=0}$ be governed by the It\^o stochastic differential equation (SDE)
\begin{align}\label{eq:OU_SDE}
\d Y_t &= - Y_t\, \d t + \sqrt{2}\, \d B_t, \qquad t \ge 0
\end{align}
with the initial condition $Y_0 \eqdist X$, where $\{B_t\}$ denotes the standard one-dimensional Wiener process (Brownian motion). The solution of this SDE (which is known as the {\em Langevin equation} \cite[p.~75]{Oksendal_book}) is given by the so-called {\em Ornstein--Uhlenbeck process}
$$Y_t = X e^{-t} + \sqrt{2} \int_0^t e^{-(t-s)} \, \d B_s, \qquad t \ge 0$$
where, by the It\^o isometry, the variance of the (zero-mean) additive Gaussian noise
is indeed
\begin{align*}\expectation \left[ \left(\sqrt{2}\int_0^t e^{-(t-s)} \, \d B_s \right)^2 \right]
&= 2 \int_0^t e^{-2(t-s)} \d s \\
&= 1 - e^{-2t}, \quad \forall \, t \geq 0
\end{align*}
(see, e.g., \cite[p.~358]{Karatzas_Shreve_book} or \cite[p.~127]{Klebaner_book}). This explains our choice of the name ``Ornstein--Uhlenbeck channel" for the random transformation \eqref{eq:OU_channel}.
\end{enumerate}
In order to state the main result to be proved in this section, we need the following definition: the {\em R\'enyi divergence} of order $\alpha \in \reals^+ \backslash \{0,1\}$ between two probability measures, $P$ and $Q$, is defined as
\begin{align}\label{eq:Renyi_div_0}
	D_\alpha(P \| Q) \deq \frac{1}{\alpha-1} \; \ln \left( \int \d\mu\left(\frac{\d P}{\d\mu}\right)^\alpha \left(\frac{\d Q}{\d \mu}\right)^{1-\alpha} \right),
\end{align}
where $\mu$ is an arbitrary $\sigma$-finite measure that dominates both $P$ and $Q$. If $P \ll Q$, we have the equivalent form
\begin{align}\label{eq:Renyi_div}
D_\alpha (P \| Q) =
\dfrac{1}{\alpha - 1} \; \ln \left( \expectation_Q \left[ \left(\dfrac{\d P}{\d Q}\right)^\alpha \right] \right).
\end{align}
We recall several key properties of the R\'enyi divergence (see, for example, \cite{Erven_Harremoes_IT14}):
\begin{enumerate}
\item The Kullback-Leibler divergence $D(P \| Q)$ is the limit of $D_\alpha( P \| Q)$ as $\alpha$ tends to~1
from {\em below}:
$$D(P \| Q) = \lim_{\alpha \uparrow 1} D_\alpha( P \| Q).$$
In addition,
$$D(P\|Q) = \sup_{0 < \alpha < 1} D_\alpha( P \| Q) \leq \inf_{\alpha > 1} D_\alpha( P \| Q)$$
and, if $D( P \| Q) = \infty$ or there exists some $\beta > 1$ such that $D_\beta( P \| Q) < \infty$,
then also
\begin{equation}
D(P \| Q) = \lim_{\alpha \downarrow 1} D_\alpha( P \| Q).
\label{eq: limit from above for the Renyi divergence}
\end{equation}
\item If we {\em define} $D_1(P \| Q)$ as $D(P \| Q)$, then the function $\alpha \mapsto D_\alpha(P \| Q)$ is nondecreasing.
\item For all $\alpha > 0$, $D_\alpha(\cdot \| \cdot)$ satisfies the {\em data processing inequality}: if we have two possible distributions $P$ and $Q$ for a random variable $U$, then for every channel (stochastic transformation) $T$ that takes $U$ as input we have
\begin{align}\label{eq:DPI_Renyi}
D_\alpha(\tilde{P} \| \tilde{Q}) \le D_\alpha(P \| Q), \qquad \forall \, \alpha > 0
\end{align}
where $\tilde{P}$ or $\tilde{Q}$ is the distribution of the output of $T$ when the input has distribution $P$ or $Q$, respectively.
\item The R\'enyi divergence is non-negative for every order $\alpha > 0$.
\end{enumerate}
Now consider the following set-up. Let $X$ be a real-valued random variable with distribution $P$, such that $P \ll G$. For every $t \ge 0$, let $P_t$ denote the output distribution of the $\OU(t)$ channel with input $X \sim P$. Then, using the fact that the standard Gaussian distribution $G$ is left invariant by the Ornstein--Uhlenbeck channel family together with the data processing inequality \eqref{eq:DPI_Renyi}, we have
\begin{align}
D_\alpha(P_t \| G) \le D_\alpha(P \| G), \qquad \forall \, t \ge 0, \; \alpha > 0.
\end{align}
This is, of course, nothing but the second law of thermodynamics for Markov chains (see, e.g., \cite[Section~4.4]{Cover_and_Thomas} or \cite{Mackey_book}) applied to the continuous-time Markov process governed by the Langevin equation \eqref{eq:OU_SDE}. We will now show, however, that the Gaussian log-Sobolev inequality of Gross (see Theorem~\ref{thm:Gross_LSI}) implies a stronger statement: For every $\alpha > 1$ and $\eps \in (0,1)$, there exists a positive constant $\tau = \tau(\alpha,\eps)$, such that
\begin{align}\label{eq:Renyi_contraction}
D_\alpha(P_t \| G) \le \eps D_\alpha(P \| G), \qquad  \forall \, t \ge \tau.
\end{align}
In other words, as we increase the noise parameter $t$, the output distribution $P_t$ starts to resemble the invariant distribution $G$ more and more, where the measure of resemblance is given by a R\'enyi divergence of an arbitrary order.
Here is the precise result:

\begin{theorem}
	\label{thm:hypercontractive} The Gaussian log-Sobolev inequality of Theorem~\ref{thm:Gross_LSI} is equivalent to the following statement: For every $\alpha, \beta$ such that $1 < \beta < \alpha  < \infty$
\begin{align}\label{eq:Renyi_hypercontractivity}
D_\alpha ( P_t \| G ) \le \left(\frac{\alpha(\beta - 1)}{\beta(\alpha-1)} \right) \,
D_\beta( P \| G), \,\, \forall \, t \ge \frac{1}{2} \ln \left( \frac{\alpha-1}{\beta-1} \right).
\end{align}
\end{theorem}

The proof of Theorem~\ref{thm:hypercontractive} is provided in Appendix~\ref{app:hypercontractive} (with a certain
equality, involved in this proof, that is proved separately in Appendix~\ref{app:OU_semigroup}).

\begin{remark} The original hypercontractivity result of Gross is stated as an inequality relating suitable norms of $g_t = \d P_t/\d G$ and $g = \d P / \d G$; we refer the reader to the original paper \cite{Gross} or to the lecture notes of Guionnet and Zegarlinski \cite{Guionnet_Zegarlinski} for the traditional treatment of hypercontractivity.
\end{remark}
\begin{remark} To see that Theorem~\ref{thm:hypercontractive} implies \eqref{eq:Renyi_contraction}, fix $\alpha > 1$ and $\eps \in (0,1)$. Let $$ \beta = \beta(\varepsilon, \alpha) \triangleq
\frac{\alpha}{\alpha - \eps(\alpha-1)}.
$$
It is easy to verify that $1 < \beta < \alpha$ and $\frac{\alpha(\beta-1)}{\beta(\alpha-1)} = \eps$. Hence, Theorem~\ref{thm:hypercontractive} implies that
$$
D_{\alpha}(P_t \| P) \leq \varepsilon D_{\beta}(P\|G), \quad \forall \, t \geq \frac{1}{2} \ln \left(1 + \frac{\alpha(1-\varepsilon)}{\varepsilon}\right) \deq \tau(\alpha, \varepsilon).
$$
Since the R\'enyi divergence $D_{\alpha}(\cdot \| \cdot)$ is non-decreasing in the parameter $\alpha$, and
$1 < \beta < \alpha$, it follows that $D_{\beta}(P||G) \leq D_{\alpha}(P||G)$. Therefore, the last inequality implies that
$$D_{\alpha}(P_t || P) \leq \varepsilon  D_{\alpha}(P||G), \quad \forall \, t \ge \tau(\alpha, \varepsilon).
$$
\end{remark}

As a consequence, we can establish a strong version of the data processing inequality for the ordinary divergence:

\begin{corollary} In the notation of Theorem~\ref{thm:hypercontractive}, we have for every $t \ge 0$
\begin{align}\label{eq:OU_data_proc}
D(P_t \| G) \le e^{-2t} D(P \| G).
\end{align}
\end{corollary}

\begin{proof} Let $\alpha = 1 + \eps e^{2t}$ and $\beta = 1 + \eps$ for some $\eps > 0$. Then
using Theorem~\ref{thm:hypercontractive}, we have
\begin{align}\label{eq:near_1_hypercontractive}
D_{1+\eps e^{2t}}( P_t \| G) \le \left(\frac{e^{-2t} + \eps}{1 + \eps}\right) D_{1+\eps}( P \| G),
\qquad \forall \, t \ge 0.
\end{align}
Taking the limit of both sides of \eqref{eq:near_1_hypercontractive} as $\eps \downarrow 0$ and using
\eqref{eq: limit from above for the Renyi divergence} (note that $D_\alpha( P \| G) < \infty$ for
$\alpha > 1$), we get \eqref{eq:OU_data_proc}.
\end{proof}

\section{Logarithmic Sobolev inequalities: the general scheme}
\label{sec:LSI}

Now that we have seen the basic idea behind log-Sobolev inequalities in the concrete case of i.i.d.\ Gaussian random variables, we are ready to take a more general viewpoint. To that end, we adopt the framework of  Bobkov and G\"otze \cite{Bobkov_Gotze_expint} and consider a probability space $(\Omega,\cF,\mu)$ together with a pair $(\cA,\Gamma)$ that satisfies the following requirements:
\begin{itemize}
 	\item {\bf (LSI-1)} $\cA$ is a family of bounded measurable functions on $\Omega$, such that if $f \in \cA$, then $a f + b \in \cA$ as well for every $a \ge 0$ and $b \in \reals$.
	\item {\bf (LSI-2)} $\Gamma$ is an operator that maps functions in $\cA$ to nonnegative measurable functions on $\Omega$.
	\item {\bf (LSI-3)} For every $f \in \cA$, $a \ge 0$, and $b \in \reals$, $\Gamma(a f + b) = a \, \Gamma f$.
\end{itemize}
Then we say that $\mu$ satisfies a {\em logarithmic Sobolev inequality} with constant $c \ge 0$, or $\LSI(c)$ for short, if
\begin{align}\label{eq:LSI}
	D(\mu^{(f)} \| \mu) \le \frac{c}{2} \, \expectation^{(f)}_\mu\left[(\Gamma f)^2\right], \qquad \forall \, f \in \cA.
\end{align}
Here, as before, $\mu^{(f)}$ denotes the $f$-tilting of $\mu$, i.e.,
\begin{align*}
	\frac{\d\mu^{(f)}}{\d\mu} = \frac{\exp(f)}{\expectation_\mu[\exp(f)]},
\end{align*}
and $\expectation^{(f)}_\mu[\cdot]$ denotes expectation with respect to $\mu^{(f)}$.

\begin{remark} We have expressed the log-Sobolev inequality using standard information-theoretic notation. Most of the mathematical literature dealing with the subject, however, uses a different notation, which we briefly summarize for the reader's benefit. Given a probability measure $\mu$ on $\Omega$ and a nonnegative function $g \colon \Omega \to \reals$, define the {\em entropy functional}
	\begin{align} \label{eq: entropy functional}
	\Ent_\mu(g) &\deq \int g \ln g\, \d \mu - \int g\, \d \mu \cdot \ln \left(\int g \, \d \mu\right) \nonumber \\
	&\equiv \expectation_\mu[g \ln g] - \expectation_\mu[g]\, \ln \expectation_\mu[g]
	\end{align}
    with the convention that $0 \ln 0 \triangleq 0$. Due to the convexity of the function $f(t) = t \ln t \; (t \ge 0)$,
    Jensen's inequality implies that $\Ent_\mu(g) \ge 0$.
	The $\LSI(c)$ condition in \eqref{eq:LSI} can be equivalently written as (cf.\ \cite[p.~2]{Bobkov_Gotze_expint})
	\begin{align}\label{eq:LSI_1}
		\Ent_\mu\big(\exp(f)\big) \le \frac{c}{2} \int (\Gamma f)^2 \exp(f)\, \d\mu.
	\end{align}
	To see the equivalence of \eqref{eq:LSI} and
    \eqref{eq:LSI_1}, note that
	\begin{align} \label{eq:ent_vs_D}
	 \Ent_\mu\big(\exp(f)\big)
    & = \int \exp(f) \ln \left(\frac{\exp(f)}{\int \exp(f) \d\mu} \right) \d\mu \nonumber \\
    & = \expectation_\mu[\exp(f)] \, \int \Bigl(\frac{\d \mu^{(f)}}{\d\mu} \Bigl)
    \ln \Bigl( \frac{\d \mu^{(f)}}{\d \mu} \Bigr) \, \d\mu \nonumber \\
    & = \expectation_\mu[\exp(f)] \cdot D(\mu^{(f)} \| \mu)
	\end{align}
    and
    \begin{align} \label{eq:another equality with Gamma f}
     \int (\Gamma f)^2 \exp(f)\, \d\mu
    & = \expectation_\mu[\exp(f)] \, \int (\Gamma f)^2 \, \d\mu^{(f)} \nonumber \\[0.1cm]
    & = \expectation_\mu[\exp(f)] \cdot \expectation^{(f)}_\mu\left[(\Gamma f)^2\right].
    \end{align}
	Substituting \eqref{eq:ent_vs_D} and \eqref{eq:another equality with Gamma f} into \eqref{eq:LSI_1}, we
    obtain \eqref{eq:LSI}. We note that the entropy functional $\Ent$ is homogeneous of degree~1:
    for every $g$ such that $\Ent_\mu(g) < \infty$ and $a > 0$, we have
	\begin{align*}
		\Ent_\mu(ag) &= a\, \expectation_\mu\left[g \ln \frac{g}{\expectation_\mu[g]}\right] = a\, \Ent_\mu(g).
	\end{align*}
\label{remark: an equivalent form of LSI}
\end{remark}

\begin{remark} Strictly speaking, \eqref{eq:LSI} should be called a modified (or exponential) logarithmic Sobolev inequality. The ordinary log-Sobolev inequality takes the form
	\begin{align}\label{eq:plain_LSI}
		\Ent_\mu(g^2) \le 2c \int (\Gamma g)^2\, \d\mu
	\end{align}
for all strictly positive $g \in \cA$. If the pair $(\cA,\Gamma)$ is such that $\psi \circ g \in \cA$ for every $g \in \cA$ and for every $C^\infty$ function $\psi \colon \reals \to \reals$, and $\Gamma$ obeys the chain rule
\begin{align}\label{eq:product_rule}
	\Gamma(\psi \circ g) = | \psi' \circ g| \; \Gamma g, \qquad \forall \, g \in \cA, \, \psi \in C^\infty
\end{align}
then \eqref{eq:LSI} and \eqref{eq:plain_LSI} are equivalent. In order to show this, recall the equivalence of
\eqref{eq:LSI} and \eqref{eq:LSI_1} (see Remark~\ref{remark: an equivalent form of LSI});
the equivalence of \eqref{eq:LSI_1} and \eqref{eq:plain_LSI} is proved in the following when the mapping
$\Gamma$ satisfies the chain rule in \eqref{eq:product_rule}. Indeed, if
\eqref{eq:plain_LSI} holds then using it with $g = \exp(f/2)$ gives
\begin{align*}
	\Ent_\mu\big(\exp(f)\big) &\le 2c \int \Bigl(\Gamma \bigl(\exp(f/2) \bigr) \Bigr)^2\, \d\mu \\
	& = \frac{c}{2} \int (\Gamma f)^2 \exp(f)\, \d\mu
\end{align*}
which is \eqref{eq:LSI_1}. The last equality in the above display follows from \eqref{eq:product_rule} which implies that
$$\Gamma \bigl(\exp(f/2) \bigr) = \frac{1}{2} \, \exp(f/2) \cdot \Gamma f.$$
Conversely, using \eqref{eq:LSI_1} with $f = 2\ln g$ gives
\begin{align*}
	\Ent_\mu\big( g^2 \big) &\le \frac{c}{2} \int \bigl( \Gamma(2 \ln g)\bigr)^2 g^2\, \d\mu \\
	&= 2c \int (\Gamma g)^2 \d\mu,
\end{align*}
which is \eqref{eq:plain_LSI}. Again, the last equality is a consequence of \eqref{eq:product_rule}, which gives $\Gamma(2 \ln g) = \frac{2 \, \Gamma g}{g}$ for
all strictly positive $g \in \cA$). In fact, the Gaussian log-Sobolev inequality we have looked at in Section~\ref{sec:Gaussian_LSI} is an instance in which this equivalence holds with $\Gamma f = ||\nabla f||$
clearly satisfying the product rule \eqref{eq:product_rule}.
\end{remark}

Recalling the discussion of Section~\ref{ssec:LSI_preview}, we now show how we can pass from a log-Sobolev
inequality to a concentration inequality via the Herbst argument. Indeed, let $\Omega = \cX^n$ and $\mu = P$,
and suppose that $P$ satisfies $\LSI(c)$ on an appropriate pair $(\cA,\Gamma)$. Suppose, furthermore, that the function of interest $f$ is an element of $\cA$ and that $\|\Gamma f\|_\infty < \infty$ (otherwise, $\LSI(c)$ is vacuously true for every $c>0$). Then $tf \in \cA$ for every $t \ge 0$, so applying \eqref{eq:LSI} to $g = tf$ we get
\begin{align}
	D \big( P^{(tf)} \big\| P \big) &\le \frac{c}{2} \; \expectation^{(tf)}_P\left[ \left( \Gamma(tf)\right)^2 \right] \nonumber \\
	&= \frac{ct^2}{2} \; \expectation^{(tf)}_P\left[ \left(\Gamma f\right)^2 \right] \nonumber \\
	& \le \frac{c \| \Gamma f \|^2_\infty t^2}{2}, \label{eq:quadratic_divergence}
\end{align}
where the second step uses the fact that $\Gamma(tf) = t\, \Gamma f$ for every $f \in \cA$ and $t \ge 0$. In other words, $P$ satisfies the bound \eqref{eq:PM_LSI} for every $g \in \cA$ with $E(g) = \| \Gamma g \|^2_\infty$. Therefore, using the bound \eqref{eq:quadratic_divergence} together with Corollary~\ref{cor:Herbst}, we arrive at
\begin{align}\label{eq:LSI_to_concentration}
	\pr \big( f(X^n) \ge \expectation f(X^n) + r \big) \le  \exp\left( -\frac{r^2}{2 c \|\Gamma f\|^2_\infty} \right), \qquad \forall \, r \ge 0.
\end{align}

\subsection{Tensorization of the logarithmic Sobolev inequality}

In the above demonstration, we have capitalized on an appropriate log-Sobolev inequality in order to derive a concentration inequality. Showing that a log-Sobolev inequality holds can be very difficult for reasons discussed in Section~\ref{ssec:tensorization}. However, when the probability measure $P$ is a product measure, i.e., the $\cX$-valued random variables $X_1,\ldots,X_n$ are independent under $P$, we can use once again the ``divide-and-conquer'' tensorization strategy: we break the original $n$-dimensional problem into $n$ one-dimensional subproblems, demonstrate that each marginal distribution $P_{X_i} \, (i = 1,\ldots,n)$ satisfies a log-Sobolev inequality for a suitable class of real-valued functions on $\cX$, and finally appeal to the tensorization bound for the relative entropy.

Let us provide the abstract scheme first. Suppose that, for each $i \in \{1,\ldots,n\}$, we have a pair $(\cA_i,\Gamma_i)$ defined on $\cX$ that satisfies the requirements (LSI-1)--(LSI-3) listed at the beginning of Section~\ref{sec:LSI}. Recall that for an arbitrary function $f \colon \cX^n \to \reals$, for $i \in \{1,\ldots,n\}$, and for an arbitrary $(n-1)$-tuple $\bar{x}^i = (x_1,\ldots,x_{i-1},x_{i+1},\ldots,x_n)$, we have defined a function $f_i(\cdot|\bar{x}^i) \colon \cX \to \reals$ via $f_i(x_i|\bar{x}^i) \deq f(x^n)$. Then, we have the following:

\begin{theorem}\label{thm:LSI_tensorization} Let $X_1,\ldots,X_n$ be $n$ independent $\cX$-valued random variables, and let $P = P_{X_1} \otimes \ldots \otimes P_{X_n}$ be their joint probability distribution. Let $\cA$ consist of all functions $f \colon \cX^n \to \reals$ such that, for every $i \in \{1, \ldots, n\}$,
	\begin{align}\label{eq:projections_in_A}
		f_i(\cdot|\bar{x}^i) \in \cA_i, \qquad \forall \, \bar{x}^i \in \cX^{n-1}.
	\end{align}
Define the operator $\Gamma$ that maps each $f \in \cA$ to
\begin{align}\label{eq:gamma_def}
	\Gamma f = \sqrt{\sum^n_{i=1} (\Gamma_i f_i)^2},
\end{align}
which is shorthand for
\begin{align}\label{eq:gamma_def_long}
\Gamma f(x^n) = \sqrt{\sum^n_{i=1} \Bigl(\Gamma_i f_i(x_i|\bar{x}^i)\Bigr)^2},
\qquad \forall \, x^n \in \cX^n.
\end{align}
Then, the following statements hold:
\begin{enumerate}
	\item\sloppypar If there exists a constant $c \ge 0$ such that, for every $i \in \{1, \ldots, n\}$, $P_{X_i}$ satisfies $\LSI(c)$ with respect to $(\cA_i,\Gamma_i)$, then $P$ satisfies $\LSI(c)$ with respect to $(\cA,\Gamma)$.
	\item\sloppypar For every $f \in \cA$ with $\expectation[f(X^n)] = 0$, and every $r \ge 0$,
	\begin{align}
		\pr\big( f(X^n) \ge r \big) \le \exp\left( - \frac{r^2}{2c \| \Gamma f \|^2_\infty}\right).
	\end{align}
\end{enumerate}
\end{theorem}
\begin{proof} We first verify that the pair $(\cA,\Gamma)$, defined in the statement of the theorem, satisfies the requirements (LSI-1)--(LSI-3). Thus, consider some $f \in \cA$, choose some $a \ge 0$ and $b \in \reals$, and let $g = af + b$. Then, for every $i \in \{1, \ldots, n\}$ and an arbitrary $\bar{x}^i$,
	\begin{align*}
		g_i(\cdot|\bar{x}^i) &= g(x_1,\ldots,x_{i-1},\cdot,x_{i+1},\ldots,x_n) \\
		&= a f(x_1,\ldots,x_{i-1},\cdot,x_{i+1},\ldots,x_n) + b \\
		&= a f_i(\cdot|\bar{x}^i) + b \in \cA_i,
	\end{align*}
	where the last step relies on \eqref{eq:projections_in_A} and the property (LSI-1) of the pair $(\cA_i,\Gamma_i)$. Hence, $f \in \cA$ implies that $g = a f + b \in \cA$ for every $a \ge 0$ and $b \in \reals$, so (LSI-1) holds. From the definition of $\Gamma$ in \eqref{eq:gamma_def} and \eqref{eq:gamma_def_long}, it is readily seen that (LSI-2) and (LSI-3) hold as well.
	
Next, for every $f \in \cA$ and $t \ge 0$, we have
	\begin{align}
		D \big( P^{(tf)} \big \| P \big) &\le \sum^n_{i=1} D\Big( P^{(tf)}_{X_i|\bar{X}^i} \Big\| P_{X_i} \Big| P_{\bar{X}^i}^{(t f)} \Big) \nonumber \\
		&= \sum^n_{i=1} \int P^{(tf)}_{\bar{X}^i}(\d \bar{x}^i) \; D\Big(P^{(tf)}_{X_i|\bar{X}^i = \bar{x}^i} \Big\| P_{X_i}\Big) \nonumber \\
		&= \sum^n_{i=1} \int P^{(tf)}_{\bar{X}^i}(\d \bar{x}^i) \; D\Big(P^{(t f_i(\cdot|\bar{x}^i))}_{X_i} \Big\| P_{X_i} \Big) \nonumber \\
		&\le \frac{c t^2}{2} \sum^n_{i=1}\int P^{(tf)}_{\bar{X}^i}(\d \bar{x}^i) \; \expectation_{P_{X_i}}^{(tf_i(\cdot|\bar{x}^i))} \left[ \left(\Gamma_i f_i(X_i|\bar{x}^i)\right)^2 \right] \nonumber \\
		&= \frac{c t^2}{2} \sum^n_{i=1} \expectation_{P_{\bar{X}^i}}^{(tf)}\left\{  \expectation_{P_{X_i|\bar{X}^i}}^{(tf)}\Big[ \left(\Gamma_i f_i(X_i|\bar{X}^i)\right)^2 \Big]\right\} \nonumber \\
		&= \frac{c t^2}{2} \cdot \expectation_P^{(tf)}\left[ (\Gamma f)^2\right], \label{eq:LSI_tensorization}
\end{align}
where the first step uses Proposition~\ref{prop:erasure_entropy_bound} with $Q = P^{(tf)}$, the second is by the definition of conditional divergence where $P_{X_i} = P_{X_i | \bar{X}^i}$, the third is due to \eqref{eq:tilting_at_i}, the fourth uses the fact that (a) $f_i(\cdot|\bar{x}^i) \in \cA_i$ for all $\bar{x}^i$ and (b) $P_{X_i}$ satisfies $\LSI(c)$ with respect to $(\cA_i,\Gamma_i)$, and the last step uses the tower property of the conditional expectation, as well as \eqref{eq:gamma_def}. We have thus proved the first part of the theorem, i.e., that $P$ satisfies $\LSI(c)$ with respect to the pair $(\cA,\Gamma)$. The second part follows from the same argument that was used to prove \eqref{eq:LSI_to_concentration}.
\end{proof}

\subsection{Maurer's thermodynamic method}

With Theorem~\ref{thm:LSI_tensorization} at our disposal, we can now establish concentration inequalities in product spaces whenever an appropriate log-Sobolev inequality can be shown to hold for each individual variable. Thus, the bulk of the effort is in showing that this is, indeed, the case for a given probability measure $P$ and a given class of functions. Ordinarily, this is done on a case-by-case basis. However, as shown recently by A.~Maurer in an insightful paper \cite{Maurer_thermo}, it is possible to derive log-Sobolev inequalities in a wide variety of settings by means of a single unified method. This method has two basic ingredients:
\begin{enumerate}
	\item A certain ``thermodynamic'' representation of the divergence $D(\mu^{(f)}\|\mu)$, $f \in \cA$, as an integral of the {\em variances} of $f$ with respect to the tilted measures $\mu^{(tf)}$ for all $t \in (0,1)$.
	\item Derivation of upper bounds on these variances in terms of an appropriately chosen operator $\Gamma$ acting on $\cA$, where $\cA$ and $\Gamma$ are the objects satisfying the conditions (LSI-1)--(LSI-3).
\end{enumerate}
In this section, we will state two lemmas that underlie these two ingredients and then describe the overall method in broad strokes. Several detailed demonstrations of the method in action will be given in the sections that follow.

Once again, consider a probability space $(\Omega,\cF,\mu)$ and recall the definition of the $g$-tilting of $\mu$:
$$
\frac{\d \mu^{(g)}}{\d\mu} = \frac{\exp(g)}{\expectation_\mu[\exp(g)]}.
$$
The variance of an arbitrary $h \colon \Omega \to \reals$ with respect to $\mu^{(g)}$ is then given by
\begin{align*}
	\var^{(g)}_\mu[h] \deq \expectation^{(g)}_\mu[h^2] - \left(\expectation^{(g)}_\mu[h]\right)^2.
\end{align*}
The first ingredient of Maurer's method is encapsulated in the following
(see \cite[Theorem~3]{Maurer_thermo}):

\begin{lemma}
	\label{lm:entropy_via_fluctuations} Let $f \colon \Omega \to \reals$ be a function such that $\expectation_\mu[\exp(\lambda f)] < \infty$ for all $\lambda > 0$. Then, the following equality holds:
	\begin{align}\label{eq:entropy_via_fluctuations}
		D\big( \mu^{(\lambda f)} \big\| \mu \big) = \int^\lambda_0 \int^\lambda_t \var^{(s f)}_\mu[f]\, \d s\, \d t, \quad \forall \, \lambda > 0.
	\end{align}
\end{lemma}
\begin{remark} The thermodynamic interpretation of the above result stems from the fact that the tilted measures $\mu^{(tf)}$ can be viewed as the {\em Gibbs measures} that are used in statistical mechanics as a probabilistic description of physical systems in thermal equilibrium. In this interpretation, the underlying space $\Omega$ is the state (or configuration) space of some physical system $\Sigma$, the elements $x \in \Omega$ are the states (or configurations) of $\Sigma$, $\mu$ is some base (or reference) measure, and $f$ is the energy function. We can view $\mu$ as some initial distribution of the system state. According to the postulates of statistical physics, the thermal equilibrium of $\Sigma$ at absolute temperature $\theta$ corresponds to that distribution $\nu$ on $\Omega$ that will globally minimize the {\em free energy functional}
	\begin{align} 
		\Psi_\theta(\nu) \deq  \expectation_\nu[f] + \theta D(\nu \| \mu).
	\end{align}
	Then we claim that $\Psi_\theta(\nu)$ is uniquely minimized by $\nu^* = \mu^{(-tf)}$, where $t = 1/\theta$ is the {\em inverse temperature}. To see this, consider an arbitrary $\nu$, where we may assume, without loss of generality, that $\nu \ll \mu$. Let $\psi \deq \d\nu/\d\mu$. Then
    \begin{equation*}
    \frac{\d \nu}{\d \mu^{(-tf)}} = \frac{\frac{\d \nu}{\d \mu}}{\frac{\d \mu^{(-tf)}}{\d \mu}}
    = \frac{\psi}{\frac{\exp(-tf)}{\expectation_{\mu}[\exp(-tf)]}} = \psi \, \exp(tf) \, \expectation_{\mu}[\exp(-tf)]
    \end{equation*}
    and
	\begin{align*}
		\Psi_\theta(\nu) &= \frac{1}{t} \, \expectation_\nu[t f + \ln \psi] \\[0.05cm]
        &= \frac{1}{t} \, \expectation_\nu\left[\ln \bigl(\psi \exp(tf) \bigr) \right] \\
		&= \frac{1}{t} \, \expectation_\nu \left[ \ln \frac{\d \nu}{\d \mu^{(-tf)}} - \Lambda(-t) \right] \\
		&= \frac{1}{t} \left[ D(\nu \| \mu^{(-tf)}) - \Lambda(-t)\right],
	\end{align*}
	where, as before, $\Lambda(-t) \deq \ln \bigl(\expectation_\mu[\exp(-t f)] \bigr)$.
Therefore, we have $\Psi_\theta(\nu) = \Psi_{1/t}(\nu) \ge - \Lambda(-t)/t$ with equality if and only if $\nu = \mu^{(-tf)}$.

We refer the reader to a recent monograph by Merhav \cite{Merhav_FnT2009} that
highlights some interesting relations between information theory and statistical physics. This monograph
relates thermodynamic potentials (like the thermodynamical entropy and free energy) to information measures
(like the Shannon entropy and information divergence); it also provides some rigorous mathematical tools that
were inspired by the physical point of view and were proved to be useful in dealing with information-theoretic problems.
\end{remark}
Now we give the proof of Lemma~\ref{lm:entropy_via_fluctuations}:
\begin{proof} We start by noting that (see \eqref{eq:LMGF_derivatives})
	\begin{align}\label{eq:MGF_derivatives}
		\Lambda'(t) = \expectation^{(tf)}_\mu[f] \qquad \text{and} \qquad \Lambda''(t) = \var^{(tf)}_\mu[f],
	\end{align}
and, in particular, $\Lambda'(0) = \expectation_\mu[f]$.
Moreover, from \eqref{eq: differential form of Herbst's bound}, we get
	\begin{align}
	D\big( \mu^{(\lambda f)} \big\| \mu \big) = \lambda^2 \; \frac{\d}{\d \lambda} \left(\frac{\Lambda(\lambda)}{\lambda} \right) = \lambda \Lambda'(\lambda) - \Lambda(\lambda).\label{eq:entropy_overall}
	\end{align}
Now, using \eqref{eq:MGF_derivatives}, we get
\begin{align}
	\lambda \Lambda'(\lambda) &= \int^\lambda_0 \Lambda'(\lambda)\, \d t \nonumber \\
	&= \int^\lambda_0 \left( \int^\lambda_0 \Lambda''(s)\, \d s + \Lambda'(0)\right) \d t \nonumber \\
	&= \int^\lambda_0 \left(\int^\lambda_0 \var^{(sf)}_\mu[f]\, \d s  + \expectation_\mu[f]\right) \d t \label{eq:entropy_1st_term}
\end{align}
and
\begin{align}
	\Lambda(\lambda) &= \int^\lambda_0 \Lambda'(t) \, \d t \nonumber\\
	&= \int^\lambda_0 \left( \int^t_0 \Lambda''(s) \, \d s + \Lambda'(0)\right) \d t \nonumber\\
	&= \int^\lambda_0 \left( \int^t_0 \var^{(sf)}_\mu[f]\, \d s + \expectation_\mu[f]\right) \d t. \label{eq:entropy_2nd_term}
\end{align}
Substituting \eqref{eq:entropy_1st_term} and \eqref{eq:entropy_2nd_term} into \eqref{eq:entropy_overall}, we get \eqref{eq:entropy_via_fluctuations}.  \end{proof}

Now the whole affair hinges on the second step, which involves bounding the variances $\var^{(tf)}_\mu[f]$,
for $t > 0$, from above in terms of expectations $\expectation^{(tf)}_\mu\left[(\Gamma f)^2\right]$ for an appropriately chosen $\Gamma$. The following is sufficiently general for our needs:

\begin{theorem}\label{thm:Maurer_method} Let the objects $(\cA,\Gamma)$ and $\{(\cA_i,\Gamma_i)\}^n_{i=1}$ be constructed as in the statement of Theorem~\ref{thm:LSI_tensorization}. Furthermore, suppose that for each
$i \in \{1, \ldots, n\}$, the operator $\Gamma_i$ maps each $g \in \cA_i$ to a constant (which may depend on
$g$), and there exists a constant $c > 0$ such that the bound
\begin{align}\label{eq:Maurer_var_bound}
\var^{(sg)}_i[g(X_i)|\bar{X}^i = \bar{x}^i] \le c \left( \Gamma_i g \right)^2, \qquad \forall \, \bar{x}^i \in \cX^{n-1}
\end{align}
holds for all $i \in \{1, \ldots, n\}$, $s > 0$, and $g \in \cA_i$, where $\var^{(g)}_i[\cdot|\bar{X}^i = \bar{x}^i]$ denotes the (conditional) variance with respect to $P^{(g)}_{X_i|\bar{X}^i = \bar{x}^i}$. Then, the pair $(\cA,\Gamma)$ satisfies $\LSI(c)$ with respect to $P_{X^n}$.
\end{theorem}
\begin{proof} Consider an arbitrary function $f \in \cA$. Then, by construction, $f_i \colon \cX_i \to \reals$ is in $\cA_i$ for each
$i \in \{1, \ldots, n\}$. We can write
	\begin{align*}
		D \Big(P^{(f)}_{X_i|\bar{X}^i = \bar{x}^i} \Big\| P_{X_i}\Big)
        &= D \Big( P^{(f_i(\cdot|\bar{x}^i))}_{X_i} \Big\| P_{X_i} \Big) \\
		&= \int^1_0 \int^1_t \var^{(s f_i(\cdot|\bar{x}^i))}_i[f_i(X_i|\bar{X}^i)|\bar{X}^i = \bar{x}^i]\,
        \d s\, \d t \\
		&\le c \left( \Gamma_i f_i \right)^2 \int^1_0 \int^1_t \d s\, \d t \\
		&= \frac{c (\Gamma_i f_i)^2}{2}
	\end{align*}
where the first step uses the fact that $P^{(f)}_{X_i|\bar{X}^i = \bar{x}^i}$ is equal to the $f_i(\cdot|\bar{x}^i)$-tilting of $P_{X_i}$, the second step uses Lemma~\ref{lm:entropy_via_fluctuations}, and the third step uses \eqref{eq:Maurer_var_bound} with $g = f_i(\cdot|\bar{x}^i)$. We have therefore established that, for each $i$, the pair $(\cA_i,\Gamma_i)$ satisfies $\LSI(c)$. Therefore, the pair $(\cA,\Gamma)$ satisfies $\LSI(c)$ by Theorem~\ref{thm:LSI_tensorization}.
\end{proof}

The following two lemmas from \cite{Maurer_thermo} will be useful for establishing bounds like \eqref{eq:Maurer_var_bound}:

\begin{lemma}\label{lm:Hoeffding_bound} Let $U$ be a random variable such that
$U \in [a,b]$ a.s.\ for some $-\infty < a \le b < +\infty$. Then
	\begin{align}\label{eq:Hoeffding_bound}
		\var[U] \le \frac{(b-a)^2}{4}.
	\end{align}
\end{lemma}
\begin{proof} Since the support of $U$ is the interval $[a,b]$, the maximal variance of $U$
is attained when the random variable $U$ is binary and equiprobable on the endpoints of this interval. The bound in \eqref{eq:Hoeffding_bound} is achieved with equality in this case.
\end{proof}

\begin{lemma}\label{lm:var_dominance} Let $f$ be a real-valued function
such that $f - \expectation_\mu[f] \le C$ for some $C \in \reals$. Then, for every $t > 0$,
	\begin{align*}
		\var^{(tf)}_\mu[f] \le \exp(tC)\, \var_\mu[f].
	\end{align*}
\end{lemma}
\begin{proof}
	\begin{align}
		\var^{(tf)}_\mu[f] &= \var^{(tf)}_\mu\Big\{f - \expectation_\mu\left[ f \right] \Big\}
\label{eq:var_dominance_0} \\
		&\le \expectation^{(tf)}_\mu \left[ \left(f - \expectation_\mu[f] \right)^2 \right] \label{eq:var_dominance_1}\\
		&= \expectation_\mu \left[ \frac{\exp(tf) \left(f - \expectation_\mu[f]\right)^2}{\expectation_\mu[\exp(tf)]}\right] \label{eq:var_dominance_2}\\
		&\le \expectation_\mu \left\{  \left(f - \expectation_\mu[f]\right)^2 \exp\left[t\left(f - \expectation_\mu[f]\right)\right] \right\} \label{eq:var_dominance_3} \\
		&\le \exp(t C) \, \expectation_\mu \left[ \left(f - \expectation_\mu\left[f\right]\right)^2 \right] \label{eq:var_dominance_4},
	\end{align}
	where:
	\begin{itemize}
        \item \eqref{eq:var_dominance_0} holds since $\var[f] = \var[f + c]$ for every constant $c \in \reals$;
		\item \eqref{eq:var_dominance_1} uses the bound $\var[U] \le \expectation[U^2]$;
		\item \eqref{eq:var_dominance_2} is by definition of the tilted distribution $\mu^{(tf)}$;
		\item \eqref{eq:var_dominance_3} is verified by applying Jensen's inequality to the denominator, and
		\item \eqref{eq:var_dominance_4} relies on the assumption that $f - \expectation_\mu[f] \le C$, and the monotonicity of the exponential  function (note that $t>0$).
	\end{itemize}
	This completes the proof of Lemma~\ref{lm:var_dominance}.
\end{proof}

\subsection{Discrete logarithmic Sobolev inequalities on the Hamming cube}

We now use Maurer's method to derive log-Sobolev inequalities for functions of $n$ i.i.d.\ Bernoulli random variables. Let $\cX$ be the two-point set $\{0,1\}$, and let $e_i \in \cX^n$ denote the binary string that has $1$ in the $i$th position and zeros elsewhere. Finally, for every $f \colon \cX^n \to \reals$, define
\begin{align}\label{eq:Hamming_Gamma}
\Gamma f(x^n) \deq \sqrt{\sum^n_{i=1} \bigl( f(x^n \oplus e_i) - f(x^n) \bigr)^2}, \qquad \forall \, x^n \in \cX^n,
\end{align}
where the modulo-$2$ addition $\oplus$ is defined componentwise. In other words, $\Gamma f$ measures the sensitivity of $f$ to local bit flips. We consider the symmetric, i.e., $\Bernoulli(1/2)$, case first:

\begin{theorem}[Discrete log-Sobolev inequality for the symmetric Bernoulli measure]
\label{log-Sobolev inequality for the symmetric Bernoulli measure} Let $\cA$ be the set of all functions $f \colon \cX^n \to \reals$. Then, the pair $(\cA,\Gamma)$ with $\Gamma$ defined in \eqref{eq:Hamming_Gamma} satisfies the conditions (LSI-1)--(LSI-3). Let $X_1,\ldots,X_n$ be $n$ i.i.d.\ $\Bernoulli(1/2)$ random variables, and let $P$ denote their joint distribution. Then, $P$ satisfies $\LSI(1/4)$ with respect to $(\cA,\Gamma)$. In other words, for every $f \colon \cX^n \to \reals$,
	\begin{align}\label{eq:Hamming_LSI}
	D \big( P^{(f)} \big \| P \big) \le \frac{1}{8} \,
    \expectation_P^{(f)}\left[ \left(\Gamma f \right)^2\right].
	\end{align}
\end{theorem}
\begin{proof} Let $\cA_0$ be the set of all functions $g \colon \{0,1\} \to \reals$, and let $\Gamma_0$ be the operator that maps every $g \in \cA_0$ to
	\begin{align}\label{eq:Hamming_Gamma_1d}
		\Gamma_0 \, g \deq |g(0) - g(1)| = |g(x) - g(x \oplus 1)|, \quad \forall \, x \in \{0, 1\}.
	\end{align}
For each $i \in \{1,\ldots,n\}$, let $(\cA_i,\Gamma_i)$ be a copy of $(\cA_0,\Gamma_0)$. Then, each $\Gamma_i$ maps every function $g \in \cA_i$ to the constant $|g(0)-g(1)|$. Moreover, for every $g \in \cA_i$, the random variable $U_i = g(X_i)$ is bounded between $g(0)$ and $g(1)$, where we can assume without loss of generality that $g(0) \le g(1)$. Hence, by Lemma~\ref{lm:Hoeffding_bound}, we have
\begin{align}
	\var^{(sg)}_{i}[g(X_i)|\bar{X}^i = \bar{x}^i] \le \frac{\bigl(g(0) - g(1)\bigr)^2}{4} = \frac{(\Gamma_i g)^2}{4}
\end{align}
for all $g \in \cA_i, \, \bar{x}^i \in \cX^{n-1}$. In other words, the condition \eqref{eq:Maurer_var_bound} of Theorem~\ref{thm:Maurer_method} holds with $c=1/4$. In addition, it is easy to see that the operator $\Gamma$ constructed from $\Gamma_1,\ldots,\Gamma_n$ according to \eqref{eq:gamma_def} is precisely the one in \eqref{eq:Hamming_Gamma}. Therefore, by Theorem~\ref{thm:Maurer_method}, the pair $(\cA,\Gamma)$ satisfies $\LSI(1/4)$ with respect to $P$, which proves \eqref{eq:Hamming_LSI}. This completes
the proof of Theorem~\ref{log-Sobolev inequality for the symmetric Bernoulli measure}.
\end{proof}

Now let us consider the case when $X_1,\ldots,X_n$ are i.i.d.\ $\Bernoulli(p)$ random variables with some $p \neq 1/2$. We will use Maurer's method to give an alternative, simpler proof of the following result of Ledoux \cite[Corollary~5.9]{Ledoux_lecture_notes} (it actually suggests a sharpened version of the latter result, as it is explained in Remark~\ref{remark: improving the bound of Ledoux}):

\begin{theorem}\label{thm:LSI_Bernoulli_p} Consider an arbitrary function $f \colon \{0,1\}^n \to \reals$ with the property that
there exists some $c>0$ such that
	\begin{align}\label{eq:bit_flips}
		\max_{i \in \{1,\ldots,n\}} \left|f(x^n \oplus e_i) - f(x^n)\right| \le c
	\end{align}
	for all $x^n \in \{0,1\}^n$. Let $X_1,\ldots,X_n$ be $n$ i.i.d.\ $\Bernoulli(p)$ random variables, and let $P$ be their joint distribution. Then
	\begin{align}\label{eq:LSI_Bernoulli_p}
		D\big(P^{(f)} \big\| P \big) \le pq \left( \frac{(qc-1) \exp(qc) + 1}{(qc)^2} \right) \expectation_P^{(f)}\left[ (\Gamma f)^2 \right],
	\end{align}
	where $q \deq 1-p$.
\end{theorem}

\begin{proof} Following the usual route, we will establish the $n=1$ case first, and then scale up to an arbitrary $n$ by tensorization. In order to capture the correct dependence on the Bernoulli parameter $p$, we will use a more refined, distribution-dependent variance bound of Lemma~\ref{lm:var_dominance}, as opposed to the cruder bound of Lemma~\ref{lm:Hoeffding_bound} that does not depend on the underlying distribution. Maurer's paper \cite{Maurer_thermo} has other examples.
	
Let $a = |\Gamma f| = |f(0) - f(1)|$, where $\Gamma$ is defined as in \eqref{eq:Hamming_Gamma_1d}. Without loss of generality, let $f(0) = 0$ and $f(1) = a$. Then
\begin{align}\label{eq:Bern_p_moments}
	\expectation[f] &= pa \qquad \text{and} \qquad \var[f] = pqa^2.
\end{align}
Using \eqref{eq:Bern_p_moments} and Lemma~\ref{lm:var_dominance}, since $f - \expectation[f] \le a - pa = qa$, it follows that for every $t > 0$
\begin{align*}
	\var_P^{(tf)}[f] \le pqa^2 \exp(tqa).
\end{align*}
Therefore, by Lemma~\ref{lm:entropy_via_fluctuations} we have
\begin{align*}
	D \big( P^{(f)} \big\| P \big) &\le pqa^2 \int^1_0 \int^1_t \exp(sqa)\, \d s\, \d t \\	
	&= pqa^2 \left( \frac{(qa-1) \exp(qa) + 1}{(qa)^2} \right) \\
	&\le pqa^2 \left(\frac{(qc-1) \exp(qc) + 1}{(qc)^2}\right),
\end{align*}
where the last step follows from the fact that the function $$u \mapsto u^{-2}[(u-1) \exp(u) + 1]$$ (defined, for continuity, to be $\frac{1}{2}$ at $u=0$) is monotonic increasing in $[0, \infty)$, and $0 \leq qa \leq qc$. Since $a^2 = (\Gamma f)^2$, we can write
\begin{align*}
	D \big( P^{(f)} \big\| P \big) \le pq \left(\frac{(qc-1) \exp(qc)+ 1}{(qc)^2}\right) \expectation_P^{(f)}
\left[ (\Gamma f)^2 \right],
\end{align*}
so we have established \eqref{eq:LSI_Bernoulli_p} for $n=1$.

Now consider an arbitrary $n \in \naturals$. Since the condition in \eqref{eq:bit_flips} can be expressed as
\begin{align*}
\left| f_i(0|\bar{x}^i) - f_i(1|\bar{x}^i) \right| \le c, \qquad \forall \, i \in \{1, \ldots, n\}, \,
\bar{x}^i \in \{0,1\}^{n-1},
\end{align*}
we can use \eqref{eq:LSI_Bernoulli_p} to write
\begin{align*}
& D\Big( P^{(f_i(\cdot|\bar{x}^i))}_{X_i} \Big\| P_{X_i} \Big) \\
& \le pq \left( \frac{(qc-1) \exp(qc) + 1}{(qc)^2} \right) \expectation_{P_{X_i}}^{(f_i(\cdot|\bar{x}^i))}
\Big[ \left(\Gamma_i f_i(X_i|\bar{x}^i) \right)^2 \Big]
\end{align*}
for every $i=1,\ldots,n$ and all $\bar{x}^i \in \{0,1\}^{n-1}$. With this, the same sequence of steps that
led to \eqref{eq:LSI_tensorization} in the proof of Theorem~\ref{thm:LSI_tensorization} can be used to
complete the proof of \eqref{eq:LSI_Bernoulli_p} for an arbitrary $n$.
\end{proof}
In Appendix~\ref{appendix: binary to Gaussian LSI}, we comment on the relations between the log-Sobolev
inequalities for the Bernoulli and the Gaussian measures.

\begin{remark} \label{remark: improving the bound of Ledoux}
Note that \eqref{eq:LSI_Bernoulli_p} improves the bound of Ledoux in \cite[Corollary~5.9]{Ledoux_lecture_notes},
which is equivalent to (see \eqref{eq:ent_vs_D} and \eqref{eq:another equality with Gamma f})
\begin{align}\label{eq:Ledoux's LSI_Bernoulli_p}
		D\big(P^{(f)} \big\| P \big) \le pq \left( \frac{(c-1) \exp(c) + 1}{c^2} \right) \expectation_P^{(f)}\left[ (\Gamma f)^2 \right].
	\end{align}
The improvement in \eqref{eq:LSI_Bernoulli_p} follows from a replacement of $c$ on the right-hand side of
\eqref{eq:Ledoux's LSI_Bernoulli_p} with $qc$; this can be verified due the fact that the function
$$u \mapsto u^{-2}[(u-1) \exp(u) + 1], \quad u > 0$$ is monotonic increasing.
\end{remark}

\subsection{The method of bounded differences revisited}
\label{ssec:bounded_differences_revisited}

As our second illustration of the use of Maurer's method, we will give an information-theoretic proof
of McDiarmid's inequality (recall that the original proof in \cite{McDiarmid_tutorial,McDiarmid_1997}
used the martingale method; the reader is referred to the
derivation of McDiarmid's inequality via the martingale approach in
Theorem~\ref{theorem: McDiarmid's inequality} of the preceding chapter).
Following the exposition in \cite[Section~4.1]{Maurer_thermo}, we have the following re-statement of
McDiarmid's inequality in Theorem~\ref{theorem: McDiarmid's inequality}:

\begin{theorem}\label{thm:bounded_differences_revisited} Let $X_1,\ldots,X_n$ be independent $\cX$-valued
random variables. Consider a function $f \colon \cX^n \to \reals$ with $\expectation [f(X^n)] = 0$, and
also suppose that there exist some constants $0 \le c_1,\ldots,c_n < + \infty$ such that, for each
$i \in \{1,\ldots,n\}$,
\begin{align}\label{eq:bounded_differences}
	\left|f_i(x|\bar{x}^i) - f_i(y|\bar{x}^i)\right| \le c_i, \qquad \forall \, x,y \in \cX, \;
    \bar{x}^i \in \cX^{n-1}.
\end{align}
Then, for every $r \ge 0$,
\begin{align}\label{eq:McDiarmid}
	\pr \Big( f(X^n) \ge r \Big) \le \exp\left( - \frac{2r^2}{\sum^n_{i=1}c^2_i}\right).
\end{align}
\end{theorem}

\begin{proof} Let $\cA_0$ be the set of all bounded measurable functions $g \colon \cX \to \reals$,
and let $\Gamma_0$ be the operator that maps every
$g \in \cA_0$ to $$\Gamma_0 \, g \deq \sup_{x \in \cX} g(x) - \inf_{x \in \cX} g(x).$$
It is easy to verify that properties (LSI-1)--(LSI-3) hold for the pair $(\cA_0, \Gamma_0)$ since
in particular
$$\Gamma_0(a g + b) = a \, \Gamma_0 \, g, \quad \forall \, a \ge 0, \; b \in \reals.$$
Now, for each $i \in \{1,\ldots,n\}$, let $(\cA_i,\Gamma_i)$ be a copy of $(\cA_0,\Gamma_0)$.
Then, each $\Gamma_i$ maps every function $g \in \cA_i$ to a non-negative constant.
Moreover, for every $g \in \cA_i$, the random variable $U_i = g(X_i)$ is bounded between
$\inf_{x \in \cX}g(x)$ and $\sup_{x \in \cX} g(x) \equiv \inf_{x \in \cX} g(x) + \Gamma_i g$.
Therefore, Lemma~\ref{lm:Hoeffding_bound} gives
\begin{align*}
\var^{(sg)}_i[g(X_i)|\bar{X}^i = \bar{x}^i] \le \frac{(\Gamma_i g)^2}{4},
\qquad \forall \, g \in \cA_i, \; \bar{x}^i \in \cX^{n-1}.
\end{align*}
Hence, the condition \eqref{eq:Maurer_var_bound} of Theorem~\ref{thm:Maurer_method} holds
with $c=1/4$. Now let $\cA$ be the set of all bounded measurable functions $f \colon \cX^n \to \reals$.
Then, for every $f \in \cA$, $i \in \{1, \ldots, n\}$ and $x^n \in \cX^n$, we have
\begin{align*}
&\sup_{x_i \in \cX_i} f(x_1,\ldots,x_i,\ldots,x_n) - \inf_{x_i \in \cX_i} f(x_1,\ldots,x_i,\ldots,x_n) \\
&= \sup_{x_i \in \cX_i} f_i(x_i|\bar{x}^i) - \inf_{x_i \in \cX_i} f_i(x_i|\bar{x}^i) \\
&= \Gamma_i f_i(\cdot|\bar{x}^i).
\end{align*}
Thus, if we construct an operator $\Gamma$ on $\cA$ from $\Gamma_1,\ldots,\Gamma_n$ according to
\eqref{eq:gamma_def}, the pair $(\cA,\Gamma)$ will satisfy the conditions of
Theorem~\ref{thm:LSI_tensorization}. Therefore, by Theorem~\ref{thm:Maurer_method}, it follows
that the pair $(\cA,\Gamma)$ satisfies $\LSI(1/4)$ for {\em every} product probability measure on
$\cX^n$. Hence, inequality~\eqref{eq:LSI_to_concentration} implies that
\begin{align}\label{eq:McDiarmid_1}
	\pr \Big( f(X^n) \ge r\Big) \le \exp\left(- \frac{2r^2}{\left\|\Gamma f\right\|^2_\infty}\right)
\end{align}
holds for every $r \ge 0$ and bounded $f$ with $\expectation[f(X^n)] = 0$. Now, if $f$ satisfies
\eqref{eq:bounded_differences}, then
\begin{align*}
	\| \Gamma f \|^2_\infty &= \sup_{x^n \in \cX^n}\sum^n_{i=1} \big( \Gamma_i f_i(x_i|\bar{x}^i) \big)^2 \\
	&\le \sum^n_{i=1}\sup_{x^n \in \cX^n} \big(\Gamma_i f_i(x_i|\bar{x}^i)\bigr)^2 \\
	&= \sum^n_{i=1}\sup_{x^n \in \cX^n, \, y \in \cX}|f_i(x_i|\bar{x}^i) - f_i(y|\bar{x}^i)|^2 \\
	&\le \sum^n_{i=1}c^2_i.
\end{align*}
Substituting this bound into the right-hand side of \eqref{eq:McDiarmid_1} gives \eqref{eq:McDiarmid}.
\end{proof}
Note that Maurer's method gives the correct constant in the exponent of McDiarmid's inequality; it
is instructive to compare it to an earlier approach in \cite{Boucheron_Lugosi_Massart} which, by also using
the entropy method, gave an exponent that is smaller by a factor of~8.

\subsection{Log-Sobolev inequalities for Poisson and compound Poisson measures}
\label{subsection: Log-Sobolev inequalities for Poisson and compound Poisson measures}

Let $\sP_\lambda$ denote, for an arbitrary $\lambda > 0$, the $\Poisson(\lambda)$ measure, i.e.,
$\sP_\lambda(n) \deq \frac{e^{-\lambda} \, \lambda^n}{n!}$ for every $n \in \naturals_0$,
where $\naturals_0 \deq \naturals \cup \{0\}$ is the set of the non-negative integers.
Bobkov and Ledoux \cite{Bobkov_Ledoux} have established the following log-Sobolev inequality: for every function $f \colon \naturals_0 \to \reals$,
\begin{align}\label{eq:Bobkov_Ledoux_LSI}
	D\Big(\sP^{(f)}_\lambda \Big\| \sP_\lambda \Big) \le \lambda\, \expectation^{(f)}_{\sP_\lambda}\left[ (\Gamma f) \; e^{\Gamma f} - e^{\Gamma f} + 1\right],
\end{align}
where $\Gamma$ is the modulus of the discrete gradient:
\begin{align}\label{eq:discrete_grad}
	\Gamma f(x) \deq |f(x) - f(x+1)|, \qquad \forall \, x \in \naturals_0.
\end{align}
(The inequality \eqref{eq:Bobkov_Ledoux_LSI} can be obtained by combining the log-Sobolev inequality in \cite[Corollary~7]{Bobkov_Ledoux}
with equality \eqref{eq:ent_vs_D}.) Using tensorization of \eqref{eq:Bobkov_Ledoux_LSI}, Kontoyiannis and Madiman \cite{Kontoyiannis_Madiman_CP}
gave a simple proof of a log-Sobolev inequality for the {\em compound Poisson distribution}. We recall that a
compound Poisson distribution is defined as follows: given $\lambda > 0$ and a probability measure $\mu$ on $\naturals$, the compound Poisson distribution
$\CP_{\lambda,\mu}$ is the distribution of the random sum
\begin{equation} \label{eq: CP}
Z = \sum^N_{i=1}X_i,
\end{equation}
where $N \sim \sP_\lambda$ and $X_1, X_2,\ldots$ are i.i.d.\ random variables with
distribution $\mu$, independent of $N$ (if $N$ takes the value zero, then $Z$ is defined to be zero).

\begin{theorem}[Log-Sobolev inequality for compound Poisson measures \cite{Kontoyiannis_Madiman_CP}]\label{thm:Kontoyiannis_Madiman} For an arbitrary
probability measure $\mu$ on $\naturals$ and an arbitrary bounded function $f \colon \naturals_0 \to \reals$,
and for every $\lambda > 0$,
	\begin{align}\label{eq:CP_LSI}
		D\Big( \CP_{\lambda,\mu}^{(f)} \Big\| \CP_{\lambda,\mu}\Big) \le \lambda \, \sum^\infty_{k=1} \mu(k)\, \expectation^{(f)}_{\CP_{\lambda,\mu}}\left[(\Gamma_k f) \, e^{\Gamma_k f} - e^{\Gamma_k f} + 1 \right],
	\end{align}
	where $\Gamma_k f(x) \deq |f(x) - f(x+k)|$ for each $k \in \naturals$ and $x \in \naturals_0$.
\label{theorem: Kontoyiannis_Madiman_CP}
\end{theorem}
\begin{proof} The proof relies on the following alternative representation of the $\CP_{\lambda,\mu}$ probability measure:
    \begin{lemma} \label{lemma: compound Poisson distribution}
    If $Z \sim \CP_{\lambda,\mu}$, then
	\begin{align}\label{eq:CP_alt}
		Z \eqdist \sum^\infty_{k=1} kY_k, \qquad Y_k \sim \sP_{\lambda \mu(k)}, \; \; \forall \, k \in \naturals
	\end{align}
	where $\{Y_k\}_{k=1}^{\infty}$ are independent random variables, and $\eqdist$ means equality in distribution.
    \end{lemma}
    \begin{proof}
    The characteristic function of $Z$ in \eqref{eq:CP_alt} is equal to
    $$\varphi_Z(\nu) \deq \expectation[\exp(j \nu Z)] = \exp \left\{\lambda \left(\sum_{k=1}^{\infty} \mu(k) \exp(j \nu k) - 1 \right) \right\}, \quad \forall \, \nu \in \reals$$
    which coincides with the characteristic function of $Z \sim \CP_{\lambda,\mu}$ in \eqref{eq: CP}.
    The statement of the lemma follows from the fact that two
    random variables are equal in distribution if and only if their characteristic functions coincide.
    \end{proof}
        For each $n \in \naturals$, let $P_n$ denote the product distribution of $Y_1,\ldots,Y_n$. Consider an arbitrary bounded function
    $f \colon \naturals_0 \to \reals$, and define the function $g \colon (\naturals_0)^n \to \reals$ by
	\begin{align*}
		g(y_1,\ldots,y_n) \deq f\left(\sum^n_{k=1}ky_k\right), \qquad \forall \, y_1,\ldots,y_n \in \naturals_0.
	\end{align*}
If we now denote by $\bar{P}_n$ the distribution of the sum $S_n \deq \sum^n_{k=1}kY_k$, then
\begin{align} \label{eq: upper bound on the divergence between tilted P and P}
	D\Big(\bar{P}^{(f)}_n \big\| \bar{P}_n \Big) &= \expectation_{\bar{P}_n} \left[ \left( \frac{\exp\big(f(S_n)\big)}{\expectation_{\bar{P}_n}[\exp\big(f(S_n)\big)]} \right) \; \ln \left( \frac{\exp\big(f(S_n)\big)}{\expectation_{\bar{P}_n}[\exp\big(f(S_n)\big)]} \right) \right] \nonumber \\[0.1cm]
	&= \expectation_{P_n} \left[ \left(\frac{\exp\big(g(Y^n)\big)}{\expectation_{P_n}[\exp\big(g(Y^n)\big)]} \right) \; \ln \left(\frac{\exp\big(g(Y^n)\big)}{\expectation_{P_n}[\exp\big(g(Y^n)\big)]}\right) \right] \nonumber \\[0.1cm]
	&= D\big(P^{(g)}_n \big\| P_n \big) \nonumber \\
	&\le \sum^n_{k=1} D\Big(P^{(g)}_{Y_k|\bar{Y}^k} \Big\| P_{Y_k} \Big| P^{(g)}_{\bar{Y}^k}\Big),
\end{align}
where the last line uses Proposition~\ref{prop:erasure_entropy_bound} and the fact that $P_n$ is a product distribution. Using the fact that
$$
\frac{\d P^{(g)}_{Y_k|\bar{Y}^k=\bar{y}^k}}{\d P_{Y_k}} = \frac{\exp\big(g_k(\cdot|\bar{y}^k)\big)}{\expectation_{\sP_{\lambda\mu(k)}} [\exp\big(g_k(Y_k|\bar{y}^k)\big)]}, \qquad P_{Y_k}= \sP_{\lambda\mu(k)}
$$
and applying the Bobkov--Ledoux inequality \eqref{eq:Bobkov_Ledoux_LSI} to $P_{Y_k}$ and all functions of the form $g_k(\cdot|\bar{y}^k)$, we can write
\begin{align}
& D\Big(P^{(g)}_{Y_k|\bar{Y}^k} \big\| P_{Y_k} \, \big| \, P^{(g)}_{\bar{Y}^k}\Big) \nonumber\\
& \le \lambda \mu(k)\, \expectation^{(g)}_{P_n}\left[ \big(\Gamma g_k(Y_k|\bar{Y}^k)\big) \, e^{\Gamma g_k(Y_k|\bar{Y}^k)} -
e^{\Gamma g_k(Y_k|\bar{Y}^k)} + 1\right] \label{eq:CP_LSI_1}
\end{align}
where $\Gamma$ is the absolute value of the ``one-dimensional'' discrete gradient in \eqref{eq:discrete_grad}.
For every $y^n \in (\naturals_0)^n$, we have
\begin{align*}
	\Gamma g_k(y_k|\bar{y}^k) &= \left|g_k(y_k|\bar{y}^k) - g_k(y_k + 1|\bar{y}^k)\right| \\
	&= \Bigg|f\left( ky_k + \sum_{j \in \{1,\ldots,n\}\backslash\{k\}} jy_j\right) \\
	& \qquad -  f\left( k(y_k+1) + \sum_{j \in \{1,\ldots,n\}\backslash\{k\}} jy_j\right)\Bigg| \\
	&= \left|f\left( \sum_{j=1}^n jy_j\right) -  f\left( \sum_{j=1}^n jy_j + k\right)\right| \\
	&= \Gamma_k f\left( \sum^n_{j=1}jy_j\right) = \Gamma_k f(S_n).
\end{align*}
Using this in \eqref{eq:CP_LSI_1} and performing the reverse change of measure from $P_n$ to $\bar{P}_n$, we can write
\begin{align}
	&D\Big(P^{(g)}_{Y_k|\bar{Y}^k} \Big\| P_{Y_k} \, \Big| \, P^{(g)}_{\bar{Y}^k}\Big) \nonumber\\
	&\le \lambda \mu(k)\, \expectation^{(f)}_{\bar{P}_n}\left[ \big(\Gamma_k f(S_n)\big) e^{\Gamma_k f(S_n)} - e^{\Gamma_k f(S_n)} + 1\right].\label{eq: upper bound on the considered conditional divergence}
\end{align}
Therefore, the combination of \eqref{eq: upper bound on the divergence between tilted P and P} and
\eqref{eq: upper bound on the considered conditional divergence} gives
\begin{align}
D\big(\bar{P}^{(f)}_n \big\| \bar{P}_n \big) &\le \lambda \sum^n_{k=1}\mu(k) \,
\expectation^{(f)}_{\bar{P}_n}\left[ (\Gamma_k f) \, e^{\Gamma_k f} - e^{\Gamma_k f} + 1\right] \nonumber \\
&\le \lambda \sum^\infty_{k=1} \mu(k) \, \expectation^{(f)}_{\bar{P}_n}\left[ (\Gamma_k f) \, e^{\Gamma_k f} - e^{\Gamma_k f} + 1\right] \label{eq:CP_LSI_2}
\end{align}
where the second line follows from the inequality $xe^x - e^x + 1 \ge 0$ that holds for all $x \ge 0$.

Now we will take the limit as $n \to \infty$ of both sides of \eqref{eq:CP_LSI_2}. For the left-hand side, we use the fact that, by \eqref{eq:CP_alt}, $\bar{P}_n$ converges in distribution to $\CP_{\lambda,\mu}$ as $n \to \infty$. Since $f$ is bounded, $\bar{P}^{(f)}_n \to \CP^{(f)}_{\lambda,\mu}$ in distribution. Therefore, by the bounded convergence theorem, we have
\begin{align} \label{eq: an implication of the bounded convergence theorem}
	\lim_{n \to \infty}D\big( \bar{P}^{(f)}_n \big\| \bar{P}_n \big) = D\Big(\CP^{(f)}_{\lambda,\mu} \, \Big\| \, \CP_{\lambda,\mu}\Big).
\end{align}
For the right-hand side of \eqref{eq:CP_LSI_2}, we have
\begin{align} \label{eq: limit of the right-hand side where n tends to infinity}
& \sum^\infty_{k=1}\mu(k) \, \expectation^{(f)}_{\bar{P}_n}\left[ (\Gamma_k f) \, e^{\Gamma_k f} - e^{\Gamma_k f} + 1\right] \nonumber\\
& = \expectation^{(f)}_{\bar{P}_n}\left\{\sum^\infty_{k=1}\mu(k) \left[ (\Gamma_k f) \, e^{\Gamma_k f} - e^{\Gamma_k f} + 1\right]\right\} \nonumber \\
& \xrightarrow{n \to \infty} \expectation^{(f)}_{\CP_{\lambda,\mu}}\left[\sum^\infty_{k=1}\mu(k) \left(
(\Gamma_k f) \, e^{\Gamma_k f} - e^{\Gamma_k f} + 1\right) \right] \nonumber \\
& = 	\sum^\infty_{k=1}\mu(k) \, \expectation^{(f)}_{\CP_{\lambda,\mu}}\left[ (\Gamma_k f) \, e^{\Gamma_k f} - e^{\Gamma_k f} + 1\right]
\end{align}
where the first and last steps follow from Fubini's theorem, and the second step follows from
the bounded convergence theorem. Putting
\eqref{eq:CP_LSI_2}--\eqref{eq: limit of the right-hand side where n tends to infinity} together,
we get the inequality in \eqref{eq:CP_LSI}.
This completes the proof of Theorem~\ref{theorem: Kontoyiannis_Madiman_CP}.
\end{proof}

\subsection{Bounds on the variance: Efron--Stein--Steele and Poincar\'e inequalities}

As we have seen, tight bounds on the {\em variance} of a function $f(X^n)$ of independent random variables $X_1,\ldots,X_n$ are key to obtaining tight bounds on the deviation probabilities $\pr \big( f(X^n) \ge \expectation f(X^n) + r \big)$ for $r \ge 0$. It turns out that the reverse is also true: assuming that $f$ has Gaussian-like concentration behavior,
\begin{align*}
	\pr \big( f(X^n) \ge \expectation f(X^n) + r \big) \le K \exp\big(-\kappa r^2 \big), \qquad \forall \, r \ge 0
\end{align*}
it is possible to derive tight bounds on the variance of $f(X^n)$.

We start by deriving a version of a well-known inequality due to Efron and Stein \cite{Efron_Stein}, with subsequent refinements by Steele \cite{Steele}:

\begin{theorem}[Efron--Stein--Steele inequality]\label{thm:ESS} Let $X_1,\ldots,X_n$ be $n$ independent $\cX$-valued random variables. Consider an arbitrary function $f \colon \cX^n \to \reals$ such that its scaled versions $tf$ are exponentially integrable for all sufficiently small
$t > 0$. Then
	\begin{align}\label{eq:Efron_Stein_Steele}
		\var[f(X^n)] \le \sum^n_{i=1} \expectation\left\{ \var\big[f(X^n)\big|\bar{X}^i\big]\right\}.
	\end{align}
\end{theorem}
\begin{proof} Let $P = P_{X_1} \otimes \ldots \otimes P_{X_n}$ be the joint probability distribution of $X_1, \ldots, X_n$.
By Proposition~\ref{prop:erasure_entropy_bound}, for every $t > 0$, we have
	\begin{align*}
		D \big( P^{(tf)} \big\| P \big) \le \sum^n_{i=1} D\big(P^{(tf)}_{X_i|\bar{X}^i} \big\| P_{X_i} \big| P_{\bar{X}^i}^{(tf)} \big).
	\end{align*}
Using Lemma~\ref{lm:entropy_via_fluctuations}, we can rewrite this inequality as
	\begin{align} \label{eq:inequlity from the erasure_entropy_bound and Maurer's method}
		& \int^t_0 \int^t_s \var^{(\tau f)}_{P}[f] \,\d \tau\,\d s  \nonumber \\
        & \le \sum^n_{i=1}\expectation_{P^{(tf)}_{\bar{X}^i}}\left[ \int^t_0 \int^t_s \var^{(\tau f_i(\cdot|\bar{X}^i))}_{P_{X_i|\bar{X}^i}}[f]\,\d \tau\, \d s\right].
	\end{align}
Dividing both sides by $t^2$, and passing to the limit as $t \to 0$, we get from L'H\^{o}pital's rule
\begin{align} \label{eq:limit1}
	\lim_{t \to 0} \frac{1}{t^2}\int^t_0 \int^t_s \var^{(\tau f)}_{P}[f]\, \d\tau\, \d s = \frac{\var_P[f]}{2} = \frac{\var[f(X^n)]}{2},
\end{align}
and
\begin{align}  \label{eq:limit2}
    & \lim_{t \to 0} \frac{1}{t^2} \, \sum^n_{i=1} \expectation_{P^{(tf)}_{\bar{X}^i}}\left[ \int^t_0 \int^t_s \var^{(\tau f_i(\cdot|\bar{X}^i))}_{P_{X_i|\bar{X}^i}}[f]\,\d \tau\, \d s\right] \nonumber \\
    & = \sum^n_{i=1} \expectation_{P_{\bar{X}^i}} \left\{ \lim_{t \to 0} \frac{1}{t^2} \int^t_0 \int^t_s \var^{(\tau f_i(\cdot|\bar{X}^i))}_{P_{X_i|\bar{X}^i}}[f]\,\d \tau\, \d s\right\} \nonumber \\
    & = \sum^n_{i=1} \expectation_{P_{\bar{X}^i}} \left\{ \frac{\var_{P_{X_i|\bar{X}^i}}[f]}{2} \right\} \nonumber \\
    & = \frac{1}{2} \sum^n_{i=1} \expectation \left\{ \var \big[f(X^n)\big|\bar{X}^i\big] \right\}
\end{align}
where the first equality in \eqref{eq:limit2} is justified by invoking the dominated convergence theorem (recall the pointwise
convergence of $P^{(tf)}_{\bar{X}^i}$ to $P_{\bar{X}^i}$, as $t \to 0$,
which holds under the assumption that the scaled functions $tf$ are exponentially integrable for all sufficiently small $t > 0$), and the second equality holds due to
L'H\^{o}pital's rule.
Inequality \eqref{eq:Efron_Stein_Steele} finally follows from \eqref{eq:inequlity from the erasure_entropy_bound and Maurer's method}--\eqref{eq:limit2}.
\end{proof}

Next, we discuss the connection between log-Sobolev inequalities and another class of functional inequalities, the so-called {\em Poincar\'e inequalities}. Consider, as before, a probability space $(\Omega,\cF,\mu)$ and a pair $(\cA,\Gamma)$ satisfying the conditions (LSI-1)--(LSI-3). Then, we say that $\mu$ satisfies a {\em Poincar\'e inequality} with constant $c \ge 0$ if
\begin{align}\label{eq:Poincare}
	\var_\mu[f] &\le c\,\expectation_\mu\left[\left(\Gamma f\right)^2\right], \qquad \forall \, f \in \cA.
\end{align}
\begin{theorem} Suppose that $\mu$ satisfies $\LSI(c)$ with respect to $(\cA,\Gamma)$. Then $\mu$ also satisfies a Poincar\'e inequality with constant $c$.
\end{theorem}
\begin{proof} For every $f \in \cA$ and $t >0$, we can use Lemma~\ref{lm:entropy_via_fluctuations} to express the corresponding $\LSI(c)$ for the function $tf$ as
	\begin{align}
		\int^t_0 \int^t_s \var^{(\tau f)}_\mu[f] \,\d\tau\,\d s &\le
        \frac{c t^2}{2} \cdot \expectation^{(t f)}_\mu\left[(\Gamma f)^2 \right].
	\end{align}
Proceeding exactly as in the proof of Theorem~\ref{thm:ESS} above (i.e., by dividing both sides of the
above inequality by $t^2$ and passing to the limit as $t \rightarrow 0$), we obtain
\begin{align*}
	\frac{1}{2}\var_\mu[f] \le \frac{c}{2} \cdot \expectation_\mu\left[(\Gamma f)^2\right].
\end{align*}
Multiplying both sides by $2$, we see that $\mu$ indeed satisfies \eqref{eq:Poincare}.
\end{proof}

Moreover, Poincar\'e inequalities tensorize, as the following analogue of Theorem~\ref{thm:LSI_tensorization} shows:
\begin{theorem}\label{thm:Poincare_tensorization} Let $X_1,\ldots,X_n$ be independent $\cX$-valued random variables, and let $P = P_{X_1} \otimes \ldots \otimes P_{X_n}$ be their joint distribution. Let $\cA$ consist of all functions $f \colon \cX^n \to \reals$, such that, for every $i$,
	\begin{align}\label{eq:Poin_projections_in_A}
		f_i(\cdot|\bar{x}^i) \in \cA_i, \qquad \forall \, \bar{x}^i \in \cX^{n-1}
	\end{align}
Define the operator $\Gamma$ that maps each $f \in \cA$ to $\Gamma f$ in \eqref{eq:gamma_def} and \eqref{eq:gamma_def_long}.
Suppose that, for every $i \in \{1, \ldots, n\}$, $P_{X_i}$ satisfies a Poincar\'e inequality with constant $c \ge 0$ with respect to $(\cA_i,\Gamma_i)$ (see \eqref{eq:Poincare}). Then $P$ satisfies a Poincar\'e inequality with constant $c$ with respect to $(\cA,\Gamma)$.
\end{theorem}
\begin{proof} The proof is conceptually similar to the proof of Theorem~\ref{thm:LSI_tensorization} (which
refers to the tensorization of the logarithmic Sobolev inequality), except that now we use the Efron--Stein--Steele inequality of Theorem~\ref{thm:ESS} to tensorize the variance of $f$.
\end{proof}

\section{Transportation-cost inequalities}
\label{sec:transportation}

So far, we have been looking at concentration of measure through the lens of various {\em functional} inequalities, primarily log-Sobolev inequalities. In a nutshell, if we are interested in the concentration properties of a given function $f(X^n)$ of a random $n$-tuple $X^n \in \cX^n$, we seek to control the divergence $D(P^{(f)} \| P)$, where $P$ is the distribution of $X^n$ and $P^{(f)}$ is its $f$-tilting, $\d P^{(f)}/\d P \propto \exp(f)$, by some quantity related to the sensitivity of $f$ to modifications of its arguments (e.g., the squared norm of the gradient of $f$, as in the Gaussian log-Sobolev inequality of Gross \cite{Gross}). The common theme underlying these functional inequalities is that every such measure of sensitivity is tied to a particular {\em metric structure} on the underlying product space $\cX^n$. To see this, suppose that $\cX^n$ is equipped with a metric $d(\cdot,\cdot)$, and consider the following generalized definition of the modulus of the gradient of an arbitrary function $f \colon \cX^n \to \reals$:
\begin{align}\label{eq:generalized_gradient}
	|\nabla f|(x^n) \deq \limsup_{y^n: d(x^n,y^n) \downarrow 0} \frac{|f(x^n)-f(y^n)|}{d(x^n,y^n)}.
\end{align}
If we also define the Lipschitz constant of $f$ by
\begin{align} \label{eq:Lipschitz constant}
	\| f \|_{\rm Lip} \deq \sup_{x^n \neq y^n} \frac{|f(x^n)-f(y^n)|}{d(x^n,y^n)}
\end{align}
and consider the class $\cA$ of all functions $f$ with $\| f \|_{\rm Lip} < \infty$, then it is easy to see
that the pair $(\cA,\Gamma)$ with $\Gamma f(x^n) \deq |\nabla f|(x^n)$ satisfies the conditions (LSI-1)--(LSI-3)
listed in Section~\ref{sec:LSI}. Consequently, suppose that a given probability distribution $P$ for a random $n$-tuple
$X^n \in \cX^n$ satisfies $\LSI(c)$ with respect to the pair $(\cA,\Gamma)$. The use of \eqref{eq:LSI_to_concentration}
and the inequality $\| \Gamma f \|_{\infty} \leq \| f \|_{\rm Lip}$, which follows
directly from \eqref{eq:generalized_gradient} and \eqref{eq:Lipschitz constant}, gives the concentration inequality
\begin{align}\label{eq:Lipschitz_conc}
	\pr \Big( f(X^n) \ge \expectation f(X^n) + r \Big) \le \exp\left( - \frac{r^2}{2c \| f \|^2_{\rm Lip}}\right),
    \quad \forall \, r > 0.
\end{align}
Some examples of concentration we have discussed so far in this chapter can be seen to fit this theme. Consider,
for instance, the following case:
\begin{example}[Euclidean metric]
For $\cX = \reals$, equip the product space $\cX^n = \reals^n$ with the ordinary Euclidean metric:
\begin{align*}
	d(x^n,y^n) = \| x^n - y^n \| = \sqrt{\sum^n_{i=1}(x_i - y_i)^2}.
\end{align*}
Then, from \eqref{eq:Lipschitz constant}, the Lipschitz constant $\| f \|_{\rm Lip}$ of an arbitrary function
$f \colon \cX^n \to \reals$ is given by
\begin{align*}
	\| f \|_{\rm Lip} = \sup_{x^n \neq y^n} \frac{|f(x^n) - f(y^n)|}{\| x^n - y^n \|},
\end{align*}
and, for every probability measure $P$ on $\reals^n$ that satisfies $\LSI(c)$, we have the concentration inequality
\eqref{eq:Lipschitz_conc}. We have already seen in \eqref{eq:Gross_LSI_2} a particular instance
of this with $P = G^n$, which satisfies $\LSI(1)$.
\end{example}
The above example suggests that the metric structure plays the primary role, while the functional concentration inequalities like \eqref{eq:Lipschitz_conc} are simply a consequence. In this section, we describe an alternative approach to concentration that works directly on the level of {\em probability measures}, rather than functions, and that makes this intuition precise. The key tool underlying this approach is the notion of {\em transportation cost}, which can be used to define a metric on probability distributions over the space of interest in terms of a given base metric on this space. This metric on distributions can then be related to the divergence via the so-called {\em transportation-cost inequalities}. The pioneering work by K.~Marton in \cite{Marton_dbar} and \cite{Marton_blowup} has shown that one can use these inequalities to deduce concentration.

\subsection{Concentration and isoperimetry}

We start by giving rigorous meaning to the notion that the concentration of measure phenomenon is fundamentally geometric in nature. In order to talk about concentration, we need the notion of a {\em metric probability space} in the sense of M.~Gromov \cite{Gromov_book}. Specifically, we say that a triple $(\cX,d,\mu)$ is a metric probability space if $(\cX,d)$ is a Polish space (i.e., a complete and separable metric space) and $\mu$ is a probability measure on the Borel sets of $(\cX,d)$.

For an arbitrary set $A \subseteq \cX$ and every $r > 0$, define the {\em $r$-blowup of $A$} by
\begin{align}\label{eq:blowup_set}
	A_r \deq \left\{ x \in \cX \colon d(x,A) < r \right\},
\end{align}
where $d(x,A) \deq \inf_{y \in A}d(x,y)$ is the distance from the point $x$ to the set $A$. We then say that the probability measure $\mu$ has {\em normal} (or {\em Gaussian}) {\em concentration} on $(\cX,d)$ if there exist positive constants $K,\kappa$, such that
\begin{align}\label{eq:isoperimetry}
	\mu(A) \ge 1/2 \qquad \Longrightarrow \qquad \mu(A_r) \ge 1 - K e^{-\kappa r^2}, \; \forall \, r > 0.
\end{align}
\begin{remark} Of the two constants $K$ and $\kappa$ in \eqref{eq:isoperimetry}, it is $\kappa$ that is more important. For that reason, sometimes we will say that $\mu$ has normal concentration with constant $\kappa > 0$ to mean that \eqref{eq:isoperimetry} holds with that value of $\kappa$ and some $K > 0$.
\end{remark}
\begin{remark}\label{rem:isoperimetry_weakened} The concentration condition \eqref{eq:isoperimetry} is often weakened to the following: there exists some $r_0 > 0$, such that
	\begin{align}\label{eq:isoperimetry_weakened}
		\mu(A) \ge 1/2 \qquad \Longrightarrow \qquad \mu(A_r) \ge 1 - Ke^{-\kappa(r-r_0)^2},\, \forall \, r \ge r_0
	\end{align}
(see, for example,  \cite[Remark~22.23]{Villani_newbook} or \cite[Proposition~3.3]{Gozlan}). It is not hard to pass from \eqref{eq:isoperimetry_weakened} to the stronger statement \eqref{eq:isoperimetry}, possibly with degraded constants (i.e.,
larger $K$ and/or smaller $\kappa$). However, since we mainly care about sufficiently large values of $r$,
\eqref{eq:isoperimetry_weakened} with {\em sharper constants} is preferable. In the sequel, therefore, whenever we talk
about Gaussian concentration with constant $\kappa > 0$, we will normally refer to \eqref{eq:isoperimetry_weakened}, unless
stated otherwise.
\end{remark}

Here are a few standard examples (see \cite[Section~1.1]{Ledoux}):
\begin{enumerate}
	\item {\bf Standard Gaussian distribution} --- if $\cX = \reals^n$, $d(x,y) = \| x - y \|$ is the standard Euclidean metric, and $\mu = G^n$ is the standard Gaussian distribution, then for every Borel set $A \subseteq \reals^n$ with $G^n(A) \ge 1/2$ we have
    \begin{align}\label{eq:Gaussian_iso}
	G^n(A_r) & \ge \frac{1}{\sqrt{2\pi}} \, \int_{-\infty}^r \exp\left(-\frac{t^2}{2}\right) \, \d t \nonumber \\
             & \ge 1 - \frac{1}{2} \, \exp\left(-\frac{r^2}{2}\right), \qquad \forall \, r > 0
	\end{align}	
	i.e., \eqref{eq:isoperimetry} holds with $K=\frac{1}{2}$ and $\kappa = \frac{1}{2}$.
	\item {\bf Uniform distribution on the unit sphere} --- if $\cX = \sphere^n \equiv \left\{ x \in \reals^{n+1} : \| x \| = 1 \right\}$, $d$ is given by the geodesic distance on $\sphere^n$, and $\mu = \sigma^n$ (the uniform distribution on $\sphere^n$), then for every Borel set $A \subseteq \sphere^n$ with $\sigma^n(A) \ge 1/2$ we have
	\begin{align}
		\sigma^n(A_r) \ge 1 - \exp\left(-\frac{(n-1)r^2}{2}\right), \qquad \forall \, r > 0.
	\end{align}
	In this instance, \eqref{eq:isoperimetry} holds with $K = 1$ and $\kappa = (n-1)/2$. Notice that $\kappa$ is increasing with the ambient dimension $n$.
	\item {\bf Uniform distribution on the Hamming cube} --- if $\cX = \{0,1\}^n$, $d$ is the normalized Hamming metric
	\begin{align*}
		d(x,y) = \frac{1}{n}\sum^n_{i=1}1_{\{x_i \neq y_i \}}
	\end{align*}
	for all $x = (x_1, \ldots, x_n), y = (y_1, \ldots, y_n) \in \{0,1\}^n$, and $\mu = B^n$ is the uniform distribution on $\{0,1\}^n$ (which is equal to the product of $n$ copies of a $\Bernoulli(1/2)$ measure on $\{0,1\}$, i.e., $B^n(A)=\frac{|A|}{2^n}$ where $|A|$ denotes the cardinality of an arbitrary set $A \subseteq \{0, 1\}^n$). Then, for every $A \subseteq \{0,1\}^n$ with $B^n(A) \ge 1/2$, we have
	\begin{align}
		B^n(A_r) \ge 1 - \exp\left(-2nr^2\right), \qquad \forall \, r > 0
	\end{align}
	so \eqref{eq:isoperimetry} holds with $K=1$ and $\kappa = 2n$.
\end{enumerate}

\begin{remark} Gaussian concentration of the form \eqref{eq:isoperimetry} is often discussed in the context of the so-called {\em isoperimetric inequalities}, which relate the full measure of a set to the measure of its boundary. To be more specific, consider a metric probability space $(\cX,d,\mu)$, and for an arbitrary Borel set $A \subseteq \cX$ define its {\em surface measure} as (see \cite[Section~2.1]{Ledoux})
\begin{align}
	\mu^+(A) \deq \liminf_{r \rightarrow 0} \frac{\mu(A_r \setminus A)}{r} =
    \liminf_{r \rightarrow 0} \frac{\mu(A_r) - \mu(A)}{r}.
\end{align}
Then, the classical Gaussian isoperimetric inequality can be stated as follows: If $H$ is a half-space in $\reals^n$, i.e., $H = \{ x \in \reals^n: \langle x,u \rangle < c \}$ for some $u \in \reals^n$ with $\| u \| =1$ and some $c \in [-\infty,+\infty]$, and if $A \subseteq \reals^n$ is a Borel set with $G^n(A) = G^n(H)$, then
\begin{align}\label{eq:Gaussian_isoperimetric_inequality}
	(G^n)^+(A) \ge (G^n)^+(H),
\end{align}
with equality if and only if $A$ is a half-space. In other words, the Gaussian isoperimetric inequality \eqref{eq:Gaussian_isoperimetric_inequality} says that, among all Borel subsets of $\reals^n$ with a given Gaussian volume, the half-spaces have the smallest surface measure. An equivalent integrated version of \eqref{eq:Gaussian_isoperimetric_inequality} says the following (see, e.g., \cite{Bobkov_iso}): Consider a Borel set $A$ in $\reals^n$ and a half-space $H = \{ x : \langle x,u \rangle < c \}$ with $\| u \| =1$, $c \ge 0$ and $G^n(A) = G^n(H)$. Then, for every $r > 0$, we have
\begin{align*}
G^n(A_r) \ge G^n(H_r),
\end{align*}
with equality if and only if $A$ is itself a half-space. Moreover, an easy calculation shows that
\begin{align*}
G^n(H_r) &= \frac{1}{\sqrt{2\pi}} \int^{c+r}_{-\infty} \exp\left(-\frac{\xi^2}{2}\right) \d\xi \\
&\ge 1 - \frac{1}{2} \, \exp\left(-\frac{(r+c)^2}{2} \right), \quad \forall \, r > 0.
\end{align*}
So, if $G(A) \ge 1/2$, we can always choose $c=0$ and get \eqref{eq:Gaussian_iso}.
\end{remark}

Intuitively, what \eqref{eq:isoperimetry} says is that, if $\mu$ has normal concentration on $(\cX,d)$, then most of the probability mass in $\cX$ is concentrated around any set with probability at least $1/2$. At first glance, this seems to have nothing to do with what we have been looking at all this time, namely the concentration of Lipschitz functions on $\cX$ around their mean. However, as we will now show, the geometric and the functional pictures of the concentration of measure phenomenon are, in fact, equivalent. To that end, let us define the {\em median} of a function $f \colon \cX \to \reals$: we say that a real number $m_f$ is a median of $f$ with respect to $\mu$ (or a {\em $\mu$-median} of $f$) if
\begin{align}\label{eq:median}
	\pr_\mu \big( f(X) \ge m_f \big) \ge \frac{1}{2} \qquad \text{and} \qquad \pr_\mu \big( f(X) \le m_f \big) \ge \frac{1}{2}
\end{align}
(note that a median of $f$ may not be unique). The precise result is as follows:

\begin{theorem}\label{thm:Lipschitz_isoperimetry} Let $(\cX,d,\mu)$ be a metric probability space. Then $\mu$ has the normal concentration property \eqref{eq:isoperimetry} (with arbitrary constants $K, \kappa > 0$) if and only if for every Lipschitz function $f \colon \cX \to \reals$ (where the Lipschitz property is defined with respect to the metric $d$) we have
	\begin{align}\label{eq:Lipschitz_isoperimetry}
		\pr_\mu \Big( f(X) \ge m_f + r \Big) \le K \exp\biggl(-\frac{\kappa r^2}{ \| f \|^2_{\rm Lip}} \biggr), \qquad \forall \, r > 0
	\end{align}
	where $m_f$ is a $\mu$-median of $f$.
\end{theorem}
\begin{proof} Suppose that $\mu$ satisfies \eqref{eq:isoperimetry}. Fix an arbitrary Lipschitz function $f$, where, without loss of generality, we may assume that $\| f \|_{\rm Lip} = 1$. Let $m_f$ be a $\mu$-median of $f$, and define the set
$$A^f \deq \Big\{ x \in \cX \colon f(x) \le m_f \Big\}.$$
By definition of the median in \eqref{eq:median}, $\mu(A^f) \ge 1/2$. Consequently, by \eqref{eq:isoperimetry}, we have
	\begin{align}\label{eq:Lipschitz_isoperimetry_1}
		\mu(A^f_r) &\equiv \pr_\mu \left( d(X,A^f) < r\right) \nonumber\\
		&\ge 1 - K \exp(-\kappa r^2), \qquad \forall \, r > 0.
	\end{align}
By the Lipschitz property of $f$, for every $y \in A^f$ we have $f(X) - m_f \le f(X) - f(y) \le d(X,y)$, so $f(X) - m_f \le d(X,A^f)$. This, together with \eqref{eq:Lipschitz_isoperimetry_1}, implies that
\begin{align*}
\pr_\mu \Big( f(X) - m_f < r \Big) &\ge \pr_\mu \Big( d(X,A^f) < r \Big) \\
&\ge 1 - K \exp(-\kappa r^2),
\qquad \forall \, r > 0
\end{align*}
which is \eqref{eq:Lipschitz_isoperimetry}.

Conversely, suppose \eqref{eq:Lipschitz_isoperimetry} holds for every Lipschitz $f$. Choose an arbitrary Borel set $A$ with $\mu(A) \ge 1/2$, and define the function $f_A(x) \deq d(x,A)$ for every $x \in \cX$. Then $f_A$ is $1$-Lipschitz, since
\begin{align*}
	|f_A(x)-f_A(y)| &= \Bigl|\inf_{u \in A}d(x,u) - \inf_{u \in A}d(y,u)\Bigr| \\
	&\le \sup_{u \in A}\left|d(x,u)-d(y,u)\right| \\
	&\le d(x,y),
\end{align*}
where the last step is by the triangle inequality. Moreover, zero is a median of $f_A$, since
\begin{align*}
	\pr_\mu \big( f_A(X) \le 0 \big) = \pr_\mu\big(X \in A\big) \ge \frac{1}{2} \qquad \text{and}\qquad
	\pr_\mu \big( f_A(X) \ge 0 \big) \ge \frac{1}{2},
\end{align*}
where the second bound is vacuously true since $f_A \ge 0$ everywhere. Consequently, with $m_f=0$, we get
\begin{align*}
	1-\mu(A_r) &= \pr_\mu \big( d(X,A) \ge r \big) \\
	&= \pr_\mu \big( f_A(X) \ge m_f + r \big) \\
	&\le K \exp( - \kappa r^2), \qquad \forall \, r > 0
\end{align*}
which gives \eqref{eq:isoperimetry}.
\end{proof}
In fact, for Lipschitz functions, normal concentration around the mean also implies
normal concentration around every median, but possibly with worse constants \cite[Proposition~1.7]{Ledoux}:
\begin{theorem} \label{thm: Lipschitz_mean_conc}
Let $(\cX,d,\mu)$ be a metric probability space, such that for every $1$-Lipschitz function
$f \colon \cX \to \reals$ we have
\begin{align}\label{eq:Lipschitz_mean_conc}
\pr_\mu \Big( f(X) \ge \expectation_\mu [f(X)] + r\Big)
\le K_0 \exp\big(-\kappa_0 r^2\big), \qquad \forall \, r > 0
\end{align}
with some constants $K_0,\kappa_0 > 0$. Then, $\mu$ has the normal concentration property \eqref{eq:isoperimetry} with $K = K_0$ and $\kappa = \frac{\kappa_0}{4}$. Consequently, the concentration inequality in \eqref{eq:Lipschitz_isoperimetry}
around every median $m_f$ is satisfied with the same constants of $\kappa$ and $K$.
\end{theorem}
\begin{proof} Let $A \subseteq \cX$ be an arbitrary Borel set with $\mu(A) \ge \frac{1}{2}$, and fix some $r > 0$. Define the function $f_{A,r}(x) \deq \min \left\{ d(x,A), r\right\}$. From the triangle inequality,
$\|f_{A,r}\|_{\rm Lip} \le 1$ and
	\begin{align}
		\expectation_\mu [f_{A,r}(X)] &= \int_\cX \min\left\{ d(x,A),r\right\} \mu(\d x) \nonumber \\
		&= \underbrace{\int_A \min\left\{ d(x,A),r\right\} \mu(\d x)}_{=0}
           + \int_{A^c} \min\left\{ d(x,A),r\right\} \mu(\d x) \nonumber \\
		&\le r \, \mu(A^c) = \bigl(1-\mu(A)\bigr) \, r. \label{eq:mean_vs_mass}
	\end{align}
Then
\begin{align*}
	1-\mu(A_r) &= \pr_\mu \Big( d(X,A) \ge r \Big) \\
	&= \pr_\mu \Big( f_{A,r}(X) \ge r \Big) \\
	&\le \pr_\mu \Big( f_{A,r}(X) \ge \expectation_\mu [f_{A,r}(X)] + r\mu(A)\Big) \\
	&\le K_0 \exp\left( - \kappa_0 \left(\mu(A)r\right)^2\right) \\
    &\le K_0 \exp\left( - \frac{\kappa_0 \, r^2}{4}\right)
\end{align*}
where the first two steps use the definition of $f_{A,r}$, the third step uses \eqref{eq:mean_vs_mass}, the fourth step uses \eqref{eq:Lipschitz_mean_conc}, and the last step holds since by assumption $\mu(A) \ge \frac{1}{2}$. Consequently, we get \eqref{eq:isoperimetry} with $K=K_0$ and $\kappa=\frac{\kappa_0}{4}$. Theorem~\ref{thm:Lipschitz_isoperimetry} therefore implies that the concentration
inequality in \eqref{eq:Lipschitz_isoperimetry} holds for every median $m_f$ with the same constants of $\kappa$ and $K$.
\end{proof}
\begin{remark}
Let $(\cX,d,\mu)$ be a metric probability space, and suppose that $\mu$ has the normal concentration property
\eqref{eq:isoperimetry} (with arbitrary constants $K, \kappa > 0$). Let $f \colon \cX \to \reals$ be an arbitrary
Lipschitz function (with respect to the metric $d$). Then we can upper-bound the distance between the mean and an arbitrary
$\mu$-median of $f$ in terms of the parameters $K,\kappa$ and the Lipschitz constant of $f$. From
Theorem~\ref{thm:Lipschitz_isoperimetry}, we have
\begin{align*}
\big| \expectation_\mu [f(X)] - m_f \big| & \le \expectation_\mu \big[|f(X) - m_f|\big]  \\
& = \int_0^{\infty} \pr_\mu ( |f(X) - m_f| \ge r) \, \d r  \\
& \le \int_0^{\infty} 2K \exp\biggl(-\frac{\kappa r^2}{ \| f \|^2_{\rm Lip}} \biggr) \, \d r  \\
& = \sqrt{\frac{\pi}{\kappa}} \, K \| f \|_{\rm Lip}
\end{align*}
where the first equality holds due to the fact that if $U$ is a non-negative random variable then
$\expectation[U] = \int_0^{\infty} \pr(U \ge r) \, \d r$ (this equality follows as a consequence of
Fubini's theorem), and the second inequality follows from the (one-sided) concentration inequality in
\eqref{eq:Lipschitz_isoperimetry} applied to $f$ and $-f$ (both functions have the same Lipschitz
constant).
\end{remark}

\subsection{Marton's argument: from transportation to concentration}

As we have just seen, the phenomenon of concentration is fundamentally geometric in nature, as captured by the isoperimetric inequality \eqref{eq:isoperimetry}. Once we have established \eqref{eq:isoperimetry} on a given metric probability space $(\cX,d,\mu)$, we immediately obtain Gaussian concentration for all Lipschitz functions $f \colon \cX \to \reals$ by Theorem~\ref{thm:Lipschitz_isoperimetry}.

There is a powerful information-theoretic technique for deriving concentration inequalities like \eqref{eq:isoperimetry}. This technique, first introduced by Marton (see \cite{Marton_dbar} and \cite{Marton_blowup}), hinges on a certain type of inequality that relates the divergence between two probability measures to a quantity called the {\em transportation cost}. Let $(\cX,d)$ be a Polish space. Given $p \ge 1$, let $\cP_p(\cX)$ denote the space of all Borel probability measures $\mu$ on $\cX$, such that the moment bound
\begin{align}\label{eq:pth_moment}
	\expectation_\mu[d^p(X,x_0)] < \infty
\end{align}
holds for some (and hence all) $x_0 \in \cX$.
\begin{definition} Given $p \ge 1$, the {\em $L^p$ Wasserstein distance} (a.k.a. the Wasserstein distance of order $p$) between a pair
$\mu,\nu \in \cP_p(\cX)$ is defined as
\begin{align}\label{eq:W_p}
	W_p(\mu,\nu) \deq \inf_{\pi \in \Pi(\mu,\nu)} \left( \int_{\cX \times \cX} d^p(x,y) \, \pi(\d x, \d y)\right)^{1/p},
\end{align}
where $\Pi(\mu,\nu)$ is the set of all probability measures $\pi$ on the product space $\cX \times \cX$ with marginals $\mu$ and $\nu$.
\end{definition}
\begin{remark}
Another equivalent way of writing down the definition of $W_p(\mu,\nu)$ is
\begin{align}\label{eq:W_p_2}
	W_p(\mu,\nu) = \inf_{X \sim \mu, \, Y \sim \nu} \bigl\{\expectation[d^p(X,Y)]\bigr\}^{1/p},
\end{align}
where the infimum is over all pairs $(X,Y)$ of jointly distributed random variables with values in $\cX$, such that $P_X = \mu$ and $P_Y = \nu$.
\end{remark}
The name ``transportation cost'' comes from the following interpretation: Let $\mu$ (resp., $\nu$) represent the initial (resp., desired) distribution of some matter (say, sand) in space, such that the total mass in both cases is normalized to one. Thus, both $\mu$ and $\nu$ correspond to sand piles of some given shapes. The objective is to rearrange the initial sand pile with shape $\mu$ into one with shape $\nu$ with minimum cost, where the cost of transporting a grain of sand from location $x$ to location $y$ is given by $c(x,y)$ for a measurable function $c \colon \cX \times \cX \to \reals$. If we allow randomized transportation policies, i.e., those that associate with each location $x$ in the initial sand pile a conditional probability distribution $\pi(\d y|x)$ for its destination in the final sand pile, then the minimum transportation cost is given by
\begin{align}\label{eq:optimal_transport}
	C^*(\mu,\nu) \deq \inf_{\pi \in \Pi(\mu,\nu)}\int_{\cX \times \cX} c(x,y) \, \pi(\d x, \d y).
\end{align}
When the cost function is given by $c = d^p$ for some $p \ge 1$ and $d$ is a metric on $\cX$, we will have $C^*(\mu,\nu) = W^p_p(\mu,\nu)$. The optimal transportation problem \eqref{eq:optimal_transport} has a rich history, dating back to a 1781 essay by Gaspard Monge, who has considered a particular special case of the problem
	\begin{align}\label{eq:Monge_problem}
		C^*_0(\mu,\nu) \deq \inf_{\varphi \colon \cX \to \cX} \left\{ \int_{\cX} c(x,\varphi(x)) \, \mu(\d x) \colon \mu \circ \varphi^{-1} = \nu \right\}.
	\end{align}
	Here, the infimum is over all {\em deterministic} transportation policies, i.e., measurable mappings $\varphi \colon \cX \to \cX$, such that the desired final measure $\nu$ is the image of $\mu$ under $\varphi$, or, in other words, if $X \sim \mu$, then $Y = \varphi(X) \sim \nu$. The problem \eqref{eq:Monge_problem} (or the {\em Monge optimal transportation problem}, as it has now come to be called) does not always admit a solution (incidentally, an optimal mapping does exist in the case considered by Monge, namely $\cX = \reals^3$ and $c(x,y) = \| x - y \|$). A stochastic relaxation of Monge's problem, given by \eqref{eq:optimal_transport}, was considered in 1942 by Leonid Kantorovich
(see \cite{Kantorovich} for a recent reprint). We recommend the books by Villani \cite{Villani_TOT,Villani_newbook} for a detailed historical overview and rigorous treatment of optimal transportation.

The following lemma introduces properties of the Wasserstein distances. For a proof, the reader is
referred to \cite[Chapter~6]{Villani_newbook}.
\begin{lemma}\label{lm:W_properties} The Wasserstein distances have the following properties:
	\begin{enumerate}
		\item For each $p \geq 1$, $W_p(\cdot,\cdot)$ is a metric on $\cP_p(\cX)$.
		\item If $1 \le p \le q$, then $\cP_p(\cX) \supseteq \cP_q(\cX)$, and $W_p(\mu,\nu) \le W_q(\mu,\nu)$ for every $\mu,\nu \in \cP_q(\cX)$.
		\item $W_p$ metrizes weak convergence plus convergence of $p$th-order moments: a sequence $\{\mu_n\}^\infty_{n=1}$ in $\cP_p(\cX)$ converges to $\mu \in \cP_p(\cX)$ in $W_p$, i.e., $W_p(\mu_n,\mu) \xrightarrow{n \to \infty} 0$, if and only if:
		\begin{enumerate}
			\item $\{\mu_n\}$ converges to $\mu$ weakly, i.e., $\expectation_{\mu_n}[\varphi] \xrightarrow{n \to \infty} \expectation_\mu[\varphi]$ for every continuous and bounded function $\varphi \colon \cX \to \reals$.
			\item For some (and hence all) $x_0 \in \cX$,
			$$
			\int_\cX d^p(x,x_0) \, \mu_n(\d x) \xrightarrow{n \to \infty} \int_\cX d^p(x,x_0) \, \mu(\d x).
			$$
		\end{enumerate}
		If the above two statements hold, then we say that $\{\mu_n\}$ converges to $\mu$ {\em weakly in $\cP_p(\cX)$}.
		\item The mapping $(\mu,\nu) \mapsto W_p(\mu,\nu)$ is continuous on $\cP_p(\cX)$, i.e., if $\mu_n \to \mu$ and $\nu_n \to \nu$ converge weakly in $\cP_p(\cX)$, then $W_p(\mu_n,\nu_n) \to W_p(\mu,\nu)$. However, it is only {\em lower semicontinuous} in the usual weak topology (without the convergence of $p$th-order moments): if $\mu_n \to \mu$ and $\nu_n \to \nu$ converge weakly, then
		$$
		\liminf_{n \to \infty} W_p(\mu_n,\nu_n) \ge W_p(\mu,\nu).
		$$
		\item The infimum in \eqref{eq:W_p} [and therefore in \eqref{eq:W_p_2}] is actually a minimum; i.e., there exists an {\em optimal coupling} $\pi^* \in \Pi(\mu,\nu)$, such that
		$$
		W^p_p(\mu,\nu) = \int_{\cX \times \cX} d^p(x,y) \, \pi^*(\d x, \d y).
		$$
		Equivalently, there exists a pair $(X^*,Y^*)$ of jointly distributed $\cX$-valued random variables with $P_{X^*} = \mu$ and $P_{Y^*} = \nu$, such that
		$$
		W^p_p(\mu,\nu) = \expectation[d^p(X^*,Y^*)].
		$$
		\item If $p=2$, $\cX = \reals$ with $d(x,y) = |x-y|$, and $\mu$ is atomless (i.e., $\mu(\{x\})= 0$ for all $x \in \reals$), then the optimal coupling between $\mu$ and every $\nu$ is given by the deterministic mapping
		\begin{align*}
			Y = \sF^{-1}_\nu \circ \sF_\mu (X)
		\end{align*}
		for $X \sim \mu$, where $\sF_\mu$ denotes the cumulative distribution function (cdf) of $\mu$, i.e., $\sF_\mu(x) = \pr_\mu ( X \le x)$, and $\sF^{-1}_\nu$ is the {\em quantile function} of $\nu$, i.e., $\sF^{-1}_\nu(\alpha) \deq \inf \left\{ x \in \reals \colon \sF_\nu(x) \ge \alpha \right\}$.
	\end{enumerate}
\end{lemma}

\begin{definition}\label{def:TC} We say that a probability measure $\mu$ on $(\cX,d)$ satisfies an $L^p$ {\em transportation-cost inequality with constant} $c > 0$, or a ${\rm T}_p(c)$ inequality for short, if for every probability measure $\nu \ll \mu$ we have
	\begin{align}\label{eq:TCI}
		W_p(\mu,\nu) \le \sqrt{2c\, D(\nu \| \mu)}.
	\end{align}
\end{definition}

\begin{example}[Total variation distance and Pinsker's inequality] Here is a specific example illustrating this abstract machinery, which should be a familiar territory to information theorists. Let $\cX$ be a discrete set, equipped with the Hamming metric $d(x,y) = 1_{\{x\neq y\}}$. In this case, the corresponding $L^1$ Wasserstein distance between every two probability measures $\mu$ and $\nu$ on $\cX$ takes the simple form
	\begin{align*}
		W_1(\mu,\nu) &= \inf_{X \sim \mu,Y \sim \nu} \pr \left( X \neq Y \right).
	\end{align*}
As we will now show, this turns out to be the total variation distance
\begin{equation}
	\| \mu - \nu \|_{\rm TV} \deq \sup_{A \subseteq \cX} |\mu(A)-\nu(A)|. \label{eq:TV-def}
\end{equation}

\begin{proposition}
\begin{align}
W_1(\mu, \nu) &= \| \mu - \nu \|_{\rm TV} \label{eq: the Wasserstein distance of 1st order equiped with the Hamming distance} \\
& = \frac{1}{2}\sum_{x \in \cX} |\mu(x) - \nu(x)|   \label{eq: total variation distance with factor 1/2}
\end{align}
(we are slightly abusing notation here, writing $\mu(x)$ for the $\mu$-probability of the singleton $\{x\}$).
\label{proposition: the Wasserstein distance of 1st order equiped with the Hamming distance}
\end{proposition}
\begin{proof}
Consider a probability measure $\pi \in \Pi(\mu,\nu)$. For every $x \in \cX$, we have
\begin{align*}
	\mu(x) = \sum_{y \in \cX}\pi(x,y) \ge \pi(x,x),
\end{align*}
and the same goes for $\nu$. Thus, $\pi(x,x) \le \min \left\{ \mu(x),\nu(x)\right\}$, and so
\begin{align}
	\expectation_\pi [d(X, Y)] &= \expectation_\pi [1_{\{X \neq Y\}}] \label{eq:Hamming_cost1} \\
    &= \pr(X \neq Y) \label{eq:Hamming_cost2} \\
    &= 1 - \sum_{x \in \cX} \pi(x,x) \label{eq:Hamming_cost3} \\
	&\ge 1 - \sum_{x \in \cX} \min \left\{ \mu(x),\nu(x)\right\}. \label{eq:Hamming_cost4}
\end{align}
From \eqref{eq:W_p_2}, \eqref{eq:Hamming_cost1} and \eqref{eq:Hamming_cost4}, we have
\begin{align}
W_1(\mu, \nu) \geq 1 - \sum_{x \in \cX} \min \left\{ \mu(x),\nu(x)\right\}.
\label{eq:lower bound on W_1}
\end{align}

In the following, equality \eqref{eq: total variation distance with factor 1/2} is proved first.
For an arbitrary $A \subseteq \cX$, we have
\begin{align}
\mu(A) - \nu(A)
& = \big(1-\mu(A^c)\big) - \big(1-\nu(A^c)\big) \nonumber \\
& = \nu(A^c) - \mu(A^c)  \label{eq:step1 - TV}
\end{align}
and, from the triangle inequality,
\begin{align}
& \big|\mu(A) - \nu(A)\big| + \big|\mu(A^c) - \nu(A^c)\big| \nonumber \\
& \leq \sum_{x \in A} \bigl| \mu(x) - \nu(x) \bigr| + \sum_{x \in A^c} \bigl| \mu(x) - \nu(x) \bigr| \nonumber \\
& = \sum_{x \in \cX} \bigl| \mu(x) - \nu(x) \bigr|. \label{eq:step2 - TV}
\end{align}
Combining \eqref{eq:step1 - TV} and \eqref{eq:step2 - TV} gives that, for every $A \subseteq \cX$,
\begin{align}
& \big|\mu(A) - \nu(A)\big| \leq \frac{1}{2} \, \sum_{x \in \cX} \bigl| \mu(x) - \nu(x) \bigr|.  \label{eq:step3 - TV}
\end{align}
Since \eqref{eq:step3 - TV} holds for every $A \subseteq \cX$, we can take the supremum over all such subsets $A$
and get (see \eqref{eq:TV-def}) that
\begin{align}
\| \mu - \nu \|_{\rm TV} \leq \frac{1}{2} \, \sum_{x \in \cX} \bigl| \mu(x) - \nu(x) \bigr|.  \label{eq:step4 - TV}
\end{align}
On the other hand, if we define
\begin{align}
A \deq \{ x \in \cX \colon \mu(x) \ge \nu(x)\}
\label{eq:definition of subset A}
\end{align}
we have from \eqref{eq:step1 - TV} and \eqref{eq:definition of subset A}
\begin{align}
    & \mu(A) - \nu(A) \nonumber \\
    & = \frac{1}{2} \, \Big[ \big(\mu(A) - \nu(A)\big) + \big(\nu(A^c) - \mu(A^c)\big) \Big] \nonumber \\
    & = \frac{1}{2} \, \left[ \sum_{x \in A} \, \big( \mu(x) - \nu(x) \big) + \sum_{x \in A^c} \big( \nu(x) - \mu(x) \big) \right] \nonumber \\
    & = \frac{1}{2} \, \sum_{x \in \cX} \, \big| \mu(x) - \nu(x) \big|  \label{eq:step5 - TV}
\end{align}
so, from \eqref{eq:TV-def} and \eqref{eq:step5 - TV},
\begin{align}
\| \mu - \nu \|_{\rm TV} \geq \frac{1}{2} \, \sum_{x \in \cX} \bigl| \mu(x) - \nu(x) \bigr|.  \label{eq:step6 - TV}
\end{align}
Equality \eqref{eq: total variation distance with factor 1/2} follows by combining
\eqref{eq:step4 - TV} and \eqref{eq:step6 - TV}, and the equality
\begin{align}
\| \mu - \nu \|_{\rm TV} = \mu(A) - \nu(A)  \label{eq:step7 - TV}
\end{align}
holds for the subset $A \subseteq \cX$ defined by \eqref{eq:definition of subset A}.
From \eqref{eq:definition of subset A} and \eqref{eq:step7 - TV}
\begin{align}
	\sum_{x \in \cX} \min\left\{ \mu(x),\nu(x) \right\}
	&= \sum_{x \in A}\nu(x) + \sum_{x \in A^c} \mu(x) \nonumber \\
	&= \nu(A) + \mu(A^c) \nonumber \\
	&= 1 - \Big(\mu(A) - \nu(A)\Big) \nonumber \\
	&= 1 - \| \mu - \nu \|_{\rm TV}.  \label{eq:Hamming_cost5}
\end{align}
Consequently, it follows from \eqref{eq:lower bound on W_1} and \eqref{eq:Hamming_cost5} that
\begin{align}\label{eq:Hamming_cost_3}
	W_1(\mu, \nu) \ge \| \mu - \nu \|_{\rm TV}.
\end{align}
Furthermore, \eqref{eq:Hamming_cost_3} holds with equality for the probability measure
$\pi^* \colon \cX \times \cX \to \reals$ which is defined as follows:
\begin{align}\label{eq:TV_optimal_coupling}
\pi^*(x,y) &= \min \left\{ \mu(x),\nu(x) \right\} 1_{\{x = y\}} \nonumber\\
& \quad+ \frac{\big( \mu(x)-\nu(x)\big)1_{\{ x \in A\}} \big(\nu(y) - \mu(y)\big)1_{\{ y \in A^c\}}}{\mu(A) - \nu(A)}
\end{align}
with the set $A$ in \eqref{eq:definition of subset A}. This can be verified by noticing that
$$\pi^*(x,x) = \min \bigl\{ \mu(x), \, \nu(x) \bigr\}, \quad \forall \, x \in \cX$$
which is the necessary and sufficient condition to satisfy an equality in \eqref{eq:Hamming_cost4};
furthermore, $\pi^*$ is indeed a probability measure (this follows from \eqref{eq:step7 - TV}
and \eqref{eq:Hamming_cost5}) with marginals $\mu$ and $\nu$. To verify this, note that for every $x \in A$
\begin{align*}
\sum_{y \in A} \pi^*(x,y) & = \min \left\{\mu(x), \nu(x)\right\} + \frac{ \bigl( \mu(x)-\nu(x) \bigr)
\, \sum_{y \in A^c} \bigl( \nu(y) - \mu(y) \bigr)}{\mu(A) - \nu(A)} \nonumber \\
& = \nu(x) + \frac{ \bigl( \mu(x)-\nu(x) \bigr) \, \bigl( \nu(A^c) - \mu(A^c) \bigr)}{\mu(A) - \nu(A)} \nonumber \\
& = \nu(x) + \bigl(\mu(x)-\nu(x) \bigr) = \mu(x)
\end{align*}
where the third equality follows from \eqref{eq:step1 - TV}, and for every $x \in A^c$
\begin{align*}
\sum_{y \in A} \pi^*(x,y) = \min \left\{ \mu(x), \nu(x) \right\} = \mu(x).
\end{align*}
A similar result holds for the second marginal distribution $\nu$. This proves
that \eqref{eq:Hamming_cost_3} holds with equality, which gives
\eqref{eq: the Wasserstein distance of 1st order equiped with the Hamming distance}.
\end{proof}

Now that we have expressed the total variation distance $\| \mu - \nu \|_{\rm TV}$ as the $L^1$ Wasserstein
distance induced by the Hamming metric on $\cX$, the well-known Pinsker's inequality
\begin{align}\label{eq:Pinsker}
	\| \mu - \nu \|_{\rm TV} \le \sqrt{\frac{1}{2}D(\nu \| \mu)}
\end{align}
can be identified as a ${\rm T}_1(1/4)$ inequality that holds for every probability measure $\mu$ on $\cX$.
\end{example}

\begin{remark} It should be pointed out that the constant $c=1/4$ in Pinsker's inequality \eqref{eq:Pinsker} is not necessarily the best possible {\em for a given distribution} $\mu$. Ordentlich and Weinberger \cite{Ordentlich_Weinberger_Pinsker} have obtained the following {\em distribution-dependent} refinement of Pinsker's inequality. Let the function $\varphi \colon [0,1/2] \to \reals^+$ be defined by
	\begin{align}\label{eq:Hoeffding_phi}
		\varphi(p) \deq \begin{cases}
		\displaystyle \left(\frac{1}{1-2p}\right) \ln \left(\frac{1-p}{p}\right), & \mbox{if} \;
        p \in \big[0, \frac{1}{2}\big) \\[0.3cm]
		\hspace*{1.5cm} 2, & \mbox{if} \; p = \frac{1}{2}
	\end{cases}
	\end{align}
(in fact, $\varphi(p) \to 2$ as $p \uparrow 1/2$, $\varphi(p) \to \infty$ as $p \downarrow 0$, and
$\varphi$ is a monotonically decreasing and convex function). Let $\cX$ be a discrete set. For every
$P \in \cP(\cX)$, where $\cP(\cX)$ is the set of all probability distributions defined on the set $\cX$,
let the {\em balance coefficient} be defined as
\begin{align*}
	\pi_P \deq \max_{A \subseteq \cX} \, \min\left\{ P(A),1-P(A) \right\} \quad \Longrightarrow \quad
    \pi_P \in \Big[0, \frac{1}{2}\Big].
\end{align*}
Then, for every $Q \in \cP(\cX)$,
\begin{align}\label{eq:Pinsker_refined}
	\| P - Q \|_{\rm TV} \le \sqrt{\frac{1}{\varphi(\pi_P)} \cdot D(Q \| P)}
\end{align}
(see \cite[Theorem~2.1]{Ordentlich_Weinberger_Pinsker}; related results have been considered in ~\cite{Pinsker_inequality_HP_techreport}).
From the above properties of the function $\varphi$, it follows that the distribution-dependent
refinement of Pinsker's inequality is more pronounced when the balance coefficient is small
(i.e., $\pi_P \ll 1$).
Moreover, this bound is optimal for a given $P$, in the sense that
\begin{align}\label{eq:Pinsker_refined_optimality}
	\varphi(\pi_P) = \inf_{Q \in \cP(\cX)} \frac{D(Q \| P)}{\| P - Q \|^2_{\rm TV}}.
\end{align}
For instance, if $\cX = \{0,1\}$ and $P$ is the distribution of a ${\rm Bernoulli}(p)$ random variable, then $\pi_P = \min\{p,1-p\} \in \left[0,\frac{1}{2}\right]$,
\begin{align*}
	\varphi(\pi_P) = \begin{cases}
	\displaystyle \left(\frac{1}{1-2p}\right) \ln \left(\frac{1-p}{p}\right), & \mbox{if} \; p \neq \frac{1}{2} \\[0.3cm]
	\hspace*{1.5cm} 2, & \mbox{if} \; p = \frac{1}{2}
\end{cases}
\end{align*}
and for every other $Q \in \cP(\{0,1\})$ we have, from \eqref{eq:Pinsker_refined},
\begin{align}\label{eq:Pinsker_Bernoulli_p}
	\| P - Q \|_{\rm TV} \le \begin{cases}
	\displaystyle\sqrt{\frac{1-2p}{\ln\left(\frac{1-p}{p}\right)} \cdot D(Q \| P)},
    & \mbox{if} \; p \neq \frac{1}{2} \\[0.7cm]
	\displaystyle\sqrt{\frac{1}{2} \; D(Q \| P)}, & \mbox{if} \; p = \frac{1}{2}.
\end{cases}
\end{align}

Inequality~\eqref{eq:Pinsker_Bernoulli_p} provides an upper bound on the total variation distance in terms of the divergence.
In general, a bound in the reverse direction cannot be derived since it is easy to come up with examples where the total variation
distance is arbitrarily close to zero, whereas the divergence is equal to infinity. However, consider an i.i.d.\ sample
of size $n$ drawn from a probability distribution $P$. Sanov's theorem implies that the
probability that the empirical distribution of the generated sample deviates in total variation from $P$
by at least some $\varepsilon \in (0,1]$ scales asymptotically like $\exp\bigl(-n \, D^*(P, \varepsilon)\bigr)$,
where $$D^*(P, \varepsilon) \triangleq \inf_{Q \colon \|P-Q\|_{\rm TV} \ge \varepsilon}  \; D(Q \| P).$$
Although a reverse form of Pinsker's inequality (or its probability-dependent refinement in
\cite{Ordentlich_Weinberger_Pinsker}) cannot be derived, it was recently proved in
\cite{Berend_Harremoes_Kontorovich} that
$$D^*(P, \varepsilon) \le \varphi(\pi_P) \, \varepsilon^2 + O(\varepsilon^3).$$
This inequality shows that the probability-dependent refinement of Pinsker's inequality in
\eqref{eq:Pinsker_refined} is actually tight for $D^*(P, \varepsilon)$ when $\varepsilon$ is small,
since both upper and lower bounds scale like $\varphi(\pi_P) \, \varepsilon^2$ if $\varepsilon \ll 1$.
\end{remark}

\begin{remark}
Apart of providing a refined upper bound on the total variation distance between two discrete probability
distributions, the refinement of Pinsker's inequality in \eqref{eq:Pinsker_refined} enables to derive a
refined lower bound on the relative entropy when a lower bound on the total variation distance is available.
This approach was studied in \cite{Sason_ITA13} in the context of the Poisson approximation,
where \eqref{eq:Pinsker_refined} was combined with a new lower bound on the total variation distance
(using the so-called Chen--Stein method) between the distribution of a sum of independent Bernoulli
random variables and the Poisson distribution with the same mean (see \cite{Sason_SPL13b}). Note that,
for a sum of i.i.d.\ Bernoulli random variables, the lower bound on this relative entropy (see
\cite{Sason_ITA13}) scales similarly to the upper bound on this relative
entropy derived by Kontoyiannis et al.\ (see \cite[Theorem 1]{KontoyiannisHJ_2005}) using the Bobkov--Ledoux
logarithmic Sobolev inequality for the Poisson distribution  \cite{Bobkov_Ledoux} (see also
Section~\ref{subsection: Log-Sobolev inequalities for Poisson and compound Poisson measures} here).
\end{remark}

\bigskip
Marton's procedure for deriving Gaussian concentration from a transportation-cost inequality \cite{Marton_dbar,Marton_blowup}
can be distilled as follows:

\begin{proposition}\label{prop:Marton_technique} Suppose $\mu$ satisfies a ${\rm T}_1(c)$ inequality. Then,
the Gaussian concentration inequality in \eqref{eq:isoperimetry_weakened} holds with $\kappa = 1/(2c)$, $K=1$, and $r_0 = \sqrt{2 c \ln 2}$.
\end{proposition}

\begin{proof} Fix two Borel sets $A,B \subset \cX$ with $\mu(A),\mu(B) > 0$. Define the conditional probability measures
	$$
	\mu_A(C) \deq \frac{\mu(C \cap A)}{\mu(A)} \qquad \text{and} \qquad \mu_B(C) \deq \frac{\mu(C \cap B)}{\mu(B)},
	$$
	where $C$ is an arbitrary Borel set in $\cX$. Then $\mu_A, \mu_B \ll \mu$, and
	\begin{align}
		W_1(\mu_A,\mu_B) &\le W_1(\mu,\mu_A) + W_1(\mu,\mu_B) \label{eq:Marton_step1}\\
		&\le \sqrt{2 c D(\mu_A \| \mu)} + \sqrt{2 c D(\mu_B \| \mu)}, \label{eq:Marton_step2}
	\end{align}
	where \eqref{eq:Marton_step1} is by the triangle inequality, while \eqref{eq:Marton_step2} is because $\mu$ satisfies ${\rm T}_1(c)$. Now, for an arbitrary Borel set $C$, we have
	$$
	\mu_A(C) = \int_C \frac{1_A(x)}{\mu(A)} \; \mu(\d x),
	$$
so it follows that $\frac{\d \mu_A}{\d\mu} = \frac{1_A}{\mu(A)}$, and the same holds for $\mu_B$. Therefore,
\begin{align}\label{eq:probability_as_divergence}
	D(\mu_A \| \mu) = \expectation_\mu \left[\frac{\d\mu_A}{\d\mu} \ln \frac{\d\mu_A}{\d\mu} \right] = \ln \frac{1}{\mu(A)},
\end{align}
and an analogous formula holds for $\mu_B$ in place of $\mu_A$. Substituting this into \eqref{eq:Marton_step2} gives
\begin{align}\label{eq:Marton_step3}
	W_1(\mu_A,\mu_B) \le \sqrt{2 c \ln \frac{1}{\mu(A)}} + \sqrt{2 c \ln \frac{1}{\mu(B)}}.
\end{align}
We now obtain a lower bound on $W_1(\mu_A,\mu_B)$. Since $\mu_A$ (resp., $\mu_B$) is supported on $A$ (resp., $B$), every $\pi \in \Pi(\mu_A,\mu_B)$ is supported on $A \times B$. Consequently, for every such $\pi$ we have
\begin{align}
	\int_{\cX \times \cX}d(x,y) \; \pi(\d x, \d y) &= \int_{A \times B} d(x,y) \; \pi(\d x, \d y) \nonumber\\
	&\ge  \int_{A \times B}\inf_{y \in B} d(x,y) \; \pi(\d x, \d y) \nonumber\\
	&= \int_A d(x,B) \; \mu_A(\d x) \nonumber\\
	&\ge \inf_{x \in A} d(x,B) \; \mu_A(A) \nonumber\\
	&= d(A,B),\label{eq:coupling_distance}
\end{align}
where $\mu_A(A) = 1$, and $d(A,B) \deq \inf_{x \in A, y \in B}d(x,y)$ is the distance between $A$ and $B$. Since \eqref{eq:coupling_distance} holds for every $\pi \in \Pi(\mu_A,\mu_B)$, we can take the infimum over all such
$\pi$ and get $W_1(\mu_A,\mu_B) \ge d(A,B)$. Combining this with \eqref{eq:Marton_step3} gives the inequality
\begin{align}\label{eq:Marton_step4}
	d(A,B) \le \sqrt{2 c \ln \frac{1}{\mu(A)}} + \sqrt{2 c \ln \frac{1}{\mu(B)}},
\end{align}
which holds for all Borel sets $A$ and $B$ that have nonzero $\mu$-probability.

Let $B = A^c_r$. Then $\mu(B) = 1 - \mu(A_r)$ and $d(A,B) \ge r$. Consequently, \eqref{eq:Marton_step4} gives
\begin{align}\label{eq:Marton_conc_1}
	r \le \sqrt{2 c \ln \frac{1}{\mu(A)}} + \sqrt{2 c \ln \frac{1}{1-\mu(A_r)}}.
\end{align}
If $\mu(A) \ge 1/2$ and $r \ge \sqrt{2 c \ln 2}$, then \eqref{eq:Marton_conc_1} gives
\begin{align} \label{eq: consequence of Marton_conc_1}
	\mu(A_r) \ge 1 - \exp\left( - \frac{1}{2c} \left(r - \sqrt{2 c \ln 2}\right)^2 \right).
\end{align}
Hence, the Gaussian concentration inequality in \eqref{eq:isoperimetry_weakened} indeed holds with $\kappa = 1/(2c)$ and $K=1$ for all $r \ge r_0 = \sqrt{2 c \ln 2}$.
\end{proof}
\begin{remark} The exponential inequality \eqref{eq: consequence of Marton_conc_1} has appeared earlier in the work of McDiarmid \cite{McDiarmid_bounded_differences_Martingales_1989} and Talagrand \cite{Talagrand95}. The major innovation that came from Marton's work was her use of optimal transportation ideas to derive a more general ``symmetric'' form \eqref{eq:Marton_step4}.
\end{remark}
\begin{remark} The formula \eqref{eq:probability_as_divergence}, apparently first used explicitly by Csisz\'ar \cite[Eq.~(4.13)]{Csiszar_Sanov}, is actually quite remarkable: it states that the probability of an arbitrary event can be expressed as an exponential of a divergence.
\end{remark}

While the method described in the proof of Proposition~\ref{prop:Marton_technique} does not produce optimal concentration estimates (which typically have to be derived on a case-by-case basis), it hints at the potential power of the transportation-cost inequalities. To make full use of this power, we first establish an important fact that, for $p \in [1,2]$, the ${\rm T}_p$ inequalities tensorize (see, for example, \cite[Proposition~22.5]{Villani_newbook}):
\begin{proposition}[Tensorization of transportation-cost inequalities]\label{prop:TC_tensorization} If $\mu$ satisfies ${\rm T}_p(c)$ on $(\cX,d)$ for an arbitrary $p \in [1,2]$, then, for every $n \in \naturals$, the product measure $\mu^{\otimes n}$ satisfies ${\rm T}_p(cn^{2/p-1})$ on $(\cX^n,d_{p,n})$ with the metric
	\begin{align}\label{eq:product_metric}
	d_{p,n}(x^n,y^n) \deq \left(\sum^n_{i=1}d^p(x_i,y_i)\right)^{1/p}, \qquad \forall \, x^n,y^n \in \cX^n.
\end{align}
\end{proposition}
\begin{proof} Suppose $\mu$ satisfies ${\rm T}_p(c)$. Fix $n \in \naturals$, and fix an arbitrary probability measure $\nu$ on $(\cX^n,d_{p,n})$. Let $X^n,Y^n \in \cX^n$ be two independent random $n$-tuples, such that
	\begin{align}
		P_{X^n} &= P_{X_1} \otimes P_{X_2|X_1} \otimes \ldots \otimes P_{X_n|X^{n-1}} = \nu \label{eq: P_X^n} \\
		P_{Y^n} &= P_{Y_1} \otimes P_{Y_2} \otimes \ldots \otimes P_{Y_n} = \mu^{\otimes n} \label{eq: P_Y^n}.
	\end{align}
	For each $i \in \{ 1,\ldots,n\}$, let us define the ``conditional'' $W_p$ distance
	\begin{align}
		&W_p(P_{X_i|X^{i-1}},P_{Y_i} | P_{X^{i-1}}) \nonumber \\
		&\qquad \deq \left(\int_{\cX^{i-1}} W^p_p(P_{X_i|X^{i-1}=x^{i-1}},P_{Y_i}) P_{X^{i-1}}( \d x^{i-1})\right)^{1/p}.
        \label{eq:conditional W_p}
	\end{align}
	We will now prove that
	\begin{align}\label{eq:W_subadditivity}
		W^p_p(\nu,\mu^{\otimes n}) &= W^p_p(P_{X^n},P_{Y^n}) \nonumber\\
		&\le \sum^n_{i=1} W^p_p(P_{X_i|X^{i-1}},P_{Y_i}|P_{X^{i-1}}),
	\end{align}
where the $L^p$ Wasserstein distance on the left-hand side is computed with respect to the $d_{p,n}$ metric. By Lemma~\ref{lm:W_properties}, there exists an optimal coupling of $P_{X_1}$ and $P_{Y_1}$, i.e., a pair $(X^*_1,Y^*_1)$ of jointly distributed $\cX$-valued random variables such that $P_{X^*_1} = P_{X_1}$, $P_{Y^*_1} = P_{Y_1}$, and
$$
W^p_p(P_{X_1},P_{Y_1}) = \expectation[d^p(X^*_1,Y^*_1)].
$$
Now for each $i = 2,\ldots,n$ and each choice of $x^{i-1} \in \cX^{i-1}$, again by Lemma~\ref{lm:W_properties}, there exists an optimal coupling of $P_{X_i|X^{i-1}=x^{i-1}}$ and $P_{Y_i}$, i.e., a pair $(X^*_i(x^{i-1}),Y^*_i(x^{i-1}))$ of jointly distributed $\cX$-valued random variables such that $P_{X^*_i(x^{i-1})} = P_{X_i|X^{i-1}=x^{i-1}}$,  $P_{Y^*_i(x^{i-1})} = P_{Y_i}$, and
$$
W^p_p(P_{X_i|X^{i-1}=x^{i-1}},P_{Y_i}) = \expectation[d^p(X^*_i(x^{i-1}),Y^*_i(x^{i-1}))].
$$
Moreover, because $(\cX, d)$ is a Polish space, all couplings can be constructed in such a way that the mapping
$$x^{i-1} \mapsto \pr\big( (X^*_i(x^{i-1}),Y^*_i(x^{i-1})) \in C\big)$$ is measurable for each Borel set
$C \subseteq \cX \times \cX$ \cite{Villani_newbook}. In other words, for each $i$, we can define the regular
conditional distributions
$$
P_{X^*_iY^*_i|{X^*}^{(i-1)}=x^{i-1}} \deq P_{X^*_i(x^{i-1})Y^*_i(x^{i-1})}, \qquad \forall \, x^{i-1} \in \cX^{i-1}
$$
such that
$$
P_{{X^*}^n {Y^*}^n} = P_{X^*_1Y^*_1} \otimes P_{X^*_2 Y^*_2 | X^*_1} \otimes \ldots \otimes P_{X^*_n Y^*_n|{X^*}^{(n-1)}}
$$
is a coupling of $P_{X^n} = \nu$ and $P_{Y^n} = \mu^{\otimes n}$, and for all $x^{i-1} \in \cX^{i-1}$ and $i \in \{1, \ldots, n\}$
\begin{align}   \label{eq:conditional_coupling}
	W^p_p(P_{X_i|X^{i-1}=x^{i-1}},P_{Y_i}) = \expectation[d^p(X^*_i,Y^*_i)|X^{*(i-1)}=x^{i-1}].
\end{align}
	 By definition of $W_p$, we then have
\begin{align}
	W^p_p(\nu,\mu^{\otimes n}) &\le \expectation[d^p_{p,n}(X^{*n},Y^{*n})] \label{eq:W_subadd_1}\\
	&= \sum^n_{i=1}\expectation[d^p(X^*_i,Y^*_i)] \label{eq:W_subadd_2}\\
	&= \sum^n_{i=1}\expectation\Big[ \expectation\big[d^p(X^*_i,Y^*_i)|X^{*(i-1)}\big] \Big] \label{eq:W_subadd_3}\\
	&= \sum^n_{i=1} W^p_p(P_{X_i|X^{i-1}},P_{Y_i}|P_{X^{i-1}}), \label{eq:W_subadd_4}
\end{align}
where:
\begin{itemize}
	\item \eqref{eq:W_subadd_1} is due to the facts that $W_p(\nu,\mu^{\otimes n})$ is the $L^p$ Wasserstein distance with respect to
    the $d_{p,n}$ metric, and $(X^{*n},Y^{*n})$ is a (not necessarily optimal) coupling of $P_{X^n} = \nu$ and $P_{Y^n} = \mu^{\otimes n}$;
	\item \eqref{eq:W_subadd_2} is by the definition \eqref{eq:product_metric} of $d_{p,n}$;
	\item \eqref{eq:W_subadd_3} is by the law of iterated expectations; and
	\item \eqref{eq:W_subadd_4} is by \eqref{eq:conditional W_p} and \eqref{eq:conditional_coupling}.
\end{itemize}
We have thus proved \eqref{eq:W_subadditivity}. By hypothesis, $\mu$ satisfies ${\rm T}_p(c)$ on $(\cX,d)$. Therefore, since $P_{Y_i} = \mu$ for every $i$, we can write
\begin{align}
	&W^p_p(P_{X_i|X^{i-1}},P_{Y_i}|P_{X^{i-1}}) \nonumber\\
	&\quad= \int_{\cX^{i-1}} W^p_p(P_{X_i|X^{i-1}=x^{i-1}},P_{Y_i}) \; P_{X^{i-1}}(\d x^{i-1}) \nonumber \\
	&\quad\le \int_{\cX^{i-1}} \left( 2 c D(P_{X_i|X^{i-1}=x^{i-1}} \| P_{Y_i})\right)^{p/2} P_{X^{i-1}}(\d x^{i-1}) \nonumber \\
	&\quad\le (2c)^{p/2}\left(\int_{\cX^{i-1}}  D(P_{X_i|X^{i-1}=x^{i-1}} \| P_{Y_i}) \; P_{X^{i-1}}(\d x^{i-1})\right)^{p/2} \nonumber \\
	&\quad= (2c)^{p/2} \left( D(P_{X_i|X^{i-1}}\|P_{Y_i}|P_{X^{i-1}})\right)^{p/2},
    \label{eq: upper bound on the conditional p-th Wasserstein distance}
\end{align}
where the second inequality follows from Jensen's inequality and the concavity of the function $t \mapsto t^{p/2}$ for $p \in [1,2]$.
Consequently, it follows that
\begin{align*}
	W^p_p(\nu,\mu^{\otimes n}) &\le (2c)^{p/2} \sum^n_{i=1} \left( D(P_{X_i|X^{i-1}}\|P_{Y_i}|P_{X^{i-1}})\right)^{p/2} \\
	&\le (2c)^{p/2}n^{1-p/2} \left( \sum^n_{i=1} D(P_{X_i|X^{i-1}}\|P_{Y_i}|P_{X^{i-1}})\right)^{p/2} \\
	&= (2c)^{p/2}n^{1-p/2} \left( D(P_{X^n} \| P_{Y^n}) \right)^{p/2} \\
	&= (2c)^{p/2}n^{1-p/2} \left( D(\nu \| \mu^{\otimes n}) \right)^{p/2},
\end{align*}
where the first line holds by using \eqref{eq:W_subadditivity}
and \eqref{eq: upper bound on the conditional p-th Wasserstein distance},
the second line is by H\"older's inequality, the third line is by the chain rule for the
divergence and since $P_{Y^n}$ is a product probability measure, and the fourth line is by
\eqref{eq: P_X^n} and \eqref{eq: P_Y^n}. This finally gives
\begin{align*}
	W_p(\nu, \mu^{\otimes n}) \le \sqrt{2c n^{2/p-1} D(\nu \| \mu^{\otimes n})},
\end{align*}
i.e., $\mu^{\otimes n}$ indeed satisfies the ${\rm T}_p(c n^{2/p-1})$ inequality.
\end{proof}
Since $W_2$ dominates $W_1$ (cf.\ Item~2 of Lemma~\ref{lm:W_properties}), a ${\rm T}_2(c)$ inequality is stronger than a ${\rm T}_1(c)$ inequality (for an arbitrary $c>0$). Moreover, as Proposition~\ref{prop:TC_tensorization} above shows, ${\rm T}_2$ inequalities tensorize {\em exactly}: if $\mu$ satisfies ${\rm T}_2$ with a constant $c > 0$, then $\mu^{\otimes n}$ also satisfies ${\rm T}_2$ for every $n$ with the {\em same} constant $c$. By contrast, if $\mu$ only satisfies ${\rm T}_1(c)$, then the product measure $\mu^{\otimes n}$ satisfies ${\rm T}_1$ with the much worse constant $cn$. As we shall shortly see, this sharp difference between the ${\rm T}_1$ and ${\rm T}_2$ inequalities actually has deep consequences. In a nutshell, in the two sections that follow, we will show that, for $p \in \{1, 2\}$, a given probability measure $\mu$ satisfies a ${\rm T}_p(c)$ inequality on $(\cX,d)$ if and only if it has Gaussian concentration with constant $1/(2c)$. Suppose now that we wish to show Gaussian concentration for the product measure $\mu^{\otimes n}$ on the product space $(\cX^n, d_{1,n})$. Following our tensorization programme, we could first show that $\mu$ satisfies a transportation-cost inequality for some $p \in [1,2]$, then apply Proposition~\ref{prop:TC_tensorization} and consequently also apply
Proposition~\ref{prop:Marton_technique}. If we go through with this approach, we will see that:
	\begin{itemize}
		\item If $\mu$ satisfies ${\rm T}_1(c)$ on $(\cX,d)$, then $\mu^{\otimes n}$ satisfies ${\rm T}_1(cn)$ on $(\cX^n,d_{1,n})$, which is equivalent to Gaussian concentration with constant $1/(2cn)$. Hence, in this
case, the concentration phenomenon is weakened by increasing the dimension $n$.
		\item If, on the other hand, $\mu$ satisfies ${\rm T}_2(c)$ on $(\cX,d)$, then $\mu^{\otimes n}$ satisfies ${\rm T}_2(c)$ on $(\cX^n,d_{2,n})$, which is equivalent to Gaussian concentration with the same constant $1/(2c)$, and this constant is {\em independent} of the dimension $n$.
\end{itemize}
These two results give the same constants in concentration inequalities for sums of independent random variables: if $f$ is a Lipschitz function
on $(\cX,d)$, then from the fact that
		\begin{align*}
			d_{1,n}(x^n,y^n) &= \sum^n_{i=1} d(x_i,y_i) \\[-0.1cm]			
			&\le \sqrt{n} \left( \sum^n_{i=1}d^2(x_i,y_i)\right)^{\frac{1}{2}} \\[-0.1cm]
			&= \sqrt{n}\, d_{2,n}(x^n,y^n)
		\end{align*}
	we can conclude that, for $f_n(x^n) \deq (1/n)\sum^n_{i=1}f(x_i)$,
		\begin{align*}
			\| f_n \|_{\rm Lip,1} & \deq \sup_{x^n \neq y^n} \frac{|f_n(x^n) - f_n(y^n)|}{d_{1,n}(x^n,y^n)} \le \frac{\| f \|_{\rm Lip}}{n} \,
		\end{align*}
		and
		\begin{align*}
				\| f_n \|_{\rm Lip,2} & \deq \sup_{x^n \neq y^n} \frac{|f_n(x^n) - f_n(y^n)|}{d_{2,n}(x^n,y^n)} \le \frac{\| f \|_{\rm Lip}}{\sqrt{n}} \, .
		\end{align*}
		Therefore, both ${\rm T}_1(c)$ and ${\rm T}_2(c)$ give
		\begin{align*}
			\pr \left( \frac{1}{n}\sum^n_{i=1} f(X_i) \ge r \right) \le \exp\left(- \frac{nr^2}{2c \| f \|^2_{\rm Lip}}\right), \quad \forall \, r>0
		\end{align*}
		where $X_1,\ldots,X_n$ are i.i.d.\ $\cX$-valued random variables whose common marginal $\mu$ satisfies either ${\rm T}_2(c)$ or ${\rm T}_1(c)$, and $f$ is a Lipschitz function on $\cX$ with $\expectation[f(X_1)] = 0$. However, the difference between concentration inequalities
that are derived from ${\rm T}_1$ and ${\rm T}_2$ inequalities becomes quite pronounced in general. Note that, in practice, it is often easier to work with ${\rm T}_1$ inequalities than with ${\rm T}_2$ inequalities.	
		
	The same strategy as above can be used to prove the following generalization of Proposition~\ref{prop:TC_tensorization}:	
	\begin{proposition}\label{prop:TC_tensorization_2} Let $\mu_1,\ldots,\mu_n$ be $n$ Borel probability measures on a Polish
space $(\cX,d)$, such that $\mu_i$ satisfies ${\rm T}_p(c_i)$ for some $c_i > 0$, for each $i \in \{1,\ldots,n\}$.
Let $c \triangleq \max_{1 \le i \le n} c_i$. Then, for an arbitrary $p \in [1,2]$, the probability measure
$\mu = \mu_1 \otimes \ldots \otimes \mu_n$ satisfies ${\rm T}_p(cn^{2/p-1})$ on $(\cX^n,d_{p,n})$ (with the metric $d_{p,n}$
in \eqref{eq:product_metric}).
	\end{proposition}

\subsection{Gaussian concentration and ${\rm T}_1$ inequalities}
\label{subsection: Gaussian concentration and $T_1$ inequalities}

As we have shown above, Marton's argument can be used to deduce Gaussian concentration from a transportation-cost inequality. As we will demonstrate here and in the following section, in certain cases these properties are {\em equivalent}. We will consider first the case when $\mu$ satisfies a ${\rm T}_1$ inequality. The first proof of equivalence between ${\rm T}_1$ and Gaussian concentration was obtained by Bobkov and G\"otze \cite{Bobkov_Gotze_expint}, and it relies on
the following variational representations of the $L^1$ Wasserstein distance and the divergence:
\begin{enumerate}
	\item {\bf Kantorovich--Rubinstein theorem} \cite[Theorem~1.14]{Villani_TOT} and \cite[Theorem~5.10]{Villani_newbook}: For every $\mu,\nu \in \cP_1(\cX)$ on a Polish probability space $(\cX, d)$,
	\begin{align}\label{eq:Kantorovich_Rubinstein}
	W_1(\mu,\nu) = \sup_{f \colon \| f \|_{\rm Lip} \le 1} \bigl| \expectation_\mu[f] - \expectation_\nu[f] \bigr|.
	\end{align}
	\item {\bf Donsker--Varadhan lemma} \cite[Lemma~6.2.13]{Dembo_Zeitouni}: For every two Borel probability measures $\mu,\nu$ on a Polish probability space $(\cX,d)$ such that $\nu \ll \mu$, the following variational representation of the divergence holds:
	\begin{align}\label{eq:Donsker_Varadhan}
		D(\nu \| \mu) = \sup_{g \in C_{\text{b}}(\cX)}\bigl\{ \expectation_\nu[g]-
\ln \expectation_\mu[\exp(g)]\bigr\}
	\end{align}
where the supremization in \eqref{eq:Donsker_Varadhan} is over the set $C_{\text{b}}(\cX)$ of
continuous and bounded real-valued functions on $\cX$. Furthermore, for every measurable function $g$ such
that $\expectation_\mu[\exp(g)] < \infty$,
\begin{align}  \label{eq: Donsker-Varadhan inequality}
\expectation_\nu[g] \le D(\nu \| \mu) + \ln \expectation_\mu[\exp(g)].
\end{align}
(In fact, the supremum in \eqref{eq:Donsker_Varadhan} can be extended to bounded Borel-measurable functions $g$ \cite[Lemma~1.4.3]{Dupuis_Ellis_book}.)
\end{enumerate}

The following theorem was introduced by Bobkov and G\"otze \cite[Theorem~3.1]{Bobkov_Gotze_expint}:
\begin{theorem}[Bobkov and G\"otze]\label{thm:Bobkov_Gotze} Let $\mu \in \cP_1(\cX)$ be a Borel probability measure, and
assume that there exists some $x_0 \in \cX$ such that $\expectation_\mu[d(X,x_0)] < \infty$. Then,
$\mu$ satisfies ${\rm T}_1(c)$ if and only if the inequality
\begin{align}\label{eq:expint}
	\expectation_\mu\left\{\exp[tf(X)]\right\} \le \exp\left(\frac{ct^2}{2}\right)
\end{align}
holds for all $1$-Lipschitz functions $f \colon \cX \to \reals$ with $\expectation_\mu[f(X)] = 0$, and all $t \in \reals$.
\end{theorem}
\begin{remark} The condition $\expectation_\mu[d(X,x_0)] < \infty$ is needed to ensure that every Lipschitz function $f \colon \cX \to \reals$ is $\mu$-integrable:
\begin{align*}
\expectation_\mu\big[|f(X)|\big] &\leq |f(x_0)| + \expectation_\mu\big[|f(X)-f(x_0)|\big] \\
&\le |f(x_0)| +
\| f \|_{\rm Lip} \; \expectation_\mu\big[d(X,x_0)\big] < \infty.
	\end{align*}
\end{remark}

\begin{proof} Without loss of generality, we may consider \eqref{eq:expint} only for $t \ge 0$.
	
	Suppose first that $\mu$ satisfies ${\rm T}_1(c)$. Consider some $\nu \ll \mu$. Using the ${\rm T}_1(c)$ property of $\mu$ together with the Kantorovich--Rubinstein formula \eqref{eq:Kantorovich_Rubinstein}, we can write
	\begin{align*}
		\int_{\cX} f \d \nu \le W_1(\mu, \nu) \le \sqrt{2 c D(\nu \| \mu)}
	\end{align*}
	for every $1$-Lipschitz $f \colon \cX \to \reals$ with $\expectation_\mu[f] = 0$. Next, from the fact that
	\begin{align}\label{eq:sqrt_var}
		\inf_{t > 0}\left( \frac{a}{t} + \frac{bt}{2}\right) = \sqrt{2ab}
	\end{align}
	for every $a,b \ge 0$, we see that every such $f$ must satisfy
	\begin{align*}
		\int_{\cX} f \, \d \nu \le \frac{D(\nu \| \mu)}{t} + \frac{ct}{2}, \qquad \forall \, t > 0.
	\end{align*}
Rearranging, we obtain
\begin{align*}
	\int_{\cX} t f \, \d \nu - \frac{ct^2}{2} \le D(\nu \| \mu), \qquad \forall \, t > 0.
\end{align*}
Applying this inequality to $\nu = \mu^{(g)}$ (the $g$-tilting of $\mu$) where $g \triangleq tf$,
and using the fact that
\begin{align*}
D(\mu^{(g)} \| \mu) &= \int_{\cX} g \,\d\mu^{(g)} - \ln \int_{\cX} \exp(g)\, \d \mu \\
& = \int_{\cX} tf \, \d \nu - \ln \int_{\cX} \exp(tf)\, \d \mu
\end{align*}
we deduce that
\begin{align*}
\ln \left(\int_{\cX} \exp(tf)\, \d \mu \right) \leq \frac{c t^2}{2}
\end{align*}
for all $t \ge 0$, and all $f$ with $\| f \|_{\rm Lip} \le 1$ and $\expectation_\mu[f] = 0$, which is precisely \eqref{eq:expint}.

Conversely, assume that $\mu$ satisfies \eqref{eq:expint} for all $1$-Lipschitz functions
$f \colon \cX \to \reals$ with $\expectation_\mu[f(X)] = 0$ and all $t \in \reals$, and let
$\nu$ be an arbitrary Borel probability measure such that $\nu \ll \mu$.
Consider an arbitrary function of the form $g \deq tf$ where $t > 0$. By the assumption in
\eqref{eq:expint}, $\expectation_\mu[\exp(g)] < \infty$; furthermore, $g$ is a Lipschitz
function, so it is also measurable. Hence, \eqref{eq: Donsker-Varadhan inequality} gives
\begin{align*}
	D(\nu \| \mu) &\ge \int_{\cX} tf\, \d\nu - \ln \int_{\cX} \exp(tf)\, \d \mu \\
	&\ge \int_{\cX} tf\, \d\nu - \int_{\cX} tf\, \d\mu - \frac{ct^2}{2}
\end{align*}
where in the second step we have used the fact that $\int_{\cX} f\,\d\mu = 0$ by hypothesis, as well as \eqref{eq:expint}. Rearranging gives
\begin{align}
	\left|\int_{\cX} f\, \d\nu - \int_{\cX} f\,\d\mu\right|
    \le \frac{D(\nu \| \mu)}{t} + \frac{ct}{2}, \qquad \forall \, t > 0
\end{align}
(the absolute value in the left-hand side is a consequence of the fact that exactly the same argument goes through with $-f$ instead of $f$). Applying \eqref{eq:sqrt_var}, we see that the inequality
\begin{align}\label{eq:BG_converse}
	\left|\int_{\cX} f\, \d\nu - \int_{\cX} f\, \d\mu\right| \le \sqrt{2 c D(\nu \| \mu)}
\end{align}
holds for all $1$-Lipschitz $f$ with $\expectation_\mu[f] = 0$. In fact, we may now drop the condition that $\expectation_\mu[f] = 0$ by replacing $f$ with $f - \expectation_\mu[f]$. Thus, taking the supremum over all $1$-Lipschitz functions $f$ on the left-hand side of \eqref{eq:BG_converse} and using the Kantorovich--Rubinstein formula \eqref{eq:Kantorovich_Rubinstein}, we conclude that $W_1(\mu,\nu) \le \sqrt{2c D(\nu \| \mu)}$ for every $\nu \ll \mu$, i.e., $\mu$ satisfies ${\rm T}_1(c)$. This completes the proof of Theorem~\ref{thm:Bobkov_Gotze}.
\end{proof}

Theorem~\ref{thm:Bobkov_Gotze} gives us an alternative way of deriving Gaussian concentration for Lipschitz functions (compare with earlier derivations using the entropy method):

\begin{corollary}\label{cor:Lipschitz_Gauss_con_revisited} Let $\cA$ be the space of all Lipschitz functions
on $\cX$, and let $\mu \in \cP_1(\cX)$ be a Borel probability measure that satisfies ${\rm T}_1(c)$. Then, the following inequality holds for every $f \in \cA$:
\begin{align} \label{eq:Lipschitz_Gauss_con_revisited}
	\pr \Big( f(X) \ge \expectation_\mu [f(X)] + r \Big) \le \exp\left( - \frac{r^2}{2c \| f \|_{\rm Lip}^2} \right), \qquad \forall \, r > 0.
	\end{align}
\end{corollary}
\begin{proof}
The result follows from the Chernoff bound and \eqref{eq:expint}.
\end{proof}

As another illustration, we prove the following bound, which includes the Kearns--Saul inequality (cf.\ Theorem~\ref{theorem: Kearns-Saul inequality}) as a special case:
\begin{theorem} Let $\cX$ be the Hamming space $\{0,1\}^n$, equipped with the metric
	\begin{align}\label{eq:product_Hamming_metric}
		d(x^n,y^n) = \sum^n_{i=1}1_{\{x_i \neq y_i\}}.
	\end{align}
Let $X_1,\ldots,X_n$ be i.i.d.\ ${\rm Bernoulli}(p)$ random variables.
Then, for every Lipschitz function $f \colon \{0,1\}^n \to \reals$,
\begin{align}\label{eq:sharp_Bernoulli_concentration}
\pr \Bigg( f(X^n) - \expectation[f(X^n)] \ge r \Bigg) \le
\exp\left(-\frac{\ln\left(\frac{1-p}{p}\right) \, r^2}{n \| f \|^2_{\rm Lip} (1-2p)} \right),
\qquad \forall \, r > 0.
\end{align}
\end{theorem}
\begin{remark}\label{rem:p_half} In the limit as $p \to 1/2$, the right-hand side of \eqref{eq:sharp_Bernoulli_concentration} becomes $\exp\left(-\frac{2r^2}{n \|f\|^2_{\rm Lip}}\right)$.
\end{remark}
\begin{proof} Taking into account Remark~\ref{rem:p_half}, we may assume without loss of generality that $p \neq 1/2$. From the distribution-dependent refinement of Pinsker's inequality \eqref{eq:Pinsker_Bernoulli_p}, it follows that the ${\rm Bernoulli}(p)$ measure satisfies ${\rm T}_1(1/(2\varphi(p)))$ with respect to the Hamming metric, where $\varphi(p)$ is defined in \eqref{eq:Hoeffding_phi}. By Proposition~\ref{prop:TC_tensorization}, the product of $n$ ${\rm Bernoulli}(p)$ measures satisfies ${\rm T}_1(n/(2\varphi(p)))$ with respect to the metric \eqref{eq:product_Hamming_metric}. The bound \eqref{eq:sharp_Bernoulli_concentration} then follows from Corollary~\ref{cor:Lipschitz_Gauss_con_revisited}.
\end{proof}

\begin{remark} If $\| f \|_{\rm Lip} \le \frac{C}{n}$ for an arbitrary $C > 0$, then \eqref{eq:sharp_Bernoulli_concentration}
implies that for every $r>0$
\begin{align} \label{eq:Lipschitz_Gaussian_concentration}
	\pr \Bigg( f(X^n) - \expectation[f(X^n)] \ge r \Bigg)
    \le \exp\left(-\frac{\ln\left(\frac{1-p}{p}\right)}{C^2 (1-2p)} \cdot nr^2 \right).
\end{align}
This will be the case, for instance, if $f(x^n) = (1/n)\sum^n_{i=1} f_i(x_i)$ for some functions $f_1,\ldots,f_n : \{0,1\} \to \reals$ satisfying $|f_i(0)-f_i(1)| \le C$ for all $i=1,\ldots,n$. More generally, every $f$ satisfying \eqref{eq:bounded_differences} with $c_i = c'_i/n$, $i=1,\ldots,n$, for some constants $c'_1,\ldots,c'_n \ge 0$,
satisfies \eqref{eq:Lipschitz_Gaussian_concentration}
for all $r > 0$ with $C = \max_{1 \le i \le n} c'_i$.
\end{remark}

In the following, we provide Marton's coupling inequality, which forms a slightly stronger form of the original
result of Marton \cite{Marton_blowup} (see \cite[Theorem~8.2]{Boucheron_Lugosi_Massart_book} for the following
stronger statement):
\begin{theorem}[Marton's coupling inequality] \label{thm:Marton's coupling inequality}
Let $\mu = \mu_1 \otimes \ldots \mu_n$ be a product probability measure of $X^n \in \cX^n$, and let $\nu$ (where
$\nu \ll \mu$) be a probability measure of $Y^n \in \cX^n$. Then,
\begin{align} \label{eq:Marton's coupling inequality}
\min_{\pi \in \Pi(\mu,\nu)} \sum_{i=1}^n \pr^2(X_i \neq Y_i) \le \frac{1}{2} \, D(\nu \| \mu).
\end{align}
\begin{proof}
For the sake of conciseness and for avoiding some overlap with the excellent textbook \cite{Boucheron_Lugosi_Massart_book},
the reader is referred to \cite[p.~241]{Boucheron_Lugosi_Massart_book}.
\end{proof}
\end{theorem}

We provide in the following an alternative proof of McDiarmid's inequality \eqref{eq:McDiarmid},
based on the earlier material in this chapter about transportation-cost inequalities (recall the
two previous proofs of this inequality in Sections~\ref{subsection: McDiarmid's inequality}
and~\ref{ssec:bounded_differences_revisited}).

\bigskip
{\bf{An alternative proof of McDiarmid's inequality}}:
For every $n \in \naturals$, constants $c_1,\ldots,c_n > 0$, and a measurable space $\cX$, let us
equip the product space $\cX^n$ with the weighted Hamming metric
\begin{align*}
	d(x^n,y^n) \deq \sum^n_{i=1} c_i 1_{\{x_i \neq y_i\}}.
\end{align*}
Let $f \colon \cX^n \to \reals$ be a Lipschitz function (with respect to the metric $d$), and
suppose that it satisfies the condition of the bounded differences in \eqref{eq:bounded_differences}.
The corresponding Lipschitz constant $\| f \|_{\rm Lip}$ is given by
\begin{align*}
	\| f \|_{{\rm Lip}} = \sup_{x^n \neq y^n} \frac{|f(x^n) - f(y^n)|}{d(x^n,y^n)}.
\end{align*}
It is easy to verify that the condition $\| f \|_{{\rm Lip}} \le 1$ is equivalent to the condition
in \eqref{eq:bounded_differences}.

Let $\mu_1, \ldots, \mu_n$ be arbitrary $n$ probability measures on $\cX$, and let $\mu = \mu_1 \otimes \ldots \mu_n$
be a product probability measure of $X^n \in \cX^n$. Let $\nu$ be an arbitrary (not necessarily a product)
probability measure on $\cX^n$, where $\nu \ll \mu$, and let $Y^n$ be a random vector that is drawn from $\nu$.
Using the condition of the bounded differences in \eqref{eq:bounded_differences} and the Cauchy-Schwarz inequality,
\begin{align} \label{eq:invoking Cauchy-Schwarz and the condition of the bounded differences}
\Big| \expectation_{\mu}[f] - \expectation_{\nu}[f] \Big| & = \Big| \expectation \big[f(X^n) - f(Y^n) \big] \Big|  \nonumber \\
& \leq \expectation \Big| f(X^n) - f(Y^n) \Big| \nonumber \\
& \leq \sum_{i=1}^n \expectation \big[ c_i \, 1_{\{X_i \neq Y_i\}} \big] \nonumber \\
& \leq \left( \sum_{i=1}^n c_i^2 \right)^{\frac{1}{2}} \left( \sum_{i=1}^n \expectation^2\big[ 1_{\{X_i \neq Y_i\}}\big]
\right)^{\frac{1}{2}} \nonumber \\[0.1cm]
& = \left( \sum_{i=1}^n c_i^2 \right)^{\frac{1}{2}} \left( \sum_{i=1}^n \pr^2(X_i \neq Y_i) \right)^{\frac{1}{2}}
\end{align}
where the last equality holds because the expectation of the indicator function of an event is the probability of the event.
By minimizing the right-hand side of \eqref{eq:invoking Cauchy-Schwarz and the condition of the bounded differences}
with respect to all the couplings  $\pi \in \Pi(\mu,\nu)$, it follows from \eqref{eq:Marton's coupling inequality} that
\begin{align}
\Big| \expectation_{\mu}[f] - \expectation_{\nu}[f] \big| \leq \sqrt{\frac{1}{2} \, \Bigl( \sum_{i=1}^n c_i^2 \Bigr) \; D(\nu \| \mu)}.
\label{eq:Marton's coupling and Cauchy-Schwarz}
\end{align}
By supremizing the left-hand side of \eqref{eq:Marton's coupling and Cauchy-Schwarz}, with respect to all the Lipschitz
functions $f \colon \cX^n \to \reals$ such that $\| f \|_{{\rm Lip}} \le 1$, it follows from the Kantorovich--Rubinstein theorem
(see \eqref{eq:Kantorovich_Rubinstein}) that
\begin{align*}
W_1(\mu, \nu) \leq \sqrt{\frac{1}{2} \, \Bigl( \sum_{i=1}^n c_i^2 \Bigr) \; D(\nu \| \mu)} \, .
\end{align*}
Hence, $\mu$ satisfies ${\rm T}_1(c)$ (relative to the weighted Hamming metric $d$) with the constant $c = \frac{1}{4} \sum_{i=1}^n c_i^2$.
By Theorem~\ref{thm:Bobkov_Gotze}, it is equivalent to the satisfiability of the inequality
\begin{align*}
\expectation_\mu\Big\{\exp[tf(X^n)]\Big\} \le \exp\left(\frac{1}{8} \, \sum_{i=1}^n c_i^2 \, t^2 \right), \quad \forall \, t \in \reals
\end{align*}
for all Lipschitz functions $f \colon \cX^n \to \reals$ with $\expectation_\mu[f(X^n)] = 0$, and $\| f \|_{{\rm Lip}} \le 1$.
Following Corollary~\ref{cor:Lipschitz_Gauss_con_revisited}, it provides an alternative proof
of McDiarmid's inequality \eqref{eq:McDiarmid}.

\subsection{Dimension-free Gaussian concentration and ${\rm T}_2$ inequalities}

So far, we have mostly confined our discussion to the ``one-dimensional'' case of a probability measure $\mu$ on a Polish space $(\cX,d)$. Recall, however, that in most applications our interest is in functions of $n$ independent random variables taking values in $\cX$. Proposition~\ref{prop:TC_tensorization} shows that the transportation-cost inequalities tensorize, so in principle this property can be used to derive concentration inequalities for
such functions. However, as suggested by Proposition~\ref{prop:TC_tensorization} and the discussion following it, ${\rm T}_1$ inequalities are not very useful in this regard, since the resulting concentration bounds will deteriorate as $n$ increases. Indeed, if $\mu$ satisfies ${\rm T}_1(c)$ on $(\cX,d)$, then the product measure $\mu^{\otimes n}$ satisfies ${\rm T}_1(cn)$ on the product space $(\cX^n,d_{1,n})$, which is equivalent to the Gaussian concentration property
$$
\pr\Big( f(X^n) \ge \expectation f(X^n) + r\Big) \le K\exp\left( - \frac{r^2}{2cn}\right)
$$
for every $f \colon \cX^n \to \reals$ with Lipschitz constant $1$ with respect to $d_{1,n}$. Since the exponent is inversely proportional to the dimension $n$, we need to have $r$ grow at least as $\sqrt{n}$ in order to guarantee a given value for the deviation probability. In particular, the higher the dimension $n$ is, the more we will need to ``inflate'' a given set $A \subset \cX^n$ to capture most of the probability mass. For these reasons, we seek a direct characterization of a much stronger concentration property, the so-called {\em dimension-free Gaussian concentration}.

Once again, let $(\cX,d,\mu)$ be a metric probability space. We say that $\mu$ has {\em dimension-free Gaussian concentration} if there exist constants $K,\kappa > 0$, such that for every $k \in \naturals$ and $r>0$
\begin{align}\label{eq:dimfree_concentration}
	A \subseteq \cX^k \text{ and } \mu^{\otimes k}(A) \ge 1/2 \quad \Longrightarrow \quad \mu^{\otimes k}(A_r) \ge 1 - Ke^{-\kappa r^2}
\end{align}
where the isoperimetric enlargement $A_r$ of a Borel set $A \subseteq \cX^k$ is defined in \eqref{eq:blowup_set} with respect to the metric
$d_k \equiv d_{2,k}$ defined according to \eqref{eq:product_metric}:
\begin{align*}
	A_r \deq \left\{ y^k \in \cX^k \colon \exists x^k \in A \; \text{s.t.} \; \sum^k_{i=1}d^2(x_i,y_i) < r^2\right\}.
\end{align*}

\begin{remark} As before, we are mainly interested in the constant $\kappa$ in the exponent. Thus, it is said that $\mu$ has dimension-free Gaussian concentration with constant $\kappa > 0$ if \eqref{eq:dimfree_concentration} holds with that $\kappa$ and some $K > 0$.
\end{remark}

\begin{remark} In the same spirit as Remark~\ref{rem:isoperimetry_weakened}, it may be desirable to relax  \eqref{eq:dimfree_concentration} to the following: there exists some $r_0 > 0$, such that for every $k \in \naturals$ and $r \geq r_0$,
\begin{align}\label{eq:dimfree_concentration_weakened}
	A \subseteq \cX^k \text{ and } \mu^{\otimes k}(A) \ge 1/2 \qquad \Longrightarrow \qquad \mu^{\otimes k}(A_r) \ge 1 - Ke^{-\kappa (r-r_0)^2}
\end{align}
(see, for example,  \cite[Remark~22.23]{Villani_newbook} or \cite[Proposition~3.3]{Gozlan}). The same considerations about (possibly) sharper constants that were stated in Remark~\ref{rem:isoperimetry_weakened} also apply here.
\end{remark}

In this section, we will show that dimension-free Gaussian concentration and ${\rm T}_2$ inequalities are equivalent. Before we get to that, here is an example of a ${\rm T}_2$ inequality:

\begin{theorem}[Talagrand \cite{Talagrand_Gaussian_T2}]\label{thm:Gaussian_T2} Let $\cX = \reals^n$ and
$d(x,y) = \| x - y \|$. Then $\mu =G^n$ satisfies a ${\rm T}_2(1)$ inequality.
\end{theorem}
\begin{proof} The proof starts for $n=1$: let $\mu = G$, let $\nu \in \cP(\reals)$ have
density $f$ with respect to $\mu$: $f = \frac{\d\nu}{\d\mu}$, and let $\Phi$ denote the standard Gaussian cdf, i.e.,
\begin{align*}
	\Phi(x) = \int^x_{-\infty}\gamma(y) \d y = \frac{1}{\sqrt{2\pi}} \int^x_{-\infty} \exp\left(-\frac{y^2}{2}\right)\d y, \quad \forall \, x \in \reals.
\end{align*}
If $X \sim G$, then (by Item~6 of Lemma~\ref{lm:W_properties}) the optimal coupling of $\mu = G$ and $\nu$, i.e.,
the one that achieves the infimum in
\begin{align*}
	W_2(\nu,\mu) = W_2(\nu,G) = \inf_{X \sim G, \, Y \sim \nu} \left(\expectation[(X-Y)^2]\right)^{1/2}
\end{align*}
is given by $Y = h(X)$ with $h = \sF^{-1}_\nu \circ \Phi$. Consequently,
\begin{align}
	W^2_2(\nu,G) = \expectation[(X-h(X))^2]
	= \int_{-\infty}^{\infty} \big(x-h(x)\big)^2 \gamma(x) \, \d x. \label{eq: W_2 squared}
\end{align}
Since $\d\nu = f\,\d\mu$ with $\mu = G$, and $\sF_\nu(h(x)) = \Phi(x)$ for every $x \in \reals$, we have
\begin{align}\label{eq:optimal_coupling_int}
\int_{-\infty}^x \gamma(y) \, \d y = \Phi(x)
= \sF_\nu(h(x))
= \int_{-\infty}^{h(x)} f \, \d \mu
= \int^{h(x)}_{-\infty} f(y) \gamma(y) \, \d y.
\end{align}
Differentiating both sides of \eqref{eq:optimal_coupling_int} with respect to $x$ gives
\begin{equation}
h'(x) f(h(x))\gamma(h(x)) = \gamma(x), \quad \forall \, x \in \reals.
\label{eq: resulting equation after differentiation}
\end{equation}
Since $h = \sF^{-1}_\nu \circ \Phi$, $h$ is a monotonically increasing function, and
$$ \lim_{x \rightarrow -\infty} h(x) = -\infty, \quad \lim_{x \rightarrow \infty} h(x) = \infty.$$
Moreover,
\begin{align}
	D(\nu \| G) &= D(\nu \| \mu) \nonumber \\
	&= \int_{\reals} \d\nu \, \ln \frac{\d\nu}{\d\mu} \nonumber \\
	&= \int_{-\infty}^{\infty} \ln \big(f(x)\big) \, \d \nu(x) \nonumber \\
	&= \int_{-\infty}^{\infty} f(x) \, \ln \big(f(x)\big) \, \d \mu(x) \nonumber \\
    &= \int_{-\infty}^{\infty} f(x) \, \ln \big(f(x)\big) \, \gamma(x) \, \d x \nonumber \\
    &= \int_{-\infty}^{\infty} f\big(h(x) \big) \, \ln \big(f\big(h(x)\big)\big) \, \gamma\big(h(x)\big)
        \, h'(x) \, \d x \nonumber \\
	&= \int_{-\infty}^{\infty} \ln \big(f(h(x))\big) \, \gamma(x) \, \d x, \label{eq:optimal_coupling_div}
\end{align}
where we have used
\eqref{eq: resulting equation after differentiation} to get the last equality. From
\eqref{eq: resulting equation after differentiation}
\begin{align*}
\ln \big(f\big(h(x)\big)\big) = \ln \left(\frac{\gamma(x)}{h'(x) \, \gamma\big(h(x)\big)} \right) = \frac{h^2(x)-x^2}{2}-\ln h'(x).
\end{align*}
Upon substituting this into \eqref{eq:optimal_coupling_div}, we get
\begin{align*}
	D(\nu \| \mu) &= \frac{1}{2} \int_{-\infty}^{\infty} \left[h^2(x) - x^2\right] \, \gamma(x) \, \d x - \int_{-\infty}^{\infty} \ln h'(x) \, \gamma(x) \, \d x \\
	&= \frac{1}{2} \int_{-\infty}^{\infty} \big(x-h(x)\big)^2 \, \gamma(x) \, \d x + \int_{-\infty}^{\infty} x\big(h(x)-x\big) \, \gamma(x) \, \d x \\
	& \qquad \qquad \qquad - \int_{-\infty}^{\infty} \ln h'(x) \, \gamma(x) \, \d x \\
	& \stackrel{(\text{a})}{=} \frac{1}{2} \int_{-\infty}^{\infty} \big(x-h(x)\big)^2 \, \gamma(x) \, d x
       + \int_{-\infty}^{\infty} (h'(x)-1) \, \gamma(x) \, \d x \\
    & \qquad \qquad \qquad   - \int_{-\infty}^{\infty} \ln h'(x) \, \gamma(x) \, \d x \\
	& \stackrel{(\text{b})}{\ge} \frac{1}{2} \int_{-\infty}^{\infty} \big(x - h(x)\big)^2 \, \gamma(x) \, \d x \\
	& \stackrel{(\text{c})}{=} \frac{1}{2} \; W^2_2(\nu,\mu)
\end{align*}
where equality~(a) relies on integration by parts, inequality~(b) follows from the inequality $\ln t \le t - 1$
for $t > 0$ and since $h$ is monotonic increasing and differentiable, and equality~(c) holds due to \eqref{eq: W_2 squared}.
This shows that $\mu = G$ satisfies ${\rm T}_2(1)$, so the proof of Theorem~\ref{thm:Gaussian_T2} for $n=1$ is complete.
Finally, this theorem is generalized for an arbitrary $n$ by tensorization via Proposition~\ref{prop:TC_tensorization}.
\end{proof}

We get in the following to the main result of this section, namely that dimension-free Gaussian concentration and ${\rm T}_2$ inequalities are equivalent:

\begin{theorem} Let $(\cX,d,\mu)$ be a metric probability space. Then, the following statements are equivalent:
	\begin{enumerate}
		\item $\mu$ satisfies ${\rm T}_2(c)$.
		\item $\mu$ has dimension-free Gaussian concentration with $\kappa = \frac{1}{2c}$.
	\end{enumerate}
\label{theorem: dimension-free Gaussian concentration and T_2 inequalities}
\end{theorem}

\begin{remark} As we will see, the implication $1) \Rightarrow 2)$ follows easily from the tensorization property of transportation-cost inequalities (Proposition~\ref{prop:TC_tensorization}). The reverse implication $2) \Rightarrow 1)$ is a nontrivial result, which was proved by Gozlan \cite{Gozlan} using an elegant probabilistic approach relying on the theory of large deviations \cite{Dembo_Zeitouni}.
\end{remark}

\begin{proof} We first prove that $1) \Rightarrow 2)$. Assume that $\mu$ satisfies ${\rm T}_2(c)$ on $(\cX,d)$. Fix some $k \in \naturals$ and consider the metric probability space $(\cX^k,d_{2,k},\mu^{\otimes k})$, where the metric $d_{2,k}$ is defined by \eqref{eq:product_metric} with $p=2$. By the tensorization property of transportation-cost inequalities (Proposition~\ref{prop:TC_tensorization}), the product measure $\mu^{\otimes k}$ satisfies ${\rm T}_2(c)$ on $(\cX^k,d_{2,k})$. Because the $L^2$ Wasserstein distance dominates the $L^1$ Wasserstein distance (by item~2 of Lemma~\ref{lm:W_properties}), $\mu^{\otimes k}$ also satisfies ${\rm T}_1(c)$ on $(\cX^k,d_{2,k})$. Therefore, by Proposition~\ref{prop:Marton_technique}, $\mu^{\otimes k}$ has Gaussian concentration \eqref{eq:isoperimetry_weakened} with respect to $d_{2,k}$ with constants $\kappa = 1/(2c), K = 1, r_0 = \sqrt{2 c \ln 2}$. Since this holds for every $k \in \naturals$, we conclude that $\mu$ indeed has dimension-free Gaussian concentration with constant
$\kappa = 1/(2c)$.
	
We now prove the converse implication $2) \Rightarrow 1)$. Suppose that $\mu$ has dimension-free Gaussian concentration with constant $\kappa > 0$, where for simplicity we assume that $r_0 = 0$ (the argument for the general case of $r_0 > 0$ is slightly more involved, and does not contribute much in the way of insight). Let $k \in \naturals$ be fixed, and consider the metric probability space $(\cX^k,d_{2,k},\mu^{\otimes k})$. Given $x^k \in \cX^k$, let $\emp_{x^k}$ be the {\em empirical measure}
\begin{align} \label{eq: empirical distribution}
	\emp_{x^k} = \frac{1}{k}\sum^k_{i=1}\delta_{x_i},
\end{align}
where $\delta_{x}$ denotes a Dirac measure (unit mass) concentrated at $x \in \cX$. Now consider a probability measure $\nu$ on $\cX$, and define the function $f_\nu \colon \cX^k \to \reals$ by
$f_\nu(x^k) \deq W_2(\emp_{x^k},\nu)$ for all $x^k \in \cX^k.$
We claim that this function is Lipschitz with respect to $d_{2,k}$ with Lipschitz constant $\frac{1}{\sqrt{k}}$.
To verify this, note that
\begin{align}
	\big|f_\nu(x^k) - f_\nu(y^k)\big| &= \left|W_2(\emp_{x^k},\nu) - W_2(\emp_{y^k},\nu)\right| \nonumber\\
	&\le W_2(\emp_{x^k},\emp_{y^k}) \label{eq:W_Lip_1}\\
	&= \inf_{\pi \in \Pi(\emp_{x^k},\emp_{y^k})}\left( \int_\cX d^2(x,y)\, \pi(\d x, \d y)\right)^{1/2} \label{eq:W_Lip_2}\\
	&\le \left( \frac{1}{k}\sum^k_{i=1}d^2(x_i,y_i) \right)^{1/2} \label{eq:W_Lip_3} \\
	&= \frac{1}{\sqrt{k}} \, d_{2,k}(x^k,y^k), \label{eq:W_Lip_5}
\end{align}
where
\begin{itemize}
	\item \eqref{eq:W_Lip_1} is by the triangle inequality;
	\item \eqref{eq:W_Lip_2} is by definition of $W_2$;
	\item \eqref{eq:W_Lip_3} uses the fact that the measure that places mass $1/k$ on each $(x_i,y_i)$
for $i \in \{1,\ldots,k\}$, is an element of $\Pi(\emp_{x^k},\emp_{y^k})$ (due to the definition of
an empirical distribution in \eqref{eq: empirical distribution}, the marginals of the above measure
are indeed $\emp_{x^k}$ and $\emp_{y^k}$); and
	\item \eqref{eq:W_Lip_5} uses the definition \eqref{eq:product_metric} of $d_{2,k}$.
\end{itemize}
Now let us consider the function $f_k(x^k) \deq W_2(\emp_{x^k},\mu)$, for which, as we have just seen, we have
\begin{align}  \label{eq:f_k Lip}
\| f_k \|_{\rm Lip,2} \le \frac{1}{\sqrt{k}}.
\end{align}
Let $X_1,\ldots,X_k$ be i.i.d.\ draws from $\mu$. Let $m_k$ denote some $\mu^{\otimes k}$-median of $f_k$. Then, by the assumed dimension-free Gaussian concentration property of $\mu$, Theorem~\ref{thm:Lipschitz_isoperimetry} yields that for every $r \ge 0$ and $k \in \naturals$
\begin{align}\label{eq:Gozlan_step1}
	\pr \Big( f_k(X^k) \ge m_k + r \Big)
    \le \exp \left(-\frac{\kappa r^2}{ \|f_k\|^2_{\text{Lip}, 2}} \right)
    \le \exp\Big( -\kappa k r^2\Big)
\end{align}
where the second inequality follows from \eqref{eq:f_k Lip}.

We now claim that every sequence $\{m_k\}^\infty_{k=1}$ of medians of the $f_k$'s converges to zero. If $X_1,X_2,\ldots$ are i.i.d.\ draws from $\mu$, then the sequence of empirical distributions $\{ \emp_{X^k}\}^\infty_{k=1}$ almost surely converges weakly to $\mu$ (this is known as Varadarajan's theorem \cite[Theorem~11.4.1]{Dudley_book}).
Therefore, since $W_2$ metrizes the topology of weak convergence together with the convergence of second moments (cf.~Lemma~\ref{lm:W_properties}), $\lim_{k \to \infty} W_2(\emp_{X^k},\mu) = 0$ almost surely. Hence, using the fact that
convergence almost surely implies convergence in probability, we have
\begin{align*}
	\lim_{k \to \infty}\pr\big( W_2(\emp_{X^k},\mu) \ge t \big) = 0, \qquad \forall \, t > 0.
\end{align*}
Consequently, every sequence $\{m_k\}$ of medians of the $f_k$'s converges to zero, as claimed. Combined with \eqref{eq:Gozlan_step1}, this implies that
\begin{align}\label{eq:Gozlan_step2}
	\limsup_{k \to \infty} \frac{1}{k}\ln \pr \Big( W_2(\emp_{X^k},\mu) \ge  r \Big) \le - \kappa r^2.
\end{align}
On the other hand, for a fixed $\mu$, the mapping $\nu \mapsto W_2(\nu,\mu)$ is lower semicontinuous in the topology of weak convergence of probability measures (cf.~Item~4 of Lemma~\ref{lm:W_properties}). Consequently, the set $\{ \mu: W_2(\emp_{X^k},\mu) > r \}$ is open in the weak topology, so by Sanov's theorem \cite[Theorem~6.2.10]{Dembo_Zeitouni}
\begin{align}\label{eq:Gozlan_step3}
	\liminf_{k \to \infty} \frac{1}{k}\ln \pr \Big( W_2(\emp_{X^k},\mu) \ge r\Big) \ge - \inf \left\{ D(\nu \| \mu) : W_2(\mu,\nu) > r \right\}.
\end{align}
Combining \eqref{eq:Gozlan_step2} and \eqref{eq:Gozlan_step3}, we get that
\begin{equation*}
\inf \bigl\{D(\nu \| \mu): \, W_2(\mu,\nu) > r \bigr\} \geq \kappa r^2
\end{equation*}
which then implies that $D(\nu \| \mu) \geq \kappa \, W^2_2(\mu,\nu)$. Upon rearranging, we obtain
$W_2(\mu,\nu) \le \sqrt{\bigl(\frac{1}{\kappa}\bigr) \, D(\nu \| \mu)}$, which is a ${\rm T}_2(c)$ inequality
with $c = \frac{1}{2\kappa}$.
This completes the proof of Theorem~\ref{theorem: dimension-free Gaussian concentration and T_2 inequalities}.
\end{proof}

\subsection{A grand unification: the HWI inequality}

At this point, we have seen two perspectives on the concentration of measure phenomenon: functional (through various log-Sobolev inequalities) and probabilistic (through transportation-cost inequalities). We now show that these two perspectives are, in a very deep sense, equivalent, at least in the Euclidean setting of $\reals^n$. This equivalence is captured by a striking inequality, due to Otto and Villani \cite{Otto_Villani}, which relates three measures of similarity between probability measures: the divergence, $L^2$ Wasserstein distance, and Fisher information distance. In the literature on optimal transport, the divergence between two probability measures $Q$ and $P$ is often denoted by $H(Q \| P)$ or $H(Q, P)$, due to its close links to the Boltzmann $H$-functional of statistical physics. For this reason, the inequality we have alluded to above has been dubbed the {\em HWI inequality}, where $H$ stands for the divergence, $W$ for the Wasserstein distance, and $I$ for the Fisher information distance (see
\eqref{eq:Fisher_dist} and \eqref{eq:Fisher_dist_a}).

As a warm-up, we first state a weaker version of the HWI inequality specialized to the Gaussian distribution, and give a self-contained information-theoretic proof following \cite{Wu}:

\begin{theorem} \label{theorem: weak HWI}
Let $G$ be the standard Gaussian probability distribution on $\reals$. Then, the inequality
\begin{align}\label{eq:weak_HWI}
D(P \| G) \le W_2(P,G) \sqrt{I(P \| G)},
\end{align}
where $W_2$ is the $L^2$ Wasserstein distance with respect to the absolute-value metric $d(x,y) = |x-y|$, holds for every Borel probability distribution $P$ on $\reals$, for which the right-hand side of \eqref{eq:weak_HWI} is finite.
\end{theorem}

\begin{proof} We first show the following:
\begin{lemma} \label{lemma: an upper bound on the divergence in terms of Wasserstein distance}
Let $X$ and $Y$ be a pair of real-valued random variables, and let $N \sim G$ be independent of $(X,Y)$.
Then, for every $t > 0$,
\begin{align}\label{eq:convolution_HW}
D(P_{X + \sqrt{t}N} \| P_{Y + \sqrt{t}N}) \le \frac{1}{2t} \, W^2_2(P_X,P_Y).
\end{align}
\end{lemma}
\begin{proof}
From the chain rule for divergence (see \cite[Theorem~2.5.3]{Cover_and_Thomas}), we have
\begin{align} \label{eq: first way to expand the divergence}
&D(P_{X,Y,X+\sqrt{t}N} \| P_{X,Y,Y+\sqrt{t}N} ) \ge D(P_{X+\sqrt{t}N} \| P_{Y+\sqrt{t}N})
\end{align}
and
\begin{align} \label{eq: second way to expand the divergence}
&D(P_{X,Y,X+\sqrt{t}N} \| P_{X,Y,Y+\sqrt{t}N}) \nonumber \\
&= D(P_{X + \sqrt{t}N|X,Y} \, \| \, P_{Y + \sqrt{t}N|X,Y} | P_{X,Y})
\nonumber \\
&\stackrel{\text{(a)}}{=} \expectation[D( \cN(X,t) \, \| \, \cN(Y,t)) \, | \, X, Y] \nonumber \\
&\stackrel{\text{(b)}}{=} \frac{1}{2t} \, \expectation[(X-Y)^2].
\end{align}
Note that equality~(a) holds since $N \sim G$ is independent of $(X, Y)$, and equality~(b)
is a special case of the equality
$$D\big( \cN(m_1, \sigma_1^2) \, \| \, \cN(m_2, \sigma_2^2) \big)
= \frac{1}{2} \, \ln \left(\frac{\sigma_2^2}{\sigma_1^2}\right) +
\frac{1}{2} \left( \frac{(m_1-m_2)^2}{\sigma_2^2} + \frac{\sigma_1^2}{\sigma_2^2} - 1 \right).$$
It therefore follows from \eqref{eq: first way to expand the divergence} and
\eqref{eq: second way to expand the divergence} that
\begin{align}\label{eq:pointwise_convolution_HW}
D(P_{X+\sqrt{t}N} \| P_{Y+\sqrt{t}N}) \le \frac{1}{2t} \, \expectation[(X-Y)^2]
\end{align}
where the left-hand side of \eqref{eq:pointwise_convolution_HW} only depends on the marginal distributions of $X$ and $Y$
(due to the independence of $(X,Y)$ and $N \sim G$).
Hence, taking the infimum of the right-hand side of \eqref{eq:pointwise_convolution_HW} with respect to all
$\mu \in \Pi(P_X,P_Y)$, we get \eqref{eq:convolution_HW} (see \eqref{eq:W_p_2}).
\end{proof}
We now proceed with the proof of Theorem~\ref{theorem: weak HWI}.
Let $X$ and $Y$ have distributions $P$ and $Q=G$, respectively.
For simplicity, we focus on the case where $X$ has zero mean and unit variance; the general case can be
handled similarly. Let
$F(t) \deq D(P_{X + \sqrt{t} N} \| P_{Y + \sqrt{t} N})$, for $t > 0$,
where $N \sim G$ is independent of the pair $(X,Y)$.
Then we have $F(0) = D(P \| G)$, and from \eqref{eq:convolution_HW}
\begin{align}\label{eq:Ft_bound}
F(t) \le \frac{1}{2t} \, W^2_2(P_X,P_Y) = \frac{1}{2t} \, W^2_2(P,G), \quad \forall \, t > 0.
\end{align}
Moreover, the function $F(t)$ is differentiable, and it follows from a result by Verd\'u
\cite[Eq.~(32)]{Verdu_mismatched_estimation} that
\begin{align}\label{eq:Ft_deriv_bound}
F'(t) &= \frac{1}{2t^2} \, \Big[ \mmse(X,t^{-1}) - \mse_Q(X,t^{-1}) \Big] \nonumber \\[0.1cm]
      &= \frac{1}{2t^2} \, \Big[ \mmse(X,t^{-1}) - \lmmse(X,t^{-1}) \Big], \quad \forall \, t > 0
\end{align}
where $\mmse(X,\cdot)$, $\mse_Q(X,\cdot)$ and $\lmmse(X,\cdot)$ have been defined in \eqref{eq:MMSE},
\eqref{eq:mismatched MSE} and \eqref{eq:LMMSE}, respectively. The second equality in \eqref{eq:Ft_deriv_bound}
holds due to \eqref{eq:optimal estimator in the Gaussian setting is linear} with $Q=G$ (recall that in the Gaussian
setting, the optimal estimator for minimizing the mean square error is linear). For every $t > 0$,
\begin{align}
&D(P||G) = F(0) \nonumber \\
&= -\big(F(t) - F(0)\big) + F(t) \nonumber \\
&= - \int^t_0 F'(s) \d s + F(t) \nonumber \\
&= \frac{1}{2} \int^t_0 \frac{1}{s^2} \left( \lmmse(X,s^{-1}) - \mmse(X,s^{-1}) \right) \d s + F(t) \label{eq:HWI_step1} \\
&\le \frac{1}{2} \int^t_0 \left( \frac{1}{s(s+1)} - \frac{1}{s(sJ(X)+1)} \right) \d s + \frac{1}{2t} \, W^2_2(P,G) \label{eq:HWI_step2} \\
&= \frac{1}{2} \left( \ln \frac{t J(X) + 1}{t+1} + \frac{W^2_2(P,G)}{t} \right) \label{eq:HWI_step3} \\
&\le \frac{1}{2} \left( \frac{t (J(X) - 1)}{t+1} + \frac{W^2_2(P,G)}{t} \right) \label{eq:HWI_step4} \\
& \le \frac{1}{2} \left( I(P \| G)\, t + \frac{W^2_2(P,G)}{t} \right) \label{eq:HWI_step5}
\end{align}
where
\begin{itemize}
\item \eqref{eq:HWI_step1} uses \eqref{eq:Ft_deriv_bound};
\item \eqref{eq:HWI_step2} uses \eqref{eq:Gaussian_LMMSE}, the Van Trees inequality \eqref{eq:VanTrees}, and \eqref{eq:Ft_bound};
\item \eqref{eq:HWI_step3} is an exercise in calculus;
\item \eqref{eq:HWI_step4} uses the inequality $\ln x \le x-1$ for $x>0$; and
\item \eqref{eq:HWI_step5} uses the formula \eqref{eq:Fisher_dist_Gs} (so $I(P||G) = J(X)-1$ since
$X \sim P$ has zero mean and unit variance, and one needs to substitute $s=1$ in \eqref{eq:Fisher_dist_Gs} to get
$G_s = G$), and the fact that $t \ge 0$.
\end{itemize}
Optimizing the choice of $t$ in \eqref{eq:HWI_step5}, we get \eqref{eq:weak_HWI}.
\end{proof}
\begin{remark} Note that the HWI inequality \eqref{eq:weak_HWI} together with the ${\rm T}_2$ inequality for the Gaussian distribution imply a weaker version of the log-Sobolev inequality \eqref{eq:Gross_LSI_1} (i.e., with a larger constant). Indeed, using the ${\rm T}_2$ inequality of Theorem~\ref{thm:Gaussian_T2} on the right-hand side of \eqref{eq:weak_HWI}, we get
\begin{align*}
D(P \| G) &\le W_2(P,G) \sqrt{I(P \| G)} \\
&\le \sqrt{2 D(P \| G)} \sqrt{I(P \| G)},
\end{align*}
which gives $D(P \| G) \le 2 I(P \| G)$.
It is not surprising that we end up with a suboptimal constant here as compared to \eqref{eq:Gross_LSI_1}:
the series of bounds leading up to \eqref{eq:HWI_step5} contributes a lot more slack than the single use
of the van Trees inequality \eqref{eq:VanTrees} in our proof of Stam's inequality (which, due to Proposition~\ref{proposition: equivalence of Gaussian LSI and Stam's inequality}, is equivalent to
the Gaussian log-Sobolev inequality of Gross).
\label{remark: HWI and TC implies a loosened version of Gaussian LSI}
\end{remark}
We are now ready to state the HWI inequality in its strong form:

\begin{theorem}[Otto--Villani \cite{Otto_Villani}]\label{thm:HWI} Let $P$ be a Borel probability measure on $\reals^n$ that is absolutely continuous with respect to the Lebesgue measure, and let the corresponding pdf $p$ be such that
	\begin{align}\label{eq:log_concave}
		\nabla^2 \ln \left(\frac{1}{p}\right) \succeq K I_n
	\end{align}
	for some $K \in \reals$ (where $\nabla^2$ denotes the Hessian, and the matrix inequality $A \succeq B$ means that $A-B$ is non-negative semidefinite). Then, every probability measure $Q \ll P$ satisfies
	\begin{align}\label{eq:HWI}
		D(Q \| P) \le W_2(Q,P) \sqrt{I(Q \| P)} - \frac{K}{2} \; W^2_2(Q,P).
	\end{align}
\end{theorem}
We omit the proof, which relies on deep structural properties of optimal transportation mappings achieving the infimum in the definition of the $L^2$ Wasserstein metric with respect to the Euclidean norm in $\reals^n$. (An alternative simpler proof was given later by Cordero--Erausquin \cite{Cordero_Gaussian_transport}.) We can, however, highlight a couple of key consequences (see \cite{Otto_Villani}):
\begin{enumerate}
	\item Suppose that $P$, in addition to satisfying the conditions of Theorem~\ref{thm:HWI}, also satisfies a ${\rm T}_2(c)$ inequality. Using this fact in \eqref{eq:HWI}, we get
	\begin{align}\label{eq:HWI_T2}
		D(Q \| P) \le \sqrt{2 c D(Q \| P)} \sqrt{I(Q \| P)} - \frac{K}{2} \; W^2_2(Q,P).
	\end{align}
If the pdf $p$ of $P$ is log-concave, so that \eqref{eq:log_concave} holds with $K = 0$, then \eqref{eq:HWI_T2} implies the inequality
\begin{align}\label{eq:HWI_LSI}
	D(Q \| P) \le 2c\, I(Q \| P)
\end{align}
for every $Q$ such that $Q \ll P$. This is an Euclidean log-Sobolev inequality that is similar to the one satisfied by $P =G^n$
(see Remark~\ref{remark: HWI and TC implies a loosened version of Gaussian LSI}). However, note that the constant in front of the Fisher information distance $I(\cdot\|\cdot)$ on the right-hand side of \eqref{eq:HWI_LSI} is suboptimal, as can be verified by letting $P = G^n$, which satisfies ${\rm T}_2(1)$; going through the above steps, as we know from Section~\ref{sec:Gaussian_LSI} (in particular, see \eqref{eq:Gross_LSI_1}), the optimal constant should be $\frac{1}{2}$, so the one in \eqref{eq:HWI_LSI} is off by a factor of $4$. On the other hand, it is quite remarkable that, up to constants, the Euclidean log-Sobolev and ${\rm T}_2$ inequalities are equivalent.
	\item If the pdf $p$ of $P$ is {\em strongly log-concave}, i.e., if \eqref{eq:log_concave} holds with some
$K > 0$, then $P$ satisfies the Euclidean log-Sobolev inequality with constant $\frac{1}{K}$. Indeed, using the simple inequality $ab \le \frac{a^2+b^2}{2}$ for every $a, b \in \reals$, we have from \eqref{eq:HWI}
\begin{align*}
	D(Q \| P) &\le \sqrt{K}W_2(Q,P) \sqrt{\frac{I(Q \| P)}{K}} - \frac{K}{2} \; W^2_2(Q,P) \\
	&\le \frac{1}{2K} \; I(Q \| P),
\end{align*}
which shows that $P$ satisfies the Euclidean $\LSI\bigl(\frac{1}{K}\bigr)$ inequality. In particular, the standard Gaussian distribution $P = G^n$ satisfies \eqref{eq:log_concave} with $K=1$, so we even get the right constant. In fact, the statement that \eqref{eq:log_concave} with $K > 0$ implies Euclidean $\LSI\bigl(\frac{1}{K}\bigr)$
was first proved in 1985 by Bakry and Emery \cite{Bakry_Emery} using very different means.
\end{enumerate}

\section{Extension to non-product distributions}
\label{sec:nonproduct}

Our focus in this chapter has been mostly on functions of independent random variables. However, there is extensive literature on the concentration of measure for weakly dependent random variables. In this section, we describe (without proof) a few results along this direction that explicitly use information-theoretic methods. The examples we give are by no means exhaustive, and are only intended to show that, even in the case of dependent random variables, the underlying ideas are essentially the same as in the independent case.

The basic scenario is exactly as before: We have $n$ random variables $X_1,\ldots,X_n$ with a given joint distribution $P$ (which is now not necessarily of a product form, i.e., $P = P_{X^n}$ may not be equal to
$P_{X_1} \otimes \ldots \otimes P_{X_n}$), and we are interested in the concentration properties of some function $f(X^n)$.

\subsection{Samson's transportation-cost inequalities for dependent random variables}

Samson \cite{Samson} has developed a general approach for deriving transportation-cost inequalities for dependent random variables that revolves around a certain $L^2$ measure of dependence. Given the distribution $P = P_{X^n}$ of $(X_1,\ldots,X_n)$, consider an upper triangular matrix $\Delta \in \reals^{n \times n}$, such that $\Delta_{i,j} = 0$ for $i > j$, $\Delta_{i,i} = 1$ for all $i$, and for $i < j$
\begin{align}  \label{eq: elements of Delta in the upper triangular}
\Delta_{i,j} &= \sup_{x_i,x'_i} \, \sup_{x^{i-1}} \sqrt{\left\| P_{X^n_j|X_i = x_i,X^{i-1}=x^{i-1}} - P_{X^n_j|X_i = x'_i,X^{i-1}=x^{i-1}} \right\|_{\rm TV}}.
\end{align}
Note that in the special case where $P$ is a product measure, the matrix $\Delta$ is equal to the $n \times n$
identity matrix.
Let $\| \Delta \|$ denote the operator norm of $\Delta$, i.e.,
\begin{align*}
	\|\Delta \| \deq \sup_{v \in \reals^n \setminus \{0\}} \frac{\| \Delta v \|}{\| v \|} =
\sup_{v \in \reals^n \colon \|v\| =1} \| \Delta v \|.
\end{align*}
Following Marton \cite{Marton_contracting_MCs}, Samson \cite{Samson} considered a Wasserstein-type distance on the space of probability measures on $\cX^n$. For every pair of probability measures $Q$ and $R$ on $\cX^n$, let $\Pi(Q,R)$ denote the set of all probability measures on $\cX^n \times \cX^n$ with marginals $Q$ and $R$; the following non-negative quantity is defined in \cite{Samson}
\begin{align} \label{eq:Samson's quantity}
	d_2(Q,R) \deq \inf_{\pi \in \Pi(Q,R)} \sup_{\alpha} \int \sum^n_{i=1}\alpha_i(y) 1_{\{x_i \neq y_i\}} \pi(\d x^n, \d y^n),
\end{align}
where $\sup_{\alpha}$ refers to the supremum over all vector-valued functions $\alpha \colon \cX^n \to \reals^n$ where $\alpha = (\alpha_1,\ldots,\alpha_n)$ is a vector of positive functions, and $\expectation_R \left[ \| \alpha(Y^n) \|^2\right] \le 1.$

\begin{remark}
Note that $d_2(Q,Q)=0$; however, in general, we have $d_2(Q,R) \neq d_2(R,Q)$ due to the difference in the two
conditions $\expectation_R \left[ \| \alpha(Y^n) \|^2\right] \le 1$ and $\expectation_Q \left[ \| \alpha(Y^n) \|^2\right] \le 1$
involved in the definition of $d_2(Q,R)$ and $d_2(R,Q)$, respectively. Therefore, $d_2$ is {\em not} a distance.
\end{remark}

The main result of \cite{Samson} goes as follows (see \cite[Theorem~1]{Samson}):
\begin{theorem}\label{thm:Samson} The probability distribution $P$ of $X^n$ satisfies the following transportation-cost inequality:
	\begin{align}\label{eq:Samson_TC}
		d_2(Q,P) \le \|\Delta \| \sqrt{2 D(Q \| P)}
	\end{align}
	for all $Q \ll P$. Furthermore,
    \begin{align}\label{eq:Samson_TC2}
		d_2(P,Q) \le \|\Delta \| \sqrt{2 D(Q \| P)}.
	\end{align}
\end{theorem}

\vspace*{0.2cm}
In the following, we examine some implications of Theorem~\ref{thm:Samson}.
\begin{enumerate}
	\item Let $\cX = [0,1]$. Theorem~\ref{thm:Samson} implies that every probability measure $P$ on the unit cube $\cX^n = [0,1]^n$ satisfies the following Euclidean log-Sobolev inequality: for an arbitrary smooth convex function $f \colon [0,1]^n \to \reals$,
	\begin{align}\label{eq:Samson_LSI}
		D \big( P^{(f)} \big\| P \big) \le 2 \|\Delta\|^2\, \expectation_P^{(f)}\left[ \left\| \nabla f(X^n) \right\|^2 \right]
	\end{align}
	(this follows from a combination of \cite[Eq.~(2.13)]{Samson} and equalities \eqref{eq:ent_vs_D} and \eqref{eq:another equality with Gamma f}). The same method as the one we used to prove Proposition~\ref{prop:Gauss_smooth_Lipschitz} and Theorem~\ref{thm:Gauss_Lipschitz} can be applied to obtain, from \eqref{eq:Samson_LSI}, the following concentration inequality for every convex function $f \colon [0,1]^n \to \reals$ with $\| f \|_{\rm Lip} \le 1$:
	\begin{align}\label{eq:Samson_concentration_looser}
		\pr \Big( f(X^n) \ge \expectation f(X^n) + r \Big) \le \exp\left( - \frac{r^2}{8 \|\Delta \|^2}\right), \qquad \forall \, r \ge 0.
	\end{align}
However, an adaptation of the approach by Bobkov and G\"otze \cite{Bobkov_Gotze_expint} that is used to prove
Theorem~\ref{thm:Bobkov_Gotze} and Corollary~\ref{cor:Lipschitz_Gauss_con_revisited} gives the
following improved concentration inequality for every smooth convex function $f \colon [0,1]^n \to \reals$
with $\| \nabla f \| \le 1$ $P$-a.s. (see \cite[Corollary~3]{Samson})
    \begin{align}\label{eq:Samson_concentration}
		\pr \Big( f(X^n) \ge \expectation f(X^n) + r \Big) \le \exp\left( - \frac{r^2}{2 \|\Delta \|^2}\right), \qquad \forall \, r \ge 0.
	\end{align}
Furthermore, inequality \eqref{eq:Samson_concentration} also holds for an arbitrary smooth concave function
$f \colon [0,1]^n \to \reals$ such that $\expectation_P \big[ \| \nabla f \|^2 \big] \le 1$.
	\item The operator norm $\| \Delta \|$ in \eqref{eq:Samson_TC}--\eqref{eq:Samson_concentration} is weakly dependent on $n$ whenever the dependence between the $X_i$'s is sufficiently weak. For instance, if $X_1,\ldots,X_n$ are independent then $\Delta = I_{n \times n}$, and
$\| \Delta \|=1$ independently of $n$. In this case, \eqref{eq:Samson_TC} becomes
	\begin{align*}
		d_2(Q,P) \le \sqrt{2 D(Q \| P)},
	\end{align*}
and we recover the usual concentration inequalities for Lipschitz functions. To see some examples with dependent random variables, suppose that $X_1,\ldots,X_n$ is a Markov chain, i.e., for each $i$, $X^n_{i+1}$ is conditionally independent of $X^{i-1}$ given $X_i$.
In that case, from \eqref{eq: elements of Delta in the upper triangular}, the upper triangular part of $\Delta$ is given by
\begin{align*}
\Delta_{i,j} = \sup_{x_i,x'_i} \sqrt{ \left\| P_{X_j|X_i=x_i}-P_{X_j|X_i=x'_i} \right\|_{\rm TV}} \, , \qquad i<j
\end{align*}
and $\|\Delta\|$ will be independent of $n$ under suitable ergodicity assumptions on the Markov chain $X_1,\ldots,X_n$. For instance, suppose that the Markov chain is homogeneous, i.e., the conditional probability distribution $P_{X_i|X_{i-1}}$ $(i > 1)$ is independent of $i$, and that
\begin{align*}
	\sup_{x_i,x'_i} \| P_{X_{i+1}|X_i = x_i} - P_{X_{i+1}|X_i = x'_i} \|_{\rm TV} \le 2\rho
\end{align*}
for some $\rho < 1$. Then it can be shown (see \cite[Eq.~(2.5)]{Samson}) that
\begin{align*}
	\| \Delta \| &\le \sqrt{2} \left( 1 + \sum^{n-1}_{k=1}\rho^{k/2} \right) \\
	& \le \frac{\sqrt{2}}{1-\sqrt{\rho}}.
\end{align*}
More generally, following Marton \cite{Marton_contracting_MCs}, we will say that the (not necessarily homogeneous) Markov chain $X_1,\ldots,X_n$ is {\em contracting} if, for every $i$,
\begin{align*}
	\delta_i \deq \sup_{x_i,x'_i} \| P_{X_{i+1}|X_i = x_i} - P_{X_{i+1}|X_i = x'_i} \|_{\rm TV} < 1.
\end{align*}
In this case, it is shown in \cite[pp.~422--424]{Samson} that $\| \Delta \|$ can be also bounded independently of $n$ as
\begin{align*}
\| \Delta \| \le \frac{1}{1-\sqrt{\delta}}, \qquad \text{where } \qquad \delta \triangleq \max_{i} \delta_i.
\end{align*}

\end{enumerate}

\subsection{Marton's transportation-cost inequalities for $L^2$ Wasserstein distance}

Another approach to obtaining concentration of measure inequalities for dependent random variables, due to Marton \cite{Marton_Euclidean,Marton_Euclidean_correction}, relies on another measure of dependence that pertains to the sensitivity of the conditional distributions of $X_i$ given $\bar{X}^i$ to the particular realization $\bar{x}^{i}$ of $\bar{X}^i$. The results of \cite{Marton_Euclidean,Marton_Euclidean_correction} are set in the Euclidean space $\reals^n$, and center around a transportation-cost inequality for the $L^2$ Wasserstein distance
\begin{align}\label{eq:W2}
W_2(P,Q) \deq \inf_{X^n \sim P, \, Y^n \sim Q} \sqrt{\expectation \| X^n - Y^n \|^2},
\end{align}
where $\| \cdot \|$ denotes the Euclidean norm.

We will state a special case of Marton's results (a more general development considers conditional distributions of
$(X_i \colon i \in S)$ given $(X_j \colon j \in S^c)$ for a suitable system of sets $S \subset \{1,\ldots,n\}$). Let
$P$ be a probability measure on $\reals^n$ which is absolutely continuous with respect to the Lebesgue measure. For
each $x^n \in \reals^n$ and  $i \in \{1,\ldots,n\}$, we denote by $\bar{x}^i$ the vector in $\reals^{n-1}$ obtained
by deleting the $i$th coordinate of $x^n$:
$$
\bar{x}^i = (x_1,\ldots,x_{i-1},x_{i+1},\ldots,x_n).
$$
Following Marton \cite{Marton_Euclidean}, the probability measure $P$ is {\em $(1-\delta)$-contractive}, with $\delta \in (0, 1)$,
if for every $y^n,z^n \in \reals^n$
\begin{align}\label{eq:Marton_contractive}
	\sum^n_{i=1}W^2_2(P_{X_i|\bar{X}^{i}=\bar{y}^i}, P_{X_i|\bar{X}^i = \bar{z}^i}) \le (1-\delta) \| y^n - z^n \|^2.
\end{align}
\begin{remark} Marton's contractivity condition \eqref{eq:Marton_contractive} is closely related to the so-called
{\em Dobrushin--Shlosman's strong mixing condition} \cite{Dobrushin--Shlosman} from mathematical statistical physics.
\end{remark}
\begin{theorem}[Marton \cite{Marton_Euclidean,Marton_Euclidean_correction}]\label{thm:Marton_contractive} Suppose that $P$ is absolutely continuous with respect to the Lebesgue measure on $\reals^n$ and also $(1-\delta)$-contractive, and that the conditional distributions $P_{X_i|\bar{X}^i}$, $i \in \{1,\ldots,n\}$, have the following properties:
	\begin{enumerate}
		\item For each $i$, the function $x^n \mapsto p_{X_i|\bar{X}^i}(x_i|\bar{x}^i)$ is continuous, where $p_{X_i|\bar{X}^i}(\cdot|\bar{x}^i)$ denotes the univariate probability density function of $P_{X_i|\bar{X}^i = \bar{x}^i}$.
		\item For each $i$ and each $\bar{x}^i \in \reals^{n-1}$, $P_{X_i|\bar{X}^i = \bar{x}^{i-1}}$ satisfies ${\rm T}_2(c)$ with respect to the $L^2$ Wasserstein distance \eqref{eq:W2} (cf.~Definition~\ref{def:TC}).
	\end{enumerate}
	Then, for every probability measure $Q$ on $\reals^n$, we have
	\begin{align}
		W_2(Q,P) \le \left(\frac{K}{\sqrt{\delta}}+1\right) \sqrt{2c D(Q \| P)},
	\end{align}
	where $K > 0$ is an absolute constant. In other words, every $P$ satisfying the conditions of the theorem admits a ${\rm T}_2(c')$ inequality with $$c' = \left(\frac{K}{\sqrt{\delta}}+1\right)^2 c.$$
\end{theorem}

The contractivity criterion \eqref{eq:Marton_contractive} is not easy to verify in general. Let us mention one sufficient condition \cite{Marton_Euclidean}. Let $p$ denote the probability density of $P$, and suppose that it takes the form
\begin{align}
	p(x^n) = \frac{1}{Z}\exp\left( - \Psi(x^n)\right)
\end{align}
for some $C^2$ function $\Psi \colon \reals^n \to \reals$, where $Z$ is the normalization factor. For every $x^n,y^n \in \reals^n$, let us define a matrix $B(x^n,y^n) \in \reals^{n \times n}$ by
\begin{align}
	B_{i,j}(x^n,y^n) \deq \begin{cases}
	\nabla^2_{i,j}\Psi(x_i \odot \bar{y}^i), & i \neq j \\
	0, & i = j
\end{cases}
\end{align}
where $\nabla^2_{i,j} F$ denotes the $(i,j)$ entry of the Hessian matrix of $F \in C^2(\reals^n)$, and $x_i \odot \bar{y}^i$ denotes the $n$-tuple obtained by replacing the deleted $i$th coordinate in $\bar{y}^i$ with $x_i$:
$$
x_i \odot \bar{y}^i = (y_1,\ldots,y_{i-1},x_i,y_{i+1},\ldots,y_n).
$$
For example, if $\Psi$ is a sum of one-variable and two-variable terms
\begin{align*}
	\Psi(x^n) = \sum^n_{i=1}V_i(x_i) + \sum_{i < j}b_{i,j}x_ix_j
\end{align*}
for some smooth functions $V_i : \reals \to \reals$ and some constants $b_{i,j} \in \reals$, which is often the case in statistical physics, then the matrix $B$ is independent of $x^n,y^n$, and has off-diagonal entries $b_{i,j}$, $i \neq j$. Then (see \cite[Theorem~2]{Marton_Euclidean}), the conditions of Theorem~\ref{thm:Marton_contractive} are satisfied provided the following holds:
\begin{enumerate}
	\item For each $i$ and $\bar{x}^i \in \reals^{n-1}$, the conditional probability distributions $P_{X_i|\bar{X}^i=\bar{x}^{i}}$ satisfy the Euclidean log-Sobolev inequality
	\begin{align*}
		D(Q \| P_{X_i|\bar{X}^{i}=\bar{x}^i}) \le \frac{c}{2} \; I(Q \| P_{X_i|\bar{X}^i = \bar{x}^i}),
	\end{align*}
	where $I(\cdot\|\cdot)$ is the Fisher information distance, cf.~\eqref{eq:Fisher_dist} for the definition.
	\item The operator norms of $B(x^n,y^n)$ are uniformly bounded as
	\begin{align*}
		\sup_{x^n,y^n} \left\| B(x^n,y^n) \right\|^2 \le \frac{1-\delta}{c^2}.
	\end{align*}
\end{enumerate}
We also refer the reader to more recent follow-up work by Marton \cite{Marton_arXiv2009,Marton_2013}, which further elaborates on the theme of studying the concentration properties of dependent random variables by focusing on the conditional probability distributions $P_{X_i|\bar{X}^i}$, $i=1,\ldots,n$. These papers describe sufficient conditions on the joint distribution $P$ of $X_1,\ldots,X_n$, such that, for every other distribution $Q$,
\begin{align}\label{eq:dependent_erasure_bound}
	D(Q \| P) \le K(P) \cdot D^{-}(Q \| P),
\end{align}
where $D^-(\cdot\|\cdot)$ is the erasure divergence (cf.\ \eqref{eq: erasure divergence} for the definition), and the $P$-dependent constant $K(P) > 0$ is controlled by suitable contractivity properties of $P$. At this point, the utility of a tensorization inequality like \eqref{eq:dependent_erasure_bound} should be clear: each term in the erasure divergence
$$
D^-(Q \| P) = \sum^n_{i=1} D(Q_{X_i|\bar{X}^i} \| P_{X_i|\bar{X}^i} | Q_{\bar{X}^i})
$$
can be handled by appealing to appropriate log-Sobolev inequalities or transportation-cost inequalities for probability measures on $\cX$ (indeed, one can just treat $P_{X_i|\bar{X}^i=\bar{x}^i}$ for each fixed $\bar{x}^i$ as a probability measure on $\cX$, in just the same way as with $P_{X_i}$ before), and then these ``one-dimensional'' bounds can be assembled together to derive concentration for the original ``$n$-dimensional'' distribution.

\section{Applications in information theory and related topics}
\label{sec:applications}

\subsection{The blowing-up lemma}

The first explicit invocation of the concentration of measure phenomenon in an information-theoretic context appears in the work of Ahlswede et al.~\cite{Ahlswede_Gacs_Korner,Ahlswede_Dueck}. These authors have shown that the following result, now known as the {\em blowing-up lemma} (see, e.g., \cite[Lemma~1.5.4]{Csiszar_Korner_book}), provides a versatile tool for proving strong converses in a variety of scenarios, including some multiterminal problems:
\begin{lemma}\label{lm:blowup} For every two finite sets $\cX$ and $\cY$ and every positive sequence
$\eps_n \rightarrow 0$, there exist positive sequences $\delta_n, \eta_n \rightarrow 0$, such that
the following holds: For every discrete memoryless channel (DMC) with input alphabet $\cX$, output
alphabet $\cY$, and transition probabilities $T(y|x), x \in \cX, y \in \cY$, and every
$n \in \naturals$, $x^n \in \cX^n$, and $B \subseteq \cY^n$,
\begin{align}\label{eq:blowup}
	T^n(B|x^n) \ge \exp\left(-n\eps_n\right) \qquad \Longrightarrow \qquad T^n(B_{n\delta_n}|x^n) \ge 1 - \eta_n.
\end{align}
Here, for an arbitrary $B \subseteq \cY^n$ and $r>0$, the set $B_r$ denotes the {\em $r$-blowup} of $B$ (see
the definition in \eqref{eq:blowup_set}) with respect to the Hamming metric
\begin{align}
\label{eq:Hamming metric}
d_n(y^n,u^n) \deq \sum^n_{i=1}1_{\{y_i \neq u_i\}}, \qquad \forall \, y^n,u^n \in \cY^n.
\end{align}
\end{lemma}

The proof of the blowing-up lemma, given in \cite{Ahlswede_Gacs_Korner}, was rather technical and made use of a delicate isoperimetric inequality for discrete probability measures on a Hamming space, due to Margulis \cite{Margulis}. Later, the same result was obtained by Marton \cite{Marton_blowup} using purely information-theoretic methods. We will use a sharper, ``nonasymptotic'' version of the blowing-up lemma, which is more in the spirit of the modern viewpoint on the concentration of measure (cf.\ Marton's follow-up paper \cite{Marton_dbar}):

\begin{lemma}\label{lm:my_blowup} Let $X_1,\ldots,X_n$ be $n$ independent random variables taking values in a finite set $\cX$. Then, for every $A \subseteq \cX^n$ with $P_{X^n}(A) > 0$,
	\begin{align}\label{eq:my_blowup}
		P_{X^n}(A_r) &\ge 1 - \exp\left[-\frac{2}{n}\Biggl(r - \sqrt{\frac{n}{2}
        \ln \left(\frac{1}{P_{X^n}(A)}\right)} \, \Biggr)^2\right], \nonumber\\
& \qquad \qquad \qquad \qquad \forall
        \, r > \sqrt{\frac{n}{2} \ln \left(\frac{1}{P_{X^n}(A)}\right)}.
	\end{align}
\end{lemma}
\begin{proof} Let $P_n$ denote the product measure $P_{X^n} = P_{X_1} \otimes \ldots \otimes P_{X_n}$. By Pinsker's inequality, every $\mu \in \cP(\cX)$ satisfies ${\rm T}_1(1/4)$ on $(\cX,d)$ where $d=d_1$ is the Hamming metric. By Proposition~\ref{prop:TC_tensorization_2}, the product measure $P_n$ satisfies ${\rm T}_1(n/4)$ on the product space $(\cX^n,d_n)$, i.e., for every $\mu_n \in \cP(\cX^n)$,
	\begin{align}\label{eq:Pinsker_tensorized}
		W_1(\mu_n ,P_n) \le \sqrt{\frac{n}{2} \, D(\mu_n \| P_n)}.
	\end{align}
The statement of the lemma follows from the proof of Proposition~\ref{prop:Marton_technique}. More precisely,
applying \eqref{eq:Marton_conc_1} to the probability measure $P_{X^n}$ with $c = \frac{n}{4}$ gives
$$r \le \sqrt{\frac{n}{2} \, \ln \frac{1}{P_{X^n}(A)}} + \sqrt{\frac{n}{2} \, \ln \frac{1}{1-P_{X^n}(A_r)}} \, , \quad \forall \, r >0$$
and \eqref{eq:my_blowup} holds by rearranging terms.
\end{proof}

We can now easily prove Lemma~\ref{lm:blowup}. To this end, given a positive sequence $\{\varepsilon_n\}_{n=1}^{\infty}$
that tends to zero, let us choose a positive sequence $\{\delta_n\}^\infty_{n=1}$ such that
\begin{align*}
\delta_n > \sqrt{\frac{\varepsilon_n}{2}}, \quad \delta_n \xrightarrow{n \to \infty} 0, \quad
\eta_n \deq \exp\left( - 2n\left(\delta_n - \sqrt{\frac{\eps_n}{2}}\right)^2\right) \xrightarrow{n \to \infty} 0.
\end{align*}
These requirements can be satisfied, e.g., by the setting
\begin{align} \label{eq:choice of sequences for the proof of the blowing-up lemma}
\delta_n \triangleq \sqrt{\frac{\varepsilon_n}{2}} + \sqrt{\frac{\alpha \, \ln n}{n}}, \quad
\eta_n = \frac{1}{n^{2\alpha}}, \quad \forall \, n \in \naturals
\end{align}
where $\alpha > 0$ can be made arbitrarily small. Using this selection for $\{\delta_n\}_{n=1}^{\infty}$ in
\eqref{eq:choice of sequences for the proof of the blowing-up lemma}, we get \eqref{eq:blowup} with the
$r_n$-blowup of the set $B$ where $r_n \triangleq n \delta_n$. Note that the above selection of $\delta_n$
does not depend on the transition probabilities of the DMC with input $\cX$ and output $\cY$ (the correspondence
between Lemmas~\ref{lm:blowup} and~\ref{lm:my_blowup} is given by $P_{X^n}=T^n(\cdot | x^n)$ where
$x^n \in \mathcal{X}^n$ is arbitrary).

\subsection{Strong converse for the degraded broadcast channel}

We are now ready to demonstrate how the blowing-up lemma can be used to obtain strong converses. Following \cite{Csiszar_Korner_book}, from this point on, we will use the  notation $T \colon \cU \to \cV$ for a DMC with input alphabet $\cU$, output alphabet $\cV$, and transition probabilities $T(v|u), u \in \cU, v \in \cV$.

Consider the problem of characterizing the capacity region of a 2-user discrete memoryless degraded broadcast
channel (DM-DBC) with independent messages, defined as follows:
\begin{definition}[DM-DBC] \label{def:DM-DBC}
Let $\cX$, $\cY$ and $\cZ$ be finite sets. A DM-DBC is specified by a pair of DMCs
$T_1 \colon \cX \to \cY$ and $T_2 \colon \cX \to \cZ$ where there exists a DMC $T_3 \colon \cY \to \cZ$ such that
\begin{align} \label{eq: stochastically degraded broadcast channel}
	T_2(z|x) &= \sum_{y \in \cY} T_1(y|x) \, T_3(z|y), \qquad \forall \, x \in \cX, \, z \in \cZ.
\end{align}
(More precisely, this is a {\em stochastically degraded} broadcast channel -- see, e.g.,
\cite[Section~15.6]{Cover_and_Thomas} and \cite[Section~5.4]{ElGamal_Kim_NIT}; a {\em physically
degraded} broadcast channel has the probability law
$$\pr(y,z|x) = T_1(y|x) \, T_3(z|y), \quad \forall \, x \in \cX, \, y \in \cY, \, z \in \cZ$$
so, to every DM-DBC, there is a corresponding physically
degraded broadcast channel with the same conditional marginal distributions.
\end{definition}
\begin{definition}[Codes] \label{def:codes for DM-DBC}
Given $n,M_1,M_2 \in \naturals$, an {\em $(n,M_1,M_2)$-code} $\cC$ for the broadcast channel consists of the following objects:
\begin{enumerate}
	\item An {\em encoding map} $f_n \colon \{1,\ldots,M_1\} \times \{1,\ldots,M_2\} \to \cX^n$;
	\item A collection $\cD_1$ of $M_1$ disjoint {\em decoding sets} for receiver~1
    $$D_{1,i}\subset \cY^n, \quad i \in \{1, \ldots, M_1\}$$ and
	a collection $\cD_2$ of $M_2$ disjoint decoding sets for receiver~2
    $$D_{2,j} \subset \cZ^n, \quad j  \in \{1, \ldots, M_2\}.$$
\end{enumerate}
Given $\eps_1,\eps_2 \in (0,1)$, we say that the code $\cC = (f_n,\cD_1,\cD_2)$ is an {\em $(n,M_1,M_2,\eps_1,\eps_2)$-code} if
\begin{align*}
	\max_{1 \le i \le M_1} \max_{1 \le j \le M_2} T^n_1\Big(D^c_{1,i}\Big|f_n(i,j)\Big) &\le \eps_1, \\
	\max_{1 \le i \le M_1} \max_{1 \le j \le M_2} T^n_2\Big(D^c_{2,j} \Big| f_n(i,j)\Big) &\le \eps_2.
\end{align*}
\end{definition}
In other words, the maximal probability of error criterion is used in Definition~\ref{def:codes for DM-DBC}.
Note that, for general multiuser channels, the capacity region with respect to the maximal probability
of error may be strictly smaller than the capacity region with respect to the average probability of error \cite{Dueck_ave_max};
nevertheless, these two capacity regions are identical for discrete memoryless broadcast channels \cite{Willems_BC}.
\begin{definition}[Achievable rates] \label{def:achievable rates}
A pair of rates $(R_1,R_2)$ (in nats per channel use) is said to be {\em $(\eps_1,\eps_2)$-achievable} if for
every $\delta > 0$, there exists an $(n,M_1,M_2,\eps_1,\eps_2)$-code (for a sufficiently large block length $n$) such that
\begin{align*}
	\frac{1}{n} \, \ln M_k \ge R_k - \delta, \qquad k = 1,2.
\end{align*}
Likewise, $(R_1,R_2)$ is said to be {\em achievable} if it is $(\eps_1,\eps_2)$-achievable for all $0 < \eps_1,\eps_2 \le 1$
(according to the criterion of the maximal probability of error in Definition~\ref{def:codes for DM-DBC}, this is equivalent
to the requirement that $(R_1,R_2)$ is $(\eps_1,\eps_2)$-achievable for arbitrarily small values of $\eps_1, \eps_2 > 0$).
Let $\cR(\eps_1,\eps_2)$ denote the set of all $(\eps_1,\eps_2)$-achievable
rates, and let $\cR$ denote the set of all achievable rates. Clearly,
\begin{align*}
	\cR = \bigcap_{(\eps_1,\eps_2) \in (0,1)^2} \cR(\eps_1,\eps_2)
\end{align*}
is the capacity region.
\end{definition}
The capacity region of a discrete memoryless broadcast channel only depends on its conditional marginal
distributions (see, e.g., \cite[Lemma~5.1]{ElGamal_Kim_NIT}). This observation implies that the capacity
region of a general DM-DBC is the same as that of a discrete memoryless physically degraded broadcast
channel when they both have the same conditional marginal distributions. Consequently, one can assume
w.l.o.g. that $X \to Y_1 \to Y_2$ forms a Markov chain (see, e.g., \cite[Section~5.4]{ElGamal_Kim_NIT}).

The capacity region of the DM-DBC is fully known. The achievability of this rate region was demonstrated by
Cover \cite{Cover72} and Bergmans \cite{Bergmans73} via the use of superposition coding. Consequently,
weak converses have been proved by Wyner \cite{Wyner73}, Gallager \cite{Gallager74}, and Ahlswede and K\"orner
\cite{Ahlswede_Korner}, and a strong converse for the capacity region of the DM-DBC has been proved by Ahlswede,
G\'acs and K\"orner \cite{Ahlswede_Gacs_Korner}.

In the absence of a common message, the capacity region of the DM-DBC is introduced in the following theorem (see, e.g., \cite[Theorem~15.6.2]{Cover_and_Thomas} or \cite[Theorem~5.2]{ElGamal_Kim_NIT}).
\begin{theorem} \label{thm: capacity region for a DM-DBC}
A rate pair $(R_1,R_2)$ is achievable for the DM-DBC $(T_1,T_2)$, characterized by
\eqref{eq: stochastically degraded broadcast channel} with $P_{Y|X} = T_1$ and $P_{Z|X} = T_2$,
if and only if
\begin{align*}
	R_1 \le I(X; Y | U), \quad R_2 \le I(U; Z)
\end{align*}
for an auxiliary random variable $U \in \cU$ such that $U \to X \to Y \to Z$ is a Markov chain,
and $|\cU| \le \min\left\{|\cX|,|\cY|,|\cZ|\right\}+1$.
\end{theorem}

The {\em strong converse} for the DM-DBC, due to Ahlswede, G\'acs and K\"orner \cite{Ahlswede_Gacs_Korner}, states that allowing for nonvanishing probabilities of error does not enlarge the achievable region:
\begin{theorem}[Strong converse for the DM-DBC]\label{thm:SC_DM-DBC}
	$$
	\cR(\eps_1,\eps_2) = \cR, \qquad \forall \, (\eps_1,\eps_2) \in (0,1)^2.
	$$
\end{theorem}
Before proceeding with the formal proof of this theorem, we briefly describe the way in which the blowing-up lemma enters the picture. The main idea is that, given an arbitrary code, one can ``blow up'' the decoding sets in such a way that the probability of decoding error can be as small as one desires (for large enough $n$). Of course, the blown-up decoding sets are no longer disjoint, so the resulting object is no longer a code according to Definition~\ref{def:codes for DM-DBC}. On the other hand, the blowing-up operation transforms the original code into a {\em list code} with a subexponential list size, and one can use a generalization of Fano's inequality for list decoding (see Appendix~\ref{app:Fano_list})
to get nontrivial converse bounds.

\begin{proof}[Proof (Theorem~\ref{thm:SC_DM-DBC})]
Given $\widetilde{\eps}_1,\widetilde{\eps}_2 \in (0,1)$, let
$\widetilde{\cC} = (f_n,\widetilde{\cD}_1,\widetilde{\cD}_2)$ be
an arbitrary $(n,M_1,M_2,\widetilde{\eps}_1,\widetilde{\eps}_2)$-code
for the DM-DBC $(T_1,T_2)$ with
	$$
	\widetilde{\cD}_{1} = \left\{ \widetilde{D}_{1,i} \right\}^{M_1}_{i=1} \qquad \text{and} \qquad \widetilde{\cD}_2 = \left\{ \widetilde{D}_{2,j} \right\}^{M_2}_{j=1}.
	$$
By hypothesis, the decoding sets in $\widetilde{\cD}_1$ and $\widetilde{\cD}_2$ satisfy
\begin{subequations}\label{eq:original_code for the DM-DBC}
\begin{align}
	\min_{1 \le i \le M_1} \min_{1 \le j \le M_2} T^n_1\Big( \widetilde{D}_{1,i} \Big| f_n(i,j) \Big) &\ge 1 - \widetilde{\eps}_1, \\
	\min_{1 \le i \le M_1} \min_{1 \le j \le M_2} T^n_2\Big( \widetilde{D}_{2,j} \Big| f_n(i,j) \Big) &\ge 1 - \widetilde{\eps}_2.
\end{align}
\end{subequations}
For an arbitrary $\alpha > 0$, define a positive sequence $\{\delta_n\}$ as
\begin{align}  \label{eq:delta for DM-DBC}
\delta_n = \sqrt{\frac{-\ln \bigl(1-\max\{\widetilde{\eps}_1, \widetilde{\eps}_2\}\bigr)}{2n}}
+ \sqrt{\frac{\alpha \ln n}{n}} \, , \quad \forall \, n \in \naturals.
\end{align}
Note that, as $n \to \infty$,
\begin{align*}
	& n^\beta \, \delta_n \to 0, \quad \forall \, \beta < \frac{1}{2}, \\
    & \sqrt{n} \, \delta_n \to \infty.
\end{align*}
For each $i \in \{1,\ldots,M_1\}$ and $j \in \{1,\ldots,M_2\}$,
define the ``blown-up'' decoding sets
\begin{align} \label{blown-up decoding sets for the DM-DBC}
	D_{1,i} \deq \left[\widetilde{D}_{1,i}\right]_{n\delta_n} \qquad \text{and} \qquad D_{2,j} \deq \left[\widetilde{D}_{2,j}\right]_{n\delta_n}.
\end{align}
We rely in the following on Lemma~\ref{lm:blowup}
with the setting in \eqref{eq:choice of sequences for the proof of the blowing-up lemma}.
From \eqref{eq:choice of sequences for the proof of the blowing-up lemma}
and \eqref{eq:original_code for the DM-DBC}, the blown-up decoding sets
in \eqref{blown-up decoding sets for the DM-DBC} with the sequence $\{\delta_n\}$ defined
in \eqref{eq:delta for DM-DBC} imply that, for every $n \in \naturals$,
\begin{subequations}\label{eq:blown_up_code}
\begin{align}
	\min_{1 \le i \le M_1} \min_{1 \le j \le M_2} T^n_1\Big( D_{1,i} \Big| f_n(i,j) \Big) & \ge 1 - n^{-2 \alpha} \\
	\min_{1 \le i \le M_1} \min_{1 \le j \le M_2} T^n_2\Big( D_{2,j} \Big| f_n(i,j) \Big) & \ge 1 - n^{-2 \alpha}.
\end{align}
\end{subequations}
Let $\cD_{1} = \left\{ D_{1,i}\right\}^{M_1}_{i=1}$, and $\cD_2= \left\{ D_{2,j}\right\}^{M_2}_{j=1}$. We have thus constructed a triple $(f_n,\cD_1,\cD_2)$ satisfying \eqref{eq:blown_up_code}. Note, however, that this new object is not a code because the blown-up sets
$\cD_{1}$ are not disjoint, and the same holds for the
blown-up sets $\cD_2$. On the other hand, each given $n$-tuple $y^n \in \cY^n$ belongs to a subexponential number of
the $D_{1,i}$'s, and the same applies to $D_{2,j}$'s. More precisely, let us define the sets
\begin{subequations}\label{eq:listsets}
\begin{align}
	& \cN_1(y^n) \deq \left\{ i \colon y^n \in D_{1,i} \right\},  \quad \forall \, y^n \in \cY^n, \\
    & \cN_2(z^n) \deq \left\{ j \colon z^n \in D_{2,j} \right\},  \quad \forall \, z^n \in \cZ^n.
\end{align}
\end{subequations}
Then, a simple combinatorial argument (see
\cite[Eq.~(37)]{Ahlswede_Gacs_Korner}) shows that there exists a positive sequence $\{\eta_n\}^\infty_{n=1}$
such that $\eta_n \to 0$ as $n \to \infty$, and
\begin{subequations}\label{eq:listsizes}
\begin{align}
	|\cN_1(y^n)| & \le \exp(n\eta_n), \qquad \forall \, y^n \in \cY^n, \\
	|\cN_2(z^n)| & \le \exp(n\eta_n), \qquad \forall \, z^n \in \cZ^n.
\end{align}
\end{subequations}
In order to get an explicit expression for $\{\eta_n\}$, for every $y^n \in \cY^n$ and $r \ge 0$, let $\cB_r(y^n) \subseteq \cY^n$
denote the ball of $d_n$-radius $r$ centered at $y^n$:
$$
\cB_{r}(y^n) \deq \left\{ v^n \in \cY^n \colon d_n(v^n,y^n) \le r\right\} \equiv  \left\{y^n\right\}_{r}
$$
where $d_n$ is the Hamming metric \eqref{eq:Hamming metric}, and $\left\{y^n\right\}_{r}$ denotes the $r$-blowup of the singleton set $\{y^n\}$.
Since $\delta_n \to 0$ as $n \to \infty$, there exists $n_0 \in \naturals$ such that $\delta_n + \frac{1}{n} \leq \frac{1}{2}$ for every
$n \ge n_0$. Consequently, it follows that for every $n \ge n_0$,
\begin{align*}
	|\cN_1(y^n)| &\le \left|\cB_{n\delta_n}(y^n)\right| \\
                 &=\sum_{i=0}^{\lceil n \delta_n \rceil} {{n}\choose{i}} |\cY|^i \\
                 &\le \bigl(\lceil n\delta_n \rceil +1 \bigr) {{n}\choose{\lceil n \delta_n \rceil}} \, |\cY|^{\lceil n \delta_n \rceil} \\
                 &\le (n \delta_n+2) \, \exp\left(n \, h\Bigl(\delta_n + \frac{1}{n}\Bigr)\right) \, |\cY|^{n \delta_n+1}, \quad
                 \forall \, y^n \in \cY^n.
\end{align*}
The second inequality holds since, for $n \ge n_0$, we have $\lceil n \delta_n \rceil \leq \lceil \frac{n}{2} \rceil$,
and the binomial coefficients $\{{{n}\choose{k}}\}$ is monotonic increasing with $k$ if $k \leq \lceil \frac{n}{2} \rceil$;
the third inequality holds since, for every $n \ge n_0$,
$${{n}\choose{k}} \leq \exp\left(n \, h\Bigl(\frac{k}{n}\Bigr)\right), \quad \text{if }   k \leq \frac{n}{2}$$
where $h$ denotes the binary entropy function; similarly, for all $n \ge n_0$,
\begin{align*}
|\cN_2(z^n)| \le (n \delta_n+2) \, \exp\left(n \, h\Bigl(\delta_n + \frac{1}{n}\Bigr)\right) \, |\cZ|^{n \delta_n+1}, \quad
                 \forall \, z^n \in \cZ^n.
\end{align*}
From \eqref{eq:listsizes}, one can define the positive sequence $\{\eta_n\}$ such that
$$\eta_n = \frac{\ln(n\delta_n+2)}{n} + h\Bigl(\delta_n + \frac{1}{n}\Bigr) + \Bigl(\delta_n + \frac{1}{n}\Bigr)
\log \Bigl(\max\bigl\{|\cY|, \, |\cZ|\bigr\}\Bigr), \; \; \forall \, n \ge n_0$$
so, we have, $\eta_n \to 0$ as $n \to \infty$.

We are now ready to apply Fano's inequality, just as in \cite{Ahlswede_Korner}. To this end, for every $j \in \{1,\ldots,M_2\}$, define
\begin{align*}
	\cT(j) \deq \left\{ f_n(i,j) \colon 1 \le i \le M_1 \right\},
\end{align*}
let $U$ be a uniformly distributed random variable over $\{1,\ldots,M_2\}$, and let $X^n \in \cX^n$ be uniformly distributed
over the set $\cT(U)$.
Finally, let $Y^n \in \cY^n$ and $Z^n \in \cZ^n$ be generated from $X^n$ via the DMCs $T^n_1$ and $T^n_2$, respectively. Now,
consider the error event of the second receiver (which corresponds to the degraded channel $T_2^n$); the error event of a list
decoder for the second receiver refers to the case where $U \notin D_{2,U}$ and, from  \eqref{eq:listsets}, it is given by
\begin{align*}
	E_n(Z^n) \deq \left\{ U \not\in \cN_2(Z^n) \right\}
\end{align*}
and let $\zeta_n \deq \pr\left(E_n(Z^n)\right)$ be the error probability of the list decoder
for the blown-up sets $\cD_2$.
Then, using a modification of Fano's inequality for list decoding (see Appendix~\ref{app:Fano_list})
together with \eqref{eq:listsizes}, we get
\begin{align}\label{eq:Fano_list_application}
	H(U|Z^n) &\le h(\zeta_n) + (1-\zeta_n) n \eta_n + \zeta_n \ln M_2.
\end{align}
On the other hand, $\ln M_2 = H(U) = I(U; Z^n) + H(U|Z^n)$, so
\begin{align*}
\frac{1}{n} \, \ln M_2 &\le \frac{1}{n} \, \Big[I(U; Z^n) + h(\zeta_n) + \zeta_n \ln M_2\Big] + (1-\zeta_n)\eta_n \\[0.1cm]
&= \frac{1}{n} \, I(U;Z^n) + o(1),
\end{align*}
where the second step uses the fact that $\eta_n \to 0$ and, by \eqref{eq:blown_up_code}, $\zeta_n \le n^{-2 \alpha}$
for some $\alpha > 0$, so also $\zeta_n \to 0$ as $n \to \infty$.
Using a similar argument, we can also prove that
\begin{align*}
	\frac{1}{n} \, \ln M_1 &\le \frac{1}{n} \, I(X^n; Y^n |U) + o(1).
\end{align*}
By the weak converse for the DM-DBC \cite{Ahlswede_Korner}, the pair $(R_1,R_2)$ with $R_1 = \frac{1}{n} \, I(X^n; Y^n | U)$ and
$R_2 = \frac{1}{n} \, I(U; Z^n)$ belongs to the achievable region $\cR$. Since every element of $\cR(\eps_1,\eps_2)$ can be expressed
as a limit of rates $\left(\frac{1}{n} \, \ln M_1, \, \frac{1}{n} \, \ln M_2\right)$ in the region $\cR$, and since the achievable
region $\cR$ is closed, we conclude that $\cR(\eps_1,\eps_2) \subseteq \cR$ for all $\eps_1,\eps_2 \in (0,1)$, and Theorem~\ref{thm:SC_DM-DBC}
is proved.
\end{proof}

\subsection{The empirical distribution of good channel codes with non-vanishing error probability}

A more recent application of concentration of measure to information theory has to do with characterizing stochastic behavior of output sequences of good channel codes. On a conceptual level, the random coding argument originally used by Shannon, and many times since, to show the existence of good channel codes suggests that the input (resp., output) sequence of such a code should resemble, as much as possible, a typical realization of a sequence of i.i.d.\ random variables sampled from a capacity-achieving input (resp., output) distribution. For capacity-achieving sequences of codes with asymptotically vanishing probability of error, this intuition has been analyzed rigorously by Shamai and Verd\'u \cite{Shamai_Verdu_empirical}, who have proved the following remarkable statement \cite[Theorem~2]{Shamai_Verdu_empirical}: given a DMC $T \colon \cX \to \cY$, every capacity-achieving sequence of channel codes with asymptotically vanishing probability of error (maximal or average) has the property that
\begin{align}\label{eq:convergence_to_caod}
	\lim_{n \to \infty} \frac{1}{n} \, D(P_{Y^n} \| P^*_{Y^n}) = 0,
\end{align}
where, for each $n$, $P_{Y^n}$ denotes the output distribution on $\cY^n$ induced by the code (assuming the messages are equiprobable), while $P^*_{Y^n}$ is the product of $n$ copies of the single-letter capacity-achieving output distribution (see below for a more detailed exposition). In fact, the convergence in \eqref{eq:convergence_to_caod} holds not just for DMCs, but for arbitrary channels satisfying the condition
\begin{align*}
	C = \lim_{n \to \infty} \frac{1}{n} \, \sup_{P_{X^n} \in \cP(\cX^n)} I(X^n; Y^n).
\end{align*}
(These ideas go back to the work of Han and Verd\'u on approximation theory of output statistics, see \cite[Theorem~15]{Han_Verdu}). In a recent paper \cite{Polyanskiy_Verdu_good_codes}, Polyanskiy and Verd\'u extended the results of \cite{Shamai_Verdu_empirical} for codes with {\em nonvanishing} probability of error, provided one uses the
maximal probability of error criterion and deterministic encoders.

In this section, we will present some of the results from \cite{Polyanskiy_Verdu_good_codes,ISIT2013_Raginsky_Sason} in the context of the material covered earlier in this chapter.
To keep things simple, we will only focus on channels with finite input and output alphabets. Thus, let $\cX$ and $\cY$ be finite sets, and consider a DMC $T \colon \cX \to \cY$. The capacity $C$ is given by solving the optimization problem
\begin{align*}
	C = \max_{P_X \in \cP(\cX)} I(X;Y),
\end{align*}
where $X$ and $Y$ are related via $T$. Let $P^*_X \in \cP(\cX)$ be a capacity-achieving input distribution (there may be several). It can be shown \cite{Topsoe_capacity,Kemperman_capacity} that the corresponding output distribution $P^*_Y \in \cP(\cY)$ is unique, and for every $n \in \naturals$, the product distribution $P^*_{Y^n} \equiv (P^*_Y)^{\otimes n}$ has the key property
\begin{align}\label{eq:info_gain_bound}
	D(P_{Y^n|X^n=x^n} \| P^*_{Y^n}) \le nC, \qquad \forall \, x^n \in \cX^n
\end{align}
where $P_{Y^n|X^n=x^n}$ is shorthand for the product distribution $T^n(\cdot|x^n)$. From the bound \eqref{eq:info_gain_bound}, we see that the capacity-achieving output distribution $P^*_{Y^n}$ dominates every output distribution $P_{Y^n}$ induced by an arbitrary input distribution $P_{X^n} \in \cP(\cX^n)$:
\begin{align*}
	P_{Y^n|X^n=x^n} \ll P^*_{Y^n}, \; \forall \, x^n \in \cX^n \; \; \Longrightarrow \; \;
P_{Y^n} \ll P^*_{Y^n}, \; \forall \, P_{X^n} \in \cP(\cX^n).
\end{align*}
This has two important consequences:
\begin{enumerate}
	\item The information density is well-defined for every $x^n \in \cX^n$ and $y^n \in \cY^n$:
	\begin{align*}
		i^*_{X^n;Y^n}(x^n;y^n) \deq \ln \frac{\d P_{Y^n|X^n=x^n}}{\d P^*_{Y^n}}(y^n).
	\end{align*}
\item For every input distribution $P_{X^n}$, the corresponding output distribution $P_{Y^n}$ satisfies
\begin{align*}
	D(P_{Y^n} \| P^*_{Y^n}) \le nC - I(X^n; Y^n).
\end{align*}
Indeed, by the chain rule for divergence, it follows that for every input distribution $P_{X^n} \in \cP(\cX^n)$
\begin{align*}
I(X^n; Y^n) &= D(P_{Y^n|X^n} \| P_{Y^n} | P_{X^n}) \\
&= D(P_{Y^n|X^n} \| P^*_{Y^n} | P_{X^n}) - D(P_{Y^n} \| P^*_{Y^n}) \\
&\le nC - D(P_{Y^n} \| P^*_{Y^n}).
\end{align*}
The claimed bound follows upon rearranging this inequality.
\end{enumerate}
Now let us bring codes into the picture. Given $n,M \in \naturals$, an {\em $(n,M)$-code} for $T$ is a pair $\cC = (f_n,g_n)$ consisting of an {\em encoding map} $f_n \colon \{1,\ldots,M\} \to \cX^n$ and a {\em decoding map} $g_n \colon \cY^n \to \{1,\ldots,M\}$. Given $0 < \eps \le 1$, we say that $\cC$ is an {\em $(n,M,\eps)$-code} if
\begin{align}\label{eq:max_det_error}
	\max_{1 \le i \le M}\pr\big(g_n(Y^n) \neq i \big| X^n = f_n(i)\big) \le \eps.
\end{align}
\begin{remark} Polyanskiy and Verd\'u \cite{Polyanskiy_Verdu_good_codes} use a more precise nomenclature and say that every such $\cC=(f_n,g_n)$ satisfying \eqref{eq:max_det_error} is an {\em $(n,M,\eps)_{\max,\det}$-code} to indicate explicitly that the encoding map $f_n$ is deterministic and that the maximal probability of error criterion is used. Here, we will only consider codes of this type, so we will adhere to our simplified terminology.\end{remark}
	
Consider an arbitrary $(n,M)$-code $\cC = (f_n,g_n)$ for $T$, and let $J$ be a random variable uniformly distributed on $\{1,\ldots,M\}$. Hence, we can think of every $i \in \{1 \ldots, M\}$ as one of $M$ equiprobable messages to be transmitted over $T$. Let $P^{(\cC)}_{X^n}$ denote the distribution of $X^n = f_n(J)$, and let $P^{(\cC)}_{Y^n}$ denote the corresponding output distribution. The central result of \cite{Polyanskiy_Verdu_good_codes} is that the output distribution $P^{(\cC)}_{Y^n}$ of every $(n,M,\eps)$-code satisfies
\begin{align}\label{eq:PV_output_bound}
	D\big(P^{(\cC)}_{Y^n} \big\| P^*_{Y^n}\big) \le nC - \ln M + o(n);
\end{align}
moreover, the $o(n)$ term was refined in \cite[Theorem~5]{Polyanskiy_Verdu_good_codes} to $O(\sqrt{n})$ for
every DMC, except those that have zeroes in their transition matrix. In the following, we present
a sharpened bound with a modified proof, in which we specify an explicit form for the term that scales like
$O(\sqrt{n})$.

Just as in \cite{Polyanskiy_Verdu_good_codes}, the proof of \eqref{eq:PV_output_bound} with the $O(\sqrt{n})$
term uses the following strong converse for channel codes due to Augustin \cite{Augustin} (see also \cite[Theorem~1]{Polyanskiy_Verdu_good_codes} and \cite[Section~2]{Ahlswede_MAC}):
\begin{theorem}[Augustin] Let $S \colon \cU \to \cV$ be a DMC with finite input and output alphabets, and let $P_{V|U}$ be the transition probability induced by $S$. For every $M \in \naturals$ and $0 < \eps \le 1$, let $f \colon \{1,\ldots,M\} \to \cU$ and $g \colon \cV \to \{1,\ldots,M\}$ be two mappings, such that
\begin{align*}
	\max_{1 \le i \le M}\pr\big(g(V) \neq i \big| U = f(i) \big) \le \eps.
\end{align*}
Let $Q_V \in \cP(\cV)$ be an auxiliary output distribution, and fix an arbitrary mapping $\gamma \colon \cU \to \reals$.
Then, the following inequality holds:
\begin{align}\label{eq:Augustin}
	M \le \frac{\exp\bigl\{\expectation[\gamma(U)]\bigr\}}{\displaystyle\inf_{u \in \cU}
P_{V|U=u}\left( \ln \frac{\d P_{V|U=u}}{\d Q_V} < \gamma(u)\right) - \eps},
\end{align}
provided the denominator is strictly positive. The expectation in the numerator is taken with respect to the distribution of $U = f(J)$ with $J \sim {\rm Uniform}\{1,\ldots,M\}$.
\end{theorem}

We first establish the bound \eqref{eq:PV_output_bound} for the case when the DMC $T$ is such that
\begin{align}\label{eq:no_zeroes}
	C_1 \deq \max_{x,x' \in \cX} D(P_{Y|X=x} \| P_{Y|X=x'}) < \infty.
\end{align}
Note that $C_1 < \infty$ if and only if the transition matrix of $T$ does not have any zeroes. Consequently,
\begin{align} \label{eq: c(T)}
	c(T) \deq 2\max_{x \in \cX}\max_{y,y' \in \cY} \left| \ln \frac{P_{Y|X}(y|x)}{P_{Y|X}(y'|x)} \right| < \infty.
\end{align}
We can now establish the following sharpened version of the bound in \cite[Theorem~5]{Polyanskiy_Verdu_good_codes}:

\begin{theorem}\label{thm:PV_output_bound_1} Let $T \colon \cX \to \cY$ be a DMC with $C > 0$ satisfying \eqref{eq:no_zeroes}. Then, every $(n,M,\eps)$-code $\cC$ for $T$ with $0 < \eps < 1/2$ satisfies
\begin{align}\label{eq:PV_output_bound_1}
	D\big(P^{(\cC)}_{Y^n} \big\| P^*_{Y^n}\big) \le nC - \ln M + \ln \frac{1}{\eps} + c(T)\sqrt{\frac{n}{2}\ln \frac{1}{1-2\eps}}.
\end{align}
\end{theorem}

\begin{remark} As it is shown in \cite{Polyanskiy_Verdu_good_codes}, the restriction to codes with deterministic encoders and to the maximal probability of error criterion is necessary both for this theorem and for the next one.
\end{remark}
\begin{proof} Fix an input sequence $x^n \in \cX^n$, and consider the function $h_{x^n} \colon \cY^n \to \reals$ defined by
	\begin{align*}
		h_{x^n}(y^n) &\deq \ln \frac{\d P_{Y^n|X^n=x^n}}{\d P^{(\cC)}_{Y^n}}(y^n), \quad \forall \, y^n \in \cY^n.
	\end{align*}
Then $\expectation[h_{x^n}(Y^n)|X^n=x^n] = D(P_{Y^n|X^n=x^n} \| P^{(\cC)}_{Y^n})$. Moreover, for every $i \in \{1,\ldots,n\}$, $y,y' \in \cY$, and $\overline{y}^i \in \cY^{n-1}$, we have (see the notation used in
\eqref{eq:f_i})
\begin{align}
	&\left|h_{i,x^n}(y|\overline{y}^i) - h_{i,x^n}(y'|\overline{y}^i)\right| \nonumber\\
	& \le \left|\ln P_{Y^n|X^n=x^n}(y^{i-1},y,y^n_{i+1}) - \ln  P_{Y^n|X^n=x^n}(y^{i-1},y',y^n_{i+1})\right| \nonumber\\
	&\qquad \qquad \qquad + \left|\ln P^{(\cC)}_{Y^n}(y^{i-1},y,y^n_{i+1}) - \ln P^{(\cC)}_{Y^n}(y^{i-1},y',y^n_{i+1})\right| \nonumber\\
	& \le \left| \ln \frac{P_{Y_i|X_i=x_i}(y)}{P_{Y_i|X_i=x_i}(y')}\right| + \left| \ln \frac{P^{(\cC)}_{Y_i|\overline{Y}^i}(y|\overline{y}^i)}{P^{(\cC)}_{Y_i|\overline{Y}^i}(y'|\overline{y}^i)} \right| \nonumber\\
	&\le 2\max_{x \in \cX}\max_{y,y' \in \cY} \left| \ln \frac{P_{Y|X}(y|x)}{P_{Y|X}(y'|x')}\right|
     \label{eq: see Appendix D}\\
	&= c(T) < \infty \label{eq:cond_info_bd}
\end{align}
(see Appendix~\ref{app:McDiarmid_details} for a detailed derivation of the inequality in \eqref{eq: see Appendix D}). Hence, for each fixed $x^n \in \cX^n$, the function $h_{x^n} \colon \cY^n \to \reals$ satisfies the bounded differences condition \eqref{eq:bounded_differences} with $c_1 = \ldots = c_n = c(T)$.  Theorem~\ref{thm:bounded_differences_revisited} therefore implies that, for every $r \ge 0$, we have
\begin{align}\label{eq:info_density_concentration}
&	P_{Y^n|X^n=x^n}\left( \ln \frac{\d P_{Y^n|X^n=x^n}}{\d P^{(\cC)}_{Y^n}}(Y^n) \ge D(P_{Y^n|X^n=x^n} \| P^{(\cC)}_{Y^n}) + r \right) \nonumber\\
& \qquad \qquad \qquad \qquad \qquad \le \exp\left( -\frac{2r^2}{nc^2(T)}\right).
\end{align}
(In fact, the above derivation goes through for every possible output distribution $P_{Y^n}$, not necessarily
one induced by a code.) This is where we have departed from the original proof by Polyanskiy and Verd\'u \cite{Polyanskiy_Verdu_good_codes}: we have used McDiarmid's (or bounded differences) inequality to control the deviation probability for the ``conditional'' information density $h_{x^n}$ directly, whereas they bounded the {\em variance} of $h_{x^n}$ using a suitable Poincar\'e inequality, and then derived a bound on the deviation probability using Chebyshev's inequality. As we will see shortly, the sharp concentration inequality \eqref{eq:info_density_concentration} allows us to explicitly identify the dependence of the constant multiplying $\sqrt{n}$ in \eqref{eq:PV_output_bound_1} on the channel $T$ and on the maximal error probability $\eps$.

We are now in a position to apply Augustin's strong converse. To that end, we let $\cU = \cX^n$, $\cV = \cY^n$,
and consider the DMC $S = T^n$ together with an $(n,M,\eps)$-code $(f,g) = (f_n,g_n)$. Furthermore, let
\begin{align}
\zeta_n = \zeta_n(\eps) \deq c(T)\sqrt{\frac{n}{2}\ln \frac{1}{1-2\eps}}  \label{eq: zeta}
\end{align}
and take $\gamma(x^n) = D(P_{Y^n|X^n=x^n}\|P^{(\cC)}_{Y^n}) + \zeta_n$.
Using \eqref{eq:Augustin} with the auxiliary distribution $Q_V = P^{(\cC)}_{Y^n}$, we get
\begin{equation} \label{eq: Applying Augustin's Theorem}
M \le \frac{\exp\bigl\{\expectation[\gamma(X^n)]\bigr\}}{\displaystyle\inf_{x^n \in \cX^n}
P_{Y^n|X^n=x^n}\left( \ln \frac{\d P_{Y^n|X^n=x^n}}{\d P_{Y^n}^{(\cC)}} < \gamma(x^n)\right) - \eps}
\end{equation}
where
\begin{equation} \label{eq: expected value of gamma}
\expectation[\gamma(X^n)] = D\bigl( P_{Y^n | X^n} \| P_{Y^n}^{(\cC)} \, | \, P^{(\cC)}_{X^n} \bigr) + \zeta_n.
\end{equation}
The concentration inequality in \eqref{eq:info_density_concentration} with $\zeta_n$
in \eqref{eq: zeta} therefore gives that, for every $x^n \in \cX^n$,
\begin{align*}
P_{Y^n|X^n=x^n}\left( \ln \frac{\d P_{Y^n|X^n=x^n}}{\d P_{Y^n}^{(\cC)}} \ge \gamma(x^n)\right) & \le
\exp\left(-\frac{2 \zeta_n^2}{n c^2(T)} \right) \nonumber \\
& = 1 - 2 \eps
\end{align*}
which implies that
$$\displaystyle\inf_{x^n \in \cX^n}
P_{Y^n|X^n=x^n}\left( \ln \frac{\d P_{Y^n|X^n=x^n}}{\d P_{Y^n}^{(\cC)}} < \gamma(x^n)\right) \ge 2\eps.$$
Hence, from \eqref{eq: Applying Augustin's Theorem}, \eqref{eq: expected value of gamma} and the
last inequality, it follows that
$$ M \le \frac{1}{\eps} \; \exp\left( D\bigl( P_{Y^n | X^n} \| P_{Y^n}^{(\cC)} \, | \, P^{(\cC)}_{X^n} \bigr) + \zeta_n \right)$$
so, by taking logarithms on both sides of the last inequality and rearranging terms, we get from
\eqref{eq: zeta} that
\begin{align}
	D(P_{Y^n|X^n} \| P^{(\cC)}_{Y^n} \, | \, P^{(\cC)}_{X^n}) &\ge \ln M + \ln \eps - \zeta_n \nonumber\\
	&= \ln M + \ln \eps - c(T)\sqrt{\frac{n}{2}\ln \frac{1}{1-2\eps}}.\label{eq:PV_output_bound_1_s0}
\end{align}
We are now ready to derive \eqref{eq:PV_output_bound_1}:
\begin{align}
	& D\big(P^{(\cC)}_{Y^n} \big\| P^*_{Y^n}\big) \nonumber \\
	& \qquad = D\big(P_{Y^n|X^n} \big\|P^*_{Y^n} \big| P^{(\cC)}_{X^n}\big) - D \big(P_{Y^n|X^n} \big\| P^{(\cC)}_{Y^n} \big| P^{(\cC)}_{X^n}\big) \label{eq:PV_output_bound_1_s1}\\
	& \qquad \le nC - \ln M + \ln \frac{1}{\eps} + c(T)\sqrt{\frac{n}{2}\ln \frac{1}{1-2\eps}} \label{eq:PV_output_bound_1_s2}
\end{align}
where \eqref{eq:PV_output_bound_1_s1} uses the chain rule for divergence, while
\eqref{eq:PV_output_bound_1_s2} uses \eqref{eq:info_gain_bound} and
\eqref{eq:PV_output_bound_1_s0}. This completes the proof of Theorem~\ref{thm:PV_output_bound_1}.
\end{proof}

\bigskip
For an arbitrary DMC $T$ with nonzero capacity and zeroes in its transition matrix, we have the following result
which forms a sharpened version of the bound in \cite[Theorem~6]{Polyanskiy_Verdu_good_codes}:

\begin{theorem}\label{thm:PV_output_bound_2} Let $T \colon \cX \to \cY$ be a DMC with $C > 0$. Then, for every $0 < \eps < 1$, every $(n,M,\eps)$-code $\cC$ for $T$ satisfies
	\begin{align*}
		D\big(P^{(\cC)}_{Y^n} \big\| P^*_{Y^n}\big) \le nC - \ln M + O\left(\sqrt{n} \, (\ln n)^{3/2}\right).
	\end{align*}
	More precisely, for every such code we have
	\begin{align}
		D\big(P^{(\cC)}_{Y^n} \big\| P^*_{Y^n}\big) &\le nC - \ln M \nonumber\\
		& \,\, + \sqrt{2n} \, (\ln n)^{3/2}  \left(1 + \sqrt{\frac{1}{\ln n} \,
		\ln \left(\frac{1}{1-\varepsilon}\right)} \right) \,
		\left(1 + \frac{\ln |\cY|}{\ln n} \right) \nonumber\\
		& \,\, + 3 \ln n + \ln \bigl(2 |\cX| |\cY|^2\bigr).
        \label{eq:PV_output_bound_2}
	\end{align}
\end{theorem}
\begin{proof} Given an $(n,M,\eps)$-code $\cC = (f_n,g_n)$, let $c_1,\ldots,c_M \in \cX^n$ be its codewords, and let $\widetilde{D}_1,\ldots,\widetilde{D}_M \subset \cY^n$ be the corresponding decoding regions:
	$$
	\widetilde{D}_i = g^{-1}_n(i) \equiv \left\{ y^n \in \cY^n \colon g_n(y^n) = i \right\},
    \qquad i = 1,\ldots,M.
	$$
If we choose
\begin{align}\label{eq:blow_up_radius}
	\delta_n &= \delta_n(\eps) = \frac{1}{n} \, \left\lceil n \left(\sqrt{\frac{\ln n}{2n}} + \sqrt{\frac{1}{2n} \ln \frac{1}{1-\eps}} \right) \right\rceil
\end{align}
(note that $n\delta_n$ is an integer) then, by Lemma~\ref{lm:my_blowup}, the ``blown-up'' decoding regions $D_i \deq \left[\widetilde{D}_i\right]_{n\delta_n}$ satisfy
\begin{align}
	P_{Y^n|X^n=c_i} \left(D^c_i\right) &\le \exp\left[ -2n\left(\delta_n - \sqrt{\frac{1}{2n}\ln \frac{1}{1-\eps}}\right)^2\right] \nonumber\\
	&\le \frac{1}{n} \, , \qquad \qquad \forall \, i \in \{1,\ldots,M\} \label{eq:blown_up_decoding_error}
\end{align}
where the last inequality holds since, from \eqref{eq:blow_up_radius},
$$\delta_n \ge \sqrt{\frac{\ln n}{2n}} + \sqrt{\frac{1}{2n} \ln \frac{1}{1-\eps}}.$$
We now complete the proof by a random coding argument. For
\begin{align}\label{eq:subcode_size}
	N \deq \left\lceil \frac{M}{n {n \choose n\delta_n} |\cY|^{n\delta_n}} \right\rceil,
\end{align}
let $U_1,\ldots,U_{N}$ be independent random variables, each uniformly distributed on the set $\{1,\ldots,M\}$. For each realization $V =U^N$, let $P_{X^n(V)} \in \cP(\cX^n)$ denote the induced distribution of $X^n(V) = f_n(c_J)$, where $J$ is uniformly distributed on the set $\{{U_1},\ldots,{U_N}\}$, and let $P_{Y^n(V)}$ denote the corresponding output distribution of $Y^n(V)$:
\begin{align}
	P_{Y^n(V)} &= \frac{1}{N}\sum^N_{i=1}P_{Y^n|X^n = c_{U_i}}.
\end{align}
It is easy to show that $\expectation\left[P_{Y^n(V)}\right] = P^{(\cC)}_{Y^n}$, the output distribution of the original code $\cC$, where the expectation is with respect to the distribution of $V = U^N$. Now, for $V = U^N$ and for
every $y^n \in \cY^n$, let $\cN_V(y^n)$ denote the list of all those indices in $\{U_1,\ldots,U_N\}$ such that
$y^n \in D_{U_j}$:
$$
\cN_V(y^n) \deq \left\{ j \colon y^n \in D_{U_j}\right\}.
$$
Consider the list decoder $Y^n \mapsto \cN_V(Y^n)$, and let $\eps(V)$ denote its conditional decoding
error probability: $\eps(V) \deq P(J \not\in \cN_V(Y^n)|V)$. From \eqref{eq:subcode_size}, it follows that
\begin{align}
	\ln N &\ge \ln M - \ln n - \ln {n \choose n \delta_n} - n\delta_n \ln |\cY| \nonumber\\
	&\ge \ln M - \ln n - n\delta_n \left(\ln n  + \ln |\cY|\right) \label{eq:subcode_size_1}
\end{align}
where the last inequality uses the simple inequality ${n \choose k} \le n^k$ for $k \le n$ with $k \triangleq
n \delta_n$ (we note that the gain in using instead the inequality ${n \choose n \delta_n} \le
\exp\bigl(n \, h(\delta_n)\bigr)$ is marginal, and it does not have any advantage asymptotically for large $n$).
Moreover, each $y^n \in \cY^n$ can belong to at most ${n \choose n\delta_n}|\cY|^{n\delta_n}$ blown-up decoding sets, so
\begin{align}
	\ln \left|\cN_V(Y^n = y^n)\right| &\le \ln {n \choose n\delta_n} + n\delta_n \ln |\cY|\nonumber\\
	&\le n \delta_n \left(\ln n +  \ln |\cY|\right), \quad \forall \, y^n \in \mathcal{Y}^n.
	\label{eq:list_size}
\end{align}
Now, for each realization of $V$, we have
\begin{align}
	& D\big(P_{Y^n(V)} \big\| P^*_{Y^n}\big) \nonumber\\
	& = D \big(P_{Y^n(V)|X^n(V)} \big\| P^*_{Y^n} \big| P_{X^n(V)}\big) - I(X^n(V); Y^n(V))
        \label{eq:PV_output_bound_2_s1} \\
	& \le n C - I(X^n(V); Y^n(V)) \label{eq:PV_output_bound_2_s2}\\
	& \le n C - I(J; Y^n(V)) \label{eq:PV_output_bound_2_s3} \\
	& = n C - H(J) + H(J|Y^n(V)) \nonumber \\
	& \le n C - \ln N + (1-\eps(V)) \, \max_{y^n \in \cY^n}\ln |\cN_V(y^n)| + n\eps(V) \ln |\cX| + \ln 2  \label{eq:PV_output_bound_2_s4}
\end{align}
where:
\begin{itemize}
	\item \eqref{eq:PV_output_bound_2_s1} is by the chain rule for divergence;
	\item \eqref{eq:PV_output_bound_2_s2} is by \eqref{eq:info_gain_bound};
	\item \eqref{eq:PV_output_bound_2_s3} is by the data processing inequality, and the fact that $J \to X^n(V) \to Y^n(V)$ is a Markov chain; and
	\item \eqref{eq:PV_output_bound_2_s4} is by Fano's inequality for list decoding (see Appendix~\ref{app:Fano_list}), and also since (i) $N \le |\cX|^n$, (ii) $J$ is
uniformly distributed on $\{U_1, \ldots, U_N\}$, so $H(J | U_1, \ldots, U_N) = \ln N$ and
$H(J) \ge \ln N$.
\end{itemize}
(Note that all the quantities indexed by $V$ in the above chain of estimates are actually random variables, since they depend on the realization $V = U^N$.)
Substituting \eqref{eq:subcode_size_1} and \eqref{eq:list_size} into \eqref{eq:PV_output_bound_2_s4}, we get
\begin{align}
D\big(P_{Y^n(V)} \big\| P^*_{Y^n}\big) &\le n C - \ln M \nonumber\\
& \qquad + \ln n + 2n\delta_n \left(\ln n + \ln |\cY|\right) \nonumber\\
& \qquad + n \eps(V)\ln |\cX| + \ln 2. \label{eq:PV_output_bound_2_s5}
\end{align}
Using the fact that $\expectation\left[P_{Y^n(V)}\right] = P^{(\cC)}_{Y^n}$, it follows from the convexity of the relative entropy
and Jensen's inequality that
$$\expectation\left[D\big(P_{Y^n(V)} \big\| P^*_{Y^n}\big)\right] \ge D\big(P_{Y^n}^{(\cC)} \big\| P^*_{Y^n}\big)$$
and, by taking expectations on both sides of \eqref{eq:PV_output_bound_2_s5}, we get
\begin{align}
D\big(P_{Y^n}^{(\cC)} \big\| P^*_{Y^n}\big)
&\le n C - \ln M \nonumber\\
& \qquad + \ln n + 2n\delta_n \left(\ln n + \ln |\cY|\right) \nonumber\\
& \qquad + n \, \expectation\left[\eps(V)\right] \, \ln |\cX| + \ln 2. \label{eq:PV_output_bound_2_s6}
\end{align}
To get \eqref{eq:PV_output_bound_2}, we use the fact that
$$
\expectation\left[\eps(V)\right] \le \max_{1 \le i \le M} P_{Y^n|X^n=c_i} \left(D^c_i\right) \le \frac{1}{n},
$$
which follows from \eqref{eq:blown_up_decoding_error}, as well as the substitution of \eqref{eq:blow_up_radius}
in \eqref{eq:PV_output_bound_2_s6}; note that, from \eqref{eq:blow_up_radius}, it follows that
$$\delta_n < \sqrt{\frac{\ln n}{2n}} + \sqrt{\frac{1}{2n} \ln \frac{1}{1-\eps}} + \frac{1}{n}.$$
This completes the proof of Theorem~\ref{thm:PV_output_bound_2}.
\end{proof}

\bigskip
We are now ready to examine some consequences of Theorems~\ref{thm:PV_output_bound_1} and \ref{thm:PV_output_bound_2}. To start with, consider a sequence $\{\cC_n\}^\infty_{n=1}$, where each code $\cC_n = (f_n,g_n)$ is an $(n,M_n,\eps)$-code for a DMC $T \colon \cX \to \cY$ with $C > 0$. We say that $\{\cC_n\}^\infty_{n=1}$ is {\em capacity-achieving} if
\begin{align}
	\lim_{n \to \infty}\frac{1}{n} \, \ln M_n = C.
\end{align}
Then, from Theorems~\ref{thm:PV_output_bound_1} and \ref{thm:PV_output_bound_2}, it follows that every such sequence satisfies
\begin{align}\label{eq:good_codes_converge}
	\lim_{n \to \infty} \frac{1}{n} \, D\big(P^{(\cC_n)}_{Y^n} \big\| P^*_{Y^n}\big) = 0.
\end{align}
Moreover, as shown in \cite{Polyanskiy_Verdu_good_codes}, if the restriction to either deterministic encoding maps or to the maximal probability of error criterion is lifted, then the convergence in \eqref{eq:good_codes_converge} may no longer hold. This is in sharp contrast to \cite[Theorem~2]{Shamai_Verdu_empirical}, which states that \eqref{eq:good_codes_converge} holds for {\em every} capacity-achieving sequence of codes with vanishing probability of error (maximal or average).

Another remarkable fact that follows from the above theorems is that a broad class of functions evaluated on the output of a good code concentrate sharply around their expectations with respect to the capacity-achieving output distribution. Specifically, we have the following version of \cite[Proposition~10]{Polyanskiy_Verdu_good_codes} (again, we have streamlined the statement and the proof a bit to relate them to earlier material in this chapter):

\begin{theorem}\label{thm:good_code_con} Let $T \colon \cX \to \cY$ be a DMC with $C > 0$ and $C_1 < \infty$ (see \eqref{eq:no_zeroes}).
Let $d \colon \cY^n \times \cY^n \to \reals_+$ be a metric, and suppose that there exists a constant $c > 0$, such that the conditional probability distributions $P_{Y^n|X^n=x^n}$, $x^n \in \cX^n$, as well as $P^*_{Y^n}$ satisfy ${\rm T}_1(c)$ on the metric space $(\cY^n,d)$. Then, for every $\eps \in (0,1)$, there exists a constant $a > 0$ that depends only on $T$ and on $\eps$ (to be defined explicitly in the following), such that for every $(n,M,\eps)$-code $\cC$ for $T$ and every Lipschitz function $f \colon \cY^n \to \reals$ with respect to the metric $d$
\begin{align}\label{eq:good_code_con}
&	P^{(\cC)}_{Y^n}\Big( \left|f(Y^n) - \expectation[f(Y^{*n})]\right| \ge r\Big) \nonumber\\
&\qquad \le \frac{4}{\eps} \cdot \exp\left( nC - \ln M + a\sqrt{n} - \frac{r^2}{8c\| f \|^2_{\rm Lip}} \right), \quad \forall \, r>0
\end{align}
where $\expectation[f(Y^{*n})]$ designates the expected value of $f(Y^n)$ with respect to the capacity-achieving output distribution $P^*_{Y^n}$,
\begin{align*}
	\| f \|_{\rm Lip} \deq \sup_{y^n \neq v^n} \frac{|f(y^n) - f(v^n)|}{d(y^n,v^n)}
\end{align*}
is the Lipschitz constant of $f$, and
\begin{align}  \label{eq: explicit constant a}
    a \deq c(T)\sqrt{\frac{1}{2}\ln \frac{1}{1-2\eps}}
\end{align}
with $c(T)$ in \eqref{eq: c(T)}.
\end{theorem}
\begin{remark} Our sharpening of the corresponding result from \cite[Proposition~10]{Polyanskiy_Verdu_good_codes} consists mainly in identifying an explicit form for the constant in front of $\sqrt{n}$ in the bound \eqref{eq:good_code_con}; this provides a closed-form expression for the concentration of measure inequality.
\end{remark}

\begin{proof} For an arbitrary $f$, define
	\begin{equation} \label{eq: definitions for the proof of CAOD theorem}
		\mu^*_f \deq \expectation[f(Y^{*n})], \quad \quad
        \phi(x^n) \deq \expectation[f(Y^n)|X^n=x^n], \; \forall \, x^n \in \cX^n.
	\end{equation}
Since each $P_{Y^n|X^n=x^n}$ satisfies ${\rm T}_1(c)$, by Corollary~\ref{cor:Lipschitz_Gauss_con_revisited} we have
	\begin{align}\label{eq:channel_concentration}
		\pr\Big( \big|f(Y^n) - \phi(x^n)\big| \ge r \Big| \, X^n=x^n\Big) \le 2 \exp\left( - \frac{r^2}{2c \| f \|^2_{\rm Lip}}\right), \,\, \forall \, r \ge 0.
	\end{align}

Now, given $\cC$, consider a subcode $\cC'$ with codewords $x^n \in \cX^n$ satisfying
$\phi(x^n) \ge \mu^*_f + r$ for $r \ge 0$. The number of codewords $M'$ of $\cC'$ satisfies
\begin{align}\label{eq:good_subcode_size}
	M' &= M P^{(\cC)}_{X^n}\left(\phi(X^n) \ge \mu^*_f + r \right).
\end{align}
Let $Q = P^{(\cC')}_{Y^n}$ be the output distribution induced by $\cC'$. Then
\begin{align}
	\mu^*_f + r &\le \frac{1}{M'} \sum_{x^n \in \,{\rm codewords}(\cC')}\phi(x^n) \label{eq:good_code_con_1}\\
	&= \expectation_{Q}[f(Y^n)] \label{eq:good_code_con_2} \\
	&\le \expectation[f(Y^{*n})] + \| f \|_{\rm Lip}\sqrt{2c D(Q_{Y^n} \| P^*_{Y^n})} \label{eq:good_code_con_3}\\
	&\le \mu^*_f + \| f \|_{\rm Lip} \sqrt{2c \left(nC - \ln M' + a\sqrt{n} + \ln \frac{1}{\eps} \right)},
     \label{eq:good_code_con_4}
\end{align}
where:
\begin{itemize}
       \item \eqref{eq:good_code_con_1} is by definition of $\cC'$;
       \item \eqref{eq:good_code_con_2} is by definition of $\phi$ in
             \eqref{eq: definitions for the proof of CAOD theorem};
       \item \eqref{eq:good_code_con_3} follows from the fact that $P^*_{Y^n}$ satisfies ${\rm T}_1(c)$ and from the Kantorovich--Rubinstein formula \eqref{eq:Kantorovich_Rubinstein}; and
       \item \eqref{eq:good_code_con_4} holds for the constant $a = a(T,\eps) > 0$ in
             \eqref{eq: explicit constant a} due to Theorem~\ref{thm:PV_output_bound_1} (see
             \eqref{eq:PV_output_bound_1}) and because $\cC'$ is an $(n,M',\eps)$-code for $T$.
             The constant $\mu^*_f $ in \eqref{eq:good_code_con_4} is defined in \eqref{eq: definitions for the proof of CAOD theorem}.
\end{itemize}
From this and \eqref{eq:good_subcode_size}, we get
$$ r \le \| f \|_{\rm Lip} \sqrt{2c \left(nC - \ln M - \ln P^{(\cC)}_{X^n}\left(\phi(X^n) \ge \mu^*_f + r \right)
+ a\sqrt{n} + \ln \frac{1}{\eps} \right)}$$
so, it follows that
\begin{align*}
P^{(\cC)}_{X^n}\Big( \phi(X^n) \ge \mu^*_f + r \Big) \le \exp\left(nC - \ln M + a \sqrt{n} + \ln \frac{1}{\eps}
- \frac{r^2}{2c \| f \|^2_{\rm Lip}} \right).
\end{align*}
Following the same line of reasoning with $-f$ instead of $f$, we conclude that
\begin{align}\label{eq:good_code_con_a0}
	& P^{(\cC)}_{X^n}\Big( \big|\phi(X^n) - \mu^*_f\big| \ge r \Big) \nonumber\\
	& \qquad \le 2\exp\left(nC - \ln M + a \sqrt{n}
    + \ln \frac{1}{\eps} - \frac{r^2}{2c \| f \|^2_{\rm Lip}} \right).
\end{align}
Finally, for every $r \ge 0$,
\begin{align}
	& P^{(\cC)}_{Y^n}\Big( \big|f(Y^n) - \mu^*_f\big| \ge r \Big) \nonumber\\
	&\le P^{(\cC)}_{X^n,Y^n}\Big( \left| f(Y^n) - \phi(X^n)\right| \ge r/2\Big) + P^{(\cC)}_{X^n}\Big(\left|\phi(X^n) - \mu^*_f \right| \ge r/2 \Big) \nonumber\\
	&\le 2 \exp \left( - \frac{r^2}{8c \| f \|^2_{\rm Lip}}\right) + 2\exp\left(nC - \ln M + a \sqrt{n}
     +\ln \frac{1}{\eps} - \frac{r^2}{8c \| f \|^2_{\rm Lip}} \right) \label{eq:good_code_con_a1} \\
	&\le 4 \exp\left(nC - \ln M + a \sqrt{n} + \ln \frac{1}{\eps} - \frac{r^2}{8c \|f \|^2_{\rm Lip}} \right), \label{eq:good_code_con_a2}
\end{align}
where \eqref{eq:good_code_con_a1} is by \eqref{eq:channel_concentration} and \eqref{eq:good_code_con_a0}, while \eqref{eq:good_code_con_a2} follows from the fact that
$$
nC - \ln M + a\sqrt{n} + \ln \frac{1}{\eps} \ge D(P^{(\cC)}_{Y^n} \| P^*_{Y^n}) \ge 0
$$
by Theorem~\ref{thm:PV_output_bound_1}, and from \eqref{eq: explicit constant a}. This proves \eqref{eq:good_code_con}.
\end{proof}

As an illustration, let us consider $\cY^n$ with the Hamming metric
\begin{align}\label{eq:product_Hamming}
	d_n(y^n,v^n) = \sum^n_{i=1} 1_{\{y_i \neq v_i\}}.
\end{align}
Then, every function $f \colon \cY^n \to \reals$ of the form
\begin{align}\label{eq:ave_of_Lipschitz}
	f(y^n) = \frac{1}{n} \sum^n_{i=1}f_i(y_i), \qquad \forall \, y^n \in \cY^n
\end{align}
where $f_1,\ldots,f_n \colon \cY \to \reals$ are Lipschitz functions on $\cY$, will satisfy
\begin{align*}
	\| f \|_{\rm Lip} &\le \frac{L}{n}, \qquad L \deq \max_{1 \le i \le n} \| f_i \|_{\rm Lip}.
\end{align*}
 Every probability distribution $P$ on $\cY$ equipped with the Hamming metric satisfies ${\rm T}_1(1/4)$ (this is simply Pinsker's inequality); by Proposition~\ref{prop:TC_tensorization_2}, every product probability distribution on $\cY^n$ satisfies ${\rm T}_1(n/4)$ with respect to the product metric \eqref{eq:product_Hamming}. Consequently, for every $(n,M,\eps)$-code for $T$ and every function $f \colon \cY^n \to \reals$ of the form \eqref{eq:ave_of_Lipschitz}, Theorem~\ref{thm:good_code_con} gives the concentration inequality
	\begin{align}\label{eq:good_code_Lip_con}
	& P^{(\cC)}_{Y^n}\Big( \big|f(Y^n) - \expectation[f(Y^{*n})]\big| \ge r\Big) \nonumber\\
	&\qquad \le \frac{4}{\eps} \; \exp\left( nC - \ln M + a\sqrt{n} - \frac{nr^2}{2L^2}
    \right)
	\end{align}
for all $r > 0$. Concentration inequalities like \eqref{eq:good_code_con}, or its more specialized version \eqref{eq:good_code_Lip_con}, can be very useful for assessing various performance characteristics of good channel codes without having to explicitly construct such codes: all one needs to do is to find the capacity-achieving output distribution $P^*_Y$ and evaluate $\expectation[f(Y^{*n})]$ for an arbitrary $f$ of interest. Then, Theorem~\ref{thm:good_code_con} guarantees that $f(Y^n)$ concentrates tightly around $\expectation[f(Y^{*n})]$, which is relatively easy to compute since $P^*_{Y^n}$ is a product distribution.

The bounds presented in Theorems~\ref{thm:PV_output_bound_1} and~\ref{thm:PV_output_bound_2} quantify the trade-offs
between the minimal blocklength required for achieving a certain gap (in rate) to capacity with
a fixed block error probability, and normalized divergence between the {\em output distribution} induced by the code
and the (unique) capacity-achieving output distribution of the channel. Moreover, these bounds  sharpen the asymptotic $O(\cdot)$ terms in the results of \cite{Polyanskiy_Verdu_good_codes} for
all finite blocklengths $n$.

These results are similar in spirit to a lower bound on the rate loss with respect to fully random
block codes (with a binomial distribution) in terms of the normalized divergence between the distance
spectrum of a code and the binomial distribution. Specifically, a combination of
\cite[Eqs.~(A17) and (A19)]{Shamai_Sason_2001} provides a lower bound on the rate loss with respect
to fully random block codes in terms of the normalized divergence between the distance spectrum of the
code and the binomial distribution where the latter result refers to the empirical {\em input distribution}
of good codes.

\subsection{An information-theoretic converse for concentration of measure}

If we were to summarize the main idea behind concentration of measure, it would be this: if a subset of a metric probability space does not have a ``too small'' probability mass, then its isoperimetric enlargements (or blowups) will eventually take up most of the probability mass. On the other hand, it makes sense to ask whether a converse of this statement is true --- given a set whose blowups eventually take up most of the probability mass, how small can this set be? This question was answered precisely by Kontoyiannis \cite{Kontoyiannis_sphere_covering} using information-theoretic techniques.

The following setting is considered in \cite{Kontoyiannis_sphere_covering}: Let $\cX$ be a finite set, together with a nonnegative distortion function $d \colon \cX \times \cX \to \reals^+$ (which is not necessarily a metric) and a strictly positive mass function $M \colon \cX \to (0,\infty)$ (which is not necessarily normalized to one). As before, let us extend the ``single-letter'' distortion $d$ to $d_n \colon \cX^n \to \reals^+$, $n \in \naturals$, where
\begin{align*}
	d_n(x^n,y^n) \deq \sum^n_{i=1} d(x_i,y_i), \qquad \forall \, x^n,y^n \in \cX^n.
\end{align*}
For every $n \in \naturals$ and for every set $C \subseteq \cX^n$, let us define
\begin{align*}
	M^n(C) \deq \sum_{x^n \in C} M^n(x^n)
\end{align*}
where
\begin{align*}
	M^n(x^n) \deq \prod^n_{i=1}M(x_i), \qquad \forall \, x^n \in \cX^n.
\end{align*}
We also recall the definition of the $r$-blowup of an arbitrary set $A \subseteq \cX^n$:
\begin{align*}
	A_r \deq \left\{ x^n \in \cX^n \colon d_n(x^n,A) \le r \right\},
\end{align*}
where $d_n(x^n,A) \triangleq \min_{y^n \in A}d_n(x^n,y^n)$. Fix a probability distribution $P \in \cP(\cX)$, where we assume without loss of generality that $P$ is strictly positive. We are interested in the following question: Given a sequence of sets $\left\{A^{(n)}\right\}_{n \in \naturals}$ such that $A^{(n)} \subseteq \cX^n$ for every $n$, and
\begin{align*}
	P^{\otimes n}\left(A^{(n)}_{n\delta}\right) \xrightarrow{n \to \infty} 1,
\end{align*}
for some $\delta \ge 0$, how small can their masses $M^n(A^{(n)})$ be?

In order to state and prove the main result of \cite{Kontoyiannis_sphere_covering} that answers this question, we need a few preliminary definitions. For every $n \in \naturals$, every pair $P_n,Q_n$ of probability measures on $\cX^n$, and every $\delta \ge 0$, let us define the set
\begin{align*}
\Pi_n(P_n,Q_n,\delta) \deq \left\{ \pi_n \in \Pi_n(P_n,Q_n) \colon \frac{1}{n} \, \expectation_{\pi_n}\left[d_n(X^n,Y^n)\right] \le \delta \right\}
\end{align*}
of all couplings $\pi_n \in \cP(\cX^n \times \cX^n)$ of $P_n$ and $Q_n$, such that the per-letter expected distortion between $X^n$ and $Y^n$ with $(X^n,Y^n) \sim \pi_n$ is at most $\delta$. With this, we define
\begin{align*}
	I_n(P_n,Q_n,\delta) \deq \inf_{\pi_n \in \Pi_n(P_n,Q_n,\delta)} D(\pi_n \| P_n \otimes Q_n),
\end{align*}
and consider the following {\em rate function}:
\begin{align} \label{eq:definition of the rate function}
	R_n(\delta) &\equiv R_n(\delta; P_n,M^n) \nonumber\\
	&\deq \inf_{Q_n \in \cP(\cX^n)} \Big\{ I_n(P_n,Q_n,\delta) + \expectation_{Q_n}[\ln M^n(Y^n)] \Big\} \nonumber \\
	&\equiv \inf_{P_{X^nY^n}} \Bigg\{ I(X^n;Y^n) + \expectation[\ln M^n(Y^n)] \colon \nonumber \\
	& \qquad \qquad \; P_{X^n} = P_n,  \; \frac{1}{n} \, \expectation[d_n(X^n,Y^n)] \le \delta \Bigg\}.
\end{align}
When $n=1$, we will simply write $\Pi(P,Q,\delta)$, $I(P,Q,\delta)$ and $R(\delta)$. For the special case when each $P_n$ is the product measure $P^{\otimes n}$, we have
\begin{align}\label{eq:R_subadditivity}
	R(\delta) = \lim_{n \to \infty} \frac{1}{n} \, R_n(\delta) = \inf_{n \ge 1} \frac{1}{n} \, R_n(\delta)
\end{align}
(see \cite[Lemma~2]{Kontoyiannis_sphere_covering}). We are now ready to state the main result of \cite{Kontoyiannis_sphere_covering}:
\begin{theorem}[Kontoyiannis]\label{thm:converse_CM} Consider an arbitrary set $A^{(n)} \subseteq \cX^n$, and denote
	$\delta \deq \frac{1}{n} \, \expectation[d_{n}(X^n,A^{(n)})].$
	Then
	\begin{align}\label{eq:Konto_converse}
		\frac{1}{n} \, \ln M^n(A^{(n)}) \ge R(\delta; P,M).
	\end{align}
\end{theorem}
\begin{proof} Given $A^{(n)} \subseteq \cX^n$, let $\varphi_n \colon \cX^n \to A^{(n)}$ be the function that maps each $x^n \in \cX^n$ to the closest element $y^n \in A^{(n)}$, i.e.,
	\begin{align*}
		d_n(x^n,\varphi_n(x^n)) = d_n(x^n,A^{(n)}), \quad \forall \, x^n \in \cX^n
	\end{align*}
(we assume some fixed rule for resolving ties). If $X^n \sim P^{\otimes n}$, then let $Q_n \in \cP(\cX^n)$ denote the distribution of $Y^n = \varphi_n(X^n)$, and let $\pi_n \in \cP(\cX^n \times \cX^n)$ be the following joint distribution of $X^n$ and $Y^n$:
\begin{align} \label{eq:special form of coupling}
	\pi_n(x^n,y^n) = P^{\otimes n}(x^n) \, 1_{\{y^n = \varphi_n(x^n)\}}, \quad \forall \, x^n, y^n \in \cX^n.
\end{align}
This implies that $\pi_n \in \Pi_n(P^{\otimes n}, Q_n)$, and
\begin{align*}
	\expectation_{\pi_n}[d_n(X^n,Y^n)] &= \expectation_{\pi_n}[d_n(X^n,\varphi_n(X^n))] \nonumber\\
	&= \expectation_{\pi_n}\left[d_n(X^n,A^{(n)})\right] \nonumber\\
	&= n\delta, 
\end{align*}
so $\pi_n \in \Pi_n(P^{\otimes n},Q_n,\delta)$. Furthermore, we have
\begin{align}
	\ln M^n(A^{(n)}) &= \ln \sum_{y^n \in A^{(n)}} M^n(y^n) \nonumber\\
	&= \ln \sum_{y^n \in A^{(n)}} Q_n(y^n) \cdot \frac{M^n(y^n)}{Q_n(y^n)} \nonumber \\
	&\ge \sum_{y^n \in A^{(n)}} Q_n(y^n) \ln \frac{M^n(y^n)}{Q_n(y^n)}\label{eq:Konto_step1} \\
	&= \sum_{x^n \in \cX^n,y^n \in A^{(n)}} \pi_n(x^n,y^n) \ln \frac{\pi_n(x^n,y^n)}{P^{\otimes n}(x^n)Q_n(y^n)} \nonumber\\
	& \qquad \qquad \qquad + \sum_{y^n \in A^{(n)}}Q_n(y^n)\ln M^n(y^n)
\label{eq:Konto_step2}\\
	&= I(X^n; Y^n) + \expectation_{Q_n}[\ln M^n(Y^n)] \label{eq:Konto_step3} \\
	&\ge R_n(\delta),\label{eq:Konto_step4}
\end{align}
where \eqref{eq:Konto_step1} is by Jensen's inequality, \eqref{eq:Konto_step2} and \eqref{eq:Konto_step3} use the fact that $\pi_n$ is a coupling of $P^{\otimes n}$ and $Q_n$ where equality \eqref{eq:Konto_step2} uses the particular coupling in \eqref{eq:special form of coupling}, and \eqref{eq:Konto_step4} is by definition of $R_n(\delta)$ in \eqref{eq:definition of the rate function}. Using \eqref{eq:R_subadditivity}, we get \eqref{eq:Konto_converse}, and the theorem is proved.
\end{proof}
\begin{remark}\label{rem:achievability_CM} In \cite{Kontoyiannis_sphere_covering}, an achievability result was also proved: For every $\delta \ge 0$ and $\eps > 0$, there is a sequence of sets $\left\{A^{(n)}\right\}_{n \in \naturals}$ such that $A^{(n)} \subseteq \cX^n$ for every $n$,
	\begin{align}
		\frac{1}{n} \, \ln M^n(A^{(n)}) \le R(\delta) + \eps, \qquad \forall \, n \in \naturals
	\end{align}
	and
	\begin{align}
		\frac{1}{n} \, d_n(X^n,A^{(n)}) \le \delta, \qquad \text{eventually a.s.}
	\end{align}
\end{remark}

We are now ready to use Theorem~\ref{thm:converse_CM} to answer the question posed at the beginning of this section. Specifically, we consider the case when $M = P$. Defining the {\em concentration exponent} $R_{\rm c}(r; P) \deq R(r; P,P)$, we have:

\begin{corollary}[Converse concentration of measure]\label{cor:converse_CM} If $A^{(n)} \subseteq \cX^n$ is an arbitrary set, then
	\begin{align}
		P^{\otimes n}\left(A^{(n)}\right) \ge \exp\big(n \, R_{\rm c}(\delta; P)\big),
	\end{align}
where $$\delta = \frac{1}{n} \, \expectation\left[d_n\left(X^n,A^{(n)}\right)\right].$$
Moreover, if the sequence of sets $\{A^{(n)}\}^\infty_{n=1}$ is such that, for some $\delta \ge 0$, $P^{\otimes n}\left(A^{(n)}_{n\delta}\right) \to 1$ as $n \to \infty$, then
	\begin{align}\label{eq:converse_CM_exponent}
		\liminf_{n \to \infty} \frac{1}{n} \, \ln P^{\otimes n}\left(A^{(n)}\right) \ge R_{\rm c}(\delta; P).
	\end{align}
\end{corollary}
\begin{remark} A moment of reflection shows that the concentration exponent $R_{\rm c}(\delta; P)$ is nonpositive. Indeed, from definitions,
	\begin{align}
		& R_{\rm c}(\delta; P) \nonumber\\
		&= R(\delta; P,P) \nonumber\\
		&= \inf_{P_{XY}} \Big\{ I(X; Y) + \expectation[\ln P(Y)] \colon P_X = P,\, \expectation[d(X,Y)] \le \delta \Big\}  \nonumber\\
		&= \inf_{P_{XY}} \Big\{ H(Y) - H(Y|X) + \expectation[\ln P(Y)] \colon P_X = P,\, \expectation[d(X,Y)] \le \delta \Big\} \nonumber\\
		&= \inf_{P_{XY}} \Big\{ -D(P_Y \| P) - H(Y|X) \colon P_X = P,\, \expectation[d(X,Y)] \le \delta \Big\} \nonumber\\
		&= - \sup_{P_{XY}} \Big\{ D(P_Y \| P) + H(Y|X) \colon P_X = P,\, \expectation[d(X,Y)] \le \delta \Big\}, \label{eq:CM_exponent_alt}
	\end{align}
	which proves the claim, since both the divergence and the (conditional) entropy are nonnegative.
\end{remark}
\begin{remark} Using the achievability result from \cite{Kontoyiannis_sphere_covering} (cf.~Remark~\ref{rem:achievability_CM}), one can also prove that there exists a sequence of sets $\{A^{(n)}\}^\infty_{n=1}$, such that
\begin{align*}
	\lim_{n \to \infty} P^{\otimes n}\left(A^{(n)}_{n\delta}\right) = 1 \qquad \text{and} \qquad \lim_{n \to \infty}\frac{1}{n} \, \ln P^{\otimes n}\left(A^{(n)}\right) \le R_{\rm c}(\delta; P).
\end{align*}
\end{remark}

As an illustration, let us consider the case when $\cX = \{0,1\}$ and $d$ is the Hamming distortion, $d(x,y) = 1_{\{x \neq y\}}$. Then $\cX^n = \{0,1\}^n$ is the $n$-dimensional binary cube. Let $P$ be the ${\rm Bernoulli}(p)$ probability measure, which satisfies a ${\rm T}_1\left(\frac{1}{2\varphi(p)}\right)$ transportation-cost inequality with respect to the $L^1$ Wasserstein distance induced by the Hamming metric, where $\varphi(p)$ is defined in \eqref{eq:Hoeffding_phi}. By Proposition~\ref{prop:TC_tensorization}, the product measure $P^{\otimes n}$ satisfies a ${\rm T}_1\left(\frac{n}{2\varphi(p)}\right)$ transportation-cost inequality on the product space $(\cX^n,d_n)$. Consequently, it follows from \eqref{eq:Marton_conc_1} that for every $A^{(n)} \subseteq \cX^n$,
\begin{align}
	P^{\otimes n}\left(A^{(n)}_{n\delta}\right) & \ge 1 - \exp\left( -\frac{\varphi(p)}{n} \, \left(n\delta - \sqrt{\frac{n}{\varphi(p)} \; \ln \frac{1}{P^{\otimes n}\left(A^{(n)}\right)}}\right)^2\right) \nonumber \\[0.1cm]
    & = 1 - \exp\left( -n \, \varphi(p) \, \left(\delta - \sqrt{\frac{1}{n \, \varphi(p)} \; \ln \frac{1}{P^{\otimes n}\left(A^{(n)}\right)}}\right)^2\right)
\end{align}
provided that
$$\delta \ge \sqrt{\frac{1}{n \, \varphi(p)} \; \ln \frac{1}{P^{\otimes n}\left(A^{(n)}\right)}}.$$
Thus, if a sequence of sets $A^{(n)} \subseteq \cX^n$, $n \in \naturals$, satisfies
\begin{align}\label{eq:Hamming_cube_exponent}
	\liminf_{n \to \infty} \frac{1}{n} \, \ln P^{\otimes n}\left(A^{(n)}\right) > - \varphi(p)\delta^2,
\end{align}
then
\begin{align}\label{eq:Hamming_cube_blowup}
	P^{\otimes n}\left(A^{(n)}_{n\delta}\right) \xrightarrow{n \to \infty} 1.
\end{align}
The converse result, Corollary~\ref{cor:converse_CM}, says that if a sequence of sets $A^{(n)} \subseteq \cX^n$ satisfies \eqref{eq:Hamming_cube_blowup}, then \eqref{eq:converse_CM_exponent} holds. Let us compare the concentration exponent $R_{\rm c}(\delta; P)$, where $P$ is the ${\rm Bernoulli}(p)$ measure, with the exponent $-\varphi(p)\delta^2$ on the right-hand side of \eqref{eq:Hamming_cube_exponent}:

\begin{theorem} If $P$ is the ${\rm Bernoulli}(p)$ measure with $p \in [0,1/2]$, then the concentration exponent $R_{\rm c}(\delta;P)$ satisfies
    \begin{align}
    R_{\rm c}(\delta; P) \le - \varphi(p)\delta^2 - (1-p) \, h\left(\frac{\delta}{1-p}\right),
    \qquad \forall \, \delta \in [0,1-p]
	\end{align}
	and
	\begin{align}\label{eq:CM_exponent_large_delta}
		R_{\rm c}(\delta; P) = \ln p, \qquad \forall \, \delta \in [1-p,1]
	\end{align}
	where $$h(x) \triangleq -x \ln x - (1-x)\ln (1-x), \quad \forall \, x \in [0,1]$$ is the binary entropy function to base~$e$ (with the convention that $0 \log 0 = 0$).
\end{theorem}
\begin{proof}
From \eqref{eq:CM_exponent_alt}, we have
\begin{align}\label{eq:CM_exponent_variational}
	R_{\rm c}(\delta; P) &= - \sup_{P_{XY}}\Big\{ D(P_Y \| P) + H(Y|X) \colon P_X = P,\, \pr(X \neq Y) \le \delta \Big\}.
\end{align}
For a given $\delta \in [0,1-p]$, let us choose $P_Y$ so that $\| P_Y - P \|_{\rm TV} = \delta$. Then from \eqref{eq:Pinsker_refined_optimality},
\begin{align}
	\frac{D(P_Y \| P)}{\delta^2} &= \frac{D(P_Y \| P)}{\| P_Y - P \|^2_{\rm TV}}
	\ge \inf_{Q} \frac{D(Q \| P)}{\| Q - P \|^2_{\rm TV}}
	= \varphi(p).\label{eq:best_Pinsker_constant}
\end{align}
By the coupling representation of the total variation distance, we can choose a joint distribution $P_{\widetilde{X}\widetilde{Y}}$ with marginals $P_{\widetilde{X}} = P$ and $P_{\widetilde{Y}} = P_Y$, such that $\pr(\widetilde{X} \neq \widetilde{Y}) = \| P_Y - P \|_{\rm TV} = \delta$. Moreover, using \eqref{eq:TV_optimal_coupling}, we can compute
\begin{align*}
	P_{\tilde{Y}|\tilde{X}=0} = {\rm Bernoulli}\left(\frac{\delta}{1-p}\right) \qquad \text{and} \qquad P_{\tilde{Y}|\tilde{X}=1}(\tilde{y}) = \delta_1(\tilde{y}) \triangleq 1_{\{\tilde{y}=1\}}.
\end{align*}
Consequently,
\begin{align}\label{eq:cond_entropy_tilde}
	H(\widetilde{Y}|\widetilde{X}) = (1-p) H(\widetilde{Y}|\widetilde{X}=0) = (1-p)h\left( \frac{\delta}{1-p}\right).
\end{align}
From \eqref{eq:CM_exponent_variational}, \eqref{eq:best_Pinsker_constant} and \eqref{eq:cond_entropy_tilde},
we obtain
\begin{align*}
	R_{\rm c}(\delta; P) &\le -  D(P_{\widetilde{Y}} \| P) - H(\widetilde{Y} | \widetilde{X})  \\
	&\le - \varphi(p)\delta^2 - (1-p)h\left( \frac{\delta}{1-p}\right).
\end{align*}
To prove \eqref{eq:CM_exponent_large_delta}, it suffices to consider the case where $\delta = 1-p$.
If we let $Y$ be independent of $X \sim P$, then $I(X;Y) = 0$, so we have to minimize
$\expectation_Q[\ln P(Y)]$ over all distributions $Q$ of $Y$. But then
\begin{align*}
\min_{Q} \expectation_Q[\ln P(Y)] = \min_{y \in \{0,1\}} \ln P(y) = \min \left\{\ln p, \, \ln (1-p) \right\} = \ln p,
\end{align*}
where the last equality holds since $p \le 1/2$.
\end{proof}

\section{Summary}
\label{section: Entropy Chapter Summary}

In this chapter, we have covered the essentials of the entropy method, an information-theoretic technique for deriving concentration inequalities for functions of many independent random variables. As its very name suggests, the entropy method revolves around the relative entropy (or information divergence), which in turn can be related to the logarithmic moment-generating function and its derivatives.

A key ingredient of the entropy method is {\em tensorization}, or the use of a certain subadditivity property of the divergence in order to break the original multidimensional problem up into simpler one-dimensional problems. Tensorization is used in conjunction with various inequalities relating the relative entropy to suitable energy-type functionals defined on the space of functions for which one wishes to establish concentration. These inequalities fall into two broad classes: functional inequalities (typified by the logarithmic Sobolev inequalities) and transportation-cost inequalities (such as Pinsker's inequality). We have examined the many deep and remarkable information-theoretic ideas that bridge these two classes of inequalities, and also showed some examples of their applications to problems in coding and information theory.

At this stage, the relationship between information theory and the study of measure concentration is heavily skewed towards the use of the former as a tool for the latter. Moreover, applications of concentration of measure inequalities to problems in information theory, coding and communications have been exemplified in Chapters~2 and~3. We hope that the present monograph may offer some inspiration for information and coding theorists to deepen the ties between their discipline and the fascinating realm of high-dimensional probability and concentration of measure.

\begin{subappendices}
\section{Van Trees inequality}
\label{app:VanTrees}

Consider the problem of estimating a random variable $Y \sim P_Y$ based on a noisy observation
$U = \sqrt{s}Y + Z$, where $s > 0$ is the SNR parameter, while the additive noise $Z \sim G$
is independent of $Y$. We assume that $P_Y$ has a differentiable, absolutely continuous density
$p_Y$ with $J(Y) < \infty$. Our goal is to prove the van Trees inequality \eqref{eq:VanTrees}
and to establish that equality in \eqref{eq:VanTrees} holds if and only if $Y$ is Gaussian.

In fact, we will prove a more general statement: Let $\varphi(U)$ be an arbitrary (Borel-measurable)
estimator of $Y$. Then
\begin{align}\label{eq:VanTrees_general}
	\expectation\left[(Y - \varphi(U))^2\right] \ge \frac{1}{s + J(Y)},
\end{align}
with equality if and only if $Y$ has a standard normal distribution and $\varphi(U)$ is the MMSE
estimator of $Y$ given $U$.

The strategy of the proof is simple. Define two random variables
\begin{align*}
	\Delta(U,Y) &\deq \varphi(U) - Y, \\
	\Upsilon(U,Y) &\deq \frac{\d}{\d y}\ln \left[ p_{U|Y}(U|y)p_Y(y) \right]\Bigg|_{y=Y} \\
	& = \frac{\d}{\d y} \ln \left[\gamma(U-\sqrt{s}y) p_Y(y)\right]\Bigg|_{y=Y} \\
	& = \sqrt{s}(U-\sqrt{s}Y) + \rho_Y(Y) \\
	& = \sqrt{s}Z + \rho_Y(Y)
\end{align*}
where $\rho_Y(y) \deq \frac{\d}{\d y} \ln p_Y(y)$ for $y \in \reals$ is the score function.
We show below that $\expectation[\Delta(U,Y)\Upsilon(U,Y)]=1$. Then, by applying the Cauchy--Schwarz
inequality,
\begin{align*}
	1 &= \left| \expectation[\Delta(U,Y)\Upsilon(U,Y)]\right|^2 \\
	&\le \expectation[\Delta^2(U,Y)] \cdot \expectation[\Upsilon^2(U,Y)] \\
	&= \expectation[(\varphi(U)-Y)^2] \cdot \expectation[(\sqrt{s}Z + \rho_Y(Y))^2] \\
	&= \expectation[(\varphi(U)-Y)^2] \cdot (s + J(Y)).
\end{align*}
Upon rearranging, we obtain \eqref{eq:VanTrees_general}. The fact that $J(Y) < \infty$ implies that
the density $p_Y$ is bounded (see \cite[Lemma~A.1]{Johnson_Barron_CLT}). Using this and the rapid
decay of the Gaussian density $\gamma$ at infinity, we have
\begin{align}\label{eq:VanTrees_term1}
\int_{-\infty}^{\infty} \frac{\d}{\d y}\left[p_{U|Y}(u|y)p_Y(y)\right] \d y &= \gamma(u - \sqrt{s}y)p_Y(y)\Bigg|^\infty_{-\infty} = 0.
\end{align}
Integration by parts gives
\begin{align}
	&\int_{-\infty}^{\infty} y \frac{\d}{\d y}\left[ p_{U|Y}(u|y)p_Y(y)\right] \d y \nonumber\\
	&\quad= y \gamma(u - \sqrt{s}y)p_Y(y)\Bigg|^\infty_{-\infty} - \int_{-\infty}^{\infty} p_{U|Y}(u|y) p_Y(y) \d y \nonumber\\
	&\quad=  - \int_{-\infty}^{\infty} p_{U|Y}(u|y) p_Y(y) \d y = - p_U(u). \label{eq:VanTrees_term2}
\end{align}
Using \eqref{eq:VanTrees_term1} and \eqref{eq:VanTrees_term2}, we have
\begin{align*}
	&\expectation[\Delta(U,Y)\Upsilon(U,Y)] \nonumber\\[0.1cm]
	&= \int_{-\infty}^{\infty} \int_{-\infty}^{\infty} \left(\varphi(u)-y\right)
       \frac{\d}{\d y}\ln\left[p_{U|Y}(u|y)p_Y(y)\right] p_{U|Y}(u|y)p_Y(y) \d u\, \d y \\[0.1cm]
	&= \int_{-\infty}^{\infty} \int_{-\infty}^{\infty} \left(\varphi(u)-y\right) \frac{\d}{\d y}
       \left[p_{U|Y}(u|y) p_Y(y)\right] \d u\, \d y \\[0.1cm]
	&= \int_{-\infty}^{\infty} \varphi(u) \underbrace{\left(\int_{-\infty}^{\infty}
       \frac{\d}{\d y}\left[ p_{U|Y}(u|y)p_Y(y)\right] \d y\right)}_{=0} \d u \\
& \qquad \qquad \qquad - \int_{-\infty}^{\infty} \underbrace{\left( \int_{-\infty}^{\infty} y \frac{\d}{\d y} \left[p_{U|Y}(u|y)p_Y(y)\right]\d y\right)}_{=-p_U(u)} \d u \nonumber\\
	&= \int_{-\infty}^{\infty} p_U(u) \d u = 1,
\end{align*}
as was claimed. It remains to establish the necessary and sufficient condition for equality in \eqref{eq:VanTrees_general}. The Cauchy--Schwarz inequality for the product of $\Delta(U,Y)$ and $\Upsilon(U,Y)$ holds if and only if $\Delta(U,Y) = c \Upsilon(U,Y)$ for some constant $c \in \reals$, almost surely. This is equivalent to
\begin{align*}
	\varphi(U) &= Y + c\sqrt{s}(U - \sqrt{s}Y) + c \rho_Y(Y) \\
	&= c\sqrt{s}U + (1-cs)Y + c \rho_Y(Y)
\end{align*}
for some $c \in \reals$. In fact, $c$ must be nonzero, for otherwise we will have $\varphi(U) = Y$, which is not a valid estimator. But then it must be the case that $(1-cs)Y + c \rho_Y(Y)$ is independent of $Y$, i.e., there exists some other constant $c' \in \reals$, such that
$$
\rho_Y(y) \deq \frac{p'_Y(y)}{p_Y(y)} = \frac{c'}{c} + (s-1/c)y.
$$
In other words, the score $\rho_Y(y)$ must be an affine function of $y$, which is the case if and only if $Y$ is a Gaussian random variable.

\section{The proof of Theorem~\ref{thm:hypercontractive}}
\label{app:hypercontractive}

As a reminder, the $L^p$ norm of a real-valued random variable $U$
is defined by $\| U \|_p \deq \left( \expectation[|U|^p] \right)^{1/p}$
for $p \geq 1$.
It will be convenient to work with the following equivalent form of the
R\'enyi divergence in \eqref{eq:Renyi_div}: For every two random variables
$U$ and $V$ such that $P_U \ll P_V$, we have
\begin{align}\label{eq:Renyi_div_altern}
D_\alpha (P_U \| P_V) = \frac{\alpha}{\alpha-1} \ln \left\| \frac{\d P_U}{\d P_V}(V) \right\|_\alpha,
\qquad \alpha > 1.
\end{align}
Let us denote by $g$ the Radon--Nikodym derivative $\d P / \d G$. It is easy to show that $P_t \ll G$ for
all $t$, so the Radon--Nikodym derivative $g_t \deq \d P_t / \d G$ exists. Moreover, $g_0 = g$. Also, let
us define the function $\alpha \colon [0,\infty) \to [\beta, \infty)$ by $\alpha(t) = 1 + (\beta-1)e^{2t}$ for
some $\beta > 1$. Let $Z \sim G$. Using \eqref{eq:Renyi_div_altern}, it is easy to verify that the
desired bound \eqref{eq:Renyi_hypercontractivity} is equivalent to the statement that the function
$F \colon [0,\infty) \to \reals$, defined by
\begin{align*}
F(t) &\deq \ln \left\| \frac{\d P_t}{\d G}(Z) \right \|_{\alpha(t)} \equiv \ln \left\| g_t(Z) \right\|_{\alpha(t)},
\end{align*}
is non-increasing. From now on, we will adhere to the following notational convention: we will use either the dot or $\d/\d t$ to denote derivatives with respect to the ``time" $t$, and the prime to denote derivatives with respect to the ``space" variable $z$. We start by computing the derivative of $F$ with respect to $t$, which gives
\begin{align}
\dot{F}(t) &= \frac{\d }{\d t} \left\{ \frac{1}{\alpha(t)} \ln \expectation\left[ \big(g_t(Z)\big)^{\alpha(t)} \right] \right\} \nonumber \\
&= - \frac{\dot{\alpha}(t)}{\alpha^2(t)}  \ln \expectation\left[ \big(g_t(Z)\big)^{\alpha(t)} \right] + \frac{1}{\alpha(t)} \frac{\dfr{\d}{\d t} \expectation \left[ \big(g_t(Z) \big)^{\alpha(t)} \right] }{\expectation \left[ \big(g_t(Z) \big)^{\alpha(t)} \right]}. \label{eq:F_deriv}
\end{align}
To handle the derivative with respect to $t$ in the second term in \eqref{eq:F_deriv}, we need to delve a bit into the theory of the so-called {\em Ornstein--Uhlenbeck semigroup}, which is an alternative representation of the Ornstein--Uhlenbeck channel \eqref{eq:OU_channel}.

For every $t \ge 0$, let us define a linear operator $K_t$ acting on an arbitrary sufficiently regular (e.g., $L^1(G)$) function $h$ as
\begin{align}\label{eq:OU_semigroup}
K_t h(x) \deq \expectation\left[h\left(e^{-t} x + \sqrt{1-e^{-2t}} Z\right)\right],
\end{align}
where $Z \sim G$, as before. The family of operators $\{K_t\}^\infty_{t=0}$ has the following properties:
\begin{enumerate}
\item $K_0$ is the identity operator, $K_0 h = h$ for every $h$.
\item For every $t \ge 0$, if we consider the $\OU(t)$ channel, given by the random transformation \eqref{eq:OU_channel}, then for every measurable function $F$ such that $\expectation\bigl|F(Y)\bigr| < \infty$
    with $Y$ in \eqref{eq:OU_channel}, we can write
\begin{align}\label{eq:OU_semigroup_1}
K_t F(x) &= \expectation[F(Y)|X=x], \qquad \forall \, x \in \reals
\end{align}
and
\begin{align}\label{eq:OU_semigroup_2}
\expectation[F(Y)] = \expectation[K_t F(X)].
\end{align}
Here, \eqref{eq:OU_semigroup_1} easily follows from \eqref{eq:OU_channel}, and \eqref{eq:OU_semigroup_2} is immediate from \eqref{eq:OU_semigroup_1}.
\item A particularly useful special case of the above is as follows. Let $X$ have distribution $P$ with $P \ll G$, and let $P_t$ denote the output distribution of the $\OU(t)$ channel. Then, as we have seen before, $P_t \ll G$, and the corresponding densities satisfy
\begin{align}\label{eq:OU_density_evolution}
g_t(x) = K_t g(x).
\end{align}
To prove \eqref{eq:OU_density_evolution}, we can either use \eqref{eq:OU_semigroup_1} and the fact that $g_t(x) = \expectation[g(Y)|X=x]$, or proceed directly from \eqref{eq:OU_channel}:
\begin{eqnarray} \label{eq: g_t}
&& g_t(x) = \frac{1}{\sqrt{2\pi(1-e^{-2t})}}
\int_\reals \exp\left(-\frac{(u-e^{-t}x)^2}{2(1-e^{-2t})}\right) g(u) \d u \nonumber \\
&& \hspace*{1cm} = \frac{1}{\sqrt{2\pi}} \int_\reals g\left(e^{-t}x + \sqrt{1-e^{-2t}}z \right)
\exp\left(-\frac{z^2}{2}\right) \d z \nonumber \\
&& \hspace*{1cm} \equiv \expectation\left[g\left(e^{-t}x + \sqrt{1-e^{-2t}} Z \right)\right]
\end{eqnarray}
where in the second line we have made the change of variables $z = \frac{u-e^{-t}x}{\sqrt{1-e^{-2t}}}$, and
in the third line $Z \sim G$.
\item The family of operators $\{K_t\}^\infty_{t=0}$ forms a semigroup, i.e., for every $t_1,t_2 \ge 0$ we have
\begin{align*}
K_{t_1 + t_2} = K_{t_1} \circ K_{t_2} = K_{t_2} \circ K_{t_1},
\end{align*}
which is shorthand for saying that $K_{t_1 + t_2} h = K_{t_2}(K_{t_1}h) = K_{t_1}(K_{t_2}h)$ for every sufficiently regular $h$. This follows from \eqref{eq:OU_semigroup_1} and \eqref{eq:OU_semigroup_2} and from the fact that the channel family $\{\OU(t)\}^\infty_{t=0}$ is ordered by degradation. For this reason, $\{K_t\}^\infty_{t=0}$ is referred to as the {\em Ornstein--Uhlenbeck semigroup}. In particular, if $\{Y_t\}^\infty_{t=0}$ is the Ornstein--Uhlenbeck process, then for every function $F \in L^1(G)$ we have
\begin{align*}
K_t F(x) = \expectation[F(Y_t)|Y_0 = x], \qquad \forall \, x \in \reals.
\end{align*}
\end{enumerate}
Two deeper results concerning the Ornstein--Uhlenbeck semigroup, which we will need, are as follows: Define the second-order differential operator $\cL$ by
\begin{align*}
\cL h(x) \deq h''(x) - x h'(x)
\end{align*}
for all $C^2$ functions $h \colon \reals \to \reals$. Then:
\begin{enumerate}
\item The {\em Ornstein--Uhlenbeck flow} $\{h_t\}^\infty_{t=0}$, where $h_t = K_t h$ with a $C^2$ initial condition $h_0 = h$, satisfies the partial differential equation (PDE)
\begin{align}\label{eq:OU_flow_ODE}
\dot{h}_t &= \cL h_t.
\end{align}
\item For $Z \sim G$ and all $C^2$ functions $g,h \colon \reals \to \reals$ we have the {\em integration-by-parts formula}
\begin{align}\label{eq:OU_parts}
\expectation[g(Z) \cL h(Z)] = \expectation[h(Z) \cL g(Z)] = - \expectation[g'(Z) h'(Z)].
\end{align}
\end{enumerate}
We provide the proofs of \eqref{eq:OU_flow_ODE} and \eqref{eq:OU_parts} in Appendix~\ref{app:OU_semigroup}.

We are now ready to tackle the second term in \eqref{eq:F_deriv}. Noting that the family of densities $\{g_t\}^\infty_{t=0}$ forms an Ornstein--Uhlenbeck flow with initial condition $g_0 = g$, we have
\begin{align}
&\frac{\d}{\d t} \expectation \left[ \big(g_t(Z)\big)^{\alpha(t)}\right] \nonumber\\
&= \expectation\left[
 \frac{\d} {\d t} \Bigl\{\big(g_t(Z) \big)^{\alpha(t)} \Bigr\} \right] \nonumber \\
 &= \dot{\alpha}(t) \; \expectation\left[ \big(g_t(Z)\big)^{\alpha(t)} \ln g_t(Z) \right] + \alpha(t)\, \expectation\left[ \big(g_t(Z)\big)^{\alpha(t)-1} \frac{\d}{\d t} g_t(Z)\right] \nonumber \\
 &= \dot{\alpha}(t) \; \expectation\left[ \big(g_t(Z)\big)^{\alpha(t)} \ln g_t(Z) \right] \nonumber\\
& \qquad \qquad \qquad  +\alpha(t)\,\expectation\left[ \big(g_t(Z)\big)^{\alpha(t)-1} \cL g_t(Z)\right] \label{eq:g_t_deriv_1}\\
 &= \dot{\alpha}(t) \; \expectation\left[ \big(g_t(Z)\big)^{\alpha(t)} \ln g_t(Z) \right] \nonumber\\
& \qquad \qquad \qquad  - \alpha(t)\,\expectation\left[ \left(\big(g_t(Z)\big)^{\alpha(t)-1}\right)' g'_t(Z)\right]
 \label{eq:g_t_deriv_2} \\
 &= \dot{\alpha}(t) \; \expectation\left[ \big(g_t(Z)\big)^{\alpha(t)} \ln g_t(Z) \right] \nonumber\\
& \qquad \qquad \qquad  -
 \alpha(t) \big( \alpha(t) - 1\big) \; \expectation\left[ \big(g_t(Z)\big)^{\alpha(t)-2}
 \left(g'_t(Z)\right)^2\right] \label{eq:g_t_deriv_3}
\end{align}
where we use \eqref{eq:OU_flow_ODE} to get \eqref{eq:g_t_deriv_1}, and \eqref{eq:OU_parts} to get \eqref{eq:g_t_deriv_2}. (Referring back to \eqref{eq: g_t}, we see that the functions $g_t$, for all $t > 0$, are $C^\infty$ due to the smoothing property of the Gaussian kernel, so all interchanges of expectations and derivatives in the above display are justified.) If we define the function $\phi_t(z) \deq \bigl(g_t(z)\bigr)^{\alpha(t)/2}$, then
we can rewrite \eqref{eq:g_t_deriv_3} as
\begin{align}
\frac{\d}{\d t} \expectation \left[ \big(g_t(Z)\big)^{\alpha(t)}\right] &= \frac{\dot{\alpha}(t)}{\alpha(t)} \; \expectation\left[\phi^2_t(Z) \ln \phi^2_t(Z) \right] \nonumber\\
& \qquad - \frac{4\big(\alpha(t)-1\big)}{\alpha(t)} \;
\expectation\left[ \left( \phi'_t(Z) \right)^2 \right].\label{eq:g_t_deriv}
\end{align}
Using the definition of $\phi_t$ and substituting \eqref{eq:g_t_deriv} into
the right-hand side of \eqref{eq:F_deriv}, we get
\begin{align}
\alpha^2(t)\, \expectation[\phi^2_t(Z)] \, \dot{F}(t)
&= \dot{\alpha}(t) \; \left( \expectation[\phi^2_t(Z) \ln \phi^2_t(Z)] - \expectation[\phi^2_t(Z)] \ln \expectation[\phi^2_t(Z)] \right) \nonumber\\
&\qquad\qquad- 4 (\alpha(t)-1) \expectation\left[ \left( \phi'_t(Z) \right)^2 \right]. \label{eq:F_deriv_2}
\end{align}
If we now apply the Gaussian log-Sobolev inequality \eqref{eq:Gross_LSI} to $\phi_t$, then from \eqref{eq:F_deriv_2} we get
\begin{align}
\alpha^2(t)\, \expectation[\phi^2_t(Z)] \, \dot{F}(t) \le 2 \left(  \dot{\alpha}(t) - 2(\alpha(t)-1) \right) \expectation\left[ \left(\phi'_t(Z) \right)^2\right]. \label{eq:F_deriv_3}
\end{align}
Since $\alpha(t) = 1 + (\beta - 1)e^{2t}$, $\dot{\alpha}(t) - 2(\alpha(t)-1) = 0$, which implies that the right-hand side of \eqref{eq:F_deriv_3} is
equal to zero. Moreover, because $\alpha(t) > 0$
and $\phi^2_t(Z) > 0$ a.s.\ (note that $\phi^2_t > 0$
if and only if $g_t > 0$, but the latter follows from
\eqref{eq: g_t} where $g$ is a probability density function), we conclude that $\dot{F}(t) \le 0$.

What we have proved so far is that, for every $\beta > 1$ and $t \ge 0$,
\begin{align}\label{eq:hypercontractivity_1}
D_{\alpha(t)}(P_t \| G) \le \left(\frac{\alpha(t)(\beta-1)}{\beta(\alpha(t)-1)}\right) D_\beta(P \| G)
\end{align}
where $\alpha(t) = 1 + (\beta-1)e^{2t}$. By the monotonicity property of the R\'enyi divergence, the
left-hand side of \eqref{eq:hypercontractivity_1} is greater than or equal to $D_\alpha(P_t \| G)$ as
soon as $\alpha \le \alpha(t)$. By the same token, because the function $u \in (1,\infty) \mapsto \frac{u}{u-1}$
is strictly decreasing, the right-hand side of \eqref{eq:hypercontractivity_1} can be upper-bounded by
$\Bigl(\frac{\alpha(\beta-1)}{\beta(\alpha-1)}\Bigr) D_\beta(P \| G)$ for all $\alpha \ge \alpha(t)$.
Putting all these facts together, we conclude that the Gaussian log-Sobolev inequality \eqref{eq:Gross_LSI}
implies \eqref{eq:Renyi_hypercontractivity}.

We now show that \eqref{eq:Renyi_hypercontractivity} implies the log-Sobolev inequality of Theorem~\ref{thm:Gross_LSI}. To that end, we recall that \eqref{eq:Renyi_hypercontractivity} is equivalent to the right-hand side of \eqref{eq:F_deriv_2} being less than or equal to zero for all $t \ge 0$ and all $\beta > 1$. Let us choose $t = 0$ and $\beta = 2$, in which case
\begin{align*}
\alpha(0) = \dot{\alpha}(0) = 2, \qquad \phi_0 = g.
\end{align*}
Using this in \eqref{eq:F_deriv_2} for $t=0$, we get
\begin{align*}
2 \left(\expectation\left[g^2(Z) \ln g^2(Z) \right] - \expectation[g^2(Z)] \ln \expectation[g^2(Z)] \right) - 4\, \expectation\left[ \left( g'(Z)\right)^2 \right] \le 0
\end{align*}
which is precisely the log-Sobolev inequality \eqref{eq:Gross_LSI} where
$\expectation[g(Z)] = \expectation_G\left[\frac{\d P}{\d G} \right] = 1$.
This completes the proof of Theorem~\ref{thm:hypercontractive} (up to the proof of the equality in
\eqref{eq:OU_parts} that is related to Appendix~\ref{app:OU_semigroup}).

\section{Details on the Ornstein--Uhlenbeck semigroup}
\label{app:OU_semigroup}

In this appendix, we will prove the formulas \eqref{eq:OU_flow_ODE} and \eqref{eq:OU_parts} pertaining to the Ornstein--Uhlenbeck semigroup. We start with \eqref{eq:OU_flow_ODE}. Recalling that
\begin{align*}
h_t(x) = K_t h(x) = \expectation\left[h\left(e^{-t}x + \sqrt{1-e^{-2t}}Z\right)\right],
\end{align*}
we have
\begin{align*}
\dot{h}_t(x) &=  \frac{\d}{\d t} \expectation\left[h\left(e^{-t}x + \sqrt{1-e^{-2t}} Z\right)\right] \\
&= -e^{-t}x\, \expectation\left[h'\left(e^{-t}x + \sqrt{1-e^{-2t}}Z\right)\right] \\
& \qquad \qquad \qquad + \frac{e^{-2t}}{\sqrt{1-e^{-2t}}}
\cdot \expectation\left[Z h'\left(e^{-t}x + \sqrt{1-e^{-2t}}Z\right)\right].
\end{align*}
For an arbitrary sufficiently smooth function $h$ and every $m,\sigma \in \reals$,
\begin{align*}
\expectation[Z h'(m + \sigma Z) ] = \sigma \expectation[h''(m + \sigma Z)]
\end{align*}
(which is proved straightforwardly using integration by parts, provided that
$\lim_{x \to \pm \infty} e^{-\frac{x^2}{2}} h'(m+\sigma x) = 0$). Using this equality, we can write
\begin{align*}
\expectation\left[Z h'\left(e^{-t}x + \sqrt{1-e^{-2t}}Z\right)\right] &= \sqrt{1-e^{-2t}} \expectation\left[h''\left(e^{-t} x + \sqrt{1-e^{-2t}} Z\right)\right].
\end{align*}
Therefore,
\begin{align}\label{eq:OU_ODE_1}
\dot{h}_t(x) = - e^{-t} x \cdot K_t h'(x) + e^{-2t} K_t h''(x).
\end{align}
On the other hand,
\begin{align}
\cL h_t(x) &= h''_t(x) - x h'_t(x) \nonumber \\
&= e^{-2t} \expectation\left[h''\left(e^{-t}x + \sqrt{1-e^{-2t}} Z\right)\right] \nonumber \\
& \qquad - x e^{-t} \expectation\left[h'\left(e^{-t} x + \sqrt{1-e^{-2t}}Z \right) \right]\nonumber \\
&= e^{-2t} K_t h''(x) - e^{-t} x K_t h'(x). \label{eq:OU_ODE_2}
\end{align}
Comparing \eqref{eq:OU_ODE_1} and \eqref{eq:OU_ODE_2}, we get \eqref{eq:OU_flow_ODE}.

The proof of the integration-by-parts formula \eqref{eq:OU_parts} is more subtle, and relies on the fact that the Ornstein--Uhlenbeck process $\{Y_t\}^\infty_{t=0}$ with $Y_0 \sim G$ is stationary and {\em reversible} in the sense that, for every $t,t' \ge 0$, $(Y_t,Y_{t'}) \eqdist (Y_{t'}, Y_t)$. To see this, let
\begin{align*}
p^{(t)}(y|x) \deq \frac{1}{\sqrt{2\pi(1-e^{-2t})}} \exp\left( - \frac{(y-e^{-t}x)^2}{2(1-e^{-2t})} \right)
\end{align*}
be the transition density of the $\OU(t)$ channel. Then it is not hard to establish that
\begin{align*}
p^{(t)}(y|x) \gamma(x) = p^{(t)}(x|y) \gamma(y), \qquad \forall \, x, y \in \reals
\end{align*}
(recall that $\gamma$ denotes the standard Gaussian pdf). For $Z \sim G$ and every two smooth functions $g,h$, this implies that
\begin{align*}
\expectation[g(Z) K_t h(Z)] &= \expectation[g(Y_0) K_t h(Y_0) ] \\
&= \expectation[g(Y_0) \expectation[h(Y_t)|Y_0]] \\
&= \expectation[g(Y_0) h(Y_t) ] \\
&= \expectation[g(Y_t) h(Y_0) ] \\
&= \expectation[K_t g(Y_0) h(Y_0) ] \\
&= \expectation[K_t g(Z) h(Z)],
\end{align*}
where we have used \eqref{eq:OU_semigroup_1} and the reversibility property of the Ornstein--Uhlenbeck process. Taking the derivative of both sides with respect to $t$, we conclude that
\begin{align}\label{eq:OU_inner_product}
\expectation[g(Z) \cL h(Z)] = \expectation[\cL g(Z) h(Z)].
\end{align}
In particular, since $\cL 1 = 0$ (where on the left-hand side $1$ denotes the constant function $x \mapsto 1$), we have
\begin{align}\label{eq:OU_inner_product_2}
\expectation[\cL g(Z)] = \expectation[1 \cL g(Z)] = \expectation[g(Z) \cL 1] = 0
\end{align}
for all smooth $g$.
\begin{remark} If we consider the Hilbert space $L^2(G)$ of all functions $g \colon \reals \to \reals$ such that $\expectation[g^2(Z)] < \infty$ with $Z \sim G$, then \eqref{eq:OU_inner_product} expresses the fact that $\cL$ is a self-adjoint linear operator on this space. Moreover, \eqref{eq:OU_inner_product_2} shows that the constant functions are in the kernel of $\cL$ (the closed linear subspace of $L^2(G)$ consisting of all $g$ with $\cL g = 0$).
\end{remark}

We are now ready to prove \eqref{eq:OU_parts}. To that end, let us first define the operator $\Gamma$ on pairs of functions $g,h$ by
\begin{align}\label{eq:carre_du_champ}
\Gamma(g,h) \deq \frac{1}{2} \left[ \cL(gh) - g\cL h - h \cL g \right].
\end{align}
\begin{remark} This operator was introduced into the study of Markov processes by Paul Meyer under the name ``carr\'e du champ'' (French for ``square of the field''). In the general theory, $\cL$ can be an arbitrary linear operator that serves as an infinitesimal generator of a Markov semigroup. Intuitively, $\Gamma$ measures how far a given $\cL$ is from being a derivation, where we say that an operator $\cL$ acting on a function space is a {\em derivation} (or that it satisfies the {\em Leibniz rule}) if, for every $g,h$ in its domain,
\begin{align*}
\cL(gh) = g\cL h + h \cL g.
\end{align*}
An example of a derivation is the first-order linear differential operator $\cL g = g'$, in which case the Leibniz rule is simply the product rule of differential calculus.
\end{remark}

Now, for our specific definition of $\cL$, we have
\begin{align}
\Gamma(g,h)(x) &= \frac{1}{2} \Big[ (gh)''(x) - x (gh)'(x) - g(x)\big(h''(x) - x h'(x)\big) \nonumber\\
& \qquad \qquad - h(x)\big(g''(x) - x g'(x)\big) \Big] \nonumber \\
&= \frac{1}{2} \Big[ g''(x)h(x) + 2g'(x)h'(x) + g(x) h''(x) \nonumber \\
& \qquad \qquad - x g'(x) h(x) - x g(x) h'(x)- g(x) h''(x) \nonumber\\
& \qquad \qquad + x g(x) h'(x) - g''(x) h(x) + x g'(x) h(x) \Big] \nonumber\\
&= g'(x) h'(x), \label{eq:carre_du_champ_2}
\end{align}
or, more succinctly, $\Gamma(g,h) = g'h'$. Therefore,
\begin{align}
\expectation[g(Z) \cL h(Z)] &= \frac{1}{2} \Big\{ \expectation[g(Z) \cL h(Z)] + \expectation[h(Z) \cL g(Z)] \Big\} \label{eq:OU_parts_proof_1} \\
&= \frac{1}{2} \expectation[\cL(gh)(Z)] - \expectation[\Gamma(g,h)(Z)] \label{eq:OU_parts_proof_2} \\
&= - \expectation[g'(Z) h'(Z)], \label{eq:OU_parts_proof_3}
\end{align}
where \eqref{eq:OU_parts_proof_1} uses \eqref{eq:OU_inner_product}, \eqref{eq:OU_parts_proof_2} uses the definition \eqref{eq:carre_du_champ} of $\Gamma$, and \eqref{eq:OU_parts_proof_3} uses \eqref{eq:carre_du_champ_2} together with \eqref{eq:OU_inner_product_2}. This proves \eqref{eq:OU_parts}.

\section{LSI for Bernoulli and Gaussian measures}
\label{appendix: binary to Gaussian LSI}

The following log-Sobolev inequality was derived by Gross \cite{Gross}:
	\begin{align}\label{eq:Gross_Hamming_LSI}
		\Ent_P[g^2] \le \frac{(g(0)-g(1))^2}{2}.
	\end{align}
We will now show that \eqref{eq:Hamming_LSI} can be derived from \eqref{eq:Gross_Hamming_LSI}. Let us define $f$ by $e^f = g^2$, where we may assume without loss of generality that $0 < g(0) \le g(1)$. Note that
	\begin{align} \label{eq: bound on the squared of g(0)-g(1)}
		\left(g(0)-g(1)\right)^2 &= \left(\exp\left(f(0)/2\right)-\exp\left(f(1)/2\right)\right)^2 \nonumber \\
		&\le \frac{1}{8}\left[\exp\left(f(0)\right)+\exp\left(f(1)\right)\right]\left(f(0)-f(1)\right)^2 \nonumber \\
		&= \frac{1}{4}\expectation_P\left[\exp(f) (\Gamma f)^2\right]
	\end{align}
	with $\Gamma f = |f(0)-f(1)|$, where the inequality follows from the easily verified fact that $(1-x)^2 \le \frac{(1+x^2) (\ln x)^2}{2}$ for all $x \ge 0$, which we apply to $x \triangleq g(1)/g(0)$. Therefore,
the inequality in \eqref{eq:Gross_Hamming_LSI} implies the following:
\begin{align}
D(P^{(f)} || P)  & = \frac{\Ent_P[\exp(f)]}{\expectation_P[\exp(f)]} \label{eq:a} \\[0.1cm]
& = \frac{\Ent_P[g^2]}{\expectation_P[\exp(f)]} \label{eq:b} \\[0.1cm]
& \leq \frac{\big(g(0)-g(1)\big)^2}{2 \, \expectation_P[\exp(f)]} \label{eq:c} \\[0.1cm]
& \leq \frac{\expectation_P[\exp(f) \, (\Gamma f)^2]}{8 \, \expectation_P[\exp(f)]} \label{eq:d} \\[0.1cm]
& = \frac{1}{8} \, \expectation_P^{(f)} \bigl[(\Gamma f)^2 \bigr] \label{eq:e}
\end{align}
where equality \eqref{eq:a} follows from \eqref{eq:ent_vs_D},
equality \eqref{eq:b} holds due to the equality $e^f = g^2$,
inequality~\eqref{eq:c} holds due to \eqref{eq:Gross_Hamming_LSI},
inequality~\eqref{eq:d} follows from \eqref{eq: bound on the squared of g(0)-g(1)},
and equality~\eqref{eq:e} follows by definition of the expectation with respect to the
tilted probability measure $P^{(f)}$. Therefore, we conclude that indeed
\eqref{eq:Gross_Hamming_LSI} implies \eqref{eq:Hamming_LSI}.

Gross used \eqref{eq:Gross_Hamming_LSI} and the central limit theorem to establish his Gaussian log-Sobolev inequality (see Theorem~\ref{thm:Gross_LSI}). We can follow the same steps and arrive at \eqref{eq:Gross_LSI_2} from \eqref{eq:Hamming_LSI}. To that end, let $g \colon \reals \to \reals$ be a sufficiently smooth function (to guarantee, at least, that both $g \exp(g)$ and the derivative of $g$ are continuous and bounded), and define the function $f \colon \{0,1\}^n \to \reals$ by
	\begin{align*}
		f(x_1,\ldots,x_n) \deq g \left( \frac{x_1 + x_2 + \ldots + x_n - n/2}{\sqrt{n/4}} \right).
	\end{align*}
If $X_1,\ldots,X_n$ are i.i.d.\ $\Bernoulli(1/2)$ random variables, then, by the central limit theorem, the sequence of probability measures $\{P_{Z_n}\}^\infty_{n=1}$ with
\begin{align*}
	Z_n \deq \frac{X_1 + \ldots + X_n - n/2}{\sqrt{n/4}}
\end{align*}
converges weakly to the standard Gaussian distribution $G$ as $n \to \infty$: $P_{Z_n} \Rightarrow G$. Therefore, by the assumed smoothness properties of $g$ we have (see \eqref{eq: entropy functional} and \eqref{eq:ent_vs_D})
\begin{align}
	&\expectation\left[\exp\big(f(X^n)\big)\right] \cdot D \big( P^{(f)}_{X^n} \big\| P_{X^n} \big) \nonumber\\
	&= \expectation\left[ f(X^n) \exp \big(f(X^n) \bigr) \right] - \expectation[\exp \big(f(X^n)\bigr)] \ln \expectation[ \exp \big(f(X^n)\bigr)] \nonumber \\
	&= \expectation\left[ g(Z_n) \exp \big(g(Z_n)\big)\right] - \expectation[\exp \big(g(Z_n)\big)] \ln \expectation[ \exp \big(g(Z_n)\big)] \nonumber \\
	& \xrightarrow{n \to \infty} \expectation\left[ g(Z) \exp \big(g(Z)\big)\right] - \expectation[\exp \big(g(Z)\big)] \ln \expectation[ \exp \big(g(Z)\big)] \nonumber \\
	&= \expectation\left[\exp\left(g(Z)\right)\right] D \big( P^{(g)}_Z \big\| P_Z \big) \label{eq:CLT_LSI_step1}
\end{align}
where $Z \sim G$ is a standard Gaussian random variable. Moreover, using the definition \eqref{eq:Hamming_Gamma} of $\Gamma$ and the smoothness of $g$, for every $i \in \{1,\ldots,n\}$ and $x^n \in \{0,1\}^n$ we have
\begin{align*}
&\left| f(x^n \oplus e_i) - f(x^n) \right|^2 \\
&= \left| g\left( \frac{x_1 + \ldots + x_n - n/2}{\sqrt{n/4}} + \frac{(-1)^{x_i}}{\sqrt{n/4}}\right) - g \left(\frac{x_1 + \ldots + x_n - n/2}{\sqrt{n/4}}\right) \right|^2 \\
&= \frac{4}{n} \left(g'\left(\frac{x_1 + \ldots + x_n - n/2}{\sqrt{n/4}}\right)\right)^2 + o\left(\frac{1}{n}\right),
\end{align*}
which implies that
\begin{align*}
\left| \Gamma f(x^n) \right|^2
&= \sum^n_{i=1} \left( f(x^n \oplus e_i) - f(x^n)\right)^2 \\
&=  4 \left(g'\left(\frac{x_1 + \ldots + x_n - n/2}{\sqrt{n/4}}\right)\right)^2 + o\left(1\right).
\end{align*}
Consequently,
\begin{align}
&\expectation\left[\exp\left(f(X^n)\right)\right] \cdot \expectation^{(f)}\left[ \left(\Gamma f (X^n)\right)^2 \right] \nonumber\\
&\qquad = \expectation\left[ \exp\left(f(X^n)\right) \left(\Gamma f(X^n) \right)^2 \right] \nonumber \\
&\qquad = 4\, \expectation\left[ \exp\left(g(Z_n)\right) \left( \left(g'(Z_n)\right)^2 + o(1)\right)\right] \nonumber \\
&\xrightarrow{n \to \infty} 4\, \expectation \left[ \exp\left(g(Z)\right) \left(g'(Z)\right)^2\right] \nonumber \\
&\qquad = 4\, \expectation\left[\exp\left(g(Z)\right)\right] \cdot \expectation_{P_Z}^{(g)} \left[ \left(g'(Z)\right)^2\right]. \label{eq:CLT_LSI_step2}
\end{align}
Taking the limit of both sides of \eqref{eq:Hamming_LSI} as $n \to \infty$ and then using \eqref{eq:CLT_LSI_step1} and \eqref{eq:CLT_LSI_step2}, we obtain
\begin{align*}
D\big(P^{(g)}_Z \big\| P_Z \big) \le \frac{1}{2} \, \expectation_{P_Z}^{(g)}\left[ \left(g'(Z)\right)^2 \right],
\end{align*}
which is \eqref{eq:Gross_LSI_2}. The same technique applies in the case of an asymmetric Bernoulli measure:  given a sufficiently smooth function $g \colon \reals \to \reals$, define $f \colon \{0,1\}^n \to \reals$ by
	\begin{align*}
		f(x^n) \deq g\left( \frac{x_1 + \ldots + x_n - np}{\sqrt{npq}} \right).
	\end{align*}
	and then apply \eqref{eq:LSI_Bernoulli_p} to it.

\section{Fano's inequality for list decoding}
\label{app:Fano_list}

The following generalization of Fano's inequality for list decoding has been used in the proof of Theorem~\ref{thm:SC_DM-DBC}: Let $\cX$ and $\cY$ be finite sets, and let $(X,Y) \in \cX \times \cY$ be a pair of jointly distributed random variables. Consider an arbitrary mapping $L \colon \cY \to 2^\cX$ which maps every $y \in \cY$ to a set $L(y) \subseteq \cX$, such that $|L(Y)| \le N$ a.s.. Let $P_{\rm e} = \pr\left(X \not\in L(Y)\right)$ designate the list decoding error. Then
\begin{align}\label{eq:Fano_list}
	H(X|Y) \le  h(P_{\rm e}) + (1-P_{\rm e}) \ln N  + P_{\rm e} \ln |\cX|
\end{align}
(see, e.g., \cite{Ahlswede_Korner} or \cite[Lemma~1]{Kim_Sutivong_Cover}). For proving \eqref{eq:Fano_list},
define the indicator random variable $E \deq 1_{\{X \not\in L(Y)\}}$. Then we can expand the conditional
entropy $H(E,X|Y)$ in two ways as
\begin{subequations}
\begin{align}
	H(E,X|Y) &= H(E|Y) + H(X|E,Y) \label{eq:Fano_H_1}\\
	&= H(X|Y) + H(E|X,Y). \label{eq:Fano_H_2}
\end{align}
\end{subequations}
Since $X$ and $Y$ uniquely determine $E$ (for the given $L$), the quantity on the right-hand side of \eqref{eq:Fano_H_2} is equal to $H(X|Y)$. On the other hand, we can upper-bound the right-hand side of \eqref{eq:Fano_H_1} as
\begin{align*}
	H(E|Y) + H(X|E,Y) &\le H(E) + H(X|E,Y) \\
	&\le h(P_{\rm e}) + (1-P_{\rm e}) \ln N + P_{\rm e} \ln |\cX|,
\end{align*}
where we have bounded the conditional entropy $H(X|E,Y)$ as follows:
\begin{align*}
	& H(X|E,Y) \nonumber \\
    &= \sum_{y \in \cY} \pr(E=0,Y=y) \, H(X|E=0,Y=y) \\
	& \qquad + \sum_{y \in \cY} \pr(E=1,Y=y) \, H(X|E=1,Y=y) \\
	&\le \sum_{y \in \cY} \Bigl\{ \pr(E=0,Y=y) \, H(X|E=0,Y=y) \Bigr\} + P_{\rm e}\ln |\cX| \\
	&= (1-P_{\rm e})\sum_{y \in \cY} \Bigl\{ \pr(Y=y|E=0) \, H(X|E=0,Y=y) \Bigr\} + P_{\rm e}\ln |\cX| \\
	&\le (1-P_{\rm e}) \, \expectation\big[\ln |L(Y)| \, \big|\, E=0\big]  + P_{\rm e}\ln |\cX| \\
	&\le (1-P_{\rm e}) \ln N + P_{\rm e}\ln |\cX|,
\end{align*}
where in the first line we have used the standard log-cardinality bound on the entropy, while in the third line we have used the fact that, given $E=0$ and $Y=y$, $X$ is supported on the set $L(y)$. Since $$H(X|Y) = H(E|Y) + H(X|E,Y) \le H(E) + H(X|E,Y),$$ we get \eqref{eq:Fano_list}.

\begin{remark} If instead of assuming that $L(Y)$ is bounded a.s.\ we assume that it is bounded in expectation, i.e., if $\expectation[\ln |L(Y)|] < \infty$, then we can obtain a weaker inequality
	\begin{align*}
		H(X|Y) \le \expectation\left[ \ln |L(Y) \right] + h(P_{\rm e}) + P_{\rm e} \ln |\cX|.
	\end{align*}
To get this, we follow the same steps as before, except the last step in the above series of bounds on $H(X|E,Y)$ is replaced by
\begin{align*}
	& (1-P_{\rm e}) \, \expectation\big[\ln |L(Y)| \, \big|\, E=0\big] \nonumber \\[0.1cm]
    &\le (1-P_{\rm e})\, \expectation\big[\ln |L(Y)| \, \big|\, E=0\big] + P_{\rm e}\, \expectation\big[\ln |L(Y)| \, \big|\, E=1\big] \\[0.1cm]
	&= \expectation\left[ \ln |L(Y)| \right]
\end{align*}
(we assume, of course, that $L(y)$ is nonempty for all $y \in \cY$).
\end{remark}

\section{Details for the derivation of \eqref{eq: see Appendix D}}
\label{app:McDiarmid_details}

Let $X^n \sim P_{X^n}$ and $Y^n \in \cY^n$ be the input and output sequences of a DMC with transition matrix $T \colon \cX \to \cY$, where the DMC is used without feedback. In other words, $(X^n,Y^n) \in \cX^n \times \cY^n$ is a random variable with $X^n \sim P_{X^n}$ and
\begin{align*}
&	P_{Y^n|X^n}(y^n|x^n) = \prod^n_{i=1}P_{Y|X}(y_i|x_i),\\
&  \forall \, y^n \in \cY^n,\, \forall \, x^n \in \cX^n \text{ s.t. } P_{X^n}(x^n) > 0.
\end{align*}
Because the channel is memoryless and there is no feedback, the $i$th output symbol $Y_i \in \cY$ depends only on the $i$th input symbol $X_i \in \cX$ and not on the rest of the input symbols $\overline{X}^i$. Consequently, $\overline{Y}^i \to X_i \to Y_i$ is a Markov chain for every $i =1,\ldots,n$, so we can write
\begin{align}
	P_{Y_i|\overline{Y}^i}(y|\overline{y}^i) &= \sum_{x \in \cX}P_{Y_i|X_i}(y|x)P_{X_i|\overline{Y}^i}(x|\overline{y}^i) \\
	&= \sum_{x \in \cX} P_{Y|X}(y|x)P_{X_i|\overline{Y}^i}(x|\overline{y}^i)
\end{align}
for all $y \in \cY$ and all $\overline{y}^i \in \cY^{n-1}$ such that $P_{\overline{Y}^i}(\overline{y}^i) > 0$. Therefore, for every $y,y' \in \cY$ we have
\begin{align*}
	\ln \frac{P_{Y_i|\overline{Y}^i}(y|\overline{y}^i)}{P_{Y_i|\overline{Y}^i}(y'|\overline{y}^i)}
	&= \ln \frac{\sum_{x \in \cX} P_{Y|X}(y|x)P_{X_i|\overline{Y}^i}(x|\overline{y}^i)}{\sum_{x \in \cX} P_{Y|X}(y'|x)P_{X_i|\overline{Y}^i}(x|\overline{y}^i)} \\
	&= \ln \frac{\sum_{x \in \cX} P_{Y|X}(y'|x) P_{X_i|\overline{Y}^i}(x|\overline{y}^i) \frac{P_{Y|X}(y|x)}{P_{Y|X}(y'|x)}}{\sum_{x \in \cX} P_{Y|X}(y'|x)P_{X_i|\overline{Y}^i}(x|\overline{y}^i)},
\end{align*}
where in the last line we have used the fact that $P_{Y|X}(\cdot|\cdot) > 0$. This shows that we can express the quantity $	\ln \dfrac{P_{Y_i|\overline{Y}^i}(y|\overline{y}^i)}{P_{Y_i|\overline{Y}^i}(y'|\overline{y}^i)} $ as the logarithm of expectation of $\dfrac{P_{Y|X}(y|X)}{P_{Y|X}(y'|X)}$ with respect to the (conditional) probability measure
\begin{align*}
Q(x|\overline{y}^i,y') = \frac{P_{Y|X}(y'|x)P_{X_i|\overline{Y}^i}(x|\overline{y}^i)}{\sum_{x \in \cX} P_{Y|X}(y'|x)P_{X_i|\overline{Y}^i}(x|\overline{y}^i)}, \qquad \forall \, x \in \cX.
\end{align*}
Therefore,
\begin{align*}
	\ln \frac{P_{Y_i|\overline{Y}^i}(y|\overline{y}^i)}{P_{Y_i|\overline{Y}^i}(y'|\overline{y}^i)} \le \max_{x \in \cX} \ln \frac{P_{Y|X}(y|x)}{P_{Y|X}(y'|x)}.
\end{align*}
Interchanging the roles of $y$ and $y'$, we get
\begin{align*}
		\ln \frac{P_{Y_i|\overline{Y}^i}(y'|\overline{y}^i)}{P_{Y_i|\overline{Y}^i}(y|\overline{y}^i)}
	&\le	\max_{x \in \cX} \ln \frac{P_{Y|X}(y'|x)}{P_{Y|X}(y|x)}.
\end{align*}
This implies, in turn, that
\begin{align*}
	\left|	\ln \frac{P_{Y_i|\overline{Y}^i}(y|\overline{y}^i)}{P_{Y_i|\overline{Y}^i}(y'|\overline{y}^i)} \right| \le \max_{x \in \cX}\max_{y,y' \in \cY} \left| \ln \frac{P_{Y|X}(y|x)}{P_{Y|X}(y'|x)}\right| &= \frac{c(T)}{2}
\end{align*}
for all $y,y' \in \cY$.

\end{subappendices}